%% file: thesis.tex
\pdfoutput=1

\documentclass[a4paper,11pt]{scrbook}

\usepackage{scrhack}

\KOMAoptions{BCOR=8mm}

\KOMAoptions{toc=graduated, toc=bibliography}

\setkomafont{section}{\rmfamily\bfseries\LARGE}
\setkomafont{subsection}{\rmfamily\bfseries\Large}
\setkomafont{subsubsection}{\rmfamily\bfseries\large}
\setkomafont{chapter}{\rmfamily\bfseries\scshape\huge}
\setkomafont{chapterentry}{\rmfamily\bfseries\scshape}
\setkomafont{partentry}{\rmfamily\bfseries\scshape}
\setkomafont{part}{\rmfamily\bfseries\scshape\huge}
\setkomafont{partnumber}{\rmfamily\bfseries\scshape\huge}
\setkomafont{partentrypagenumber}{\rmfamily\bfseries\scshape}

\setkomafont{descriptionlabel}{\normalfont\bfseries}

\usepackage[automark]{scrpage2}
\KOMAoptions{headsepline=true}
\setheadsepline{.4pt}
\setkomafont{pageheadfoot}{\normalfont\normalcolor\upshape} 

\usepackage[T1]{fontenc} % to make special characters as in Sa\"ibi searchable

\usepackage[protrusion=true,expansion=true]{microtype}

\usepackage[charter]{mathdesign}

\usepackage{ownstyle}

\title{Initiality for Typed Syntax and Semantics}
\author{Benedikt Ahrens}

\begin{document}

% \maketitle

%%%   roman page numbers
\frontmatter 

%%%   front page has special formatting, in particular symmetric (oneside)
 \input{front_page}

%%%   you might want to put your own title page
%  \input{title_page}

%%%    for the abstract we remove page numbers and headers
 \pagestyle{empty}

% \footnotesize
% \small

 \input{abstract_en}

\cleardoublepage

 \input{abstract_fr}            %%%%%%%%%%%%%%     

\cleardoublepage

%%%   back to normal size and page style
% \normalsize 
\pagestyle{plain}

\input{thanks}      %%%%%%%%%%%%%%%%%%%%%%%%   

\input{long_resume_fr}

\input{introduction_fr}

\input{conclusion_fr}

\cleardoublepage

%%%   we switch to the headings style we defined in the preamble
\pagestyle{scrheadings}

%%%  for the table of contents, protrusion should be disabled
\microtypesetup{protrusion=false}
 \tableofcontents
\microtypesetup{protrusion=true}

%%%   switch to regular numbering, starting from 1
\mainmatter

%%%   i had a TODO file at the beginning of the document
% \input{TODO}                     %%%%%%%%%%%%%%%%%  remove for final

\chapter{Introduction}

\input{introduction}

\input{contribution}

\input{summary}

\input{rel_work}

\input{publications}

\part{Theory}\label{part:theory}

\chapter{Category--Theoretic Constructions}\label{chap:cat}

\input{monad_colax}

\input{relative_monads}

\chapter{Simple Type Systems} \label{sec:compilation}

In this chapter we present two generalizations to simple type systems 
of Hirschowitz and Maggesi's initiality theorem for untyped syntax
\cite{DBLP:conf/wollic/HirschowitzM07}:
\begin{itemize}
 \item in \autoref{sec:sts_ju} we review Zsid\'o's theorem \cite[\chapterautorefname~6]{ju_phd}.
 \item In \autoref{sec:ext_zsido} we prove a variant of Zsid\'o's theorem which accounts for translations between languages over
\emph{different} sets of object types.
\end{itemize}

\noindent
We explain the difference between the two abovementioned theorems in more detail:

in Zsid\'o's theorem, the underlying set of types of a signature --- and thus of the term language the signature specifies --- is given as a fixed parameter.
In particular, all the models --- representations --- of the signature have the same underlying set of types.
Furthermore, this set does not necessarily have inductive structure, as opposed to the sets of types we characterize via initiality
in \autoref{sec:type_sigs} --- the content of \autoref{sec:sts_ju} is independent of that of \autoref{sec:type_sigs}.

In our variant of Zsid\'o's theorem we prove in \autoref{sec:ext_zsido}, a language is specified by \emph{a pair} $(S,\Sigma)$ of signatures,
a signature $S$ for \emph{types} as presented in \autoref{sec:type_sigs}, and a signature $\Sigma$ for \emph{terms} over the signature $S$.
A representation of such a signature is given by a pair of a representation of $S$ and a representation of $\Sigma$. In particular,
we consider models of $(S,\Sigma)$ whose underlying set of types is different from the set freely generated by the signature $S$.
The initiality result of \autoref{sec:ext_zsido} thus characterizes both the types and terms freely generated by a signature as initial object 
in a category of representations.

As running examples, we consider the simply--typed lambda calculus and Plotkin's \PCF~\cite{Plotkin1977223}. 
In \autoref{ex:logic_trans} we present a logic translation from classical to intuitionistic propositional logic
  as an instance of our theorem of \autoref{sec:ext_zsido}.
Before focusing on \emph{term signatures}, however, we review, in \autoref{sec:type_sigs},
\emph{algebraic} signatures as treated by Birkhoff \cite{birkhoff1935}.
Algebraic signatures are used in \autoref{sec:ext_zsido} for the specification of the set of types of a language.

\input{type_sigs}

\input{sts}

\input{n_compilation_sig}

\chapter{Reductions for Untyped Syntax} \label{sec:prop_arities}

\input{prop_arities}

\chapter{Simple Type Systems with Reductions}\label{chap:comp_types_sem}

\input{comp_sem_monads}

\part{Computer Implementation}\label{part:impl}

\chapter{Formalizing Category Theory in \texorpdfstring{$\mathsf{Coq}$}{Coq}}\label{chap:cats_in_coq}

In this chapter we describe our computer formalization of general concepts of category theory as 
presented in \autoref{chap:cat}.
We start with a brief introduction to our favourite theorem prover \textsf{Coq} \cite{coq}.  
We then describe the challenges one encounters when one attempts to formalize algebraic structures in general,
and category theory in particular, in \textsf{Coq}.
Finally we present our implementation of monads and modules over monads as well as their relative counterparts.
Throughout the chapter we explain features of \textsf{Coq} when we first encounter them.

\input{about_coq}
\input{formalizing_alg}

\input{formalizing_cats}

\input{basic_formal}
\input{endomonads_formal}

\input{relative_monads_formal}

\input{sts_formal}

\chapter{Initiality for Untyped 2--Signatures, Formalized}\label{chap:2--signatures_formal}

\input{prop_arities_formal}

\chapter{A Faithful Translation of \texorpdfstring{$\PCF$}{PCF} to \texorpdfstring{$\ULC$}{ULC}}\label{chap:comp_sem_formal}

\input{comp_sem_monads_formal}

\chapter{Conclusions and Further Work}

\input{conclusion}

\input{appendix}

%%%%%%%%%%%%%%%%%%%       bibliography gets backmatter 
\backmatter

%%% not necessary
%\addcontentsline{toc}{chapter}{Bibliography}

%%%   some possible styles for bibliography
%   \bibliographystyle{plainnat}
%  \bibliographystyle{abbrvnat}
\bibliographystyle{alpha}

%%%  producing biblio using bibtex
\bibliography{literature}

\end{document}

%% file: front_page.tex
\KOMAoptions{BCOR=0mm,twoside=off}
\recalctypearea

\begin{titlepage}
 
\centering

%\begin{center}
{\Large UNIVERSIT\'E NICE SOPHIA ANTIPOLIS   --   UFR Sciences}\\
\vspace*{0.5cm}
\'Ecole Doctorale Sciences Fondamentales et Appliqu\'ees\\
\vspace*{1.5cm}
{\Large\bf TH\`ESE}\\
\vspace*{0.3cm}
pour obtenir le titre de\\
\vspace*{0.1cm}
~~~{\Large\bf Docteur en Sciences}\\
\vspace*{0.1cm}
Sp\'ecialit\'e 
{\Large {\sc Math\'ematiques}}\\
\vspace*{0.8cm}
pr\'esent\'ee et soutenue par\\
\vspace*{0.1cm}
{\large\bf Benedikt AHRENS}\\
\vspace*{1.0cm}
{\LARGE\bf Initiality for Typed Syntax and Semantics}\\
\vspace*{1.0cm}
Th\`ese dirig\'ee par {\bf Andr\'e HIRSCHOWITZ}\\
\vspace*{0.2cm}
soutenue le 23 mai 2012\\

\vspace*{2.3cm}
Membres du jury :\\
\vspace*{0.4cm}
\begin{tabular}{lllll}
M. & Pierre--Louis & CURIEN & & Rapporteur et Examinateur\\
M. & Andr\'e & HIRSCHOWITZ & & Directeur de th\`ese\\
M. & Marco & MAGGESI & & Examinateur \\
M. & Laurent &  REGNIER & & Rapporteur et Examinateur\\
M. & Carlos & SIMPSON & & Examinateur\\
\end{tabular}

%\vspace*{3cm}
%\pagefill
\vfill
Laboratoire Jean-Alexandre Dieudonn\'e, Universit\'e de Nice, Parc Valrose, 06108 NICE
%\end{center}
\end{titlepage}
\KOMAoptions{BCOR=9mm,twoside=on}
\recalctypearea
\cleardoublepage

%% file: abstract_en.tex
\section*{Abstract}

In this thesis we give an algebraic characterization of the syntax and semantics 
of simply--typed languages.
More precisely, we characterize \emph{simply--typed binding syntax equipped with reduction rules}
via a \emph{universal property}, namely as the initial object of some category.

We specify a language by a \emph{2--signature} $(\Sigma,A)$, that is, a signature on two levels: 
the \emph{syntactic} level $\Sigma$ specifies the sorts and terms of the language, and 
associates a sort to each term.
The \emph{semantic} level $A$ specifies, through \emph{inequations}, reduction rules
on the terms of the language.
To any given 2--signature $(\Sigma,A)$ we associate a category of ``models'' of $(\Sigma,A)$. 
We prove that this category has an initial object, which integrates the
terms freely generated by $\Sigma$ and the reduction relation --- on those terms --- 
generated by $A$.
We call this object the \emph{programming language generated by $(\Sigma, A)$}.

Initiality provides an \emph{iteration principle}
which allows to specify translations on the syntax, possibly to a language over different sorts. 
Furthermore, translations specified via the iteration principle are by construction \emph{type--safe} 
and \emph{faithful with respect to reduction}.

To illustrate our results, we consider two examples extensively: firstly, we specify a double negation translation from 
classical to intuitionistic propositional logic via the category--theoretic iteration principle.
Secondly, we specify a translation from $\PCF$ to the untyped lambda calculus which is faithful with
respect to reduction in the source and target languages.

In a second part, we formalize some of our initiality theorems in the proof assistant \textsf{Coq}.
The implementation yields a machinery which, when given 
a 2--signature, returns an implementation of its associated abstract syntax together with certified substitution operation,
iteration operator and a reduction relation generated by the specified reduction rules.

%% file: abstract_fr.tex
\section*{R\'esum\'e}

Dans cette th\`ese, on donne une caract\'erisation alg\'ebrique de la syntaxe et de la s\'emantique des langages
simplement typ\'es.
Plus pr\'ecisement, on caract\'erise la syntaxe simplement typ\'ee avec liaison de variables, 
\'equip\'ee des r\`egles de r\'eduction, via une propri\'et\'e universelle, \`a savoir 
comme l'objet initial d'une cat\'egorie.

Nous sp\'ecifions un langage par une 2--signature $(\Sigma,A)$, c'est--\`a--dire, une
signature \`a deux niveaux: le niveau \emph{syntaxique} $\Sigma$ sp\'ecifie les types et les termes
du langage, et associe un type \`a chaque terme.
Le niveau \emph{s\'emantique} $A$ sp\'ecifie, via des \emph{in\'equations}, des r\`egles de
r\'eduction sur les termes du langage.
A chaque 2--signature $(\Sigma,A)$ donn\'ee on associe une cat\'egorie 
des \guillemotleft mod\`eles\guillemotright\ de $(\Sigma, A)$. 
Nous d\'emontrons que cette cat\'egorie admet un objet initial, qui int\`egre
les termes librement engendr\'es par $\Sigma$ et la relation de r\'eduction --- sur ces termes ---
 engendr\'ee par $A$.
Nous appelons cet objet le \emph{langage engendr\'e par $(\Sigma,A)$}.

Initialit\'e fournit un \emph{principe d'it\'eration} qui permet de sp\'ecifier
des traductions sur la syntaxe, possiblement vers un langage sur des types diff\'erents.
De plus, les traductions qui sont sp\'ecifi\'ees via ce principe d'it\'eration sont
\emph{fid\`eles relativement au typage et la r\'eduction}.

Afin d'illustrer nos r\'esultats, nous considerons deux exemples en d\'etail: premi\`erement, nous specifions une traduction de la logique classique \`a 
la logique intuitioniste propositionnelle via le principe d'it\'eration cat\'egorique.
Deuxi\`emement, nous specifions une traduction de $\PCF$ au lambda calcul non--typ\'e qui est fid\`ele
par rapport aux r\'eductions aux langages source et but.

Dans une deuxi\`eme partie, nous formalisons quelques uns de nos th\'eor\`emes d'initialit\'e dans 
l'assistant de preuves \textsf{Coq}.
L'impl\'ementation apporte un m\'ecanisme qui, \'etant donn\'ee une 2--signature, 
rend une impl\'ementation de sa syntaxe associ\'ee, \'equip\'ee d'une op\'eration de
substitution certifi\'ee, d'un op\'erateur d'it\'eration et d'une relation de r\'eduction
engendr\'ee par les r\`egles de reduction specifi\'ees.

%% file: thanks.tex
\chapter*{Hello and Thank You, ...}

\begin{packitem}

\item Andr\'e, for all your time and energy spent in working with me, and for advice on any subject

\item Laurent and Pierre--Louis, for carefully reading this thesis and suggesting many improvements

\item Carlos, for help and advice in various situations throughout my doctorate

\item GGhh and Marco, for fun talk about science and stuff, and for receiving me in Florence

\item Ingrid, for smoothing my path to Nice and further

\item Jean--Marc and Julien, for tech support and guitar and linux talk 

\item LJAD and EDSFA administration crew, for making coping with administrative stuff a pleasure

\item my Erasmus friends:
Charline, Chiara, Daniela, Karo, Kerstin, No\'emie, Oph\'elia, Sarah, Susanna, Tomke,
GGhh, Henry, Marco, Martin, Nils, for going through the Erasmus experience with me

\item my Florentine flatmates: Silvia \& Carlo and Marzia \& Antonio, and my office mates in stanza T1: Giulia, Loredana, John and Simone,
   for making me feel at home during my stay in Florence

\item 
Amel, Audrey, Cindy, Ioana \& Pierre, Irene \& Marco, Julie \& S\'ebastien, Laura \& Benjamin, Monica, Nahla, Nancy, Olivia \& Joan, Salima \& Paul Eric, Sara, Silvia, St\'ephanie, 
Vanessa,
Ahed, Amine, Benjamin, Brahim, Brice, Giovanni, Hamad, Hugo, Luca, Marc, Nicolas, Osman, Rapha\"el, R\'emy, Sarrage, Tolgahan, Tom V and Xavier,
for giving me a nice time in Nice

\item Julianna, for patiently answering  my questions, and for writing a paper with me

% \item Edward Morehouse and Tadeusz Litak, for fighting the angry squids with me

\item Debian and upstream, for providing the best operating system and tools, and, in particular, the Fossil SCM community

\item Tobias and Michael, for news from Bayreuth and technical assistance
\item Krissi, Nicki and ``Volker'', for the fun time spent together

\item Anne--Laure, for bearing with me, and family, for receiving me with such warmth

\item Rike \& Uwe, Vroni \& Matze + Max, Feli \& Clemi

\end{packitem}

\begin{comment}
g.gherdovich@gmail.com, tolgahan.k@gmail.com, paris.slv@gmail.com, ahed.hindawi@gmail.com, io_e_annie@hotmail.com, marc.lasson@ens-lyon.fr, benjamin.marteau@laposte.net, contact@kyol.fr,
ireara@gmail.com, calderon.sara07@gmail.com, lalabenji@gmail.com, monica.courtney@unice.fr, vlandaverde@gmail.com,
tobias.sesselmann@googlemail.com, michael.gerhaeuser@gmx.de,
mailkrissi@gmx.de, volker.gall@gmx.de, nicole.fortmueller@gmx.de
mail@arndtposer.de
ira_nzh@yahoo.com
dimitrova_desislava@gmx.net
domi.laget@free.fr
\end{comment}

%% file: long_resume_fr.tex
\chapter*{R\'esum\'e Long} 

Dans cette th\`ese, on donne une caract\'erisation alg\'ebrique de la syntaxe et de la s\'emantique des langages
simplement typ\'es.
Plus pr\'ecis\'ement, on caract\'erise la syntaxe simplement typ\'ee avec liaison de variables, 
\'equip\'ee des r\`egles de r\'eduction, via une propri\'et\'e universelle, \`a savoir 
comme l'objet initial d'une cat\'egorie.

\section*{S\'emantique Initiale}

La \emph{S\'emantique Initiale} caract\'erise les termes d'un langage associ\'es \`a une \emph{signature $S$}
 comme l'objet initial d'une cat\'egorie --- dont on appelera les objets les \emph{S\'emantiques de $S$} ---,
ce qui fournit une d\'efinition concise de haut niveau de la syntaxe abstraite associ\'ee a $S$.
Plus pr\'ecisement, les ingr\'edients suivants sont utilis\'es:

\begin{description}
 \item [Signature] Une \emph{signature} sp\'ecifie, de fa\c con abstraite et concise, la syntaxe et la s\'emantique
                  d'un langage.
 \item [Cat\'egorie de Repr\'esentations] A chaque signature $S$, on associe une  cat\'egorie de \guillemotleft models\guillemotright\ de cette signature,
                      que l'on appelera des \emph{repr\'esentations de $S$}.
 \item [Initialit\'e] Dans cette cat\'egorie de repr\'esentations de $S$, on exhibe l'objet initial, le \emph{langage g\'en\'er\'e par $S$}.
\end{description}
Les motivations pour la S\'emantique Initiale sont doubles:
premi\`erement, la S\'emantique Initiale fournit une d\'efinition cat\'egorique --- via une propri\'et\'e universelle --- 
          de la syntaxe et de la s\'emantique engendr\'ees librement par une signature.
Deuxi\`emement, l'initialit\'e donne lieu \`a un \emph{op\'erateur d'it\'eration} qui permet de sp\'ecifier 
       de fa\c con \'economique et conviviale des morphismes --- \emph{traductions} --- de l'objet initial vers des autres langages.

Selon la \guillemotleft richesse\guillemotright\ du langage qu'on veut sp\'ecifier, on a besoin d'une notion 
de signature adapt\'ee et, en cons\'equence, d'une repr\'esentation de cette signature. Les caract\'eristiques que l'on consid\`ere dans
cette th\`ese sont:

\begin{description}
 \item [Liaison de Variables] On consid\`ere des constructions liantes au \emph{niveau des termes}, tels que l'abstraction lambda.
 \item [Typage] On consid\`ere des syst\`emes de types \emph{simples}, tels que le lambda calcul simplement typ\'e et, 
            via l'isomorphisme de Curry--Howard, la logique propositionelle (cf.\ \autoref{ex:logic_trans}).
 \item [R\'eduction] On consid\`ere de la s\'emantique sous forme de r\`egles de r\'eduction sur des termes, telles que la r\'eduction b\^eta,
                \[ \lambda x.M (N) \leadsto M [x:= N] \enspace . \]
\end{description}
Pour l'int\'egration de chacune des caract\'eristiques ci--dessus, les notions de signature et de repr\'esentation
 n\'ecessitent d'\^etre adapt\'ees pour tenir compte de la quantit\'e croissante d'information qui doit \^etre
fournie pour sp\'ecifier un langage.

Un de nos buts, c'est d'utiliser la S\'emantique Initiale pour traiter la question suivante:
nous voudrons traduire d'un langage \`a un autre --- possiblement sur des ensembles de types diff\'erents ---, 
en utilisant une construction universelle cat\'egorique.
Cette construction devrait prendre en compte le plus de \guillemotleft structure\guillemotright\ possible.
Par cela nous entendons que la traduction consid\'er\'ee devrait, par construction, \^etre compatible, par exemple,
avec le typage et r\'eduction aux langages source et but.

\section*{Contributions}

Dans cette th\`ese, nous donnons, via une propri\'et\'e universelle, une caract\'erisation alg\'ebrique
de la syntaxe simplement typ\'ee \'equip\'ee d'une s\'emantique sous forme de r\`egles de r\'eduction.
Plus pr\'ecis\'ement, \'etant donn\'ee une \emph{signature} --- qui sp\'ecifie les types et les termes
d'un langage --- et des \emph{in\'equations} sur cette signature --- qui sp\'ecifient des r\`egles de r\'eduction ---,
nous caract\'erisons les termes du langage associ\'e \`a cette signature, \'equip\'es des r\`egles de r\'eduction
selon les in\'equations donn\'ees, comme l'objet initial d'une cat\'egorie des \guillemotleft mod\`eles\guillemotright.

Notre point de d\'epart est un travail sur l'initialit\'e de la syntaxe non--typ\'ee effectu\'e par 
Hirschowitz et Maggesi \cite{DBLP:conf/wollic/HirschowitzM07},
et sur son extension sur la syntaxe simplement typ\'ee par Zsid\'o \cite{ju_phd}.
Dans un premier temps nous \'etendons le th\'eor\`eme de Zsid\'o \cite[Chap.~6]{ju_phd}
pour tenir compte des variations des types (cf.\ \autoref{sec:compilation}).
Puis, nous int\`egrons des r\`egles de r\'eduction dans le r\'esultat d'initialit\'e purement syntaxique
d'Hirschowitz et Maggesi \cite{DBLP:conf/wollic/HirschowitzM07}, cf.\ \autoref{sec:prop_arities}.
Finalement nous obtenons notre th\'eor\`eme principal, qui tient compte des variations des types ainsi que
des r\`egles de r\'eduction, en combinant les deux r\'esultats susmentionn\'es, cf.\ \autoref{chap:comp_types_sem}.

De plus, pour le cas non--typ\'e, nous fournissons une preuve formalis\'ee dans l'assistant de preuves \textsf{Coq}
de notre r\'esultat, ce qui donne un m\'ecanisme qui, \'etant donn\'ee une signature pour des termes et un 
ensemble d'in\'equations, produit la syntaxe abstraite associ\'ee a cette signature, \'equip\'ee de la
relation de r\'eduction engendr\'ee par les in\'equations.
Pour  le cas simplement typ\'e, nous formalisons l'instance de notre r\'esultat principal (cf.\ \autoref{thm:init_w_ineq_typed})
pour la signature du langage de programmation $\PCF$ \cite{Plotkin1977223}.

Nous d\'ecrivons maintenant nos contributions en d\'etail:

\subsection*{Une variante du th\'eor\`eme de Zsid\'o}

Dans sa th\`ese, \cite[Chap.~6]{ju_phd}, Zsid\'o d\'emontre un th\'eor\`eme d'initialit\'e pour la syntaxe 
abstraite associ\'ee a une signature simplement typ\'ee.
Pourtant, les mod\`eles qu'elle consid\`ere, dont la syntaxe abstraite est initiale, 
sont tous des mod\`eles sur le m\^eme ensemble de types.
Ainsi, le principe d'it\'eration obtenu par initialit\'e ne permet pas la sp\'ecification d'une traduction
vers un langage sur un ensemble diff\'erent de types.
Nous adaptons son th\'eor\`eme en introduissant des \emph{signatures typ\'ees}.
Une signature typ\'ee $(S,\Sigma)$ sp\'ecifie un ensemble de \emph{types} via une signature alg\'ebrique $S$, ainsi 
qu'un ensemble de \emph{termes} simplement typ\'es sur ces types via une signature de termes $\Sigma$ sur $S$.

Une repr\'esentation $R$ d'une telle signature typ\'ee est alors donn\'ee par une repr\'esentation de sa signature $S$ pour
les types dans un ensemble $T = T_R$ ainsi qu'une repr\'esentation de $\Sigma$ dans une monade --- aussi appel\'ee $R$ ---
sur la cat\'egorie $\TS{T}$.
Un morphisme de repr\'esentations $P\to R$ est constitu\'e d'un morphisme $f$ entre les repr\'esentations de $S$ sous-jacentes,
et d'un morphisme de repr\'esentations de $\Sigma$ qui est compatible dans un sens appropri\'e avec la 
\guillemotleft traduction des types\guillemotright\ $f$.
Nous d\'emontrons que la cat\'egorie des repr\'esentations de $(S,\Sigma)$ ainsi d\'efinie admet un objet initial,
qui int\`egre les types librement engendr\'es par $S$ et les termes librement engendr\'es par $\Sigma$, typ\'es
sur les types de $S$.
Notre d\'efinition de morphismes assure que, pour toute traduction sp\'ecifi\'ee par le principe d'it\'eration,
la traduction des termes est compatible avec la traduction des types par rapport au typage des langages source et but.

\subsection*{Syntaxe non--typ\'ee et r\`egles de r\'eduction}

Pour int\'egrer des r\`egles de r\'eduction \`a nos r\'esultats d'initialit\'e,
nous d\'efinissons la notion de \emph{2--signature}.
Une 2--signature $(\Sigma,A)$ est donn\'ee par une (1--)signature $\Sigma$ qui sp\'ecifie les termes d'un langage,
et un ensemble  $A$ d'\emph{in\'equations} sur $\Sigma$. 
Intuitivement, chaque in\'equation sp\'ecifie une r\`egle de r\'eduction, par exemple la r\`egle b\^eta.

Les \emph{mod\`eles} --- ou \emph{repr\'esentations} --- d'une telle 2--signature sont construits \`a partir des \emph{monades relative} et
des \emph{modules sur des monades relatives}:
\'etant donn\'ee une 1--signature $\Sigma$,
nous d\'efinissons une  repr\'esentation de $\Sigma$ comme \'etant donn\'ee par une monade relative sur le
foncteur appropri\'e $\Delta : \Set \to \PO$ (cf.\ \autoref{def:delta}),
accompagn\'ee d'un morphisme de modules (sur des monades relatives) appropri\'e
pour chacune des arit\'es de $\Sigma$.
Etant donn\'e un ensemble $A$ d'in\'equations sur $\Sigma$, nous d\'efinissons un pr\'edicat de satisfaction pour les mod\`eles de $\Sigma$;
nous appelons \emph{repr\'esentation de $(\Sigma,A)$} chaque repr\'esentation de $\Sigma$ qui satisfait chacune des in\'equations de $A$.
Ce pr\'edicat sp\'ecifie une sous--cat\'egorie pleine de la cat\'egorie des repr\'esentations de $\Sigma$.
Nous appelons cette sous--cat\'egorie la \emph{cat\'egorie des repr\'esentations de $(\Sigma,A)$}.
Nous d\'emontrons que cette cat\'egorie admet un objet initial, qui est construit en \'equipant la repr\'esentation initiale de $\Sigma$
--- donn\'ee par les termes librement engendr\'es par $\Sigma$ --- d'une relation de r\'eduction appropri\'ee engendr\'ee par les 
in\'equations de $A$.

Avec ce th\'eor\`eme d'initialit\'e de $(\Sigma,A)$ nous obtenons un nouveau principe d'it\'eration, et chaque traduction 
qui est sp\'ecifi\'ee via ce principe est, par construction, compatible avec la relation de r\'eduction aux langages 
source et but.

\subsection*{Th\'eor\`eme principal: Syst\`emes de types simples et r\'eductions}

Finalement, nous combinons les deux th\'eor\`emes susmentionn\'es pour obtenir un r\'esultat d'initialit\'e
qui tient compte de notre exemple principal, une traduction de $\PCF$ vers le lambda calcul
non--typ\'e.
Plus pr\'ecis\'ement, nous d\'efinissons une \emph{2--signature} comme \'etant donn\'ee par une signature
typ\'ee $(S,\Sigma)$, accompagn\'ee d'un ensemble $A$ d'in\'equations sur $(S,\Sigma)$ qui 
sp\'ecifie des r\`egles de r\'eduction.

Nous d\'efinissons une cat\'egorie de repr\'esentations $(S,\Sigma)$ et nous d\'emontrons que cette cat\'egorie
admet un objet initial.
Cette repr\'esentation initiale int\`egre les types et les termes librement engendr\'es par $(S,\Sigma)$,
les termes \'etant \'equip\'es d'une relation de r\'eduction engendr\'ee par les in\'equations de $A$.

\subsection*{Une impl\'ementation sur machine pour la sp\'ecification de syntaxe et s\'emantique}

Les th\'eor\`emes susmentionn\'es sont faits pour \^etre impl\'ement\'es dans un assistant de preuves.
Une telle impl\'ementation permet la sp\'ecification de syntaxe et r\`egles de r\'eduction via des 2--signatures,
fournissant un m\'ecanisme fortement automatis\'e pour produire de la syntaxe \'equip\'ee d'une substitution 
certifi\'ee et d'un principe d'it\'eration.

Nous d\'emontrons le th\'eor\`eme pour syntaxe non--typ\'ee avec r\`egles de r\'eduction d\'ecrit en haut
dans l'assistant de preuves \textsf{Coq} \cite{coq}.
Comme illustration, nous d\'ecrivons comment obtenir le lambda calcul avec r\'eduction b\^eta via initialit\'e.

De plus, nous formalisons une instance du th\'eor\`eme principal, \'egalement en \textsf{Coq}.
Plus pr\'ecisement, nous d\'efinissons la cat\'egorie des repr\'esentations de la signature typ\'ee de $\PCF$ avec des r\'eductions et
nous d\'emontrons que cette cat\'egorie admet un objet initial.
Apr\`es, nous donnons une repr\'esentation de cette signature dans la monade relative du lambda calcul avec r\'eduction b\^eta
$\ULCB$, ce qui fournit une traduction de $\PCF$ vers $\LC$.
Des instructions sur comment obtenir le code source complet de notre biblioth\`eque \textsf{Coq} sont disponible sur

\begin{center}
 \url{http://math.unice.fr/laboratoire/logiciels}.
\end{center}

%% file: introduction_fr.tex
\chapter*{Introduction}

\section*{Motivation: Traductions de \texorpdfstring{$\PCF$}{PCF} vers \texorpdfstring{$\ULC$}{ULC}} \label{sec:trans_pcf_ulc_fr}

Comme exemple introductif, on consid\`ere des traductions de \PCF, introduit par Plotkin \cite{Plotkin1977223},
vers le lambda calcul de Church \cite{Church_1936}.
Une description d\'etaill\'ee des deux langages est donn\'ee dans 
 \autoref{chap:syntax_semantics_pcf_ulc}. 
Ces deux langages sont paradigmatiques au sens o\`u $\PCF$ peut \^etre vu comme un langage de haut niveau, \'equip\'e
d'un syst\`eme de types, tandis que le lambda calcul repr\'esente un langage non--typ\'e de bas niveau.

Nous sp\'ecifions une application $f$ de l'ensemble de termes de \PCF~vers le lambda calcul comme dans \autoref{fig:trans_pcf_ulc} (cf.\ \cite{DBLP:conf/lics/Phoa93}),
avec une fonction $g$ des constantes de $\PCF$ vers des lambda termes, e.g., $g(\mathbf{T}) := \lambda x y.x$, et
constantes du lambda calcul, e.g.,
\begin{align*}
%  \mathbf{Y} &:= \lambda f.(\lambda x.f (x x))(\lambda x.f (x x)) \enspace \text{ and} \\
{\Theta} &:= \bigl(\lambda x. \lambda y.(y (x x y))\bigr)\bigl(\lambda x.\lambda y.(y(x x y))\bigr) \quad \text{ (Turing fixed point combinator) and}\\
  \Omega &:= (\lambda x . xx)(\lambda x.xx) \enspace .
\end{align*} 
Bien entendu, diff\'erentes traductions existent; par exemple, on pourrait 
traduire $\PCFFix$ vers un combinateur de point fixe diff\'erent.

Dans cette th\`ese on pr\'esente un cadre cat\'egorique pour la sp\'ecification des tels traductions d'un langage vers un autre.
Les challenges sont:
\begin{itemize}
 \item les ensembles de types diff\'erents des langages source et but 
        et
 \item int\'egrer la compatibilit\'e de telles traductions avec la structure --- substitution et r\'eduction --- 
      des langages source et but.
\end{itemize}
Nous d\'efinissons une cat\'egorie dans laquelle les langages comme $\PCF$ et $\LC$ sont des objets, et
dans laquelle une traduction comme d\'ecrite plus haut est un morphisme $f: \PCF\to\LC$.
Plus pr\'ecis\'ement, dans la cat\'egorie qu'on construit, la traduction $f$ est un \emph{morphisme initial} $f:\PCF\to\LC$, c'est--\`a--dire,
sa source $\PCF$ est l'objet initial.
Il y a plusieurs traductions possibles de $\PCF$ vers $\LC$, et le morphisme 
 $f: \PCF\to \LC$
ne peut pas \^etre initial dans une cat\'egorie o\`u les objets ne sont \guillemotleft que\guillemotright\ des langages
--- autrement on aurait $f = f'$ pour toute traduction $f' : \PCF \to \LC$.
Donc les objets dans la cat\'egorie qu'on construit sont des langages avec de la structure de plus, 
qui permet de distinguer des morphismes initiaux $f, f':\PCF\to\LC$,

\[
 \begin{xy}
  \xymatrix @C=3pc @R=0.5pc{
     {} &     (\LC, \psi) \\
     (\PCF, \phi) \ar[ru]^{f} \ar[rd]_{f'} & {} \\
     {} & (\LC, \psi').
   }
 \end{xy}
\]
Dans cette cat\'egorie, initialit\'e de $(\PCF, \phi)$ donne le principe d'it\'eration suivant: sp\'ecifier
 une traduction it\'erative $f: \PCF\to\LC$ est \'equivalent \`a sp\'ecifier la ``structure additionelle'' $\psi$ du
lambda calcul $\LC$.

Une question naturelle est si --- ou mieux, dans quel sens --- la traduction $f$ sp\'ecifi\'ee dans \autoref{fig:trans_pcf_ulc}
est compatible
avec les r\'eductions respectives des langages source et but.
Phoa \cite{DBLP:conf/lics/Phoa93} repond \`a cette question; en particulier,
la traduction $f$  est \emph{fid\`ele} %\cite{DBLP:conf/lics/Phoa93} 
au sens que
\[  t\twoheadrightarrow_{\PCF} t' \quad \text{implique} \quad f(t) \twoheadrightarrow_{\beta} f(t') \enspace . \]
Dans cette th\`ese nous fournissons un cadre cat\'egorique qui permet de sp\'ecifier, via une propri\'et\'e universelle, 
de telles traductions fid\`eles entre des langages avec liaison sur des ensembles de types diff\'erents.

\subsection*{Exemple: Axiomes de Peano}\label{subsec:informal_intro_fr}

On introduit la notion de \emph{signature} et \emph{repr\'esentation} 
\`a l'exemple des nombres naturels; 
on donne la signature des nombres naturels ainsi que la cat\'egorie des 
repr\'esentations associ\'ee.
Comme \emph{signature}, nous consid\'erons l'application suivante:
\[ \mathcal{N} := \{z \mapsto 0 \enspace ,\quad s \mapsto 1\} \enspace . \] 
Les nombres naturels sont construit \`a partir de deux constructeurs, notamment
un operateur d'arit\'e 0, disons, $z$, --- la constante z\'ero ---
ainsi qu'un operateur unaire, disons, $s$ ---  la fonction successeur.

Une \emph{repr\'esentation} de la signature $\mathcal{N}$ est donn\'ee par un triplet $(X,Z,S)$ d'un ensemble $X$ 
avec une constante $Z \in X$ et une op\'eration unaire
$S:X\to X$. 
Un morphisme vers un autre triplet $(X_0,Z_0, S_0)$ est donn\'e par une application $f : X \to X_0$
telle que
\begin{equation*} 
 f(Z) = Z_0 \quad \text{ et } \quad \comp{S}{f} = \comp{f}{S_0}  \enspace . 
\end{equation*}
Cette cat\'egorie admet un \emph{objet initial} $(\mathbb{N}, \ZERO, \SUCC)$ donn\'e par les nombres naturels $\mathbb{N}$
\'equip\'es de la constante $\ZERO = 0$ et de l'application successeur $\SUCC : \mathbb{N}\to\mathbb{N}$.

\subsection*{Liaison des Variables} \label{subsec:adding_variables_fr}

Les techniques suivantes sont utilis\'ees fr\'equemment pour mod\'eliser la liaison des variables:

\begin{packitem}
 \item %\label{list:nominal}  
       Syntaxe nominelle utilisant l'abstraction nomm\'ee ($\mathbb{A}$ \'etant un ensemble d'\emph{atomes}), e.g.,
         \[  \lambda : [\mathbb{A}]  T \to T\]
 \item %\label{list:hoas} 
       Higher--Order Abstract Syntax (HOAS), e.g.,
          \[ \lambda : (T \to T) \to T\] et sa variante \emph{faible}, e.g.,
          \[\lambda : (\mathbb{A} \to T) \to T\] 
 \item Nested Data Types comme present\'es par \cite{BirdMeertens98:Nested}, e.g., %\label{list:brujin}
          \[ \lambda : T(X + 1) \to T(X) \]
\end{packitem}
L'encodage via nested data types est diff\'erent des autres techniques au sens qu'ici, l'ensemble des termes $T$
est param\'etris\'e par un \emph{contexte}.
Donc $T(X)$ d\'enote l'ensemble des termes du langage $T$ avec des variables libres dans l'ensemble $X$.
L'ensemble $X + 1$ correspond \`a un \emph{contexte \'elargi} d'une variable libre additionelle, qui
sera li\'ee par le constructeur lambda.

\begin{exemple*}
Nous repr\'esentons le lambda calcul comme un \emph{nested data type}:
 consid\'erons le type inductif $\LC:\Set\to\Set$:
\begin{lstlisting}
Inductive ULC (V : Type) : Type :=
  | Var : V -> ULC V
  | Abs : ULC (option V) -> ULC V
  | App : ULC V -> ULC V -> ULC V.
\end{lstlisting}
\end{exemple*}

\noindent
Pour la syntaxe avec liaison, les arit\'es doivent donner de l'information sur les liaisons du constructeur associ\'e.
Nous sp\'ecifions les arit\'es avec des listes de nombres naturels.
La longueur d'une liste sp\'ecifie le nombre d'arguments d'un constructeur, 
et sa composante $i$ donne le nombre de variables que le constructeur lie dans l'argument $i$.
La signature $\Lambda$ de $\LC$ est donn\'ee par 
\[ \Lambda := \{ \app : [0,0]\enspace , \quad \abs : [1] \} \enspace . \]
L'application $ V \mapsto \LC(V)$ est functorielle: pour $f : V \to W$,
l'application $\LC(f) : \LC(V) \to \LC(W)$ \emph{renomme} chaque variable libre $v\in V$ d'un terme par $f(v)$,
ce qui donne un terme avec des variables libres dans  $W$. 
Alors, la signature $\Lambda$ doit \^etre repr\'esent\'ee dans des functeurs $F:\Set\to \Set$ au lieu des ensembles,
et on consid\`ere des transformations naturelles au lieu des applications.

% $\mathcal{LC}$

\subsection*{Substitution}

Nous souhaitons int\'egrer le plus de structures possible dans notre cat\'egorie de \guillemotleft mod\`eles\guillemotright.
Une de ces structures est la \emph{substitution sans capture} des variables libres. 
Pour cela, nous ne consid\'erons pas des functeurs simples $F:\Set\to\Set$, 
mais des \emph{monades} sur la cat\'egorie $\Set$ des ensembles. 
Une monade est un functeur
\'equip\'e de structure additionelle, que l'on explique en utilisant l'exemple du lambda calcul.
L'application $ V\mapsto \LC(V)$ vient avec une op\'eration de substitution simultan\'ee sans capture:
soient $V$ and $W$ deux ensembles (de variables) et $f$ une application $f : V\to \LC(W)$.
Etant donn\'e un lambda terme $t\in \LC(V)$, on remplace chaque variable libre $v\in V$
dans $t$ par son image sous $f$, ce qui donne un terme $t' \in \LC(W)$.
De plus, nous consid\'erons le constructeur $\Var_V$ comme une application ``variable--comme--terme'', index\'ee par un ensemble de variables $V$, 
\[\Var_V : V\to \LC(V) \enspace . \]
Altenkirch et Reus \cite{alt_reus} observent que la structure de monade
capture ces deux op\'erations et leurs propri\'et\'es:
substitution et variable--comme--termes font de
$\LC$ une monade sur la cat\'egorie des ensembles.

La  structure de monade de $\LC$ devrait \^etre compatible dans un sens avec les constructeurs $\Abs$ et $\App$ de $\LC$:
\emph{substitution distribue sur les constructeurs}.
Pour capturer cette distributivit\'e, 
Hirschowitz et Maggesi \cite{DBLP:conf/wollic/HirschowitzM07} consid\`erent des \emph{modules sur une monade} (cf.\ \autoref{def:module}) 
--- qui g\'en\'eralisent la substitution monadique ---, et des morphismes de modules --- qui sont des transformations naturelles
qui sont compatibles avec la substitution de modules.
En effet, les applications
\begin{align*}\LC &: V\mapsto \LC(V) \enspace , \\ 
              \LC' &: V\mapsto \LC(V+1) \text{ et } \\ 
              \LC\times\LC &: V \mapsto \LC(V)\times\LC(V) 
\end{align*} sont
des application sous--jacentes de tels modules (cf.\ \autorefs{ex:ulc_taut_mod}, \ref{ex:ulc_const_mod}), et
les constructeurs $\Abs$ et $\App$ sont des morphismes de modules (cf.\  \autorefs{ex:ulc_mod_mor}, \ref{ex:ulc_mod_mor_kl}).

\subsection*{Types} 

Des \emph{syst\`emes de types} existent avec des caract\'eristiques vari\'ees, de la syntaxe simplement typ\'ee
\`a la syntaxe avec des types d\'ependents, polymorphisme etc.
Par syntaxe simplement typ\'ee  nous entendons une syntaxe non--polymorphe dont l'ensemble de types est ind\'ependent 
de l'ensemble des termes, c'est--\`a--dire les constructeurs de types ne prennent que des types comme arguments.

Dans des syst\`emes de types plus sophistiqu\'es, les types peuvent d\'ependre des termes, ce qui am\`ene a des d\'efinitions plus 
complexes d'arit\'es et de signature.
Ce travail--ci ne traite que les langages simplement typ\'es, comme le lambda calcul simplement typ\'e ou $\PCF$.
Nous appellerons l'ensemble de types sous--jacent les \emph{types objet}.

Le but du typage est de \emph{classifier} les termes selon des crit\`eres.
Par exemple, on pourrait se demander si un terme est de type fonction, et ainsi peut \^etre appliqu\'e
\`a un autre terme.
Une fois qu'une telle classification est mis en place, on peut utiliser l'information de typage pour filtrer
les termes selon leurs types, pour ne choisir que les termes avec le type d\'esir\'e.

Une fa\c con d'ajouter des types serait de les int\'egrer dans les termes comme dans \guillemotleft $\lambda x:\NN.x + 4$\guillemotright.
Par contre, pour les \emph{syst\`emes de types simples} on peut s\'eparer les univers des types et des termes et
consid\'erer le typage comme une application des termes vers les types, ainsi donnant une structure simple math\'ematique au typage.

Comment peut--on assurer que nos termes sont bien typ\'es ? Bien qu'on s\'epare les types des termes, on voudrait maintenir
une int\'egration forte du typage dans le processus de construction des termes, pour \'eviter de construire des termes mal typ\'es.
La s\'eparation des termes et des types semble contredire ce but.
La r\'eponse est de ne pas consid\'erer qu'\emph{un} ensemble de termes avec une application de typage vers l'ensemble, disons, $T$ de types,
mais \emph{une famille d'ensembles}, index\'ee par l'ensemble $T$ de types objet.
Les constructeurs de termes peuvent ainsi choisir quels termes ils accepteront comme argument.
Nous consid\'erons aussi les variables libres comme \'etant \'equip\'ees d'un type objet.
Autrement dit, nous ne consid\'erons pas des termes sur \emph{un} ensemble de variables,
mais sur une famille d'ensembles de variables, index\'ee par l'ensemble des types objet.
Encore autrement dit, nous consid\'erons un contexte comme donn\'e par une famille $(V_t)_{t\in T}$ 
d'ensembles, d'o\`u  $V_t := V(t)$ est l'ensemble de variables de type $t$.
Nous illustrons notre point de vue \`a l'aide de l'exemple du lambda calcul simplement typ\'e $\SLC$: 

\begin{exemple*}
 Soit
\[\TLCTYPE ::= \enspace * \enspace \mid \enspace \TLCTYPE \TLCar \TLCTYPE\] 
l'ensemble de types du lambda calcul simplement typ\'e. 
L'ensemble des lambda termes avec des variables libres dans $V$ est donn\'e par la famille inductive suivante:
\begin{lstlisting}
Inductive TLC (V : T -> Type) : T -> Type :=
  | Var : forall t, V t -> TLC V t
  | Abs : forall s t TLC (V + s) t -> TLC V (s ~> t)
  | App : forall s t, TLC V (s ~> t) -> TLC V s -> TLC V t.
\end{lstlisting}
d'o\`u $V + s := V + \lbrace {*s} \rbrace$ 
est l'extension du contexte par une variable de type $s\in \TLCTYPE$ --- la variable qui sera
li\'ee par le constructeur $\Abs{(s,t)}$. 
Les variables $s$ et $t$ prennent des valeurs dans l'ensemble $\TLCTYPE$ des types.
La signature du lambda calcul simplement typ\'e est donn\'ee dans \autorefs{ex:tlc_sig} et \ref{ex:tlc_sig_higher_order}. 
Le paragraphe pr\'ecedent sur les monades et modules s'applique au lambda calcul simplement typ\'e quand on remplace les ensembles par
des familles d'ensembles index\'ees par $\TLCTYPE$:
le lambda calcul simplement typ\'e peut \^etre \'equip\'e d'une structure de monade (cf.\ \autoref{ex:tlc_syntax_monadic})
\[ \SLC : \TS{\TLCTYPE} \to \TS{\TLCTYPE} \enspace . \]
Les constructeurs de $\SLC$ sont des morphismes de modules (cf.\ \autorefs{ex:slc_modules}, \ref{ex:deriv_mod_slc}, \ref{ex:fibre_mod_slc}).
\end{exemple*}
Cette m\'ethode de d\'efinir pr\'ecis\'ement les termes bien typ\'es en les organisant dans une famille d'ensembles parametris\'ee par
 les types objet s'appelle \emph{typage intrins\`eque} \cite{dep_syn} --- l'oppos\'e du \emph{typage extrins\`eque}, o\`u d'abord on d\'efinit un ensemble 
de termes \emph{bruts}, qui est filtr\'e apr\`es via un pr\'edicat de typage.
Le typage intrins\`eque d\'el\`egue le typage objet au syst\`eme de type du m\'eta langage, comme \textsf{Coq} dans \autoref{ex:slc_def}. 
Ainsi, le syst\`eme de types m\'eta (e.g.\ \textsf{Coq}) trie les termes mal typ\'es automatiquement: \'ecrire un tel terme donne 
une erreur de type au niveau m\'eta.

De plus, l'encodage intrins\`eque vient avec un principe de r\'ecursion plus conviviale; une application vers un autre
syst\`eme de types peut \^etre donn\'ee en sp\'ecifiant son image sur les termes bien typ\'es.
En utilisant le typage extrins\`eque, une application sur les termes serait sp\'ecifi\'ee sur l'ensemble des termes \emph{bruts},
y compris les termes mal typ\'es, ou seulement sur les termes bien typ\'es en donnant un argument propositionel de plus
qui exprime le fait que le terme soit bien typ\'e.
Benton et al.\ donnent une explication d\'etaill\'ee du typage intrins\`eque \cite{dep_syn}.

\subsection*{R\'eductions}

La \emph{s\'emantique} d'un langage de programmation d\'ecrit comment des logiciels de ce langage
sont \emph{\'evalu\'es}.
Pour les langages fonctionnels comme on les consid\`ere dans cette th\`ese, l'\'evaluation
est faite par des \emph{r\'eductions}.
Par exemple, l'\'evaluation du terme $7 + 5$ d'un langage arithm\'etique vers sa valeur $12$ est faite en une s\'erie
de r\'eductions, dont la forme pr\'ecise d\'epend de la s\'emantique du langage.
Des \emph{r\`egles} typiques, qui sp\'ecifient comment des termes r\'eduisent, sont donn\'ees dans \autoref{subsec:semantic_ulc_pcf}
pour les langages du lambda calcul et $\PCF$.

Etant donn\'e un ensemble $A$ de r\`egles de r\'eduction, on peut consid\'erer la relation engendr\'ee par ces r\`egles.
Plus pr\'ecis\'ement, suivant
Barendregt et Barendsen \cite{barendregt_barendsen}, nous consid\'erons plusieurs \emph{cl\^otures} de ces r\`egles:

\begin{description}
 \item [Propagation dans des sous--termes] Une relation $R$ est appel\'e \emph{compatible} si elle est close sous propagation dans des sous--termes, i.e.\
           si pour tout constructeur $f$ d'arit\'e $n$ et tout $i \leq n$,
           \[ M \leadsto_R N \Rightarrow   f(x_1,\ldots, x_{i_1}, M, x_{i+1}, \ldots, x_n) \leadsto_R f(x_1,\ldots, x_{i_1},N, x_{i+1}, \ldots, x_n) \enspace . \]
\item [R\'eduction] Une relation $R$ est une \emph{relation de r\'eduction} si elle est compatible, r\'eflexive et transitive. 
\item [Equivalence] Une relation $R$ est une \emph{congruence} si elle est une relation d'\'equivalence compatible.
\end{description}
A l'ensemble $A$ de r\`egles nous associons trois relations engendr\'ees par $A$, qui sont les relations les plus petites contenant $A$ et \'etant 
une relation compatible, une relation de r\'eduction et une relation d'\'equivalence, respectivement.
Nous \'ecrivons ces relations, dans cet ordre, par
$\to_A$, $\twoheadrightarrow_A$ and $=_A$, respectivement.

Dans cette th\`ese nous consid\'erons la \emph{relation de r\'eduction} engendr\'ee par un ensemble de r\`egles.
Par rapport \`a la congruence, il lui manque une r\`egle de \emph{symmetrie}, ce qui, bien qu'ad\'equat pour le 
raisonnement math\'ematique, donne lieu a une relation trop grossi\`ere du point de vue du \emph{calcul}.
Comme l'\'ecrit Girard \cite{proofs_and_types}, tandis que la congruence engendr\'ee par $A$ accentue le point de vue \emph{statique}
des math\'ematiques, la relation de r\'eduction associ\'ee \`a $A$ accentue le point de vue \emph{dynamique} du calcul.

Afin de tenir compte des r\'eductions, nous consid\'erons des foncteurs et monades dont le codomaine n'est pas la cat\'egorie des
(familles d') ensembles, mais des (familles d') \emph{ensembles pr\'eordonn\'es}.
La d\'efinition de monade demande du foncteur sous--jacent d'\^etre un endofoncteur, mais nous ne voudrons pas consid\'erer des 
contextes pr\'eordonn\'es --- quelle serait la signification de ce pr\'eordre ?
La restriction \`a des endofoncteurs a \'et\'e abolie par Altenkirch et al.\ \cite{DBLP:conf/fossacs/AltenkirchCU10}
en introduisant les monades \emph{relatives}.
Une monade relative est donn\'ee par un foncteur --- pas n\'ecessairement endo --- accompagn\'e de deux op\'erations tr\`es similaires
aux op\'erations monadiques variables--comme--termes et substitution.
Nous consid\'erons ainsi, par exemple, le lambda calcul comme une monade relative qui associe, \`a chaque \emph{ensemble} $X$
de variables, un \emph{ensemble pr\'eordonn\'e de lambda termes} $(\LC(X),\twoheadrightarrow_{\beta})$, 
o\`u le pr\'eordre sur $\LC(X)$ est donn\'e par la relation de r\'eduction $\twoheadrightarrow_{\beta}$
engendr\'ee par la r\`egle b\^eta de \autoref{eq:beta}, cf.\ \autoref{ex:ulcbeta}.

%% file: conclusion_fr.tex
\chapter*{Conclusions et Travaux Ult\'erieurs}%\addcontentsline{toc}{chapter}{Conclusions et Travaux Ult\'erieurs}

Nous r\'esumons les contributions de cette th\`ese et abordons des travaux ult\'erieurs.

\section*{Contributions}%\addcontentsline{toc}{section}{Contributions}

Nous avons d\'emontr\'e un r\'esultat d'initialit\'e pour de la syntaxe simplement typ\'ee, 
\'equip\'ee des r\`egles de r\'eduction.
Le principe d'it\'eration cat\'egorique obtenu par la propri\'et\'e universelle d'intialit\'e
est suffisamment g\'en\'eral pour permettre la sp\'ecification de traductions de la repr\'esentation des termes
vers des langages typ\'es sur des ensembles \emph{diff\'erents} des types.

Nous avons caract\'eris\'e la syntaxe liante avec des r\'eductions --- par exemple, le lambda calcul
avec la r\'eduction b\^eta --- comme une monade relative sur le foncteur $\Delta$ (cf.\ \autoref{ex:ulcbeta}),
ce qui n'encode pas seulement des propri\'et\'es de commutativit\'e de la substitution,
mais \'egalement sa monotonicit\'e dans l'argument d'ordre premier.
Une autre propri\'et\'e de monotonicit\'e pour l'argument d'ordre sup\'erieur peut \^etre assur\'ee 
par un renforcement appropri\'e de la d\'efinition de monade relative dans un contexte 2--cat\'egorique, 
cf.\ \autoref{rem:about_substitution}.
Nous avons \'egalement transf\'er\'e la d\'efinition de \emph{module sur une monade} et plusieurs constructions 
de modules vers des modules sur les monades \emph{relatives}.

Ensuite, nous avons d\'emontr\'e plusieurs th\'eor\`emes dans l'assistant de preuves \textsf{Coq}:
premi\`ere\-ment, nous avons impl\'ement\'e le th\'eor\`eme d'initialit\'e de Zsid\'o \cite[Chap.~6]{ju_phd},
r\'esum\'e dans ce travail pour r\'ef\'erence dans \autoref{sec:sts_ju}.
Deuxi\`emement, nous avons d\'emontr\'e le th\'eor\`eme de \autoref{sec:prop_arities},
fournissant un outil qui, \'etant donn\'ee une 2--signature $(S,A)$, g\'en\`ere la syntaxe associ\'ee a $S$,
\'equip\'ee de la relation de r\'eduction engendr\'ee par les in\'equations de $A$.
Troisi\`emement, nous avons d\'emontr\'e une instance de notre th\'eor\`eme principal, 
\autoref{thm:init_w_ineq_typed} de \autoref{chap:comp_types_sem}, pour la 2--signature particuli\`ere du langage de
programmation $\PCF$, \'equip\'e des r\`egles de r\'eduction comme dans \autoref{eq:pcf_reductions}.
La re\-pr\'e\-sen\-tation de la signature de $\PCF$ dans la monade du lambda calcul non typ\'e avec r\'eduction b\^eta
donne une traduction ex\'ecutable de $\PCF$ vers $\LC$ qui est certifi\'ee d'\^etre compatible 
avec la substitution et la r\'eduction des langages source et but.

\section*{Travaux Ult\'erieurs} %\addcontentsline{toc}{section}{Travaux Ult\'erieurs}

D\'esormais, nous esp\'erons d\'emontrer et impl\'ementer des th\'eor\`emes d'initialit\'e pour des syst\`emes de types
plus riches. En particulier, on voudrait prendre en compte des types \emph{d\'ependants} et le \emph{polymorphisme},
deux \'etapes importantes vers des logiciels certifi\'es et reutilisation de code, respectivement.

De plus, la mod\'elisation de la s\'emantique devrait \^etre am\'elior\'ee pour permettre le raisonnement sur des propri\'et\'es 
importantes telles que la terminaison.

Comme susmentionn\'e, l'impl\'ementation des r\'esultats d'initialit\'e dans un assistant de preuves peut servir comme un cadre pour 
la recherche sur des langages de programmation et des logiques. Pour cette raison, nous envisageons l'impl\'ementation dans un 
assistant de preuves de \autoref{thm:init_w_ineq_typed} en toute g\'en\'eralit\'e.

On pr\'esente ces aspects en d\'etail:

\begin{description}

% \item [Fine--grained modelling of reduction]
\item [Mod\'elisation de r\'eduction plus nuanc\'ee]

Etant donn\'ee une 2--signature (une signature avec un ensemble d'in\'equations), les mod\`eles pour
cette 2--signature \'etaient jusqu'\`a maintenant  princpalement des foncteurs qui associent, \`a chaque 
ensemble \guillemotleft de variables\guillemotright\ un ensemble pr\'eordonn\'e ---
 intuitivement un mod\`ele des \guillemotleft termes\guillemotright\ sur l'ensemble des variables\footnote{%
On ignore le cas typ\'e pour l'instant, qui est analogue.}.
Le pr\'eordre $\leq$ sur un tel mod\`ele correspond \`a la relation de r\'eduction sur ce mod\`ele, 
c'est--\`a--dire le \guillemotleft terme\guillemotright\ $t$ r\'eduit vers $t'$ si et seulement si $t \leq t'$.

La mod\'elisation des r\'eductions via des \emph{pr\'eordres} peut \^etre consid\'er\'ee comme \'etant trop grossi\`ere
\`a plusieurs \'egards:

\begin{itemize}
 \item des r\'eductions diff\'erentes peuvent amener d'un terme vers un autre. Par contre, l'utilisation des pr\'eordres
    pour la mod\'elisation des r\'eductions ne permet pas de distinguer deux r\'eductions de m\^eme source et but.
 \item La r\`egle de reflexivit\'e cod\'ee en dur rend difficile le raisonnement sur la \emph{normalisation} --- en particulier, 
    la \emph{terminaison}.

\end{itemize}

Au lieu de consid\'erer des ensembles \emph{pr\'eordonn\'es} (index\'es par des ensembles de variables libres) comme
des mod\`eles d'une 2--signature, il serait int\'eressant de consid\'erer une structure qui permet
un traitement plus nuanc\'e de r\'eduction, comme par exemples les graphes ou les cat\'egories.
Autrement dit, on pourrait construire des mod\`eles d'une 2--signature \`a partir des monades relatives
vers la cat\'egorie des \emph{graphes} ou (petites) \emph{cat\'egories}.
En utilisant cette nouvelle d\'efinition de mod\`ele, on pourrait envisager de d\'emontrer un th\'eor\`eme d'initialit\'e
analogue \`a celui d\'ej\`a d\'emontr\'e, et d'utiliser la structure de plus obtenue en 
travaillant avec des graphes ou des cat\'egories pour raisonner sur les propri\'et\'es mentionn\'ees plus haut.

\item [In\'equations, Syntaxiquement]
  Fiore et Hur \cite{DBLP:conf/csl/FioreH10} developpent une th\'eorie syntaxique d'\'equations sur 
  une signature d'ordre superieure, ce qui permet de prouver suret\'e et compl\'etude
  par rapport aux mod\`eles de la signature et aux \'equations.
   Des t\'echniques pareilles devraient permettre de pr\'esenter nos in\'equations de fa\c con 
   syntaxique.
   En plus du but \'evident de suret\'e et compl\'etude, une telle pr\'esentation syntaxique
  faciliterait aussi la sp\'ecification des r\'eductions dans l'impl\'ementation en \textsf{Coq}:
  en particulier, il serait possible de sp\'ecifier des r\'eductions sans aucune connaissance
  des concepts cat\'egoriques.

 Un but minimal, ce serait d'avoir un data type --- qui depend de la
 1-signature sous-jacant --- qui permet de sp\'ecifier les demi--\'equations
 habituelles, principalement obtenues par la substitution et en composant
des arit\'es, p.ex.\ $\comp{(\abs\times\id)}{\app}$.
 A un terme de ce data type on pourrait associer une famille de
 morphismes de modules, qui forment le carrier d'une demi--\'equation: les
 propri\'et\'es alg\'ebriques (d'\^etre un morphisme de modules, ce qui
 correspond a la compatibilit\'e entre substitution et meta--substitution
 dans \cite{DBLP:conf/csl/FioreH10}) pourraient \^etre prouv\'ees une fois pour tout par r\'ecurrence.

\item [Syst\`emes de types plus sophistiqu\'es]
% \item [More sophisticated type systems]

Les nouveaux langages de programmation sont \'equip\'es de syst\`emes de types de plus
en plus sophistiqu\'es:
des \emph{types d\'ependants} permettent d'assurer des propri\'et\'es des r\'esultats d'une 
fonction et ainsi la composition fiable des fonctions.
Le \emph{polymorphisme} permet la r\'eutilisation de code dans des situations diverses.
Une caract\'erisation alg\'ebrique de tels syst\`emes de types sophistiqu\'es avec liaison de 
variables par une propri\'et\'e universelle n'existe pas encore.
Nous esp\'erons g\'en\'eraliser nos r\'esultats d'initialit\'e pour prendre en compte ces syst\`emes de types.

\item [Une classe plus large d'arit\'es]
% \item [A wider class of arities]

Les th\'eor\`emes d'initialit\'e jusqu'\`a maintenant prennent en compte des arit\'es, c'est--\`a--dire
 des constructeurs de termes, de nature plut\^ot simple: les seules op\'erations consid\'er\'ees 
 sont le produit --- pour des constructeurs qui prennent \emph{plusieurs} arguments ---
 et l'extension de contexte, pour mod\'eliser la liaison de variables.

%  The present initiality theorems encompass arities, i.e.\ term constructors, of quite simple
%  nature: the only operations considered are product --- for constructors with \emph{multiple} arguments --- 
%  and context extension, for modelling variable binding.

On devrait tenir compte des constructeurs de termes plus g\'en\'eraux.
Hirschowitz et Maggesi \cite{hirschowitz_maggesi_fics2012}
ont introduit une notion d'arit\'e renforc\'ee qui permet, par exemple, de traiter
un constructeur d'aplatissement $\mu : \comp{T}{T} \to T$.
Finalement, nous esp\'erons trouver un crit\`ere \emph{simple}
tr\`es g\'en\'eral pour des arit\'es et des signatures pour lesquelles 
un mod\`ele initial peut \^etre construit.

\item [Un outil de recherche certifi\'e] % pour recherche sur langages de programmation et logiques]
% \item [A certified tool for research on programming languages and logics]

Les r\'esultats obtenus devraient --- comme on l'a d\'ej\`a fait pour la syntaxe non typ\'ee
avec r\'eductions --- \^etre impl\'ement\'es dans un assistant de preuves tel que \textsf{Coq}. Ainsi,
un th\'eor\`eme d'initialit\'e peut \^etre utilis\'e comme un outil pratique pour faire facilement des exp\'eriences
avec des langages differents.
Changer un langage correspondrait \`a simplement changer sa signature sp\'ecifiante, et toutes les donn\'ees et
propri\'et\'es telles que la substitution certifi\'ee et le principe d'it\'eration, mais \'egalement des r\'eductions, seraient
fournies par le syst\`eme.
Pour cette impl\'ementation sur la machine et pour avoir des r\`egles de r\'eduction appropri\'ees, nous souhaitons aussi obtenir, de fa\c con automatique, une \emph{fonction}
de r\'eduction $r$ en plus de la relation de r\'eduction. Cette fonction de r\'eduction pourrait ainsi \^etre valid\'ee 
par rapport \`a la relation au sens o\`u l'on pourrait d\'emontrer que pour chaque terme $t$, on a $t \leq r(t)$.

\end{description}

%% file: introduction.tex
%%%%%%%%    NEW INTRO

In this thesis we give a characterization, via a universal property, of the syntax and semantics of 
simply--typed languages with variable binding.
More precisely, we characterize the terms and sorts associated to a signature, 
equipped with reduction rules, as the initial object in some category.
Via the iteration principle stemming from initiality, 
translations between languages, possibly over different sets of sorts,
can be specified in a convenient and economic way. 
Furthermore, translations thus specified are ensured to be faithful with respect to 
reduction in the source and target languages, as well as compatible in a suitable sense
with substitution on either side.

\section{Motivation: Translations from \texorpdfstring{$\PCF$}{PCF} to \texorpdfstring{$\ULC$}{ULC}} \label{sec:trans_pcf_ulc}

As an introductory example, consider translations from the programming language \PCF, introduced by Plotkin \cite{Plotkin1977223},
to the untyped lambda calculus $\LC$, invented by Church \cite{Church_1936}.
A detailed account of both languages is given in \autoref{chap:syntax_semantics_pcf_ulc}. 
These two languages are paradigmatic in the sense that \PCF~may be considered a rather high--level language, equipped 
with a type system, whereas the untyped lambda calculus represents a low--level, untyped language.
We specify a map $f$ from the set of \PCF~terms to the set of lambda terms as in \autoref{fig:trans_pcf_ulc} (cf.\ \cite{DBLP:conf/lics/Phoa93}),
\begin{figure}[bt]
\centering
\fbox{%
 \begin{minipage}{7cm}

\vspace{-1ex}

\begin{align*}
           f(\bot_A) &= \Omega \\	
           f(c_A) &= g(c) \\
          f(x_A) &= x \\
         f (s@t) &= f(s)@f(t) \\
         f(\lambda x.M) &= \lambda x.f(M) \\
         f(\PCFFix_A(M)) &= {\Theta}@f(M)
\end{align*}

\vspace{0ex}

\end{minipage}
}
% \vspace{-2em}
 \caption{Translation from \PCF~to $\ULC$}\label{fig:trans_pcf_ulc}
% \end{minipage}
% }
\end{figure}
with a suitable function $g$ from the set of constants of $\PCF$ to lambda terms, e.g., $g(\mathbf{T}) := \lambda x y.x$, and suitable
constants of the lambda calculus, e.g.,
\begin{align*}
%  \mathbf{Y} &:= \lambda f.(\lambda x.f (x x))(\lambda x.f (x x)) \enspace \text{ and} \\
{\Theta} &:= \bigl(\lambda x. \lambda y.(y (x x y))\bigr)\bigl(\lambda x.\lambda y.(y(x x y))\bigr) \quad \text{ (Turing fixed point combinator) and}\\
  \Omega &:= (\lambda x . xx)(\lambda x.xx) \enspace .
\end{align*} 
Of course, different such translations exist; for instance, one may choose to 
translate $\PCFFix$ to a different fixed point combinator or one chooses a different representation $g'$ for the constants of \PCF~in the 
lambda calculus, yielding a different translation $f':\PCF\to\LC$.

In this thesis we present a category--theoretic framework to specify such translations of a language to another.
The challenges are 
\begin{itemize}
 \item the varying sets of sorts in source and target 
        languages\footnote{Here we consider untyped languages to be single--sorted.} and
 \item to capture compatibility of such translations with structure --- such as substitution and reduction --- in the 
       source and target languages.
\end{itemize}
We construct a category in which ``languages such as \PCF~and $\LC$ are objects'', 
and in which the above translation $f: \PCF\to\LC$ is a morphism. 
As it turns out, the preceding sentence is imprecise and needs to be refined:
more precisely, in the category we construct the translation $f$ is an \emph{initial morphism} $f:\PCF\to\LC$, that is,
its source \PCF~is the initial object.
Now, as we have seen, there are several possible translations from \PCF~to the lambda calculus, and the above translation $f: \PCF\to \LC$
cannot be an initial morphism in a 
category where objects are ``just'' languages --- otherwise we would have $f = f'$ for any translation $f':\PCF\to\LC$. %, which clearly is not the case.
Thus the objects in the category we construct are not just languages, but languages with additional 
structure, allowing us to distinguish
different initial morphisms $f, f':\PCF\to\LC$,

\[
 \begin{xy}
  \xymatrix @C=3pc @R=0.5pc{
     {} &     (\LC, \psi) \\
     (\PCF, \phi) \ar[ru]^{f} \ar[rd]_{f'} & {} \\
     {} & (\LC, \psi').
   }
 \end{xy}
\]
In this category, initiality of $(\PCF, \phi)$ yields the following iteration principle: specifying
 an iterative translation $f: \PCF\to\LC$ is equivalent to specifying the ``extra structure'' $\psi$ of 
the lambda calculus $\LC$.
We do not yet explain what this additional structure, here denoted $\phi, \psi$ and $\psi'$, looks like,
but refer instead to \autoref{subsec:informal_intro} for an instructive example.

A natural question then is whether --- or better, in what sense --- the translation $f$ specified in \autoref{fig:trans_pcf_ulc}
is compatible
with the respective reductions in the source and target languages.
Phoa \cite{DBLP:conf/lics/Phoa93} gives an answer to this question; in particular,
the translation $f$  is \emph{faithful} %\cite{DBLP:conf/lics/Phoa93} 
in the sense that
\[  t\twoheadrightarrow_{\PCF} t' \quad \text{implies} \quad f(t) \twoheadrightarrow_{\beta} f(t') \enspace . \]
In this thesis we provide a category--theoretic framework which allows to specify, via a universal property, 
such faithful translations between languages with variable binding over different sets of sorts.

\section{Initial Semantics} \label{sec:initial_semantics}
\emph{Initial Semantics} characterizes the terms of a language associated to a \emph{signature $S$}
as the initial object in some category --- whose objects we call \emph{Semantics of $S$} ---,
yielding a concise, high--level, definition of the abstract syntax  associated to $S$. %.
In more detail, the following ``ingredients'' are used:

\begin{description}
 \item [Signature] A \emph{signature} specifies abstractly and concisely the syntax and semantics
                  of a language.
 \item [Category of Representations] To any signature $S$ we associate a category of ``models'' of that signature,
                      the objects of which we call \emph{representations of $S$}.
 \item [Initiality] In this category of representations of $S$ we exhibit the initial object, the \emph{language generated by $S$}.
\end{description}
The motivation for Initial Semantics are twofold:
firstly, Initial Semantics provides a category--theoretic definition --- via a universal property --- 
          of the syntax and semantics freely generated by a signature.
Secondly, initiality yields an \emph{iteration operator} which allows for an economic and convenient specification
        of morphisms --- \emph{translations} --- from the initial object to other languages.

Depending on the ``richness'' of the language we want to define, we need a suitable notion of signature and, accordingly,
of representation of that signature. The language features we consider in this thesis are the following:

\begin{description}
 \item [Variable binding] We consider binding constructors on the \emph{term level}, such as lambda abstraction.
 \item [Typing] We consider \emph{simple} type systems, such as the simply--typed lambda calculus and, 
            via the Curry--Howard isomorphism, propositional logic (cf.\ \autoref{ex:logic_trans}).
 \item [Reduction] We consider semantics in form of reduction rules on terms, such as beta reduction,
                \[ \lambda x.M (N) \leadsto M [x:= N] \enspace . \]
\end{description}
For the integration of each of the features above, the notions of signature and representation 
have to be adapted to accommodate the increasing amount of information which must be given to uniquely
specify a language.

One of our goals is to use Initial Semantics in order to treat the last question 
of the preceding section:
we would like to translate from one language into another --- possibly over different sets of sorts ---, using
a universal, category--theoretic construction. This construction should take into account as much 
``structure'' as possible. By this we mean that the translations under consideration should
by construction be compatible, for instance, with typing and reduction in the source and target language.
A more in--depth description of those structures is given in
\autorefs{subsec:adding_variables}, \ref{subsec:adding_substitution}, \ref{subsec:adding_types} and
\ref{subsec:adding_semantics}. % 

In \autoref{subsec:informal_intro} we explain the notion of \emph{signature} and
\emph{representation} for a simple inductive data type, the natural numbers. The following sections
sketch the changes that have to be made in order to integrate variable binding,
substitution, typing and reduction rules, respectively.
In \autoref{sec:contribution} we summarize the contributions of this thesis,
whereas in \autoref{subsec:overview} we give a section--wise overview of its contents.

\subsection{Example: Peano Axioms}\label{subsec:informal_intro}

We introduce the notion of \emph{signature} and \emph{representation} 
using the example of the natural numbers; in line with the triple structure
mentioned at the beginning of \autoref{sec:initial_semantics},
our goal is to give a signature for the natural 
numbers and to associate to it a category of representations whose initial object is 
given by the natural numbers.

As a suitable \emph{signature}, consider the following map from a two elements set to natural numbers:
\[ \mathcal{N} := \{z \mapsto 0 \enspace ,\quad s \mapsto 1\} \enspace . \] 
Intuitively, it says that the natural numbers are built from two \emph{constructors}, 
namely a 0--ary operator (i.e.\ a constant), say, $z$, 
 --- the zero constant ---
and a unary operator, say, $s$ --- the successor function.

A \emph{representation} of the signature $\mathcal{N}$ is given by a triple $(X,Z,S)$ of a set $X$ 
together with a constant $Z \in X$ and a unary operation
$S:X\to X$. A morphism to another such triple $(X_0,Z_0, S_0)$ is a map $f : X \to X_0$
such that 
\begin{equation} \label{eq:iteration_nat}
 f(Z) = Z_0 \quad \text{ and } \quad \comp{S}{f} = \comp{f}{S_0}  \enspace . 
\end{equation}
This category has an \emph{initial object} $(\mathbb{N}, \ZERO, \SUCC)$ given by the natural numbers $\mathbb{N}$
equipped with the constant $\ZERO = 0$ and the successor function $\SUCC : \mathbb{N}\to\mathbb{N}$.

Initiality of $\mathbb{N}$ gives a way to specify \emph{iterative} functions \cite{vene}
from $\mathbb{N}$ to any set $X$ by equipping $X$ with a constant $Z\in X$ and 
a unary map $S:X\to X$,
i.e.\ making the set $X$ the carrier of an object $(X,Z,S)\in\mathcal{N}$. 
A different choice of $Z'\in X$ and $S':X\to X$ yields a different iterative map $\N \to X$.

Put differently, reading \autoref{eq:iteration_nat} dynamically rather than statically, i.e.\ as
a reduction from left to right rather than as equations, shows that functions on the initial object $\mathbb{N}$
can be defined by \emph{pattern matching}, where the right--hand side of the matching must obey a particular
form.

\begin{remark}[Digression on Natural Numbers Object]
 The very same definition is also used to define a \emph{natural numbers object} in any category $\C$ with a terminal object $\mathbf{1}$;
  just replace $Z\in X$ and $S : X \to X$ by morphisms $z : \mathbf{1} \to X$ and $s : X \to X$ in $\C$.
  More precisely, we call natural numbers object the triple $(\mathbb{N} : \C, \ZERO : \mathbf{1} \to \mathbb{N}, \SUCC : \mathbb{N} \to \mathbb{N})$
  if, for any triple $(X, z, s)$ of an object $X \in \C$ and morphisms $z$ and $s$ as above, there exists a unique morphism $f : \mathbb{N}\to X$ 
such that the following diagrams commute:
\[
 \begin{xy}
  \xymatrix@=3pc{ \mathbf{1} \ar[r]^{\ZERO} \ar [rd]_{z} & \mathbb{N}\ar[d]^{f} \\ {} & X
}
 \end{xy} \qquad  \qquad
\begin{xy}
 \xymatrix @=3pc{ \mathbb{N} \ar[r]^{\SUCC} \ar[d]_{f} & \mathbb{N}\ar[d]^{f} \\ X \ar[r]_{s} & X
}
\end{xy}
\]
For details we refer to Mac Lane and Moerdijk's book \cite{sheaves}.
 
\end{remark}

\subsection{Initial Algebras}

The term ``Initial Algebra'' is best explained using another viewpoint, where a signature is given by a 
\emph{signature functor $\Sigma : \Set\to \Set$}.
The category in question then is the category $\Sigma$--Alg of algebras of the functor $\Sigma$, that is, the category whose objects 
are pairs $(X, f)$ of a set $X$ and a map $f:\Sigma X \to X$. A morphism to another such algebra $(Y,g)$ is given by a map $h : X \to Y$ such 
that 
\[
 \begin{xy}
  \xymatrix{
      \Sigma X \ar[r]^{\Sigma h} \ar[d]_{f} & \Sigma Y \ar[d]^{g}\\
           X \ar[r]_{h} & Y
}
 \end{xy}
\]
commutes.
The example of \autoref{subsec:informal_intro} is equivalently given by the signature functor $\mathcal{N} : X \mapsto 1+X$,
with initial algebra
 \[    1 + \NN \xrightarrow{[\ZERO,\SUCC]} \NN  \enspace . \]
Another example is that of \emph{lists} (of finite length) of a given type $A$:
let $F(X) := 1 + A \times X$.
The initial $F$--algebra is given by the set $[A]$ of lists over $A$,
\[    1 + A\times [A] \xrightarrow{[\text{nil},\text{cons}]} [A]  \enspace . \]

\subsection{Adding Variable Binding} \label{subsec:adding_variables}

When passing to syntax \emph{with variable binding}, the question of how to model binding arises.
The following representations of binding are among the most frequently used:

\begin{packitem}
 \item %\label{list:nominal}  
       Nominal syntax using named abstraction ($\mathbb{A}$ being a set of \emph{atoms}), e.g.,
         \[  \lambda : [\mathbb{A}]  T \to T\]
 \item %\label{list:hoas} 
       Higher--Order Abstract Syntax (HOAS), e.g.,
          \[ \lambda : (T \to T) \to T\] and its \emph{weak} variant, e.g.,
          \[\lambda : (\mathbb{A} \to T) \to T\] 
 \item Nested Data Types as presented in \cite{BirdMeertens98:Nested}, e.g., %\label{list:brujin}
          \[ \lambda : T(X + 1) \to T(X) \]
\end{packitem}
Note that the encoding via nested data types differs conceptually from the others in that here the set $T$ of terms is
parametrized explicitly by a \emph{context}, i.e.\ a set $X$ of variables possibly appearing freely in the terms of $T(X)$.
Thus $T(X)$ denotes the set of terms of the language $T$ with free variables in the set $X$.
The set $X + 1$ corresponds to an \emph{extended context} with one additional free variable, which
is bound in the abstracted term. It is usually implemented through an inductive data type 
(\texttt{option} in \textsc{Ocaml} or the \texttt{Maybe} monad in \textsc{Haskell}) --- whence the term ``\emph{Nested}''.
It is also known under the name ``Heterogenous data type'' \cite{alt_reus}.

\begin{example}\label{ex:ulc_def}
We represent the untyped lambda calculus as a \emph{nested data type} as done, e.g., by Bird and Paterson \cite{DBLP:journals/jfp/BirdP99}:
 consider the following inductive type $\LC:\Set\to\Set$ of terms of the untyped lambda 
calculus%
\footnote{We use ``Set'' synonymously to ``Type''. Note however, that types behave differently from sets in some aspects.
          In particular, given two (propositionally) equal types $A = B$ and $a:A$, we do not have $a : B$.}%
:
\begin{lstlisting}
Inductive ULC (V : Type) : Type :=
  | Var : V -> ULC V
  | Abs : ULC (option V) -> ULC V
  | App : ULC V -> ULC V -> ULC V.
\end{lstlisting}
\end{example}

\noindent
For syntax with binding, arities need to carry information about the binding behaviour of their associated constructor.
One way to define such arities is using lists of natural numbers. 
The length of a list then indicates the number of arguments of
the constructor, 
and the $i$-th entry denotes the number 
of variables that the constructor binds in the $i$-th argument.
Continuing \autoref{ex:ulc_def}, the signature $\Lambda$ of $\LC$ is given by 
\[ \Lambda := \{ \app : [0,0]\enspace , \quad \abs : [1] \} \enspace . \]
The map $ V \mapsto \LC(V)$ is in fact functorial: given a map $f : V \to W$,
the map $\LC(f) : \LC(V) \to \LC(W)$ \emph{renames} any free variable $v\in V$ in a term by $f(v)$,
yielding a term with free variables in $W$. 
Accordingly, the signature $\Lambda$ should be represented in functors $F:\Set\to \Set$ instead of in sets,
and natural transformations take the place of maps.

\subsection{Adding Substitution}\label{subsec:adding_substitution}

As mentioned at the beginning of \autoref{sec:initial_semantics}, 
we would like to integrate as much structure as possible into our category of ``models''.
One such structure is \emph{(capture--avoiding) substitution} of free variables. 
To account for substitution, we consider not plain functors $F:\Set\to\Set$ as in the preceding paragraph, 
but instead \emph{monads} on the category $\Set$ of sets. 
Monads are functors
equipped with some extra structure, which we explain by the example of the untyped
lambda calculus.
The map $ V\mapsto \LC(V)$ comes with a (capture--avoiding) simultaneous substitution operation:
let $V$ and $W$ be two sets (of variables) and $f$ be a map $f : V\to \LC(W)$.
Given a lambda term $t\in \LC(V)$, we can replace each free variable $v\in V$
in $t$ by its image under $f$, yielding a term $t' \in \LC(W)$.
Furthermore we consider the constructor $\Var_V$ as a ``variable--as--term'' map, indexed by a set of variables $V$, 
\[\Var_V : V\to \LC(V) \enspace . \]
Altenkirch and Reus \cite{alt_reus} observed that the  well--known algebraic structure of monad
captures those two operations and their properties:
substitution and variable--as--term map 
turn $\LC$ into a monad (\autoref{def:endomonad}) on the category of sets.

The monad structure of $\LC$ should be compatible in a suitable sense with the constructors $\Abs$ and $\App$ of $\LC$:
\emph{substitution distributes over constructors}.
To capture this distributivity, 
Hirschowitz and Maggesi \cite{DBLP:conf/wollic/HirschowitzM07} consider \emph{modules over a monad} (cf.\ \autoref{def:module}) 
--- which generalize monadic substitution ---, and morphisms of modules --- which are natural transformations
that are compatible with the module substitution in a suitable sense.
Indeed, the maps 
\begin{align*}\LC &: V\mapsto \LC(V) \enspace , \\ 
              \LC' &: V\mapsto \LC(V+1) \text{ and } \\ 
              \LC\times\LC &: V \mapsto \LC(V)\times\LC(V) 
\end{align*} are
the underlying maps of such modules (cf.\ \autorefs{ex:ulc_taut_mod}, \ref{ex:ulc_const_mod}), and
the constructors $\Abs$ and $\App$ are morphisms of modules (cf.\  \autorefs{ex:ulc_mod_mor}, \ref{ex:ulc_mod_mor_kl}).

\subsection{Adding Types} \label{subsec:adding_types}

\emph{Type systems} exist with varying features, ranging from simply--typed
syntax to syntax with dependent types, kinds, polymorphism, etc.
By simply--typed syntax we mean a non--polymorphic syntax where the set of types is independent 
from the set of terms, i.e.\ type constructors only take types as arguments.
In more sophisticated type systems, types may depend on terms, leading to 
more complex definitions of arities and signatures.
The present work is only concerned with simply--typed languages, such as the simply--typed
$\lambda$--calculus and \PCF. We refer to the underlying set of types of a language as \emph{object types}
or \emph{sorts}.

The goal of typing is to \emph{classify} terms according to some criteria. As an example,
one may ask whether a term is of function type, that is, whether it would make sense to apply it to another term.
Once such a classification of terms is achieved, one can use typing information to \emph{filter} terms according 
to their types, in order to pick out only those terms that have the desired type.
The classification of terms through typing thus has a semantic flavour. However, we still subsume 
typing under the \emph{syntactic} aspect, since it has an impact on the set of terms of the language.

One way to add types would be to make them part of the terms, as in ``$\lambda x:\NN.x + 4$''.
However, for \emph{simple type systems} it is possible to separate the worlds of types and terms and 
consider typing as a map from terms to types, thus giving a simple mathematical structure to typing.
How can we be sure that our terms are well--typed? Despite the separation of types and terms we still want typing to be 
tightly integrated into the process of building terms, in order to avoid constructing ill--typed terms. 
Separation of terms and types seems to contradict this goal.
The answer lies in considering not \emph{one} set of terms with a ``typing map'' to the set, say, $T$, of types, but \emph{a family of sets}, indexed by 
the set $T$ of object types.
Term constructors then can be ``picky'' about what terms they take as arguments, accepting only those 
terms that have the suitable type. 
We also consider free variables to be equipped with an object type. 
Put differently, we do not consider terms over \emph{one} set of variables, 
but over a family of sets of variables, indexed by the set of object types. 
In other words, we consider a context to be given by a family $(V_t)_{t\in T}$ of sets of variables, where $V_t := V(t)$ is the set of
variables of object type $t$.
We illustrate our point of view by means of the example of the simply--typed lambda calculus $\SLC$: 

\begin{example}\label{ex:slc_def}
 Let 
\[\TLCTYPE ::= \enspace * \enspace \mid \enspace \TLCTYPE \TLCar \TLCTYPE\] 
be the set of types of the simply--typed lambda calculus. 
The set family of simply--typed lambda terms with free variables in $V$ is given by the following inductive family:
\begin{lstlisting}
Inductive TLC (V : T -> Type) : T -> Type :=
  | Var : forall t, V t -> TLC V t
  | Abs : forall s t TLC (V + s) t -> TLC V (s ~> t)
  | App : forall s t, TLC V (s ~> t) -> TLC V s -> TLC V t.
\end{lstlisting}
where $V + s := V + \lbrace {*s} \rbrace$ 
denotes context extension by a variable of type $s\in \TLCTYPE$ --- the variable which is
bound by the constructor $\Abs{(s,t)}$. % (cf.\ Sec.\ \ref{sec:deriv}). 
The variables $s$ and $t$ range over the set $\TLCTYPE$ of types.
The signature describing the simply--typed lambda calculus is given in \autorefs{ex:tlc_sig} and \ref{ex:tlc_sig_higher_order}. 
The preceding paragraph about monads and modules applies to the simply--typed lambda calculus when replacing sets by 
families of sets indexed by $\TLCTYPE$:
the simply--typed lambda calculus can be given the structure of a monad (cf.\ \autoref{ex:tlc_syntax_monadic})
\[ \SLC : \TS{\TLCTYPE} \to \TS{\TLCTYPE} \enspace . \]
The constructors of $\SLC$ are morphisms of modules (cf.\ \autorefs{ex:slc_modules}, \ref{ex:deriv_mod_slc}, \ref{ex:fibre_mod_slc}).

\end{example}
This method of defining exactly the well--typed terms by organizing them into a family of sets parametrized by 
 object types is 
called \emph{intrinsic typing} \cite{dep_syn} --- as opposed to the \emph{extrinsic typing}, where first a set of \emph{raw}
 terms is defined, which is then filtered via a typing predicate.
Intrinsic typing delegates object level typing to the meta language type system, such as the \textsf{Coq} type system in \autoref{ex:slc_def}. 
In this way, the meta level type
checker (e.g.\ \textsf{Coq}) sorts out ill--typed terms automatically: writing such a term yields a type error on the meta level.

Furthermore, the intrinsic encoding comes with a much more convenient recursion principle; a map to any other type system
can simply be defined by specifying its image on the well--typed terms. 
When using extrinsic typing, a map on terms would either have to be defined on the set of \emph{raw} terms, 
including ill--typed ones, or on just the well--typed terms by specifying an additional propositional argument
 expressing the welltypedness of the term argument.
Benton et al.\ give detailed explanation about intrinsic typing in a recently published paper \cite{dep_syn}.

\subsection{Adding Reductions}\label{subsec:adding_semantics}

The \emph{semantics} of a programming language describes how programmes of that language \emph{evaluate}.
For functional programming languages as considered in this thesis, 
evaluation --- or \emph{computation} --- is done by \emph{reduction}. 
As an example, 
the evaluation of the term $7 + 5$ of a hypothetical arithmetic programming language 
to its ``value'' $12$ is done by a series of reductions,
whose precise form depends on the semantics of the language in question.
Typical \emph{rules}, which specify how terms reduce, are given in \autoref{subsec:semantic_ulc_pcf}
for the example languages of the lambda calculus and $\PCF$.

Given a set $A$ of such reduction rules, one may consider the relation generated by these rules. More precisely,
following Barendregt and Barendsen \cite{barendregt_barendsen}, we consider several \emph{closures} of those rules:

\begin{description}
 \item [Propagation into subterms] A relation $R$ is called \emph{compatible} if it is closed under propagation into subterms, that is,
           if for any constructor $f$ of arity $n$ and any $i \leq n$,
           \[ M \leadsto_R N \Rightarrow   f(x_1,\ldots, x_{i_1}, M, x_{i+1}, \ldots, x_n) \leadsto_R f(x_1,\ldots, x_{i_1},N, x_{i+1}, \ldots, x_n) \enspace . \]
\item [Reduction] A relation $R$ is a \emph{reduction relation} if it is compatible, reflexive and transitive. 
\item [Equivalence] A relation $R$ is a \emph{congruence} if it is a compatible equivalence relation.
\end{description}
To the set $A$ of rules we associate three relations generated by $A$, which are the smallest relations that contain $A$ and are a compatible relation, 
a reduction relation and a congruence, 
respectively. We denote these relations, in this order, by 
$\to_A$, $\twoheadrightarrow_A$ and $=_A$, respectively.

\begin{remark}[Digression on Reduction Strategies]\label{dig:on_reduction_strategies}

Suppose we have a term in which reduction rules are applicable in several places,
such as in the term
        \[  ((\lambda x. M)N) ((\lambda y. M') N') \enspace , \]
which is $\beta$--reducible in the operator and in the operand.
Here the natural question arises where one should reduce at first, 
in the operator or in the operand (or both in parallel) --- the question about the \emph{reduction strategy}.
More precisely, one considers the following two properties of rewrite systems:
\begin{packdesc}
 \item [Termination] Are there infinite --- non--terminating --- chains of reductions?
 \item [Confluence] Suppose a term $t$ reduces both to $t'$ as well as to $t''$ via two different reductions.
                 Is there a term $t'''$ such that both $t'$ and $t''$ reduce to $t'''$?
\end{packdesc}
Termination and confluence together yield (strong) normalization, an important pro\-perty of rewriting systems:
in a strongly normalizing rewriting system, any reduction strategy yields the same value for a given 
term --- in particular, any reduction strategy arrives at a value, i.e.\ at a term without any more reducible subterms.
To illustrate the concept of termination, we give an example of a lambda term such that one reduction strategy
terminates whereas another one does not;
consider the term $(\lambda x . y) (\Omega \Omega)$ with $\Omega = (\lambda x. x x)$ and a free variable $y$.
Reducing the outermost beta redex results in an irreducible term $y$ in one step, 
whereas the strategy of reducing at first the operand $(\Omega\Omega)$
leads to an infinite chain of reductions.

\end{remark}

In this thesis we are interested in the \emph{reduction relation} generated by a set of rules.
It differs from the congruence by the absence of a \emph{symmetry} rule, which, while adequate for 
mathematical reasoning, yields a relation that is too coarse from a point of view of \emph{computation}.
In the words of Girard \cite{proofs_and_types}, while the congruence generated by $A$ emphasizes the \emph{static} point of view of 
mathematics, the reduction relation associated to $A$ emphasizes the \emph{dynamic} point of view of computation.

To account for reductions, we consider functors and monads whose codomain is not the category of (families of) sets,
but of (families of) \emph{preordered sets}.
The definition of monad requires the underlying functor to be an endofunctor, but we do not want
to consider preordered contexts --- what would be the meaning of this preorder?
The restriction to endofunctors was lifted by Altenkirch et al.\ \cite{DBLP:conf/fossacs/AltenkirchCU10} 
through the introduction of \emph{relative} monads.
A relative monad is given by a functor --- not necessarily endo --- together
with two operations very similar to monadic variables--as--terms and substitution.
We thus consider, e.g., the lambda calculus, as a relative monad associating to any \emph{set} $X$
of variables a \emph{preordered set of lambda terms} $(\LC(X),\twoheadrightarrow_{\beta})$, 
where the preorder on $\LC(X)$ is given by the reduction relation $\twoheadrightarrow_{\beta}$ generated
by the beta rule of \autoref{eq:beta}, cf.\ \autoref{ex:ulcbeta}.

%% file: contribution.tex
\section{Contributions}\label{sec:contribution}

In this thesis
we give, via a universal property, an algebraic characterization of simply--typed syntax
equipped with semantics
in form of reduction rules.
More precisely, given a pair of a \emph{signature} --- specifying the types and terms of a language --- and
\emph{inequations} over this signature --- specifying reduction rules ---, 
we characterize the terms of the language associated to this signature, equipped with 
reduction rules according to the given inequations, as the initial object of a category of ``models''.

Our starting point is work on initiality for untyped syntax done by Hirschowitz and Maggesi \cite{DBLP:conf/wollic/HirschowitzM07},
and on its generalization to simply--typed syntax by Zsid\'o \cite{ju_phd}.
In a first step we extend Zsid\'o's theorem \cite[Chap.~6]{ju_phd} to account for varying sorts, cf.\ \autoref{sec:contrib_comp}.
Afterwards, we integrate reduction rules into Hirschowitz and Maggesi's \cite{DBLP:conf/wollic/HirschowitzM07} 
purely syntactic initiality result, cf.\ \autoref{sec:contrib_sem}.
Finally we obtain our main theorem, which accounts for varying object types as well as reduction rules, by combining 
the aforementioned two results, cf.\ \autoref{sec:contrib_comp_sem}.

Furthermore, for the untyped case (cf.\ \autoref{sec:contrib_sem}), 
we provide a formalized proof in the proof assistant \textsf{Coq} of our result,
yielding a machinery which, when fed with a signature for terms and a set of inequations, 
produces the abstract syntax associated to the signature, together with the reduction relation generated 
by the given inequations.
For the simply--typed case, we formalize the instantiation of our main result (cf.\ \autoref{sec:contrib_comp_sem}) 
to the signature of the programming language \PCF~\cite{Plotkin1977223}.

We now explain our contributions and approaches in more detail:

\subsection{Extended Initiality for Varying Sorts} 
\label{sec:contrib_comp}

In her PhD thesis \cite[Chap.~6]{ju_phd}, Zsid\'o proves an initiality theorem for the abstract syntax
associated to a simply--typed signature.
However, the ``models'' (or representations) she considers, among which the abstract syntax is the initial one, are all models over 
the same set of sorts.
In this way, the iteration principle obtained by initiality does not allow the specification of a
translation to a term language over a different set of sorts.
We adapt 
Zsid\'o's theorem  by introducing \emph{typed signatures}.
A typed signature $(S,\Sigma)$ specifies a set of \emph{sorts} via an algebraic signature $S$, as well as a set of simply--typed \emph{terms} 
over these sorts via a term signature $\Sigma$ over $S$. 
A representation $R$ of such a typed signature is then given by a representation of its signature $S$ for sorts in some 
set $T = T_R$ as well as a representation of $\Sigma$ in a monad --- also called $R$ --- over the category $\TS{T}$.
A morphism of representations $P \to R$ consists of a morphism $f$ of the underlying representations of $S$,
together with a morphism of representations of $\Sigma$, that is compatible in a suitable sense with the ``translation of sorts'' $f$.
We show that the category of representations of $(S,\Sigma)$ thus defined has an initial object,
which integrates the sorts freely generated by $S$ and the terms freely generated by $\Sigma$, typed over the 
sorts of $S$. 
Our definition of morphisms ensures that, for any translation specified via the iteration principle, 
the translation of terms is compatible with the translation of sorts with respect to the typing in 
the source and target languages.

To summarize, compared to Zsid\'o's theorem \cite[Chap.~6]{ju_phd} we
consider representations of a signature for terms over \emph{varying} sets of sorts.
However, since we specify the set of sorts via a signature $S$ and thus implement the variation of sorts through
morphisms of representations of $S$, our ``initial set of sorts'' necessarily has \emph{inductive}
structure.

\subsection{Integrating Reduction Rules} \label{sec:contrib_sem}

In order to integrate reduction rules into our initiality results, 
we define a notion of \emph{2--signature}.
A 2--signature $(\Sigma,A)$ is given by a (1--)signature $\Sigma$ which specifies the terms of a language,
and a set $A$ of \emph{inequations} over $\Sigma$. 
Intuitively, each inequation specifies a reduction rule, for instance the beta rule.

The \emph{models} --- or \emph{representations} --- of such a 2--signature are built from \emph{relative monads} and 
\emph{modules over relative monads}:
given a 1--signature $\Sigma$,
we define a representation of $\Sigma$ to be given by a relative monad on the 
appropriate functor $\Delta : \Set \to \PO$ (cf.\ \autoref{def:delta})
together with a suitable morphism of modules (over relative monads) 
for each arity of $\Sigma$.
Given a set $A$ of inequations over $\Sigma$, we define a satisfaction predicate for the models of $\Sigma$;
we call \emph{representation of $(\Sigma,A)$} each representation of $\Sigma$ that satisfies each inequation of $A$.
This predicate specifies a full subcategory of the category of representations of $\Sigma$.
We call this subcategory the 
\emph{category of representations of $(\Sigma,A)$}.
We prove that this category has an initial object, which is built by equipping the initial representation of $\Sigma$
--- given by the terms freely generated by $\Sigma$ ---
with a suitable reduction relation generated by the inequations of $A$.

With this initiality theorem for $(\Sigma,A)$ we obtain a new iteration principle, and any translation specified
via this principle is, by construction, compatible with the reduction relation in the source and target languages.

\subsection{Main Theorem: Initiality for Simply--Typed Syntax with Reduction}\label{sec:contrib_comp_sem}

Finally, we combine the above two theorems in order to obtain an initiality result which accounts for 
the motivating example of \autoref{sec:trans_pcf_ulc}.
More precisely, we define a \emph{2--signature} to be given by a typed signature $(S,\Sigma)$ as in \autoref{sec:contrib_comp}
together a set $A$ of $(S,\Sigma)$--inequations analogous to \autoref{sec:contrib_sem}, specifying reduction rules.

We define a category of representations of $((S,\Sigma),A)$ and prove that this category has an initial object.
This initial representation integrates the types and terms freely generated by $(S,\Sigma)$,
the terms being equipped with the reduction relation generated by the inequations of $A$.

\subsection{A Computer Implementation for Specifying Syntax and Semantics}

Above theorems are really meant to be implemented in a proof assistant.
Such an implementation allows the specification of syntax and reduction rules via 2--signatures,
yielding a highly automated mechanism to produce syntax together with certified substitution
and iteration principle.

We prove the initiality theorem described in \autoref{sec:contrib_sem} in the proof assistant \textsf{Coq} \cite{coq}.
As an illustration we describe how to obtain the untyped lambda calculus with beta reduction via initiality.

Furthermore we formalize an instance of the theorem explained in \autoref{sec:contrib_comp_sem}, also in \textsf{Coq}.
More precisely, we define the category of representations of the typed signature of $\PCF$ with inequations and prove that
this category has an initial object. Afterwards, we give a representation of this signature in the relative monad 
$\ULCB$ of the untyped lambda calculus with beta reduction, yielding a translation from \PCF~to $\LC$.
Instructions on how to obtain the complete source code of our \textsf{Coq} library are available on 
\begin{center}
 \url{http://math.unice.fr/laboratoire/logiciels}.
\end{center}

%% file: summary.tex
\section{Synopsis}
\label{subsec:overview}

This thesis consists of two parts: \Autoref{part:theory} (\autorefs{chap:cat} to \ref{chap:comp_types_sem}) 
describes and proves informally the theorems which constitute this thesis,
whereas \Autoref{part:impl} (\autorefs{chap:cats_in_coq} to \ref{chap:comp_sem_formal}) describes
their implementation and verification in the proof assistant \textsf{Coq} \cite{coq}.

\begin{description}\setkomafont{descriptionlabel}{\normalfont\bfseries}

 \item [\Compref{chap:cat}] 
          We recall the notions of monad and module over a monad, together with some important constructions of modules.

Afterwards we state equivalent definitions of monads, modules and their morphisms in the style of Manes, 
emphasizing their \emph{substitution} structure.

Then we recall Altenkirch et al.'s  definition of \emph{relative monads} and define suitable morphisms for such monads.

Finally we define \emph{modules over relative monads} and show that the constructions of modules over monads
carry over to modules over relative monads.

 \item [\Compref{sec:compilation}]
     We present two initiality theorems for simple type systems:

In  \autoref{sec:sts_ju} we present Zsid\'o's initiality theorem \cite[Chap.\ 6]{ju_phd}:
it characterizes the syntax associated to a simply--typed signature $S$ over a set $T$ of object types
as the initial object in a category of representations of $S$.

In \autoref{sec:ext_zsido} we prove a variant of Zsid\'o's theorem which allows 
for representations of a term signature over varying sets of sorts.
We introduce the notion of \emph{typed signature} in order to account for translations of sorts.
A typed signature $(S, \Sigma)$ is a pair consisting of a first--order \emph{algebraic signature} $S$ for sorts, 
and a higher--order signature $\Sigma$ for terms over those sorts.
A representation of a typed signature $(S,\Sigma)$ is again a pair given by 
a representation of the sort signature $S$ in a set $T$ 
and a representation of the term signature $\Sigma$ in a monad $P$ over the category $\TS{T}$.
We show that the category of representations of a typed signature has an initial object.

Finally, as an example, we use the iteration principle stemming from initiality in order to specify a \emph{double negation
translation} from classical to intuitionistic propositional logic, viewing propositions as types via the
Curry--Howard isomorphism.

 \item [\Compref{sec:prop_arities}]

 We prove an initiality theorem for untyped languages with variable binding, equipped with reduction rules.

For the specification of such languages,  
we define a notion of \emph{2--signature}, i.e.\ a signature consisting of two levels: 
a \emph{syntactic} level --- called \emph{1}--signature ---, 
which specifies the terms of the language, 
and a \emph{semantic} level, which specifies reduction rules for those terms through \emph{inequations}.
A representation of such a \emph{2--signature} $(\Sigma,A)$ is any representation of the underlying 1--signature $\Sigma$ which 
satisfies each inequation of $A$.

We define the category of representations of $(\Sigma,A)$ as the full subcategory 
of representations of $\Sigma$ whose objects satisfy the inequations of $A$.
We prove that this subcategory has an initial object, integrating the terms 
generated by $\Sigma$ and the reduction relation generated by the rules of $A$.

As a running example we consider the 2--signature of the untyped lambda calculus with beta reduction.

The implementation of the theorem in \textsf{Coq} is explained in \autoref{chap:2--signatures_formal}.

 \item [\Compref{chap:comp_types_sem}]
    We prove the main result of this thesis: 
 we generalize the initiality result from the preceding \autoref{sec:prop_arities} 
   to \emph{simply--typed} syntax with reduction rules,
        in a way that allows for change of object types as in 
\autoref{sec:ext_zsido}. %

    More precisely,
    we generalize the definition of \emph{2--signature} to allow for the underlying 1--signature to specify 
     a simple type system as in \autoref{sec:ext_zsido}. Accordingly, the definition of inequation is 
    extended to allow for the specification of reduction rules on such simple type systems.
    The main theorem of this chapter states that the category of representations of such a 2--signature has an initial object.
   This initial representation integrates the types and terms specified by the underlying 1--signature, and is equipped with the reduction relation
    generated by the inequations of the 2--signature. 

  \item [\Compref{chap:cats_in_coq}]
         This chapter serves as an introduction to the proof assistant \textsf{Coq} in general and our library of 
         category theory used in the following chapters in particular.
         We describe the formalization of basic concepts such as categories, (relative) monads and modules over (relative) monads.
         In the course of the chapter we also describe some of the features of \textsf{Coq} that we use, such as \emph{implicit} arguments,
        the \lstinline!Program! framework and coercions.

  \item [\Compref{chap:sts_formal}]
	Building up on the library presented in \autoref{chap:cats_in_coq}, we describe the formalization
         of Zsid\'o's initiality theorem from \autoref{sec:sts_ju} in \textsf{Coq}.
        At first we define a \textsf{Coq} data type of simply--typed signatures over a given object type $T$. 
        Afterwards we 
       associate a category of representations to any such signature and prove that this category has 
       an initial object.

  \item [\Compref{chap:2--signatures_formal}] 
        We describe the implementation in \textsf{Coq} of the theorem proved informally in \autoref{sec:prop_arities}:
        the category of representations of a 2--signature has an initial object.
        The formal proof follows the informal proof very closely; the only noteworthy difference is 
        that the initial object of the underlying 1--signature is constructed directly rather than 
        through the adjunction proved in \autoref{sec:prop_arities}.

       Finally we demonstrate how to specify the untyped lambda calculus with beta reduction through
       a 2--signature in our implementation.

  \item [\Compref{chap:comp_sem_formal}]	
        We formalize in \textsf{Coq} an instance of the main theorem of the thesis (cf.\ \autoref{chap:comp_types_sem}), 
         for the 2--signature of \PCF, equipped with reduction rules as presented in \autoref{eq:pcf_reductions}.
	In particular, we explain where we encounter difficulties when using intrinsic typing in an intensional
         type system.
       
        By representing the signature of \PCF~in the monad of the untyped lambda calculus, we obtain a translation
        from $\PCF$ to $\ULC$ that is compatible with reductions in the source and target languages.

\end{description}

%% file: rel_work.tex
\section{Related Work} \label{sec:rel_work}

In this section we review related work, in particular in the field of 
\emph{Initial Semantics} (cf.\ \autoref{subsec:rel_work_initial}), i.e.\ algebraic 
characterization of syntax (and their semantics) and in the field of formalization of syntax in proof 
assistants, cf.\ \autoref{subsec:rel_work_formal}.

\subsection{Translations from \texorpdfstring{\PCF}{PCF}}
Our main example is given by the programming language \PCF, introduced by Plotkin \cite{Plotkin1977223}.
This language and its various semantics have been studied extensively. 
The following work is not concerned with algebraic characterization of programming languages, and thus not
directly related to this thesis; it rather answers questions that we do not (yet) consider in 
 our categorical setting:

Phoa \cite{DBLP:conf/lics/Phoa93} studies the semantic aspect of a specific translation of \PCF~to the untyped
lambda calculus, i.e.\ the behaviour of this translation and its compatibility with respect to reduction in the source and target language.
The translation he considers is also the one we specify via initiality in \autoref{chap:comp_sem_formal}.
The main result of this work is that this translation is \emph{adequate} in the sense that a \PCF~programme 
reduces to a natural number constant $n$ of \PCF~if and only if its translation into the lambda calculus 
reduces to the corresponding 
church numeral $c_n$.

Riecke \cite{DBLP:journals/mscs/Riecke93} studies translations from \PCF~into itself, where source and target 
are equipped with different \emph{reduction strategies} (cf.\ \autoref{dig:on_reduction_strategies}).
We do not consider reduction strategies in this thesis.

\subsection{Initial Semantics}\label{subsec:rel_work_initial}

We classify work in Initial Semantics according to the features it covers.
We are interested, in no particular order, in the following features:

\begin{packitem}
 \item Typing
 \item Variable binding
 \item Semantics through (in)equations
\end{packitem}

%\subsubsection{Syntactic}

\noindent
Initial Semantics for untyped syntax without variable binding is a result by Birkhoff \cite{birkhoff1935}.
Goguen et al.\ \cite{gtww} give an overview of the literature about initial algebra and 
spell out explicitly the connection between initial algebras and abstract syntax.
In fact, Goguen et al.\ also treat the example of a programming language with variable binding, 
which they call ``\emph{S}imple \emph{A}pplicative \emph{L}anguage'' (SAL).
However, they circumvent the algebraic treatment of variable binding by modelling binding
through a family of unary constructors $\abs_x : exp \to exp$ where $x$ varies over a fixed set of variables.

\subsubsection{Variable binding}

When looking for an algebraic treatment of variable binding,
the question of how to model binding arises.
Some possible encodings have already been mentioned in \autoref{subsec:adding_variables},
we repeat the list --- in no particular order --- for reasons of convenience:
 
\begin{packenum}
  \item  \label{list:nominal} 
        Nominal syntax using \emph{atom} abstraction:
           \[  \lambda : [\mathbb{A}] T \to T\]
  \item \label{list:hoas} 
      Higher--Order Abstract Syntax (HOAS):
           \[ \lambda : (T \to T) \to T\] and its \emph{weak} variant:
           \[ \lambda : (\mathbb{A} \to T) \to T\] 
  \item \label{list:brujin}
        Nested Data Types: %
           \[ \lambda: T(X + 1) \to T(X) \]
 \end{packenum}
In the following, the numbers in parentheses indicate the technique used for modelling variable binding
in the respective work, according to the list given above.
Initial Semantics for untyped syntax was presented by  
Gabbay and Pitts \cite[(\ref{list:nominal})]{gabbay_pitts99},
Hofmann \cite[(\ref{list:hoas})]{hofmann}, Fiore et al.\ \cite[(\ref{list:brujin})]{fpt} and
Hirschowitz and Maggesi \cite[(\ref{list:brujin})]{DBLP:conf/wollic/HirschowitzM07}. 

While Gabbay and Pitts work in a \emph{set theory} enriched with \emph{atoms} --- which serve as object level variables ---,
Hofmann, Fiore et al.\ and Hirschowitz and Maggesi use \emph{category--theoretic} notions to formalize syntax. 
The nominal approach initiated by Gabbay and Pitts is the only one among those mentioned that allows for a study 
of alpha conversion. For all others the notion of alpha convertibility and syntactic equality coincide.

Fiore et al.'s approach is based on the notion of signature functor and $\Sigma$--monoid, where the central concept
of substitution is expressed in terms of strengths.
Hirschowitz and Maggesi model substitution through monads, following Altenkirch and Reus' (cf.\ \cite{alt_reus}) characterization
of the untyped lambda calculus as a monad on the category of sets.
The connection between those two approaches is made precise in Zsid\'o's PhD thesis \cite{ju_phd} in 
form of adjunctions between the respective categories of models.

Later Gabbay and Hofmann \cite{DBLP:conf/lpar/GabbayH08} exhibit the relation between nominal techniques and presheaves,
showing that through the nominal approach one considers in fact presheaves $F$ that 
preserve pullbacks of monomorphisms, i.e. %\ --- using less fancy vocabulary --- 
presheaves that are stable under intersection, $F(X \cap Y) = FX \cap FY$.

Fiore et al.'s approach was extended by Fiore \cite{fio02} to the simply--typed lambda calculus,
and for general simply--typed syntax by Miculan and Scagnetto \cite[(\ref{list:hoas})]{DBLP:conf/ppdp/MiculanS03}.
Both use an encoding of binding via nested data types. 
The relation to Higher--Order Abstract Syntax --- as ``terms with holes'' --- 
is made precise in the latter work \cite[Proposition 1]{DBLP:conf/ppdp/MiculanS03}.
Hirschowitz and Maggesi's approach was generalized to simply--typed syntax in Zsid\'o's thesis \cite{ju_phd}.
It was also generalized to account for more general term formers such as explicit flattening 
$\mu : \comp{T}{T}\to T$ \cite{hirschowitz_maggesi_fics2012}.

Some of the mentioned lines of work have been extended to integrate 
\emph{semantic aspects} in form of reduction relations on terms into initiality results:

\subsubsection{Incorporating Semantics}

Ghani and L\"uth \cite{DBLP:journals/njc/GhaniL03} 
present rewriting for algebraic theories without variable binding;
they characterize 
equational theories (with a \emph{symmetry} rule) resp.\ rewrite systems (with \emph{reflexivity} and \emph{transitivity} rule, but without \emph{symmetry})
  as \emph{coequalizers} resp.\ \emph{coinserters} in a category of monads on
the categories $\Set$ resp.\ $\PO$.

Fiore and Hur \cite{DBLP:conf/icalp/FioreH07} have extended Fiore's work to ``second--order universal algebras'',
thus integrating semantic aspects in form of \emph{equations} into initiality results.
In particular, Hur's thesis \cite{hur_phd} is dedicated to \emph{equational systems} for syntax with variable binding.
In a ``Further research'' section \cite[Chap.\ 9.3]{hur_phd}, Hur suggests the use of preorders, or more generally, 
arbitrary relations to model \emph{in}equational systems.

Hirschowitz and Maggesi \cite{DBLP:conf/wollic/HirschowitzM07} prove initiality of the set of lambda terms modulo beta and eta conversion
in a category of \emph{exponential monads}.
In an unpublished paper \cite{journals/corr/abs-0704-2900} they introduce the notion of \emph{half--equation} and 
\emph{equation} --- as a pair of parallel half--equations --- that we adopt in this thesis. However, we reinterpret
a pair of parallel half--equations as an \emph{in}equation rather than as an equation.
Accordingly, we use preorders to model semantic aspects of syntax.
This emphasizes the \emph{dynamic} viewpoint of reductions as \emph{directed} equalities --- or \emph{rewrite rules} --- rather than
the \emph{static}, mathematical viewpoint one obtains by considering symmetric relations.

However, we consider not (traditional) monads but instead \emph{relative} monads 
--- on the appropriate functor $\Delta: \Set \to \PO$ (cf.\ \autoref{def:delta}) ---
as defined 
by Altenkirch et al.\ \cite{DBLP:conf/fossacs/AltenkirchCU10}, 
that is, monads with different source and target categories:
we consider \emph{variables} as elements of unstructured sets, whereas the set of \emph{terms} of a language carries
structure in form of a reduction relation. In our approach variables and terms thus live in \emph{different} 
categories, which is realized mathematically through the use of relative monads instead of regular monads.

T.\ Hirschowitz \cite{HIRSCHOWITZ:2010:HAL-00540205:2} defines a category \textsf{Sig} 
of 2--signatures for \emph{simply--typed} syntax with reductions, and constructs an adjunction
between \textsf{Sig} and the category $\mathsf{2CCCat}$
of small cartesian closed 2--categories. He thus associates, to any 2--signature, a 2--category of 
types and terms satisfying a universal property. 
His approach differs from ours in the way in which variable binding is modelled:
Hirschowitz encodes binding in a Higher--Order Abstract Syntax (HOAS) style through exponentials.
 Reduction relations are expressed by the existence of 2--cells.

\subsection{Formalization of Syntax}\label{subsec:rel_work_formal}

The implementation and formalization of syntax has been studied by a variety of people. 
The \popl challenge \cite{poplmark} is a benchmark which aims to evaluate readability and provability when using different techniques of variable binding. 
However, the benchmark only concerns \emph{one specific} language, not arbitrary syntax specified by a signature.
The technique we use, called \emph{Nested Abstract Syntax}, is used in a partial solution by Hirschowitz and Maggesi \cite{NAS}, 
but was proposed earlier by others, see e.g.\ \cite{BirdMeertens98:Nested, alt_reus}. 
The use of \emph{intrinsic typing} by dependent types of the meta--language was advertised in \cite{dep_syn}.

During our work we became aware of Capretta and Felty's framework for reasoning about programming languages \cite{HOAU}.
They implement a tool --- also in the \textsf{Coq} proof assistant --- which, given a signature, provides 
the associated abstract syntax as a data type dependent on the object types, hence intrinsically typed as well. 
Their data type of terms does not, however, depend on the set of free variables of those terms. Variables are encoded with de Bruijn indices.
There are two different constructors for free and bound variables which serve to control the binding behaviour of object level constructors.
In our theorem, there is only one constructor for (free) variables, and binding a variable is done by
removing it from the set of free variables. 
Capretta and Felty then add a layer to translate those terms into syntax using named abstraction, and provide suitable induction and recursion principles. 
However, they do not consider \emph{semantic} aspects, such as reduction rules, in their work.

The tool \textsf{Ott} \cite{DBLP:journals/jfp/SewellNOPRSS10} allows the specification of syntax and reduction rules, even for 
polymorphic type systems, in a system--independent ASCII file with subsequent translation into
several different formal systems, including \textsf{Coq}, \textsf{Isabelle} \cite{DBLP:conf/cade/Paulson88} and others. 
However, no algebraic characterization of the produced syntax is given.

%% file: publications.tex
\subsection{Published Work}

This thesis is partly based on the following articles:

\begin{description}\setkomafont{descriptionlabel}{\normalfont}
 \item [\textit{Initial Semantics for higher--order typed syntax in Coq} (with J.\ Zsid\'o) \cite{ahrens_zsido}] ~ \\
        The content of this article corresponds to the contents of \autoref{sec:sts_ju} and \autoref{chap:sts_formal}.
 \item [\textit{Extended Initiality for Typed Abstract Syntax} \cite{ahrens_ext_init}] ~ \\
        The content of this article corresponds to the contents of \autoref{sec:ext_zsido} and \autoref{ex:logic_trans}.
 \item [\textit{Modules over relative monads for syntax and semantics} \cite{ahrens_relmonads}] ~ \\
        The content of this article corresponds to the contents of \autoref{sec:prop_arities} and \autoref{chap:2--signatures_formal}.
\end{description}

%% file: monad_colax.tex
In this chapter, we first present some basic category--theoretic definitions (cf.\ \autoref{sec:categories_functors}).
Afterwards, we review two different definitions of monads and modules over monads (cf.\ \autorefs{mon_mod} and \ref{sec:alternative_defs_monad}).
Finally, we present relative monads and define colax morphisms of relative monads as well as modules over relative monads (cf.\ \autoref{sec:rel_monads}).

\section{Categories, Functors \& Transformations}\label{sec:categories_functors}

In order to fix notations, we state some basic definitions of category theory,
in particular those of category, functor and natural transformation.
The examples we give in this section are used in later chapters.
The reader might want to skip this section --- throughout the thesis we link back to the definitions and examples
where necessary.

The present section is not meant to constitute an introduction to category theory,
nor does it define \emph{all} of the concepts we use in the course of this work.
For both an introduction to category theory as well as a reference for notions whose definitions are
not given in this thesis, we refer to Mac Lane's book \cite{maclane}.

\subsection{Two Definitions of Categories}

\begin{definition}[Category, \autoref{lst:cat}]\label{def:category}
A \emph{category} $\C$ is given by
\begin{packitem}
 \item a class --- which we will also call $\C$ --- of \emph{objects},
 \item for any two objects $c$ and $d$ of $\C$, a class of \emph{morphisms}, written $\C(c,d)$,
 \item for any object $c$ of $\C$, a morphism $\id_c \in \C(c,c)$ and
 \item for any three objects $c, d, e$ of $\C$, a \emph{composition} operation
         \[ (\comp{\_}{\_})_{c,d,e} : \C(d,e) \times \C(c,d) \to \C(c,e) \]
\end{packitem}
such that the composition is associative and the morphisms of the form $\id_c$ for suitable objects $c$
are left and right neutral with respect to this composition%
\footnote{We omit the ``object'' parameters from the composition operation, since those are deducible from 
  the morphisms we compose. This omission is done in our library as well, via \emph{implicit arguments} (cf.\ \autoref{sec:alg_structure}).}:
\begin{align*}
  &\forall a~b~c~d : \C, \forall f : \C(a,b), g : \C(b,c), h : \C(d,e), \enspace \comp{f}{(\comp{g}{h})} = \comp{(\comp{f}{g})}{h} \\
  &\forall c~d : \C, \forall f : \C(c,d), \enspace \comp{f}{\id_d} = f \text{ and } \comp{\id_c}{f} = f \enspace .
\end{align*}
  We also write $f:c\to d$ for a morphism $f\in\C(c,d)$.
\end{definition}

\begin{remark}
 We omit a fifth condition stating that the classes of morphisms are pointwise disjoint. This condition
 is automatically satisfied when implementing the morphisms of a category as a dependent type
 of an intensional type theory, which we do in \autoref{chap:cats_in_coq}.
\end{remark}

\begin{remark}[(Equivalent Def.\ of Category)]\label{rem:def_cat_alt}
 Equivalently to \autoref{def:category}, a category $\C$ is given by
  \begin{itemize}
   \item a class $\C_0$ of objects and a class $\C_1$ of morphisms,
   \item two maps denoting the source and target object of any morphism,
       \[ \src, \tgt : \C_1 \to \C_0 \enspace , \]
   \item a partially defined composition function 
         \[ (\comp{\_}{\_}) : \C_1 \times \C_1 \to \C_1 \enspace , \]
     such that $\comp{f}{g}$ is defined only for \emph{composable morphisms} $f$ and $g$, 
     i.e.\ for morphisms $f$ and $g$ such that $\tgt(f) = \src(g)$ --- in which case we require that
            $\src(\comp{f}{g}) = \src(f)$ and $\tgt(\comp{f}{g}) = \tgt(g)$ ---,
   \item an identity morphism for each object, i.e.\ a map
          \[ \id : \C_0 \to \C_1 \enspace , \]
          such that $\src(\id(c)) = \tgt(\id(c)) = c$ 
 and 
    \item properties analogous to those of the preceding definition. 
       The associative law, e.g., reads as
       \[ \forall f~g~h : \C_1, \enspace \tgt(f) = \src(g) \Longrightarrow \tgt(g) = \src(f) \Longrightarrow
                          \comp{f}{(\comp{g}{h})} = \comp{(\comp{f}{g})}{h} \enspace .
       \]
  \end{itemize}
\end{remark}

\noindent
 While the two definitions of categories of \autoref{def:category} and of \autoref{rem:def_cat_alt} are equivalent,
  they both have some advantages and inconveniences when implementing them in a dependent type theory such as \textsf{Coq}.
We expand on these differences in \autoref{subsec:dep_hom_types}.

\begin{definition}\label{def:SET}
  The category $\Set$ has sets as objects. Morphisms from a set $A$ to a set $B$
  are the total maps from $A$ to $B$, together with 
  the usual composition of maps.

\end{definition}

Given a category $\C$, a morphism $f : c \to d$ from object $c$ to object $d$ is called \emph{invertible},
if there exists a left-- and right--inverse $g : d \to c$, that is, a morphism $g : d \to c$
such that $\comp{f}{g} = id_c$ and $\comp{g}{f} = id_d$.
In this case the objects $c$ and $d$ are called \emph{isomorphic}.

The following \emph{universal property} plays a central r\^ole in this thesis:
\begin{definition} \label{def:init_object}
  Let $\C$ be a category. The object $c$ of $\C$ is called \emph{initial} if there exists 
  precisely \emph{one} morphism $i_d : c \to d$ in $\C$ \emph{to} any object $d$ of $\C$.
\end{definition}

Any two initial objects of a category $\C$ are canonically isomorphic. We usually do not
distinguish canonically isomorphic objects of a category, which
explains  the (standard) use of the definite article.
Whenever it exists, we also write $0_\C$ --- or simply $0$, when the category in question can be deduced from the 
context --- for the initial object of $\C$.
The dual concept is that of a \emph{terminal} object:
\begin{definition} \label{def:term_object}
  Let $\C$ be a category. The object $d$ of $\C$ is called \emph{terminal} if there exists 
  precisely \emph{one} morphism $t_c : c \to d$ in $\C$ \emph{from} any object $c$ of $\C$.
\end{definition}

\begin{example}
 The empty set is initial in the category $\Set$ of sets. The singleton set is terminal in $\Set$.
\end{example}

Later we also use the following categories:

\begin{definition}\label{def:cat_PO}
 The category $\PO$ of preorders has, as objects, sets equipped with a preorder, and, as morphisms between any two 
 preorders $A$ and $B$, the monotone functions from $A$ to $B$.
\end{definition}

\begin{definition}\label{def:cat_wPO}
 The category $\PS$ has, as objects, sets equipped with a preorder, and, as morphisms between any two preordered 
  sets $A$ and $B$, all set--theoretic maps from $A$ to $B$, not necessarily monotone.
\end{definition}

\begin{example}\label{ex:cat_discrete}
  Any set $T$ can be regarded as a \emph{discrete} category, with objects the elements of $T$, and just identity morphisms.
\end{example}

\begin{notation}[Product, Coproduct]
 We refer to Mac Lane's book \cite{maclane} for the definition of product and coproduct.
 Whenever they exist, we write $a \times b$ for the product of objects $a$ and $b$ of $\C$, and 
  $a + b$ for the coproduct.
 Notation for arrows is informally explained in the following diagrams:
 \[
  \begin{xy}
   \xymatrix @!=2.5pc {
              {}  & a \ar@{.>}[d]|{(f,g)} \ar[dl]_{f} \ar[dr]^{g} & {}  & a \times b\ar[d]|{f\times k} &  
                                      a \ar[r]^-{\inl{}} \ar[dr]_{f}& a + b \ar@{.>}[d]|{[f,h]}& \ar[l]_-{\inr{}} b \ar[dl]^{h} & a + b \ar[d]|{f+k}\\
                c & c\times d \ar[l]^-{\pi_1}\ar[r]_-{\pi_2}         & d   & c \times d &  {}  & c & {} & c + d 
   }
  \end{xy}
%   \begin{xy}
%    \xymatrix{ 
%                 a\ar[d]^{\inl{}} \ar[rd]^{f} & {}\\
%                 a + b  \ar[r]^{f} & c \\
%                 b\ar[u]^{\inr{}} \ar[ru]^{g} & {}
%     }
%   \end{xy}
% \qquad
% \begin{xy}
%    \xymatrix @!=3pc {   a \ar[r]^-{\inl{}} \ar[dr]_{f}& a + b \ar@{.>}[d]|{[f,h]}& \ar[l]_-{\inr{}} b \ar[dl]^{h}\\
%                {}  & c & {}
%     }
%   \end{xy}
 \]
\end{notation}

\subsection{Functors \& Natural Transformations}

Given two categories $\C$ and $\D$, a functor $F:\C\to\D$ maps objects of $\C$ to objects of $\D$, and 
morphisms of $\C$ to morphisms of $\D$, while preserving source and target as well as composition and identity:

\begin{definition} A functor $F$ from $\C$ to $\D$ is given by
 \begin{packitem}
  \item a map $F : \C \to \D$ on the objects of the categories involved and
  \item for any pair of objects $(c,d)$ of $\C$, a map
     \[ F_{(c,d)} : \C(c,d) \to \D(Fc, Fd) \enspace , \]
 \end{packitem}
 such that
  \begin{packitem}
   \item $\forall c : C, \enspace F(id_c) = id_{Fc}$ and%for any object $c$ in $\C$, $F(id_c) = id_{Fc}$ and 
   \item $\forall c~d~e : C, \forall f:c\to d, \forall g:d\to e, \enspace F(\comp{f}{g}) = \comp{Ff}{Fg}$.
  \end{packitem}
Here we use the same notation for the map on objects and that on morphisms. For the latter we also omit the 
subscript ``$(c,d)$'' as \emph{implicit} arguments. 
\end{definition}

\begin{definition}[Functor $\Delta:\Set\to\PO$ and Forgetful Functor]\label{def:delta}
 We call $\Delta : \SET\to\PO$ the functor from sets to preordered sets which associates  
 to each set $X$ the set itself together with the smallest preorder, i.e.\ the diagonal of $X$,
\[ \Delta(X) := (X,\delta_X). \] 
 In other words, for any $x,y\in X$ we have $x\delta_X y$ if and only if $x = y$.
 The functor $\Delta:\Set\to \PO$ is a \emph{full embedding}, i.e.\ it is fully faithful and injective on objects.

 In the other direction we have a \emph{forgetful} functor $U:\PO\to\Set$ which maps any preordered set $(X, \leq)$
 to the set $X$. 
 We have $\comp{\Delta}{U} = \Id_{\Set}$.
\end{definition}

\begin{definition}[Natural Transformation]
  Given two functors $F, G : \C \to \D$, a natural transformation $\gamma : F \to G$ (also written $\gamma:F \Rightarrow G$) 
 is given by a family of morphisms
  \[ \gamma_c : \D(Fc, Gc) \]
  indexed by objects of $\C$
  such that, for any morphism $f : c \to d$ in $\C$, the following diagram commutes:
 \[
   \begin{xy}
    \xymatrix{  
                **[l]Fc  \ar[r]^{\gamma_c} \ar[d]_{Ff} & **[r]Gc \ar[d]^{Gf} \\
                **[l]Fd \ar[r]_{\gamma_d} & **[r]Gd
}
   \end{xy}
 \]

\end{definition}

\begin{definition}[Adjunction]
 Let $\C$ and $\D$ be categories. An adjunction from $\C$ to $\D$ is given by 
  \begin{packitem}
   \item a functor $F : \C \to \D$,
   \item a functor $G : \D \to \C$,
   \item a natural transformation $\eta : \Id_\C \to \comp{F}{G}$, called \emph{unit}, and
   \item a natural transformation $\epsilon : \comp{G}{F} \to \Id_\D$, called \emph{counit},
  \end{packitem}
  such that the transformations
 \[ G \stackrel{\eta G}{\longrightarrow} GFG \stackrel{G\epsilon}{\longrightarrow} G \enspace , \quad 
    F \stackrel{F\eta}{\longrightarrow} FGF \stackrel{\epsilon F}{\longrightarrow} F \]
 both are the identity transformation.
 We write $F \dashv G$ for such an adjunction, leaving the unit and counit implicit.
\end{definition}
\begin{remark}
  The functors $F$ and $G$ as above are adjoint if and only there is a family of  bijections
  \[ \varphi = \bigl(\varphi_{c,d} : \D(Fc,d) \cong \C(c,Gd)\bigr) \]
  indexed by objects $c,d \in \C$,
  which is natural in both $c$ and $d$.
  
\end{remark}

\begin{definition}[Coreflection]\label{def:coreflection}
 Let $F : \C \to \D$ be an embedding, that is, a faithful functor which is injective on objects --- 
e.g., the inclusion of a subcategory.
 Then $F$ is a \emph{coreflection} if it has a right adjoint.
\end{definition}

The following lemma gives an example of a coreflection:

\begin{lemma}\label{lem:adj_set_po}
 The forgetful functor $U : \PO\to\Set$ %from the category $\PO$ of preordered sets to the category of sets
  is right adjoint to the diagonal functor $\Delta:\Set\to\PO$:
\begin{equation*}   
       \begin{xy}
        \xymatrix @C=4pc {
                  **[l]\Set \rtwocell<5>_U^{\Delta}{'\bot} & **[r]\PO
}
       \end{xy} \enspace ,
%\Rep^{\Delta}(S) : \Delta \dashv U : \Rep(S) 
%  \label{eq:adjunction}
\end{equation*}
 that is, the embedding $\Delta:\SET\to\PO$ is a coreflection.
We denote by $\varphi$ the family of isomorphisms
   \[   \varphi_{X,Y} : \PO(\Delta X, Y) \cong \Set(X,UY) \enspace . \]
We omit the indices of $\varphi$ whenever they can be deduced from the context.
%     \[  \Delta \dashv U, \quad \varphi_{X,Y} : \PO(\Delta X, P) \cong \Set(X, UP) \]
\end{lemma}
\begin{proof}
  The unit is given by a family of identity maps $\eta_X := \id_X: \Set(X, U\Delta X)$.
  The counit is given by a family of maps $\epsilon_Y : \PO(\Delta U Y, Y)$ whose carrier
  map on $UY$ is the identity map on $UY$. 
\end{proof}

We later use the following result about left adjoints:

\begin{lemma}[Left adjoints are cocontinuous]\label{lem:left_adj_cocont}
 Left adjoints are cocontinuous, i.e.\ commute with colimits. In particular, the image of an initial object under a left adjoint is
 initial. 
\end{lemma}
\noindent
For the proof we refer to Mac Lane's book \cite[V.5.Thm.1]{maclane}.

\subsection{More Examples, Notations}

The following categories and functors will appear in different places throughout the thesis. 
Again, the reader may skip these examples for the moment;
we will point to the definitions from the place where they are used.

\begin{definition}[Category of Families]\label{def:TST}\label{def:TS}
 Let $\C$ be a category and $T$ be a set, i.e. a discrete category (cf.\ \autoref{ex:cat_discrete}).
 We denote by $\family{\C}{T}$ the functor category, an object of which is a $T$--indexed family of objects of $\C$.  
 Given two families $V$ and $W$, a morphism $f : V \to W$ is a family of morphisms in $\C$,
  \[ f : t \mapsto f(t) : V(t) \to W(t) \enspace . \]
We write $V_t := V(t)$ for objects and morphisms.
 Given another category $\D$ and a functor $F : \C\to \D$, we denote by $\family{F}{T}$ the functor
  defined on objects and morphisms as
 \[ \family{F}{T} : \family{\C}{T} \to \family{\D}{T}, \quad f \mapsto \bigl(t \mapsto F (f_t)\bigr) \enspace . \]
\end{definition}

\begin{remark}\label{rem:adj_set_po_typed}
 Given a set $T$, the adjunction of \autoref{lem:adj_set_po} induces an adjunction
%  \[ \family{\Delta}{T} \dashv \family{U}{T},\quad \varphi_{V,W} : \TP{T}(\Delta V,W) \cong \TS{T}(V,UW) \]
%    \[ {\Delta} \dashv {U},\quad \varphi_{V,W} : \TP{T}(\Delta V,W) \cong \TS{T}(V,UW) \]
\[
   \begin{xy}
        \xymatrix@C=4pc{
                   **[l]\family{\Set}{T}\rtwocell<5>_{\family{U}{T}}^{\family{\Delta}{T}}{'\bot} & **[r]\family{\PO}{T}
}
       \end{xy} \enspace .
\]

\end{remark}

\begin{definition}[Retyping Functor]\label{rem:retyping_adjunction_kan}
Let $T$ and $T'$ be sets and $g:T\to T'$ be a map.
Let $\C$ be a cocomplete category.
The map $g$ induces a functor 
 \[ g^*:\family{\C}{T'} \to \family{\C}{T} \enspace, \quad W \mapsto \comp{g}{W} \enspace . \]
  The \emph{retyping functor associated to $g:T\to T'$},
 \[ \retyping{g}:\family{\C}{T} \to \family{\C}{T'} \enspace,  \]
  is defined as the left Kan extension operation 
  along $g$, that is, we have an adjunction
% $\retyping{g} \dashv g^*$.
\begin{equation}
  \begin{xy}
   \xymatrix @C=4pc {
            **[l]\family{\C}{T} \rtwocell<5>_{g^*}^{\retyping{g}}{'\bot} &  **[r]\family{\C}{T'}
}  
  \end{xy} \enspace .
\label{eq:adjunction_retyping}
 \end{equation}
  
\end{definition}

\begin{remark}[Retyping Functor Explicitly, \autoref{code:retype}]\label{def:retyping_functor}

In the context of \autoref{rem:retyping_adjunction_kan}, we define the functor 
\[\retyping{g} : \family{\C}{T} \to \family{\C}{T'} \enspace , \quad
 X = t\mapsto X_t  \quad \mapsto \quad \retyping{g}(X) :=  t' \mapsto \coprod_{\{t ~\mid~ g(t) = t'\}} X_t \enspace .\]
In particular, for any $V \in \family{\C}{T}$ --- considered as a functor --- we have a natural transformation
 \[ V \Rightarrow \comp{g}{\retyping{g}V} : T \to \C  \]
given pointwise by the morphism $V_t \to \coprod_{ \{s | g(s) = g(t) \}} V_s$ in the category $\C$.
Put differently, the map $g : T\to T'$ induces an endofunctor $\bar{g}$ on $\family{\C}{T}$ with object map
 \[ \bar{g}(V) := \comp{g}{\retyping{g}(V)} \]
and we have a natural transformation \lstinline!ctype! --- the unit of the adjunction of \autoref{eq:adjunction_retyping},
 \[  \text{\lstinline!ctype!}:\Id \Rightarrow \bar{g} : \family{\C}{T} \to \family{\C}{T} \enspace . \]

\end{remark}
%
%
%\[
%  \begin{xy}
%    \xymatrix { X \ar[dd]^x & {} & X \ar[dd]^{x} & {} \\
%	         {}  & \mapsto & {} & {}                                        \\
%                T & {} & T \ar[r]^{\alpha}& T'
%    }
%  \end{xy}
%\]

\begin{remark}
One can interpret the map $g : T\to T'$ as a translation of object sorts and the 
functor $\retyping{g}$ as a ``retyping functor'' which changes the sorts of contexts and terms (or more generally, 
models of terms) according to the translation of sorts.
The monads we are interested in are monads over some category 
$\TS{T}$ and our monad morphisms are over retyping functors.
In \autoref{sec:compilation} we interpret the syntax of a language $P$ over a set of types $T$ 
as a monad $P$ over the category $\TS{T}$. 
Given another language $Q$ over a set of types $U$, we consider a translation from $P$ to $Q$ 
to be a translation of object types $g : T \to U$ and a 
colax monad morphism $P\to Q$ over the retyping functor $\retyping{g} : \TS{T} \to \TS{U}$ (cf.\ \autoref{def:colax_mon_mor}).
\end{remark}

\begin{remark}[about maps on coproducts and pattern matching]\label{rem:about_coproduct_matching}
 In the proof assistant \textsf{Coq} we implement retyping (cf.\ \autoref{def:retyping_functor}) 
  via an inductive family, cf.\ \autoref{code:retype}.
 In this context, passing from the left to the right in the adjunction isomorphism
 \[ \family{C}{T'}(\retyping{g}V,W) \cong \family{C}{T}(V, g^*W)\]
 is done by precomposing with \emph{pattern matching on the constructor} \lstinline!ctype!, 
 cf.\ \autoref{sec:comp_sem_formal_initiality}.
\end{remark}

\begin{definition}[Pointed index sets]\label{def:cat_indexed_pointed}\label{def:cat_set_pointed}
  Given a category $\C$, a set $T$ and a natural number $n$, we denote by $\family{\C}{T}_n$ the category
  with, as objects, diagrams of the form
   \[ n \stackrel{\vectorletter{t}}{\to} T \stackrel{V}{\to} \C \enspace , \]
  written $(V, t_1, \ldots, t_n)$ with $t_i := \vectorletter{t}(i)$.
  A morphism $h$ to another such $(W,\vectorletter{t}) $
  with the same pointing map $\vectorletter{t}$ is given by a morphism $h : V\to W$ in $\family{\C}{T}$.
  Note that there is are no morphisms between families with different points, that is, 
$\family{\C}{T}_n\left((V,\vectorletter{t}), (V',\vectorletter{t'})\right) = \emptyset$ if $\vectorletter{t} \neq \vectorletter{t'}$.
  Any functor $F : \family{\C}{T} \to \family{\D}{T}$ extends to $F_n : \family{\C}{T}_n \to \family{\D}{T}_n$ via
   \[ F_n (V,t_1,\ldots,t_n) := (FV, t_1,\ldots,t_n) \enspace . \]
\end{definition}

\begin{remark}
 The category $\family{\C}{T}_n$ consists of $T^n$ copies of $\family{\C}{T}$, which do not interact.
 Due to the ``markers'' $(t_1, \ldots, t_n)$ we can act differently on each copy, 
  cf., e.g., \autorefs{def:derived_mod_II} and \ref{def:fibre_mod_II}.
 The reason why we consider categories of this form is explained at the beginning of \autoref{sec:ext_zsido} %\autoref{sec:comp_term_sigs} 
 and in \autoref{rem:family_of_mods_cong_pointed_mod}.
\end{remark}

Retyping functors generalize to categories with pointed indexing sets;
when changing types according to a map of types $g:T\to U$, the markers must be adapted as well:

\begin{definition}\label{def:retyping_functor_pointed}
Given a map of sets $g:T\to U$, by postcomposing the pointing map with $g$, the retyping functor generalizes to the functor
 \[ \retyping{g}(n) : \family{\C}{T}_n \to \family{\C}{U}_n \enspace , \quad (V, \vectorletter{t}) \mapsto \bigl(\retyping{g} V, g_*(\vectorletter{t})\bigr) \enspace ,  \] 
 where $g_*(\vectorletter{t}) := \comp{\vectorletter{t}}{g} : n\to U$.
\end{definition}

Finally there is also a category where families of objects of $\C$ over different indexing sets are mixed together:

\begin{definition}\label{def:cat_TEns}
 Given a category $\C$, we denote by $\T\C$ the category where an object is a pair $(T,V)$ of 
a set $T$ and a family $V\in \family{\C}{T}$ of objects of $\C$ indexed by $T$.
   A morphism $(g,h)$ to another such $(T',W)$ is given by a map $g : T\to T'$ and a 
  morphism $h : V\to\comp{g}{W}$ in $\family{\C}{T}$, that is, a
  family of morphisms in $\C$, indexed by $T$,
  \[ h_t : V_t \to W_{g(t)} \enspace . \]

\noindent
Suppose $\C$ has an initial object, denoted by $0_{\C}$. 
Given $n\in \mathbb{N}$, we call $\hat{n} = (n, k\mapsto 0_\C)$ the object of $\T\C$ that associates to any 
$1\leq k \leq n$ the initial object of $\C$.
 We call $\T\C_n$ the slice category $\hat{n} \downarrow \T\C$.
An object of this category consists of an object $(T,V) \in\T\C$ whose indexing set ``of types'' $T$ is 
pointed $n$ times, written $(T,V,\vectorletter{t})$, where $\vectorletter{t}$ is a vector of
elements of $T$ of length $n$. A morphism $(g,h) :(T,V,\vectorletter{t}) \to (T',V',\vectorletter{t'})$
is a morphism $(g,h) :(T,V) \to (T',V')$ as above, such that $\vectorletter{t'} = \comp{g}{\vectorletter{t}}$.

We call 
$\T U_n : \T\C_n \to \Set$
the forgetful functor associating to any pointed family $(T,V, t_1,\ldots,t_n)$ the indexing set $T$.
Note that for a fixed set $T$, the category $\family{\C}{T}_n$ (cf.\ \autoref{def:cat_indexed_pointed}) 
is the fibre over $T$ of this functor.
\end{definition}

\begin{remark}[Picking out Sorts]\label{rem:nat_trans_picking_sort}
   Let $1: \T\C_n \to \Set$ denote the constant functor which maps objects to the terminal object
   of the category $\Set$.
   A natural transformation $\tau:1 \to \T U_n$ associates to any object $(T,V,\mathbf{t})$ of the 
   category $\T\C_n$
   an element of $T$. 
   Naturality imposes that $\tau(T',V',\mathbf{t'}) = g \left(\tau(T,V,\mathbf{t})\right)$ for any 
     $(g,h) : (T,V,\mathbf{t}) \to (T',V',\mathbf{t'})$.
\end{remark}
\begin{notation}\label{not:tau_simpl_notation}
 Given a natural transformation $\tau : 1 \to \T U_n$ as in \autoref{rem:nat_trans_picking_sort}, we write
  \[ \tau (T,V,\mathbf{t}) := \tau (T,V,\mathbf{t})(*) \in T \enspace , \]
  i.e.\ we omit the argument $*\in 1_{\Set}$ of the singleton set.
\end{notation}

\begin{example}
 For $ 1 \leq k \leq n$, we denote by $k : 1 \Rightarrow \T U_n : \T\C_n \to \Set$ the natural transformation such that $k(T,V,\vectorletter{t}) := \vectorletter{t}(k)$.
\end{example}

\section{Monads \& Modules}\label{mon_mod}

We state the widely known definition of monad and the less known definition of \emph{module over a monad}, 
together with their respective morphisms.
Modules have been used in the context of Initial Semantics by Hirschowitz and Maggesi \cite{DBLP:conf/wollic/HirschowitzM07, DBLP:journals/iandc/HirschowitzM10}
and Zsid\'o \cite{ju_phd}.
The monad morphisms we are interested in are, more precisely, \emph{colax} monad morphisms, see, e.g., Leinster's book \cite{Leinster2003}. 

\subsection{Definitions}

\begin{definition}[Monad]\label{def:monad_mu} A \emph{monad} $T$ over a category $\C$ is given by 
  \begin{packitem}
   \item a functor $T : \C \to \C$ (which we denote by the same name as the monad),
   \item a natural transformation $\eta : \Id_{\C} \to T$ and 
   \item a natural transformation $\mu : \comp{T}{T}\to T$
  \end{packitem}
 such that the following diagrams commute:
  \begin{equation*}
\begin{xy}
\xymatrix @C=3pc{
T  \ar[r]^{T \eta} \ar[rd]_{\id} %\ar@2{-}[rd] 
   & T^2  \ar[d]^{\mu} \ar@{}[ld] \ar@{}[rd] & T \ar[l]_{\eta_{T}} \ar[dl]^{\id} %\ar@2{-}[dl] 
   \\
 {} & T  & {}
}
\end{xy} \quad  \qquad
\begin{xy}
 \xymatrix@C=3pc{
  T^3  \ar[r]^{\mu_{T}} \ar[d]_{T\mu} & T^2  \ar[d]^{\mu} \\
  T^2  \ar[r]_{\mu} & T.
}
\end{xy}
%
%\label{mu}
\end{equation*}
\end{definition}

\begin{example}[List Monad]
 The functor $[\_] : \Set\to\Set$ which to any set $X$ associates the set of finite lists over $X$, 
is equipped with a structure as monad
 by defining $\eta$ and $\mu$ as ``singleton list'' and flattening, respectively:
  \[ \eta_X(x) := [x] \quad\text{ and } \] 
  \[
      \mu_X \left(\bigl[ [x_{1,1}, \ldots , x_{1,m_1}],\ldots,[x_{n,1},\ldots,x_{n, m_n}]\bigr]\right) := 
       [x_{1,1},\ldots,x_{1,m_1},\ldots,x_{n,1},\ldots,x_{n,m_n}] .
  \]
\end{example}

\begin{remark}[Kleisli Operation (Monadic Bind)]\label{rem:def_monad_alt}
 Given a monad $(T,\eta,\mu)$ on the category $\C$, 
 the Kleisli operation $\sigma$ is defined, for any $a,b \in\C$ and $f\in \C(a,Tb)$, by setting
   \begin{align*} 
       \sigma_{a,b} : \C(a,Tb)&\to \C(Ta,Tb) \enspace , \\
                          f &\mapsto  \comp{Tf}{\mu_b} \enspace .
   \end{align*}
%  is defined, for any $a,b \in\C$ and $f\in \C(a,Tb)$, by setting
%   \[  (f)_{(a,b)}^*:= \comp{Tf}{\mu_b} \enspace .\]
 Indeed, a monad $(T,\eta,\mu)$ can equivalently be defined as a triple $(T,\eta,\sigma)$ with an adapted set of axioms,
   see \autoref{def:endomonad}.
  We often leave the object arguments $a$ and $b$ implicit, i.e.\ we write $\sigma(f) := \sigma_{a,b}(f)$.
%  We refer to \cite{manes} for details.
\end{remark}

\begin{example}[Monadic Syntax, untyped]\label{ex:ulc_monad}
 Syntax as a monad (in form of a Kleisli triple) was presented by Altenkirch and Reus \cite{alt_reus}:
 consider the syntax of the untyped lambda calculus $\LC$ as given in \autoref{ex:ulc_def} in \autoref{subsec:informal_intro}.
 As mentioned there, the map $V\mapsto \LC(V)$ is functorial, its map on morphisms is given by \emph{renaming}
of free variables. 
This functor is equipped with a monad structure
 by defining $\eta$ as variable--as--term operation
 \[ \eta_V(v) := \Var(v) \in \LC(V)\]
and the multiplication $\mu : \comp{\LC}{\LC} \to \LC$ as flattening which, 
given a term of $\LC$ with terms of $\LC(V)$ as variables, 
returns a term of $\LC(V)$ by removing a layer of intermediate $\Var$ constructors.
These definitions turn $(\LC, \eta, \mu)$ into a monad on the category $\Set$. 
The Kleisli operation associated to this monad corresponds to simultaneous substitution \cite{alt_reus}. 

\end{example}

\begin{example}[Monadic Syntax, typed] \label{ex:tlc_syntax_monadic}
 Consider the syntax of the simply--typed lambda calculus as defined in \autoref{ex:slc_def}.
The map \[\SLC : \TS{\TLCTYPE}\to\TS{\TLCTYPE}\enspace, \quad V \mapsto \SLC(V) \enspace , \]
associating to any set family $V$ the family of lambda terms with free variables in $V$,
is the object map of a functor.
Similarly to the untyped lambda calculus (cf.\ \autoref{ex:ulc_monad}), the natural transformations $\eta : \Id\to\SLC$ 
  and $\mu : \comp{\SLC}{\SLC}\to \SLC$ are defined 
  as variable--as--term operation and flattening, respectively.
These definitions turn $(\SLC, \eta, \mu)$ into a monad on the category $\TS{\TLCTYPE}$.
\end{example}

Our definition of \emph{colax} monad morphisms and their \emph{transformations} is taken from Leinster's book \cite{Leinster2003}:

\begin{definition}[Colax Monad Morphism] \label{def:colax_mon_mor}
Let $(T,\eta,\mu)$ be a monad on the category $\C$ and $(T', \eta', \mu')$ be a monad on the category $\D$.
A \emph{colax morphism of monads} $(\C,T) \to (\D,T')$ is given by
\begin{packitem}
 \item a functor $F : \C\to\D$ and 
 \item a natural transformation $\gamma : FT\to T'F$ as in

\[
 \begin{xy}
  \xymatrix@!=2.5pc{\C \ar[r]^{T} \ar[d]_{F} \drtwocell<\omit>{\;\;\gamma}& \C \ar[d]^{F}\\
             \D\ar[r]_{T'} & \D
}
 \end{xy}
\]
\end{packitem}
such that the following diagrams commute:
\[
  \begin{xy}
   \xymatrix{
    FTT \ar[r]^{\gamma T} \ar[d]_{F\mu} & T'FT \ar[r]^{\gamma} & T'T'F \ar[d]^{\mu' F}  \\
    FT \ar[rr]_{\gamma} & {} & T'F,  
  }
  \end{xy}
  \qquad
  \begin{xy}
    \xymatrix{
  F \ar[d]_{F\eta} \ar[dr]^{\eta' F}& {} \\
  FT \ar[r]_{\gamma} & T'F .
  }
  \end{xy}
\]

\end{definition}

From now on we simply say ``monad morphism over $F$'' when speaking about a colax monad morphism with underlying functor $F$. 

\begin{definition}[Composition of Monad Morphisms]
  Suppose given a monad morphism as in \autoref{def:colax_mon_mor}. Given a third monad $(T'', \eta'', \mu'')$ on category $\E$ 
 and a monad morphism $(F', \gamma') : (T',\eta', \mu') \to (T'', \eta'',\mu'')$, we define the composition of 
  $(F, \gamma)$ and $(F', \gamma')$ to be the monad morphism given by the pair consisting of the functor $F'F$ and the transformation
 \[ 
       \begin{xy}
        \xymatrix{ F'FT \ar[r]^{F'\gamma} & F'T'F \ar[r]^{\gamma' F} &  T''F'F \enspace .}
       \end{xy}
 \]
 The verification of the necessary commutativity properties is done --- for the equivalent definition given in \autoref{def:colax_mon_mor_alt} --- in the 
 \textsf{Coq} library, cf.\ \lstinline!colax_Monad_Hom_comp!.
\end{definition}

\begin{definition}[Transformation]\label{def:transformation}
 Given two morphisms of monads 
 \[(F,\gamma), (F', \gamma') : (\C,T)\to (\D,T') \enspace,\] a \emph{transformation} 
  $ (F,\gamma)\to (F', \gamma') $ is given by a natural transformation $\beta : F\Rightarrow F'$ such that
  the following diagram commutes:
\begin{equation*}
 \begin{xy}
  \xymatrix @=3pc{
      FT \ar[r]^{\gamma} \ar[d]_{\beta T} & T'F \ar[d]^{T'\beta} \\
      F'T \ar[r]_{\gamma'} & T'F'.
   }
 \end{xy}
\end{equation*} 
\end{definition}

A \emph{2--category} is a category with ``morphisms between morphisms''. We refer to Mac Lane's book \cite{maclane} 
for the definition.

\begin{definition}[2--Category of Monads, \cite{Leinster2003}]
 We call \Mcol the 2--category an object of which is a pair $(\C,T)$ of a category $\C$ and
 a monad $T$ on $\C$. A morphism to another object $(\D,T')$ is a colax monad morphism $(F,\gamma) : (\C,T)\to(\D,T')$.
 A 2--cell $(F,\gamma) \Rightarrow (F',\gamma')$ is a transformation.
\end{definition}

\begin{notation}
 For any category $\C$, we write $\Id_\C$ for the object $(\C,\Id)$ of \Mcol.
\end{notation}

We are interested in \emph{modules over monads}. These are particular  monad morphisms
whose codomain is the identity monad on some category%
 \footnote{The characterization of modules over monads as particular colax monad morphisms is due to an anonymous referee, 
 whom I hereby thank for his helpful comments.}.
Modules and, more specifically, their morphisms, capture the distributivity of substitution over the 
constructors of a language, cf.\ \autoref{ex:ulc_mod_mor} and \autoref{ex:ulc_mod_mor_kl}.

\begin{definition}[Module over a Monad]\label{def:module}
  Given categories $\C$ and $\D$ and a monad $T$ on $\C$, a \emph{module over $T$ with codomain $\D$} 
   (or \emph{$T$--module towards $\D$})
  is a colax monad morphism $(M,\gamma) : (\C,T) \to (\D,\Id_{\D})$ from $T$ to the identity monad on $\D$.
 Given parallel $T$--modules $M$ and $N$, a \emph{morphism of modules from $M$ to $N$} is a transformation from $M$ to $N$
 as in \autoref{def:transformation}.
 We denote the category of $T$--modules towards $\D$ by
       \[ \Mod{T}{\D} := \Mcol\bigl((\C,T), (\D,\Id)\bigr) \enspace . \] 
 
\end{definition}

\noindent
Before giving some examples of modules over monads, we state a more explicit definition of modules:

\begin{remark}[(Modules and their Morphisms, explicitly \cite{DBLP:conf/wollic/HirschowitzM07})]\label{rem:def_module_hom}
 By unfolding the preceding definition and simplifying, we obtain that a $T$--module towards $\D$ is a functor $M : \C\to\D$ 
 together with a natural transformation $\sigma : MT\to M$ such that the following diagrams commute:
\[
  \begin{xy}
   \xymatrix{
    MTT \ar[r]^{\sigma T} \ar[d]_{M\mu} & MT \ar[d]^{\sigma}     &  M \ar[d]_{M\eta} \ar[dr]^{\id}& {} \\
    MT \ar[r]_{\sigma} &  M,  &   MT \ar[r]_{\sigma} & M .
  }
  \end{xy}
\]
Such a module can hence be regarded as a kind of generalized monad over a functor that is not necessarily 
an endofunctor; indeed, this is our intuition behind modules. 
In particular, every monad gives rise to 
a module over itself, the \emph{tautological module} (cf.\ \autoref{def:taut_mod}).
Furthermore, the category of modules $\Mod{T}{\D}$ allows for products, provided the target category $\D$ is equipped with a product.

 A morphism of $T$--modules from $(M, \sigma)$ to $(M', \sigma')$ then is given by a natural transformation $\beta : M\Rightarrow M'$
such that the following diagram commutes:
 \[
  \begin{xy}
    \xymatrix{
     MT \ar[r]^{\beta T} \ar[d]_{\sigma}  & M'T \ar[d]^{\sigma'} \\
      M \ar[r]_{\beta} &  M'.
}
  \end{xy} 
 \]
 
\end{remark}

\noindent
 We anticipate the constructions of the next section by giving some examples of modules:
\begin{example}[Tautological Module, \autoref{ex:ulc_monad} cont.]\label{ex:ulc_taut_mod}
   Any monad $T$ on a category $\C$ can be considered as a module over itself, the \emph{tautological module} (cf.\ \autoref{def:taut_mod}).
   In particular, the monad of the untyped lambda calculus $\LC$ (cf.\ \autoref{ex:ulc_monad}) is a $\LC$--module with codomain $\Set$.
\end{example}

\begin{example}\label{ex:ulc_const_mod}
  The map  \[ \LC' : V \mapsto \LC(V') \enspace ,  \]
   with $V':= V+1$,
  inherits the structure of a $\LC$--module from the tautological module $\LC$ --- 
  we obtain the \emph{derived module} (cf.\ \autoref{sec:endo_deriv}) of the module $\LC$.
  Also, the map \[\LC\times \LC : V \mapsto \LC(V)\times\LC(V)\] inherits a $\LC$--module structure, cf.\ \autoref{def:product_module}.
\end{example}

The constructors of our example languages are, accordingly, \emph{morphisms of modules}:

\begin{example}[\autoref{ex:ulc_const_mod} cont.]\label{ex:ulc_mod_mor}
  The map \[V \mapsto \App_V : \LC(V)\times \LC(V) \to \LC(V)\]
    satisfies the diagram of \autoref{rem:def_module_hom} and is hence a morphism of $\LC$--modules from $\LC\times\LC$ to $\LC$.
  The map \[ V \mapsto \Abs_V : \LC(V') \to \LC(V) \] is a morphism of $\LC$--modules from $\LC'$ to $\LC$.
  Later we consider this example using an alternative definition of module morphism (cf.\ \autoref{def:endo_mod_hom}) 
  and explain in detail the meaning of its commutative diagrams for the constructors $\App$ and $\Abs$, cf.\ \autoref{ex:ulc_mod_mor_kl}.
\end{example}

\subsection{Constructions on Monads and Modules}\label{subsection:mod_examples}

We present some constructions of modules which will be used in the next section.
They were previously defined in Zsid\'o's thesis \cite{ju_phd} and 
works of Hirschowitz and Maggesi \cite{DBLP:conf/wollic/HirschowitzM07, DBLP:journals/iandc/HirschowitzM10}.

\begin{definition}[Tautological Module] \label{def:taut_mod}

  Given a monad $(\C,T)$, we call $(T,\mu_T)$ (or simply $T$) the \emph{tautological module} $(T,\mu_T) : (\C,T)\to (\C,\Id)$.

\end{definition}

\begin{definition}[Constant and Terminal Module]
 Given a monad $(\C,T)$ and a category $\D$ with an object $d\in \D$, 
  the constant functor $F_d:\C\to\D$ mapping any object of $\C$ to $d\in\D$ and any morphism to the identity on $d$ 
     yields a module
   \[(F_d,\id) : (\C,T)\to (\D,\Id) \enspace . \] 
In particular, if $\D$ has a terminal object $1_\D$, then the constant module $(F_{1_\D},\id)$ is terminal in $\Mod{T}{\D}$.
\end{definition}

\begin{remark}
 Given a monad $(\C,T)$, a $T$--module $(M,\sigma)$ with codomain category $\D$ and a functor $F:\D\to\E$,
then the pair $(\comp{M}{F},F\sigma)$ is a $T$--module with codomain category $\E$.
 For $(M,\sigma) := (T,\mu_T)$ and $F := F_e$ for some $e\in \E$ one obtains the constant module as above.
\end{remark}

\begin{definition}[Pullback Module]

Let $(\C,T)$ and $(\D,T')$ be monads over $\C$ and $\D$, respectively.
Given a morphism of monads $(F,\gamma):(\C,T) \to (\D,T')$ and a $T'$-module $(M,\sigma)$ with codomain $\E$, we 
call \emph{pullback of $M$ along $(F,\gamma)$} the $T$--module $(F,\gamma)^*(M,\sigma):= \comp{(F,\gamma)}{(M,\sigma)}$.

\end{definition}

\begin{definition}[Module Morphism induced by a Monad Morphism]\label{def:induced_module_mor}
With the same notation as in the previous example, the monad morphism $(F,\gamma)$ induces a morphism of $T$--modules
--- which we call $\gamma$ as well --- 
\[\gamma : \comp{T}{(F,\id)}\Rightarrow (F,\gamma)^*(T',\mu_{T'}) \]
as in
\[
 \begin{xy}
  \xymatrix{ {} & (\C,T) \ar[ld]_{(T,\mu_T)} \ar[rd]^{(F,\gamma)}  & {} \\
            (\C, \Id)\ar[rd]_{(F,\id)} & {} & (\D,T') \ar[ld]^{(T',\mu_{T'})} \\
             {}& (\D,\Id) \uutwocell <\omit>{\gamma}& **[r].
  }
 \end{xy}
\]
Note that the above diagram can be read as a structure--enriched version of the square diagram specifying
 the type of $\gamma$ in \autoref{def:colax_mon_mor}.
\end{definition}

\begin{definition}[Product Module] \label{def:product_module}
Suppose the category $\D$ is equipped with a product. %
Given any monad $(\C,T)$, the product of $\D$ lifts to a product on the category $\Mod{T}{\D}$ 
of $T$--modules with codomain $\D$.
\end{definition}

\subsection{Monads on Set Families}\label{sec:monads_on_set_families}

We are particularly interested in monads over families of sets and monad morphisms over retyping functors.

\subsubsection{Derivation}\label{sec:endo_deriv}

Roughly speaking, a binding constructor makes free variables disappear. Its inputs are hence terms 
``with (one or more) additional free variables'' compared to the output, 
i.e.\ terms in an \emph{extended context}. 
Context extension is captured mathematically by \emph{derivation}:
let $T$ be a set and $u\in T$ an element of $T$.
We define $D(u)$ to be the object of $\TS{T}$ such that 
 \[D(u)(u)=\lbrace *\rbrace \quad\text{and}\quad D(u)(t) = \emptyset \text{ for } t\neq u \enspace .\]
We enrich the object $V$ of $\TS{T}$ with respect to $u$ by setting
  \[ V^{*u} := V + D(u) \enspace , \]
that is, we add a fresh variable of type $u$.
This yields a monad $ (\_)^{*u}$ on $\TS{T}$. % 

\begin{definition}[Derivation Monad Morphism] 
 Given any monad $P$ on $\TS{T}$, we define a monad endomorphism on $P$
 over the functor $V\mapsto V^{*u}$.
On a set family $V\in\TS{T}$ its natural transformation $\gamma$ is defined as the coproduct map
\begin{equation} 
  \gamma_V := [P(\inl{}), \comp{\inr}{\eta} %
                                            ] : (PV)^{*u} \to P(V^{*u}) \enspace , 
  \label{eq:deriv_mon_mor}
\end{equation}
where $[\inl{},\inr{}] = \id : V^{*u} \to V^{*u}$.
\end{definition}

\begin{definition}\label{def:derivation_simple}
 Given a monad $P$ over $\TS{T}$ and a $P$--module $M$, we call $M^u$ the module obtained as the composition  $\comp{(\_)^{*u}}{M}$.
\end{definition}

\begin{example}%[Ex.\ \ref{ex:slc_modules} cont.] 
 \label{ex:deriv_mod_slc}
    We consider $\SLC$ (cf.\ \autoref{ex:tlc_syntax_monadic}) as the tautological module over itself.
    Given any element $s\in \TLCTYPE$, the derived module with respect to $s$, 
    \[\SLC^{s} : V \mapsto \SLC(V^{*s})\enspace ,\]
  assigns to any type family $V$ --- the \emph{context} ---  the type family of terms of $\SLC$ 
    over $V$ enriched with one additional variable of sort $s$.
\end{example}

More generally, given a natural transformation as in \autoref{rem:nat_trans_picking_sort},
 \[\tau : 1\Rightarrow \T U_n : \T\Set_n \to \Set \enspace ,\]
we can derive, with respect to $\tau$, any module defined on a category of the form $\TS{T}_n$ for 
any set $T$:

\begin{definition}[Derived Module]\label{def:derived_mod_II}
  Let $\tau : 1 \to \T U_n$ be a natural transformation.
 Given a set $T$ and a monad $P$ on $\TS{T}_n$, 
 the functor $(\_)^{*\tau} : (T,V,\vectorletter{t})\mapsto (T,V^{*\tau(T,V,\vectorletter{t})},\vectorletter{t})$ is given the structure of a morphism of monads
  as in \autoref{eq:deriv_mon_mor}.
Given any $P$--module $M$, we call \emph{derivation of $M$ with respect to $\tau$} the module 
 $M^{\tau} := \comp{(\_)^{*\tau}}{M}$.
\end{definition}

\begin{remark}
  In the preceding definition the natural transformation $\tau : 1 \to \T U_n$ 
  supplies more data than necessary, since we only evaluate it on families of sets indexed by 	
  the fixed set $T$. 
  However, in \autoref{sec:ext_zsido} we derive different modules --- each defined on a category $\TS{T}_n$
  with varying sets $T$ --- with respect to one and the same natural transformation $\tau$. 
\end{remark}

\subsubsection{Fibres} \label{sec:fibre}

Given a typed language over a nonempty set of types $T$, 
we occasionally want to pick terms of a specific type $u\in T$.
Let $\D$ be a category --- think of $\D$ as the category $\Set$ --- 
and $V\in \family{\D}{T}$ a $T$--indexed family, e.g., of terms of said language. 
Then picking ``terms of type $u\in T$'' corresponds to projecting to the fibre $V(u)$.

Given a monad $P$ on a category $\C$ and a $P$--module $M$ towards $\family{\D}{T}$, we define the 
\emph{fibre module of $M$ with respect to $u\in T$} to be the module which associates
the fibre $M(c)(u)$ to any object $c\in \C$. This construction is expressed via postcomposition
with a particular module:
we define the \emph{fibre with respect to $u\in T$} to be the monad morphism
\[ \bigl((\_)(u), \id\bigr) : (\family{\D}{T},\Id) \to (\D,\Id) \]
over the functor $V \mapsto V(u)$.
Postcomposition of the module $M$ with this module then precisely yields the fibre module 
$\fibre{M}{u}$ of $M$
with respect to $u\in T$.
Analogously to derivation we define the fibre with respect to a natural transformation:
\begin{definition}[Fibre Module]\label{def:fibre_mod_II}
 Let the natural transformation $\tau$ be as in \autoref{def:derived_mod_II}.
  We call \emph{fibre with respect to $\tau$} the monad morphism 
  \[ (\_)_{\tau} : V \mapsto V(\tau_V) : (\family{\D}{T}_n,\Id) \to (\D,\Id) \]
  over the functor $V\mapsto V_{\tau(V)}$.
  Given a module $M$ towards $\family{\D}{T}_n$ (over some monad $P$), we call 
  the \emph{fibre module of $M$ with respect to $\tau$} the module $\fibre{M}{\tau} := \comp{M}{(\_)_{\tau}}$.
\end{definition}

\begin{example}%[Ex.\ \ref{ex:slc_modules} cont.]
 \label{ex:fibre_mod_slc}
   We consider $\SLC$ as the tautological module over itself.
   Given any element $t\in \mathcal{T}$, the fibre module with respect to $t$,
denotes the set of terms of $\SLC$ of type $t$ in context $V$:
     \[\fibre{\SLC}{t} : V \mapsto \SLC(V)_t \enspace . \]
\end{example}

\begin{example}\label{ex:slc_modules}
   Consider the monad $\SLC : \TS{\TLCTYPE} \to \TS{\TLCTYPE}$ of \autoref{ex:tlc_syntax_monadic}.
The two operations of derivation (cf.\ \autoref{ex:deriv_mod_slc}) and fibre 
 (cf.\ \autoref{ex:fibre_mod_slc}) can be combined, yielding a module over $\SLC$ with carrier
\[ V\mapsto \SLC^{s}_t(V) := \SLC(V^{*s})_t \enspace .\]
This module is actually the domain module of the abstraction constructor, cf.\ \autoref{ex:slc_mod_mor}.
The product of modules yields our final example: for any $s,t \in \TLCTYPE$, 
the domain of the application $\App(s,t)$ of simply--typed
lambda calculus is a module over $\SLC$, 
\[ \fibre{\SLC}{s\TLCar t} \times\fibre{\SLC}{s} : V \mapsto \SLC(V)_{s\TLCar t} \times \SLC(V)_{s} \enspace . \]
\end{example}

\begin{example}[\autoref{ex:slc_modules} cont.]\label{ex:slc_mod_mor}
  Given $s,t\in \TLCTYPE$, the map 
    \[ \App(s,t) : V \mapsto %\bigl(
          \App_V (s,t) : \SLC(V)_{s\TLCar t} \times \SLC(V)_{s} \to \SLC(V)_t %\bigr)
    \]
  satisfies the diagram of the preceding definition and is hence a morphism of modules.
 In the same way the constructor $\Abs(s,t)$ is a morphism of modules; we have
 \begin{align*} \App(s,t) &: \fibre{\SLC}{s\TLCar t} \times \fibre{\SLC}{s} \to \fibre{\SLC}{t} \\
  \Abs(s,t) &: \fibre{\SLC^s}{t} \to \fibre{\SLC}{s\TLCar t} \enspace .
 \end{align*}
\end{example}

The pullback operation commutes with products, derivations and fibres:

\begin{remark}
  \label{rem:pullback_deriv_fibre_endo}
   Let $(\C,P)$ and $(\D,Q)$ be monads, and let $\rho : P \to Q$ be a monad morphism.
  Let $M$ be a $Q$--module with codomain $\E$. Suppose $T$ is a set, and let $u\in T$ be an element of $T$.
\begin{enumerate}
 \item More specifically, let $Q$ be a monad on $\family{\Set}{T}$. Then
      \[ \rho^*(M^u) = (\rho^*M)^u \enspace . \]
 \item More specifically, let $\E = \family{\C}{T}$. Then
      \[ \rho^* \fibre{M}{u} = \fibre{\rho^*M}{u} \enspace . \]
 \item Let $N$ be another $Q$--module with codomain $\E$ and suppose $\E$ is equipped with a product. Then
          the pullback functor is cartesian:
 \[ \rho^* (M \times N) = \rho^*M \times \rho^*N \enspace .\]
   
\end{enumerate}
The first two properties are just instances of associativity of composition of monad morphisms.

\end{remark}

\begin{remark}
  \label{rem:about_module_eq}
In \textsf{Coq} the equality of modules is not as trivial as in informal mathematics,
 since there are two different notions of equality: 
\emph{definitional equality}, also called \emph{convertibility}, and \emph{propositional} equality.
While the latter is to be proved by the user, the former is computed by the system and thus
cannot be influenced by the user.

While the above equalities of \autoref{rem:pullback_deriv_fibre_endo} hold propositionally (using appropriate axioms, such as proof irrelevance),
they do not hold definitionally.
The consequences of this lack of definitional equality are discussed in \autoref{subsec:sts_formal_reps}.
In summary, in our formalization, monads, modules and module morphisms behave more like in a \emph{bicategory} rather than in a strict 2--category.
\end{remark}

\section{Alternative Definitions for Monads \& Modules} \label{sec:alternative_defs_monad}

Monads can be defined in terms of the Kleisli operation (cf.\ \autoref{rem:def_monad_alt})
instead of the natural transformation $\mu$ of \autoref{def:monad_mu}. 
A similar alternative definition exists for modules.
In this section we state those alternative definitions in full detail, for several reasons:
firstly, the alternative definition of monad is well--known for its prominent use in the \textsc{Haskell} programming language.
Secondly, it is also the definition we chose to implement in the proof assistant \textsf{Coq}.
Furthermore, it is also this alternative definition which generalizes to \emph{relative monads} (cf.\ \autoref{def:relative_monad}), 
that is, monads that are not necessarily endofunctors.

\begin{definition}[Alt.\ Def.\ for Monad (\autoref{def:monad_mu}), \autoref{code:endomonad}] \label{def:endomonad} 
A \emph{monad} $T$ over a category $\C$ (in Kleisli form) is given by 
\begin{packitem}
 \item 
a map $T\colon \C \to\C$ on the objects of $\C$, carrying the same name as the monad,
 \item 
for each object $c$ of $\C$, a morphism $\eta_c\in \C(c,Tc)$ and
 \item 
for all objects $c$ and $d$ of $\C$, a \emph{Kleisli} map 
 \[\sigma_{c,d}\colon \C (c,Td) \to \C (Tc,Td)\]
\end{packitem}
 such that the following diagrams commute for all suitable morphisms $f$ and $g$:
% \begin{equation*}
% \begin{xy}
% \xymatrix @=3pc{
% c \ar [r] ^ {\we_c} \ar[rd]_{f} & Pc \ar[d]^{\kl{f}} & Pc \ar@/^1pc/[rd]^{\kl{\we_c}} \ar@/_1pc/[rd]_{\id}& {}\\
% {} & Pd & {} & Pc \\
% }
% \end{xy}
% \end{equation*}
% \begin{equation*}
% \begin{xy}
% \xymatrix @=3pc{
%  Pc \ar[r]^{\kl{f}} \ar[rd]_{\kl{\comp{f}{\kl{g}}}} & Pd\ar[d]^{\kl{g}} \\
%   {} & Pe \\
% }
% \end{xy}
% %\label{mu}
% \end{equation*}
\begin{equation*}
\begin{xy}
\xymatrix @=3pc{
c \ar [r] ^ {\we_c} \ar[rd]_{f} & Tc \ar[d]^{\kl{f}} & Tc \ar@/^1pc/[rd]^{\kl{\we_c}} \ar@/_1pc/[rd]_{\id}& {} &
     Tc \ar[r]^{\kl{f}} \ar[rd]_{\kl{\comp{f}{\kl{g}}}} & Td\ar[d]^{\kl{g}}\\
{} & Td , & {} & Tc , &{} & Te .\\
%c \ar[r]^{T \eta_X} \ar@2{-}[rd] & T^2 X \ar[d]|{\mu_X} \ar@{}[ld] \ar@{}[rd] & T X \ar[l]_{\eta_{T X}} \ar@2{-}[dl] & T^3 X \ar[r]^{\mu_{TX}} \ar[d]_{T\mu_X} & T^2 X \ar[d]^{\mu_X} \\
%& T X, & &T^2 X \ar[r]_{\mu_X} & T X.
}
\end{xy}
\end{equation*}
\end{definition}

\noindent
We also refer to the Kleisli map as ``substitution map'': when $\C$ is instantiated, for example, by the category of sets and $TX$ is a set of 
terms with free variables in the set $X$, then \emph{simultaneous substitution} as Kleisli map turns $T$ into a monad. 
In this case the diagrams express the well--known substitution properties \cite{alt_reus}.
More precisely, the first diagram determines the value of substitution on variables, the second diagram states that substituting each variable by itself 
in a term does not change the term, and the third diagram shows how two consecutive substitutions can be expressed by just one substitution.
Inspired by \textsc{Haskell} syntax, we frequently use the infixed symbol $\bind{}{}$ to denote simultaneous substitution (or more generally, Kleisli maps): 
given a term $M\in TX$ with free variables in $X$ and  
$f : X \to TY$, then \[\bind{M}{f} \enspace := \enspace \sigma(f)(M)\] denotes the term obtained by replacing any free variable $x\in X$ occurring in $M$ 
by its image $f(x)\in TY$, yielding a term in $TY$.

The following remarks recover the definition of monad given in \autoref{def:monad_mu} from the
definition of \autoref{def:endomonad}.

\begin{remark}[Functoriality for Monads in Kleisli Form, \autoref{code:endomonad_functor}]\label{rem:monad_kleisli_funct}
Given a monad $T$ over $\C$ as in \autoref{def:endomonad} and a morphism $f:c \to d$ in $\C$, 
we equip $T$ with a functorial structure by setting
%a is equipped with a functorial structure (\lstinline!lift!) by setting 
\[
   T(f) := \lift_T(f) := \kl{\comp{f}{\we_d}}\enspace .
\]
\end{remark}
\begin{remark}[Naturality of $\eta$ and Multiplication for Monads in Kleisli form]

Given a monad in Kleisli form $T$, the family of morphisms $\eta = (\eta_c : \C(c,Tc))_{c\in\C}$
is natural with respect to the functorial structure defined in \autoref{rem:monad_kleisli_funct}.
A multiplication $\mu : T^2 \to T$ can be defined as substitution with identity:
\[ \mu_c := \sigma(\id_{Tc}) : TTc \to Tc  \enspace . \]
Naturality of $\mu$ is a consequence of the axioms for monads in Kleisli form.
Finally, the monad multiplication $\mu$ thus defined is compatible with the unit $\eta$ in the sense
of \autoref{def:monad_mu}.
\end{remark}

\begin{remark}[Naturality of Substitution]\label{rem:subst_natural_endo}
 Given a monad in Kleisli form $T$ over $\C$, then its substitution $\sigma$ is natural in $c$ and $d$.
 For naturality in $c$ we check that the diagram 
 \[
  \begin{xy}
   \xymatrix{
     c \ar[d]_{f}  &    **[l]\C(c,Td) \ar[r]^{\sigma_{c,d}}  &  **[r]\C (Tc,Td)\\
     c'             &    **[l]\C(c',Td) \ar[r]_{\sigma_{c',d}} \ar[u]^{f^*}  &  **[r]\C (Tc',Td) \ar[u]_{(Tf)^*}
}
  \end{xy}
 \]
commutes, where $f^*(h) := \comp{f}{h}$. Given $g\in \C(c',Td)$, we have 
\begin{align*} 
        \comp{Tf}{\sigma(g)} &= \comp{\sigma(\comp{f}{\eta_{c'}})}{\sigma(g)} \\
                             &\stackrel{3}{=} \sigma(\comp{\comp{f}{\eta_{c'}}}{\sigma(g)}) \\ 
                             &\stackrel{1}{=} \sigma(\comp{f}{g})  \enspace ,   
\end{align*}
where the numbers correspond to the diagrams of \autoref{def:endomonad} used
to rewrite in the respective step.
Similarly we check naturality in $d$. Writing $h_*(g) := \comp{g}{h}$, the diagram

 \[
  \begin{xy}
   \xymatrix{
     d \ar[d]_{h}  &    **[l]\C(c,Td) \ar[r]^{\sigma_{c,d}} \ar[d]_{(Th)_*} &  **[r]\C (Tc,Td)\ar[d]^{(Th)_*} \\
     d'             &    **[l]\C(c',Td) \ar[r]_{\sigma_{c',d}}   &  **[r]\C (Tc',Td) 
}
  \end{xy}
 \]
commutes: given $g\in \C(c, Td)$, we have
\begin{align*} 
        \comp{\sigma(g)}{Th} &= \comp{\sigma(g)}{\sigma(\comp{h}{\eta_{d'}})} \\
                             &\stackrel{3}{=} \sigma (\comp{g}{\sigma(\comp{h}{\eta_{d'}})})\\
                             &= \sigma(\comp{g}{Th})  \enspace .   
\end{align*}
 
\end{remark}

\begin{definition}[Morphism of Monads, Alt.\ to \autoref{def:colax_mon_mor}, \autoref{code:colax_mon_mor_alt}] 
\label{def:colax_mon_mor_alt}
  Let $(\C,T)$ and $(\D,T')$ be two monads.
 A \emph{colax morphism of monads} $\tau : T \to T'$ is given by 

\begin{packitem}
 \item a functor $F\colon \C \to \D$ and 
 \item for any $c\in \C$, a morphism $\tau_c : FTc \to T'Fc$
\end{packitem}
such that the following diagrams commute for all suitable morphisms $f$:
\begin{equation*}
 \begin{xy}
  \xymatrix @=5pc{
  FTc \ar[r]^{F(\kl[T]{f})} \ar[d]_{\tau_c}& FTd \ar[d]^{\tau_d} & Fc \ar[r]^{F\we^T_c} \ar[rd]_{\we^{T'}_{Fc}} & FTc \ar[d]^{\tau_c} \\
  T'Fc \ar[r]_{\kl[T']{\comp{Ff}{\tau_d}}} & T'Fd \enspace ,& {} & T'Fc \enspace .\\
}
 \end{xy}
\end{equation*}
\end{definition}

\begin{remark}
  Naturality of the family $(\tau_c)_{c\in \C}$ of a colax morphism of monads as in the preceding definition 
    is provable from the other axioms, yielding a natural transformation
      \[\tau : FT \to T'F \enspace . \]
   Here we use \autoref{rem:monad_kleisli_funct} by considering $T$ and $T'$ as functors. 
  The naturality of $\tau$ is proved in Lemma \lstinline!colax_Monad_Hom_NatTrans! in the 
  \textsf{Coq} library.
\end{remark}

\begin{definition}[Module, Alt.\ to \autoref{rem:def_module_hom}, \autoref{code:endo_module}]
 \label{def:endo_module}
Let $\D$ be a category. A \emph{module $M$ over $T$ with codomain $\D$} is given by
\begin{packitem}
 \item a map $M\colon \C \to \D$ on the objects of the categories involved and 
 \item for all objects $c,d$ of $\C$, a map 
      \[        \varsigma_{c,d}\colon \C (c,Td) \to \D (Mc,Md)      \]
\end{packitem}
such that the following diagrams commute for all suitable morphisms $f$ and $g$:
\begin{equation*}
\begin{xy}
\xymatrix @=3pc{
 Mc \ar[r]^{\mkl{f}} \ar[rd]_{\mkl{\comp{f}{\kl{g}}}} & Md\ar[d]^{\mkl{g}} & Mc  \ar@/^1pc/[rd]^{\mkl{\we_c}} \ar@/_1pc/[rd]_{\id} & {} \\
  {} & Me ,& {} & Mc .\\
}
\end{xy}
%\label{mu}
\end{equation*}
\end{definition}

\begin{remark}
Functoriality for such a module $M$ is defined similarly to that for monads:
for any morphism $f:c\to d$ in $\C$ we set
\[M(f) :=  \mlift_M(f) := \mkl{\comp{f}{\we^T}}\enspace .\]
\end{remark}

\noindent
A \emph{module morphism} is a family of morphisms that is compatible with module substitution:
\begin{definition}[Module Morphism, Alt.\ to \autoref{rem:def_module_hom}, \autoref{code:endo_mod_hom}]
     \label{def:endo_mod_hom} %[Mod\textunderscore Hom]
 Let $M$ and $N$ be two modules over $T$ with codomain $\D$. 
 A \emph{morphism of $T$--modules} from $M$ to $N$ is given by a family of 
  morphisms $\rho_c\in\D(Mc,Nc)$ such that for all morphisms $f\in \C(c,Td)$ the following diagram commutes:
\begin{equation*}
 \begin{xy}
  \xymatrix @=3pc{
  Mc \ar[r]^{\mkl[M]{f}} \ar[d]_{\rho_c}& Md \ar[d]^{\rho_d}  \\
  Nc \ar[r]_{\mkl[N]{f}} & Nd. \\
}
 \end{xy}
\end{equation*}
\end{definition}

\noindent
A module morphism $M \to N$ also constitutes a natural transformation between the functors $M$ and $N$ induced by the modules, cf.\ \lstinline!Module_Hom_NatTrans!.

\begin{example}[\autoref{ex:ulc_mod_mor} cont.]
       \label{ex:ulc_mod_mor_kl}
  We consider \autoref{ex:ulc_mod_mor} under the alternative definition of module morphism.
  The map \[V \mapsto \App_V : \LC(V)\times \LC(V) \to \LC(V)\]
    satisfies the diagram of the preceding definition and is hence a morphism of $\LC$--modules from $\LC\times\LC$ to $\LC$.
The property of being a module morphism expresses distributivity of substitution for any substitution map  $f : X\to \LC(Y)$:

\[ \bind{\App(M,N)}{f} \enspace = \enspace \App(\bind{M}{f}, \bind{N}{f}) \enspace .  \]
Similarly, the map 
\[ V \mapsto \Abs_V : \LC(V') \to \LC(V) \] 
is a morphism of $\LC$--modules from $\LC'$ to $\LC$.
  For $f:X\to\LC(Y)$ as before, the commutative diagram here expresses the equation

\[  \bind{\Abs(M)}{f} \enspace = \enspace \Abs(\bind{M}{f'}) \enspace , \]
where $f': X' \to \LC(Y')$ is obtained by shifting the map $f$ to account for the extended context under
the binder $\Abs$.

\begin{comment}
\begin{equation*}
 \begin{xy}
  \xymatrix @=3pc{
  **[l]\LC(X') \ar[r]^{(\bind{\_}{f'})} \ar[d]_{\Abs_X}& **[r]\LC(Y') \ar[d]^{\Abs_Y}  \\
  \LC(X) \ar[r]_{\bind{\_}{f}} & \LC(Y). \\
}
 \end{xy}
\end{equation*}
\end{comment}
\end{example}

Modules on $P$ with codomain $\D$ and morphisms between them form a category 
called $\Mod{P}{\D}$ (in the library: \lstinline!MOD P D!), similar to the category of monads.

%% file: relative_monads.tex
\section{Relative Monads and Modules}\label{sec:rel_monads}

The functors underlying the monads presented in the preceding section all are endo\-functors.
This is enforced by the type of monadic multiplication and substitution.
\emph{Relative monads} were defined by Altenkirch et al.\ \cite{DBLP:conf/fossacs/AltenkirchCU10} to overcome this
restriction. One of their motivations was to consider the untyped lambda calculus over \emph{finite} contexts as 
a monad--like structure --- similar to the monad structure on the lambda calculus over 
arbitrary contexts exhibited by Altenkirch and Reus \cite{alt_reus}.

We review the definition of relative monads and define suitable \emph{colax morphisms of relative monads}.
Afterwards we define \emph{modules over relative monads} and port the constructions on 
modules over monads (cf.\ \autorefs{subsection:mod_examples} and \ref{sec:monads_on_set_families}) to modules over relative monads.

\subsection{Definitions}\label{sec:rel_mon_defs}

We review the definition of relative monad as given by Altenkirch et al.\ \cite{DBLP:conf/fossacs/AltenkirchCU10}
and define suitable morphisms for them.
As an example we consider the lambda calculus as a relative monad from sets to preorders, on the functor $\Delta$ (cf.\ \autoref{def:delta}).
Afterwards we define \emph{modules over relative monads} and carry over the constructions on modules over regular monads
of the preceding section to modules over relative monads. 

The definition of relative monads is analogous to that of monads in Kleisli form (cf.\ \autoref{def:endomonad}),
except that the underlying map of objects is between \emph{different} categories.
Thus, for the operations to remain well--typed, one needs an additional ``mediating'' functor, in the following 
usually called $F$, which is inserted wherever necessary:

\begin{definition}[Relative Monad, \cite{DBLP:conf/fossacs/AltenkirchCU10}, \autoref{code:relative_monad}]\label{def:relative_monad}
Given categories $\C$ and $\D$ and a functor $F : \C\to \D$, a \emph{relative monad $P : \C \stackrel{F}{\to}\D$ on $F$}
is given by the following data:
\begin{packitem}
 \item a map $P\colon \C\to\D$ on the objects of $\C$,
 \item 
for each object $c$ of $\C$, a morphism $\eta_c\in \D(Fc,Pc)$ and
 \item 
for each two objects $c,d$ of $\C$, a \emph{substitution} map 
 \[\sigma_{c,d}\colon \D (Fc,Pd) \to \D(Pc,Pd)\]
\end{packitem}
such that the following diagrams commute for all suitable morphisms $f$ and $g$:
\begin{equation*}
\begin{xy}
\xymatrix @=4pc{
Fc \ar [r] ^ {\we_c} \ar[rd]_{f} & Pc \ar[d]^{\kl{f}} & Pc \ar@/^1pc/[rd]^{\kl{\we_c}} \ar@/_1pc/[rd]_{\id}& {} & 
 Pc \ar[r]^{\kl{f}} \ar[rd]_{\kl{\comp{f}{\kl{g}}}} & Pd\ar[d]^{\kl{g}} \\
{} & Pd \enspace , & {} & Pc \enspace , & {} & Pe \enspace .\\
%c \ar[r]^{T \eta_X} \ar@2{-}[rd] & T^2 X \ar[d]|{\mu_X} \ar@{}[ld] \ar@{}[rd] & T X \ar[l]_{\eta_{T X}} \ar@2{-}[dl] & T^3 X \ar[r]^{\mu_{TX}} \ar[d]_{T\mu_X} & T^2 X \ar[d]^{\mu_X} \\
%& T X, & &T^2 X \ar[r]_{\mu_X} & T X.
}
\end{xy}
\end{equation*}
% \begin{equation*}
% \begin{xy}
% \xymatrix @=3pc{
%  Pc \ar[r]^{\kl{f}} \ar[rd]_{\kl{\comp{f}{\kl{g}}}} & Pd\ar[d]^{\kl{g}} \\
%   {} & Pe \\
% }
% \end{xy}
% %\label{mu}
% \end{equation*}
% 
\end{definition}

\begin{example}[Lambda Calculus over Finite Contexts, \cite{DBLP:conf/fossacs/AltenkirchCU10}]
  Altenkirch et al.\ \cite{DBLP:conf/fossacs/AltenkirchCU10} consider the untyped lambda calculus
 as a relative monad on the functor $J:\Fin_{\mathrm{skel}} \to\Set$.
 Here the category $\Fin_{\mathrm{skel}}$ is the category of finite cardinals, i.e.\ 
 the skeleton of the category $\Fin$ of \emph{finite}
  sets and maps between finite sets.
\end{example}

\begin{remark}
 Relative monads on the identity functor $\Id:\C\to\C$ precisely correspond to monads as
 presented in \autoref{def:endomonad}.
\end{remark}

\begin{notation}
For this section we reserve the term ``monad'' for monads as defined in \autoref{def:endomonad}, and
explicitly state the ``relative'' when talking about relative monads.
In later sections we sometimes omit the attribute ``relative'' and instead 
refer to traditional monads (i.e.\ with $F= \Id$) as \emph{regular} or \emph{plain} monads. 
\end{notation}

\begin{remark}[Restricting a Monad yields a Relative Monad, \textup{\cite{DBLP:conf/fossacs/AltenkirchCU10}}]
 Given a monad $T$ on $\D$ and a functor $F:\C\to\D$, then the monad $T$ 
  restricts to a relative monad $T^\flat : \C\stackrel{F}{\to}\D$ by precomposing with $F$.
\end{remark}

\begin{remark}[Relative Monads are functorial, \autoref{code:rel_monad_functorial}]
  \label{rem:rel_monad_functorial}
Given a monad $P$ over $F : \C\to\D$ and a morphism $f:c\to d$ in $\C$, a functorial structure (\lstinline!rlift!)
for $P$ is defined by setting
\[
   P(f) := \lift_P(f) := \kl{\comp{Ff}{\we}} \enspace.
\]
The functor axioms are easily proved from the monadic axioms.
\end{remark}

\begin{remark}[Relative Monads as Monoids in a Functor Category, \textup{\cite{DBLP:conf/fossacs/AltenkirchCU10}}]
  A monad $(T,\eta,\mu)$ over a category $\C$ is the same as a monoid object in the functor category 
 $[\C,\C]$, where the monoidal structure is given by functor composition.
  Altenkirch et al.\ \cite{DBLP:conf/fossacs/AltenkirchCU10} recover 
  a similar characterization for relative monads on a functor $F:\C\to\D$, 
  provided that the left Kan extension along $F$,
    \[ \Lan_F : [\C,\D] \to [\D,\D] \enspace , \]
 exists:
they define a lax monoidal structure on $[\C,\D]$ by
   \begin{align*}   
         (\cdot^F):[\C,\D]\times [\C,\D] &\to [\C,\D] \\
                       (H, G) &\mapsto H \cdot^F G := \comp{G}{\Lan_F H} \enspace .
   \end{align*}
  They then show that relative monads on $F$ correspond precisely to lax monoid objects in $([\C,\D], \cdot^F )$.
 Besides, they show that under some coherence conditions, this result can be sharpened to obtain a strict
  monoidal structure, where relative monads correspond to proper monoids with respect to this structure.
 Under the same assumptions, a relative monad $P$ on $F:\C\to\D$ can be extended to a traditional monad $P^\sharp$ 
 on $\D$, yielding an adjunction $(\_)^\sharp \dashv (\_)^\flat$. This adjunction furthermore is 
  a coreflection.
\end{remark}

\begin{remark}[Naturality of Substitution]\label{rem:subst_natural_rel}
  Analogously to \autoref{rem:subst_natural_endo}, 
    the substitution $\sigma = (\sigma_{c,d})$ of a relative monad $P$ on a functor $F:\C\to\D$
   is binatural. 
\end{remark}

We are interested in monads on the category $\Set$ of sets and relative monads on 
$\Delta:\Set\to\PO$ as well as their relationship:

\begin{lemma}[Relative Monads on $\Delta$ and Monads on $\Set$]\label{lem:rmon_delta_endomon}
  Let $P$ be a relative monad on $\Delta:\Set\to\PO$ (cf.\ \autoref{def:delta}).
 By postcomposing with the forgetful functor $U:\PO\to \Set$ we obtain a monad 
   \[UP : \Set\to \Set\enspace . \]
   The substitution is defined, for $m : X \to UPY$ by setting
   \[U\sigma : m\mapsto  U \left(\kl{\varphi^{-1} m}\right) \enspace ,  \]
as indicated by the diagram
\[ 
     \begin{xy}
      \xymatrix {  **[l] \Set(X, UPY) \ar[r]^{U \sigma} \ar[d]_{\varphi^{-1}} & **[r] \Set (UP X , UP Y) \\
                   **[l] \PO(\Delta X, PY) \ar[r]_{\sigma}&   **[r]    \PO (P X, P Y) \ar[u]_{U}
      } 
     \end{xy}
    \]
making use of the adjunction $\varphi$ of \autoref{lem:adj_set_po}.

Conversely, to any monad $T$ over $\Set$, given as a Kleisli triple, we associate a relative monad over $\Delta$ by 
 postcomposing with $\Delta$. The substitution map $\Delta\sigma$ is defined, for $m : \Delta X \to \Delta T Y$, as the following composition:
    \[ 
     \begin{xy}
      \xymatrix {  **[l] \PO(\Delta X, \Delta T Y) \ar[r]^{\Delta \sigma} \ar[d]_{U} & **[r] \PO (\Delta T X , \Delta T Y) \\
                    **[l] \Set(X, T Y) \ar[r]_{\sigma}&    **[r]   \Set (T X, T Y) \ar[u]_{\varphi^{-1}}
      } 
     \end{xy}
    \]
The maps thus defined are object functions of an adjunction between monads on sets and relative monads on $\Delta$, 
cf.\ \autoref{lem:adj_mon_rmon}.

\end{lemma}

The above construction actually is an instance of a more general construction:

\begin{lemma}[Monads from Relative Monads and conversely]
    Let $F:\C\rightleftarrows \D : G$ be an adjunction with a family of isomorphisms 
     \[\varphi_{X,Y} : \D(FX,Y) \cong \C(X,GY):\varphi^{-1}_{X,Y} \enspace . \]

 \begin{enumerate}  
  \item 
Given a relative monad $P : \C\stackrel{F}{\to}\D$ with unit $\eta$ and substitution $\sigma$, 
  we define a monad $P^+$ on $\C$
  by setting
  \begin{packitem}\renewcommand{\labelitemi}{}
   \item $P^+(c) := GPc$,
   \item $\eta^+_c := \varphi(\eta_c) : \C(c, GPc) $ and
   \item $\sigma^+_{c,d}(f) := G\Bigl(\sigma\bigl(\varphi^{-1}(f)\bigr)\Bigr)$.
  \end{packitem}

 \item Let furthermore $GF = \Id$ be the identity on $\C$.
   Given a monad $(P,\eta,\sigma)$ on $\C$, we define a relative monad $P^-:\C\stackrel{F}{\to}\D$
   by setting
 \begin{packitem}\renewcommand{\labelitemi}{}
   \item $P^-(c) := FPc$,
   \item $\eta^-_c := F(\eta_c)$ and
   \item $\sigma^-_{c,d} := \varphi^{-1}\bigl(\sigma(Gf)\bigr)$.
 \end{packitem}

\end{enumerate}

\end{lemma}

% \begin{proof}
\proof
 We check the commutativity of the corresponding diagrams:
 \begin{enumerate}
  \item for the data $(P^+,\eta^+,\sigma^+)$:
    \begin{itemize}
         \item $\comp{\eta_c}{\sigma^+(f)} = \comp{\varphi(\eta_c)}{G(\sigma(\varphi^{-1}f))} = 
                  \varphi(\comp{\eta_c}{\sigma(\varphi^{-1}f)}) = \varphi\varphi^{-1}f = f$
          \item $\sigma^+(\eta^+_c) = G(\sigma(\varphi^{-1}(\varphi(\eta_c)))) = G(\sigma(\eta_c)) = G\id = \id$
          \item \begin{align*}
                     \comp{\sigma^+(g)}{\sigma^+(g)} &= \comp{G\sigma\varphi^{-1}f}{G\sigma\varphi^{-1}g} \\
                                                     &= G\bigl(\comp{\sigma(\varphi^{-1}f)}{\sigma(\varphi^{-1}g)}\bigr)  \\ % fusion for sigma
                               &= G\bigl(\sigma(\comp{\varphi^{-1}f}{\sigma(\varphi^{-1}g)})\bigr) \\  % naturality of \varphi^{-1}
                               &= G\bigl( \sigma(\varphi^{-1}(\comp{f}{G(\sigma(\varphi^{-1}g))})) \bigr) \\
                               &= \sigma^+(\comp{f}{\sigma^+(g)})
                \end{align*}

    \end{itemize}
   
   \item for the data $(P^-, \eta^-, \sigma^-)$:
    \begin{itemize}
     \item $\comp{\eta^-_c}{\sigma^-(f)} = \comp{F\eta_c}{\varphi^{-1}(\sigma(Gf))} = \varphi^{-1}(\comp{\eta_c}{\sigma(Gf)}) = \varphi^{-1}(Gf) = f$ 
     \item $\sigma^-(\eta^-_c) = \varphi^{-1}(\sigma(G\eta^-_c)) = \comp{F(\sigma(GF\eta_c))}{\epsilon_{FPC}} = \comp{F(\sigma(\eta_c))}{\epsilon_{FPc}} = 
                                                      \comp{F\id}{\epsilon_{FPc}} = \id$
     \item \begin{align*}
            \sigma^-(\comp{f}{\sigma^-g}) &= \varphi^{-1}\Bigl( \sigma\bigl( G(\comp{f}{\sigma^-g}) \bigr)\Bigr) \\
                                          &= \varphi^{-1}\Bigl( \sigma\bigl( \comp{Gf}{G\sigma^-g} \bigr)\Bigr) \\
                                          &= \varphi^{-1}(\sigma(\comp{Gf}{\sigma Gg})) \\
                                          &= \varphi^{-1}(\comp{\sigma Gf}{\sigma Gg}) \\
                                          &= \comp{F(\comp{\sigma Gf}{\sigma Gg})}{\epsilon_F} \\
                                          &= \comp{\comp{F\sigma Gf}{F\sigma Gg}}{\epsilon_F} \\
                                          &= \comp{\comp{\comp{F\sigma Gf}{\epsilon_F}}{F\sigma Gg}}{\epsilon_F} \\
                                          &= \comp{\varphi^{-1}\sigma Gf}{\varphi^{-1}\sigma Gg}\\
                                          &= \comp{\sigma^-f}{\sigma^-g}
           \end{align*}
\qed
    \end{itemize}
 
 \end{enumerate}

% \end{proof}
\noindent
This construction is functorial, and yields an adjunction between a category of 
monads on $\C$ and relative monads on $F$. Details will be reported elsewhere.

\begin{example}[Lambda Calculus as Relative Monad on $\Delta$]\label{ex:ulcbeta}
Consider the set of all lambda terms indexed by their set of free variables as defined in \autoref{ex:ulc_def}.
%\[ \ULC (V) ::= Var : V \to \ULC(V) \mid \lambda : ULC (V^*) \to ULC (V) \mid \app : ULC(V) \to ULC(V)\to ULC(V) \]
% \begin{align*} 
% \ULC (V) ::=\quad &\Var : V \to \ULC(V) \\ 
%         {}\mid{} &\Abs : \ULC (V^*) \to \ULC (V) \\ 
%         {}\mid{} &\App : \ULC(V) \to \ULC(V)\to \ULC(V),
% \end{align*}
We write $\lambda M$ and $MN$ for $\Abs M$ and $\App M N$, respectively.
We equip each $\ULC(V)$ with a preorder taken as the reflexive--transitive closure of the 
relation generated by the rule
\begin{equation*}   \quad (\lambda M) N ~ \leq ~ M [*:= N] \end{equation*}
and its propagation into subterms.
%\begin{packenum}
  %\setlength{\itemsep}{.5ex}%
  %\setlength{\parskip}{0cm}%
% \item $(\lambda M) N \leq M [*:= N]$\label{eq:generating},
% \item $M \leq M' \Longrightarrow \App (M,N) \leq \App (M',N)$\label{eq:app1},
% \item $N \leq N' \Longrightarrow \App (M,N) \leq \App (M,N')$\label{eq:app2} and
% \item $M \leq M' \Longrightarrow \lambda M  \leq \lambda M'$\label{eq:abs}.
%\end{packenum}
 This defines a monad \lstinline!ULCBETA! from sets to preorders over the functor $\Delta$,
    \[\ULCB : \SET\stackrel{\Delta}{\to}\PO.\] 
 The family $\we^{\ULC}$ is given by the constructor $\Var$, and the substitution map 
  \[\sigma_{X,Y} : \PO\bigl(\Delta (X), \ULCB(Y)\bigr) \to \PO\bigl(\ULCB(X),\ULCB(Y)\bigr)\]
 is given by capture--avoiding simultaneous substitution. % (cf.\ Code snippet \ref{code:ulc_subst}). 
Via the adjunction of \autoref{lem:adj_set_po} the substitution can also be read as
%The functor $\Delta$ is left adjoint to the forgetful functor $F : \PO\to\SET$, so that the substitution can also be read as
  \[\sigma_{X,Y} : \SET\bigl(X, \ULC(Y)\bigr) \to \PO\bigl(\ULCB(X),\ULCB(Y)\bigr) \enspace .\]

\end{example}

\begin{remark}[about Substitution] \label{rem:about_substitution}
 The substitution in \autoref{ex:ulcbeta} is compatible with the order on terms in the following sense:
 \begin{packenum}
  \item $M \leq N \enspace \text{implies} \enspace M[*:=A] \leq N[*:=A] \enspace$ and
  \item $A \leq B \enspace \text{implies} \enspace M[*:=A] \leq M[*:=B]$. %\label{item:covar_in_subst_arg}
 \end{packenum}
 The first implication is a general fact for any relative monad $P$ on $\Delta$: it is a special case of $\sigma_{X,Y}(f)$ being 
   a morphism in the category $\PO$ for any $f\in \PO(\Delta V,PW)$.
The second monotony property, however, is \emph{false} in general.
As an example, consider the monad given by
\begin{align*} 
F(V) ::=\quad & \enspace \Var : V \to F(V) \\ 
         {}\mid{}& \enspace \bot : F(V) \\
        {}\mid{} & \enspace (\Rightarrow) : F(V) \times F(V)\to F(V)
\end{align*}
equipped with a preorder which is contravariant in the first argument of the arrow constructor $\Rightarrow$.
Substituting in this position, the first argument of $(\Rightarrow)$, does in fact \emph{reverse}
the order on terms, i.e.\ we obtain (using $\Rightarrow$ infixed)
\[ A < B  \quad\text{implies} \quad (* \Rightarrow M) [*:=B] < (*\Rightarrow M)[*:=A] \enspace . \]
A different definition of monad which would enforce the second implication to hold --- and hence not include
 the example $F$ --- can be given easily by considering $\PO$ as a 2--category enriched over itself: 
given morphisms $f, g \in \PO(X,Y)$ we say that there is precisely one 2--cell 
\[ f \Rightarrow g  \quad\text{iff}\quad
 f \leq g \quad\text{iff}\quad \forall x : X, f(x) \leq g(x) \enspace . \]
A monad $P$ would then have to be equipped with a substitution action that is given, for any two sets $V$ and $W$, by \emph{a functor} (of preorders)
\[ \sigma_{V,W} : \PO(\Delta V,PW) \to \PO(PV,PW) \enspace . \]
\Autoref{def:hat_P_subst} explains one of the consequences of our monadic substitution lacking ``higher--order monotonicity''.

\end{remark}

We generalize the definition of colax monad morphisms to relative monads:

\begin{definition}[Colax Morphism of Relative Monads, \autoref{code:colax_rel_mon_mor}] 
 \label{def:colax_rel_mon_mor}
Let $P : \C\stackrel{F}{\to}\D$ and 
$Q : \C'\stackrel{F'}{\to}\D'$ be two relative monads. A \emph{colax morphism of relative monads
from $P$ to $Q$} is given by a quadruple $(G,G', N, \tau)$ consisting of
a functor $G\colon \C\to\C'$ and a functor $G': \D \to D'$ as well as a natural transformation
$N: F'G \to G'F$ as in
\[
 \begin{xy}
  \xymatrix@!=2.5pc{\C \ar[r]^{F} \ar[d]_{G} & \D \ar[d]^{G'}\\
             \C'\ar[r]_{F'} & \D'. \ultwocell<\omit>{N}
}
 \end{xy}
\]
 and a  natural transformation
 $\tau : \comp{P}{G'} \to \comp{G}{Q}$ as in 
\[
 \begin{xy}
  \xymatrix@!=2.5pc{\C \ar[r]^{P} \ar[d]_{G} \drtwocell<\omit>{\;\;\tau}& \D \ar[d]^{G'}\\
             \C'\ar[r]_{Q} & \D',
}
 \end{xy}
\]
%collection of morphisms $\tau_c\in \D'(G'Pc,QGc)$ 
such that the following diagrams commute for all suitable morphisms $f$:
\begin{equation*} %\label{eq:colax_rmon_mor}
 \begin{xy}
  \xymatrix @=5pc{
  G'Pc \ar[r]^{G'\kl[P]{f}} \ar[d]_{\tau_c}& G'Pd \ar[d]^{\tau_d} \\
  QGc \ar[r]_{\kl[Q]{\comp{Nc}{\comp{G'f}{\tau_d}}}} & QGd 
}
 \end{xy}
% \end{equation}
% \begin{equation}
\qquad
\begin{xy}
  \xymatrix @=5pc {
 F'Gc \ar[r]^{Nc} \ar[rrd]_{\we^Q_{Gc}} & 
                 G'F c \ar[r]^{G'\we^P_c}& G'Pc \ar[d]^{\tau_c} \\
{} & {} & QGc.
}
\end{xy}
\end{equation*}
\end{definition}
\begin{remark}
 Naturality of $\tau$ in the preceding definition is actually a consequence of the commutative diagrams of \autoref{def:colax_rel_mon_mor}, 
 cf.\
 Lemma \lstinline!colax_RMonad_Hom_NatTrans! in the \textsf{Coq} library.
\end{remark}

\begin{remark} \label{rem:rel_mon_mor_case}
  In \autoref{chap:comp_types_sem} we are going to use the following instance of the preceding definition:
  the categories $\C$ and $\C'$ are instantiated by $\TS{T}$ and $\TS{T'}$, respectively, for sets $T$ and $T'$.
  The functor $G$ is the retyping functor (cf.\ \autoref{def:retyping_functor}) associated to some 
  translation of types $g : T \to T'$.
  Similarly, the categories $\D$ and $\D'$ are instantiated by $\TP{T}$ and $\TP{T'}$, and 
  the functor $F$ by
  \[F := \TDelta{T} : \TS{T} \to \TP{T} \enspace ,\]
  and similar for $F'$:
\[
 \begin{xy}
  \xymatrix@!=2.5pc{
       **[l] \TS{T} \ar[r]^{\TDelta{T}} \ar[d]_{\retyping{g}}   & **[r] \TP{T} \ar[d]^{\retyping{g}}\\
       **[l] \TS{T'}\ar[r]_{\TDelta{T'}} & **[r]\TP{T'}. \ultwocell<\omit>{\Id}
}
 \end{xy}
\]
\end{remark}

\noindent
Given a monad $P$ on $F:\C\to\D$, the notion of \emph{module over $P$} generalizes the notion of monadic substitution:

\begin{definition}[Module over a Relative Monad, \autoref{code:rmodule}]
 \label{def:rmodule}
Let $P\colon\C\stackrel{F}{\to}\D$ be a relative monad and let $\E$ be a category. A \emph{module $M$ over $P$ with codomain $\E$} is given by
\begin{packitem}
 \item a map $M: \C \to \E$ on the objects of the categories involved and 
 \item for all objects $c,d$ of $\C$, a map 
      \[ 
          \varsigma_{c,d} : \D (Fc,Pd) \to \E (Mc,Md)
      \]
\end{packitem}
such that the following diagrams commute for all suitable morphisms $f$ and $g$:
\begin{equation*}
\begin{xy}
\xymatrix @=3pc{
 Mc \ar[r]^{\mkl{f}} \ar[rd]_{\mkl{\comp{f}{\kl{g}}}} & Md\ar[d]^{\mkl{g}} & Mc  \ar@/^1pc/[rd]^{\mkl{\we_c}} \ar@/_1pc/[rd]_{\id} & {} \\
  {} & Me & {} & Mc. \\
}
\end{xy}
%\label{mu}
\end{equation*}
\end{definition}

\noindent
A functoriality (\lstinline!rmlift!) for such a module $M$ is then defined similarly to that for monads:
 for any morphism $f : c \to d$ in $\C$ we set
 \[ M(f) := \rmlift_M(f) := \varsigma({ \comp{Ff}{\eta}}) \enspace .\]

\noindent
The following examples of modules are instances of constructions explained in the next section:

\begin{example}[\Autoref{ex:ulcbeta} cont.]\label{ex:ulcb_taut_mod}
 The map $\ULCB : V \mapsto \ULCB(V)$ yields a module over the relative monad $\ULCB$, the \emph{tautological module} $\ULCB$.
\end{example}

\begin{example}\label{ex:ulcb_der_mod}
 Recall that $V' := V + 1$. The map $\ULCB' : V \mapsto \ULCB(V')$ inherits the structure of an $\ULCB$--module from the
  tautological module $\ULCB$ (cf.\ \autoref{ex:ulcb_taut_mod}).
We call $\ULCB'$ the \emph{derived module} of the module $\ULCB$; cf.\ also \autoref{subsection:rmod_examples}.
\end{example}
\begin{example}\label{ex:ulcb_prod_mod}
 The map $V\mapsto \ULCB(V)\times\ULCB(V)$ inherits a structure of an $\ULCB$--module from the tautological module $\ULCB$.
\end{example}

%\noindent
A \emph{module morphism} is a family of morphisms that is compatible with module substitution in the source and target modules:
\begin{definition}[Morphism of Relative Modules, \autoref{code:rel_mod_mor}]\label{def:rel_mod_mor}
 Let $M$ and $N$ be two relative modules over $P\colon\C\stackrel{F}{\to}\D$ with codomain $\E$. 
A \emph{morphism of relative $P$--modules} from $M$ to $N$ is given by a collection of morphisms $\rho_c\in\E(Mc,Nc)$ 
such that for all morphisms $f\in \D(Fc,Pd)$ the following diagram commutes:
\begin{equation*}
 \begin{xy}
  \xymatrix @=3pc{
  Mc \ar[r]^{\mkl[M]{f}} \ar[d]_{\rho_c}& Md \ar[d]^{\rho_d}  \\
  Nc \ar[r]_{\mkl[N]{f}} & Nd.\\
}
 \end{xy}
\end{equation*}
\end{definition}

\noindent
The modules over $P$ with codomain $\E$ and morphisms between them form a category called $\RMod{P}{\E}$ 
(in the digital library: \lstinline!RMOD P E!). 
Composition and identity morphisms of modules are defined by pointwise composition and identity, similarly to the category of monads.

\begin{example}[\Autorefs{ex:ulcb_taut_mod}, \ref{ex:ulcb_der_mod}, \autoref{ex:ulcb_prod_mod} cont.]\label{ex:ulcb_constructor_mod_mor}
 Abstraction and application are morphisms of $\ULCB$--modules:
 \begin{align*} \Abs &: \ULCB' \to \ULCB \enspace , \\
                \App &: \ULCB \times \ULCB \to \ULCB \enspace .
 \end{align*}
\end{example}

\subsection{Constructions on Relative Monads and Modules}\label{subsection:rmod_examples}
The following constructions are analogous to those of \autoref{subsection:mod_examples}.

\begin{definition} [Tautological Module] 
Every monad $P$ on $F:\C\to\D$ yields a module $(P,\sigma^P)$ --- also denoted by $P$ --- over itself, 
i.e.\ an object in the category $\RMod{P}{\D}$. % (\lstinline!Taut_RMod!).
\end{definition}

\begin{definition}[Constant and Terminal Module] 
Let $P$ be a monad on $F:\C\to\D$.
For any object $e \in \E$ the constant map $T_e\colon\C\to\E$, $c\mapsto e$ for all $c\in \C$, is equipped with the structure of a $P$--module
by setting $\varsigma_{c,d}(f) = \id_e$. 
In particular, if $\E$ has a terminal object $1_\E$, then the constant module $T_{1_\E} : c \mapsto 1_\E$ is terminal in $\RMod{P}{\E}$.
% {(\lstinline!Const_RMod!, \lstinline!RMOD_Terminal!)}
\end{definition}

\begin{definition}[Postcomposition with a functor]
 Let $P$ be a monad on $F:\C\to\D$, and let $M$ be a $P$--module with codomain $\E$.
 Let $G: \E \to \X$ be a functor. Then the object map $\comp{M}{G}:\C\to \X$ defined by $c\mapsto G(M(c))$ 
 is equipped with a $P$--module structure by setting, for $c, d \in \C$ and $f\in \D(Fc,Pd)$,
  \[  \varsigma^{\comp{M}{G}}(f) := G(\varsigma^M(f)) \enspace . \]
 For $M:=P$ and $G$ a constant functor mapping to an object $x\in \X$ and its identity morphism $\id_x$, 
 we obtain the constant module $(T_x,\id)$ as in the preceding definition.
\end{definition}

\begin{definition}[Pullback Module] 
 Suppose given two relative monads $P$ and $Q$ and a morphism $\tau : P \to Q$ as in \autoref{def:colax_rel_mon_mor}.
Let $N$ a $Q$-module with codomain $\E$. We define a $P$-module $h^* M$ to $\E$ with object map 
 \[c\mapsto M (Gc) \]
by defining the substitution map, for $f : Fc \to Pd$, as
 \[\mkl[h^*M] f := \mkl[M]{\comp{N_c}{\comp {G'f}{h_d}}} \enspace . \] 
The module thus defined is called the \emph{pullback module of $N$ along $h$}. 
The pullback extends to module morphisms and is functorial. % (\lstinline!colax_PbRMod!).
\end{definition}

\begin{definition}[Induced Module Morphism]
\label{def:ind_rmod_mor}%PbRMod_ind_Hom 
With the same notation as before, the monad morphism $h$ induces a morphism of $P$--modules $h:G'P \to h^*Q$. % cf.\ \lstinline!colax_PbRMod_ind_Hom!.
Note that the domain module is the module obtained by postcomposing $P$ with $G'$, 
whereas for (plain) monads the module was just the tautological module of the domain monad. % {(\lstinline!colax_PbRMod_ind_Hom!)}.
\end{definition}

\begin{definition}[Product] 
Suppose the category $\E$ is equipped with a product. Let $M$ and $N$ be $P$--modules with codomain $\E$. Then the map 
\[ M \times N : \C\to\E, \quad c \mapsto Mc \times Nc \] 
 is canonically equipped with a substitution and thus constitutes a module called the \emph{product of $M$ and $N$}. 
 This construction extends to a product on $\RMod{P}{\E}$. % {(\lstinline!RMOD_PROD!)}.
\end{definition}

\subsection{Derivation \& Fibre} \label{sec:deriv_and_fibre}

We are particularly interested in monads on the functor $\TDelta{T}: \TS{T}\to\TP{T}$ for some set $T$, 
 and modules over such monads.
The constructions on modules over monads of \autoref{sec:monads_on_set_families},
derivation (cf.\ \autoref{sec:endo_deriv}) and fibre modules (cf.\ \autoref{sec:fibre}), carry over to 
modules over monads on $\family{\Delta}{T}$.

\begin{definition}
  \label{def:rel_module_deriv}
Given a monad $P$ over $\TDelta{T}$ and a $P$--module $M$ with codomain $\E$, 
we define the derived module of $M$ with respect to $u\in T$ by setting
\[ M^u(V) := M (V^{*u}) \enspace . \]
The module substitution is defined, for $f\in \TP{T}(\TDelta{T}V, PW)$, by 
 \[ \mkl[M^u]{f} := \mkl[M]{\shift{f}{u}} \enspace .\]
Here the ``shifted'' map 
   \[ \shift{f}{u} \in \TP{T}\bigl(\TDelta{T}(V^{*u}), P(W^{*u})\bigr) \]
is the adjunct under the adjunction of \autoref{rem:adj_set_po_typed}  %$\shift{f}{t}$ 
of the coproduct map
%\[ \varphi(\shift{f}{t}) := \defaultmap(\comp{f}{P(i)}, \we(*)), \]
 \[ \varphi (\shift{f}{u}) := [\comp{f}{P(\inl)}, \we(\inr(*))] : V^{*u} \to UP(W^{*u}) \enspace , \]
where $[\inl, \inr] = \id : W^{*u}\to W^{*u}$.
Derivation is an endofunctor on the category of $P$--modules with codomain $\E$. 
\end{definition}

\begin{notation}\label{not:deriv_rmod_untyped}
 In case the set $T$ of types is $T = \{*\}$ the singleton set of types, i.e.\ when talking about untyped syntax,
we denote by $M'$ the derived module of $M$.
  Given a natural number $n$, we denote by $M^n$ the module obtained by deriving $n$ times the module $M$.
\end{notation}

Analogously to \autoref{sec:monads_on_set_families}, we derive more generally with respect to a natural transformation $\tau : 1 \to \T U_n$ as in
 \autoref{def:derived_mod_II}:
\begin{definition}[Derived Module]\label{def:derived_rel_mod_II}
  Let $\tau : 1 \to \T U_n$ be a natural transformation.
 Let $T$ be a set and $P$ be a relative monad on $\TDelta{T}_n$.
 \noindent
Given any $P$--module $M$, we call \emph{derivation of $M$ with respect to $\tau$} the module with object map
$M^{\tau} (V):= M\left(V^{\tau(V)}\right)$.
\end{definition}

\begin{definition}
 Let $P$ be a relative monad over $F$, and $M$ a $P$--module with codomain $\E^T$ for some category $\E$.
 The \emph{fibre module $\fibre{M}{t}$ of $M$ with respect to $t\in T$} has object map
  \[ c \mapsto M(c)(t) = M(c)_t \]
  and substitution map
  \[ \mkl[{\fibre{M}{t}}]{f} := \bigl(\mkl[M]{f}\bigr)_t \enspace .  \]
\end{definition}

\noindent
This definition generalizes to fibres with respect to a natural transformation as in \autoref{def:derived_rel_mod_II}.

The pullback operation commutes with products, derivations and fibres :

\begin{lemma} \label{lem:rel_pb_prod}Let $\C$ and $\D$ be categories and $\E$ be a category with products. Let $P\colon \C\to \D$ and $Q\colon \C\to D$ be monads
over $F:\C\to\D$ and $F':\C'\to \D'$, resp., and $\rho : P \to Q$ a monad morphism. 
Let $M$ and $N$ be $P$--modules with codomain $\E$. The pullback functor is cartesian:
 \[ \rho^* (M \times N) \cong \rho^*M \times \rho^*N \enspace .\]
\end{lemma}
\begin{lemma} \label{lem:rel_pb_comm}
Consider the setting as in the preceding lemma, with $F = \TDelta{T}$, and $t\in T$.
 Then we have
\[ \rho^* (M^t) \cong (\rho^*M)^t \enspace . \]
%and 
%\[ \rho^* (M_u) \to (\rho^*M)_u .\]
\end{lemma}

\begin{lemma}\label{lem:rel_pb_fibre}
  Suppose $N$ is a $Q$--module with codomain $\family{\E}{T}$, and $t\in T$. Then
 \[ \rho^*\fibre{M}{t} \cong \fibre{\rho^*M}{t} \enspace . \]
\end{lemma}

\begin{definition}\label{def:forget_hat_module}
  Recall that the category $\PS$ is the category of preordered sets and set--theoretic maps (not necessarily monotone)
 between them (\autoref{def:cat_wPO}). 
Given a relative monad $P$ on some functor $F$ and a $P$--module $M$ with codomain $\PO$, 
we can consider $M$ as a $P$--module with codomain $\PS$.
 We denote this module by $\hat{M}$. 
 In other words, we have a functor 
 \[ \hat{\_} : \RMod{P}{\PO} \to \RMod{P}{\PS} \]
 obtained by postcomposition with the forgetful functor from $\PO$ to $\PS$.

\end{definition}

\begin{definition}[Substitution of \emph{one} Variable]\label{def:hat_P_subst}

Let $P$ be a monad over $\Delta$.
 For any set $X$, we define a binary substitution operation
 \begin{align*} 
   \subst(X) : P(X^*) \times P(X) &\to P(X), \\ 
    (y,z)&\mapsto y [*:= z] := \kl{\defaultmap(\we_X , z) }(y) \enspace , 
\end{align*}
 where ``$\defaultmap$'' is a coproduct map; for $f : A \to B$ and $z \in B$,
 \[ \defaultmap(f, z) := [f, x \mapsto z] : A+\{*\} \to B\enspace . \]
 This defines a morphism of $P$--modules with codomain $\PS$,
\[\subst^P : \hat{P}' \times \hat{P}\to \hat{P} \enspace .\] %
  The reason why we have to consider the category $\PS$ with \emph{all set--theoretic} maps instead of just monotone maps
  is that $\subst^P$ is not necessarily monotone in its second argument, cf.\ \autoref{rem:about_substitution}.
\end{definition}

The untyped substitution of \autoref{def:hat_P_subst} actually is a special case of the following typed substitution:

\begin{definition}[Substitution of \emph{one} Variable, typed]\label{def:hat_P_subst_typed}
  Let $T$ be a (nonempty) set and let $P$ be a monad over $\family{\Delta}{T}$.
  For any $s,t\in T$ and $X\in \TS{T}$ we define a binary substitution operation
      \begin{align*} 
   \subst_{s,t}(X) : P(X^{*s})_t \times P(X)_s &\to P(X)_t, \\ 
    (y,z)&\mapsto y [*:= z] := \kl{\defaultmap(\we_X , z) }(y) \enspace .
\end{align*}
  For any pair $(s,t)\in T^2$, we thus obtain a morphism of $P$--modules
  \[ \subst^{P}_{s,t} : \fibre{\hat{P}^{s}}{t} \times \fibre{\hat{P}}{s} \to \fibre{\hat{P}}{t} \enspace . \]
\end{definition}

%% file: type_sigs.tex
\section{Signatures for Types}\label{sec:type_sigs}

We present \emph{algebraic signatures}, which later are used to specify the \emph{object types} of the languages
we consider.
Algebraic signatures and their models were first considered by Birkhoff \cite{birkhoff1935}.

\begin{definition}[Algebraic Signature]\label{def:raw_sig}
An \emph{algebraic signature $S$} is a family of natural numbers, i.e.\ a set $J_S$ and a map 
(carrying the same name as the signature) $S : J_S\to \mathbb{N}$.
For $j\in J_S$ and $n\in \mathbb{N}$, we also write $j:n$ instead of $j \mapsto n$.
An element of $J$ resp.\ its image under $S$ is called an \emph{arity} of $S$.
\end{definition}

\begin{example}[Algebraic Signature of \autoref{ex:slc_def}]\label{ex:type_sig_SLC}
  The algebraic signature of the types of the simply--typed lambda calculus is given by
  \[ S_{\SLC} := \{* : 0 \enspace , \quad (\TLCar) : 2 \}\enspace .\]
\end{example}

\noindent
To any algebraic signature we associate a category of \emph{representations}.
We call \emph{representation of $S$} any set $U$ equipped with operations according to the signature $S$. % (Def.\ \ref{def:rep_alg_sig}).
A \emph{morphism of representations} is a map between the underlying sets that is compatible with the 
operations on either side in a suitable sense.
Representations and their morphisms form a category.
We give the formal definitions:

\begin{definition}[Representation of an Algebraic Signature $S$, $S$--Algebra]\label{def:rep_alg_sig}

A \emph{representation} $R$ of an algebraic signature $S$ --- also known as \emph{$S$--algebra} --- is given by 
\begin{packitem}
 \item a set $X$ and
 \item for each $j\in J_S$, an operation $j^R:X^{S(j)} \to X$.
\end{packitem}
In the following, given a representation $R$, we write $R$ also for its underlying set. %The representation of $n=S(j)$ of $R$ is denoted by $j^R$.
\end{definition}

\begin{example}\label{ex:type_PCF}
  The language \PCF~\cite{Plotkin1977223, Hyland00onfull} (see also \autoref{subsec:intro_pcf_lc_syntax}) 
is a simply--typed lambda calculus with a fixed point operator
  and arithmetic constants.
Let $J:= \{\Nat, \Bool, (\PCFar)\}$. The signature of the types of \PCF~is given by the arities 
   \[S_{\PCF}:= \lbrace\Nat:0\enspace ,\quad \Bool: 0 \enspace,\quad (\PCFar): 2 \rbrace \enspace .\]
 A representation $T$ of $S_{\PCF}$ is given by a set $T$ and three operations,
   \[
      \Nat^T : T\enspace, \quad \Bool^T : T \enspace , \quad  (\PCFar)^T : T\times T\to T \enspace .
   \]
\end{example}

\noindent
A morphism of representations is given by a map between the underlying sets that is compatible with
the representation structure:

\begin{definition}[Morphisms of Representations]\label{def:mor_raw_rep}
 Given two representations $T$ and $U$ of the algebraic signature $S$, a \emph{morphism} from $T$ to $U$ is 
  a map $f : T\to U$ such that, for any arity $n=S(j)$ of $S$, we have
  \[  \comp{j^T}{f} = \comp{(\underbrace{f\times\ldots\times f}_{n \text{ times}})}{j^U} \enspace . \]

% the diagram
%   \[
%    \begin{xy}
%     \xymatrix@!=3pc{
%          T^n \ar[d]_{f^n} \ar[r]^{j^T} & T \ar[d]^{f} \\
%          U^n \ar[r]_{j^U} & U
%      }
%    \end{xy}
%   \]
%  \[
%   \begin{xy}
%    \xymatrix{
%         **[l]\dom(\alpha,T) \ar[d]_{\dom(\alpha, f)} \ar[r]^{\alpha^T} & **[r]\cod(\alpha,T) \ar[d]^{\cod(\alpha,f)} \\
%         **[l]\dom(\alpha,U) \ar[r]_{\alpha^U} & **[r]\cod(\alpha,U)
%     }
%   \end{xy}
%  \]
%commutes.
\end{definition}

\begin{example}[\autoref{ex:type_PCF} continued]
 Given two representations $T$ and $U$ of $S_{\PCF}$, a morphism from $T$ to $U$ is a map $f : T\to U$ between the underlying sets
 such that, for any $s,t\in T$,
\begin{align*}
      f(\Nat^T) &= \Nat^U \enspace ,  \\
      f(\Bool^T) &= \Bool^U \quad \text{ and}\\
     f(s \PCFar^T t) &= f(s) \PCFar^U f(t) \enspace .
\end{align*}
\end{example}

\noindent
Representations of an algebraic signature $S$ and their morphisms form a category.

\begin{lemma}\label{lem:initial_sort}
 Let $(J,S)$ (or $S$ for short) be an algebraic signature. The category of representations of $S$ has an initial object $\hat{S}$. 
\end{lemma}
\begin{proof}
   We cut the proof into small steps:
  \begin{itemize}
  \item 
In a type--theoretic setting the set --- also called $\hat{S}$ --- 
 which underlies the initial representation $\hat{S}$ is defined as an inductive set
 with a family of constructors indexed by $J_S$: 
    \[ \hat{S} \enspace  ::= \quad C : \forall j\in J, \enspace \hat{S}^{S(j)} \to \hat{S} \enspace . \]
  That is, for each arity $j\in J$, we have a constructor
     $   C_j : \hat{S}^{S(j)} \to \hat{S}$. % \enspace $
  \item
   For each arity $j\in J$, we must specify an operation $j^{\hat{S}} : \hat{S}^{S(j)} \to \hat{S}$.
   We set 
    \[   j^{\hat{S}} := C_j : \hat{S}^{S(j)} \to \hat{S} \enspace , \]
   that is, the representation $j^{\hat{S}}$ of an arity $n=S(j)$ is given precisely by its corresponding constructor.

  \item
      Given any representation $R$ of $S$, we specify a map $i_R:\hat{S} \to R$ between the underlying sets
       by structural recursion:
  \[ i_R : \hat{S} \to R \enspace, \quad  i_R \bigl(C_j(a)\bigr) := {j}^{R} \bigl((i_R)^{S(j)} (a)\bigr) \enspace , \]
    for $a\in \hat{S}^{S(j)}$.
   That is, the image of a constructor function $C_j$ maps recursively on the image of the corresponding 
    representation $j^R$ of $R$.
   \item We must prove that $i_R$ is a morphism of representations, that is, that for any $j\in J$ with $S(j) = n$,
     \[ \comp{j^{\hat{S}}}{i_R} = \comp{(i_R)^{n}}{j^R} \enspace . \]
     Replacing $j^{\hat{S}}$ by its definition yields that this equation is precisely the specification of $i_R$, see above.
   
   \item
 It is the diagram of \autoref{def:mor_raw_rep} which ensures uniqueness of $i_R$; since any morphism of 
    representations $i' : \hat{S} \to R$ must 
 make it commute, one can show by structural induction that $i' = i_R$. More precisely:
   \begin{align*}  
i'(C_j(a))= i'(C_j(a_1,\ldots,a_{S(j)})) &= j^R (i'(a_1),\ldots,i'(a_{S(j)})) \stackrel{i'(a_k) = i_R(a_k)}{=} \\
                                          &= j^R (i_R(a_1),\ldots,i_R(a_{S(j)})) = i_R(C_j(a)) \enspace .
   \end{align*}

   \end{itemize}
\end{proof}

\begin{example}[\autoref{ex:type_PCF} continued]\label{ex:pcf_type_initial}
  The set $T_{\PCF}$ underlying the initial representation of the algebraic signature $S_{\PCF}$ is given by
  \[ T_{\PCF} \enspace ::= \quad \Nat \enspace \mid \enspace \Bool \enspace \mid \enspace T_{\PCF} \PCFar T_{\PCF} \enspace . \]
 For any other representation $R$ of $S_{\PCF}$ the initial morphism 
  $ i_R : T_{\PCF} \to R $
  is given by the clauses

\begin{align*}
    i_R(\Nat) &= \Nat^R \\
    i_R(\Bool) &= \Bool^R \\
    i_R(s\PCFar t) &= i_R(s) \PCFar^R i_R(t) \enspace .
\end{align*}

% 
%   \begin{align*} 
%           i_R (a) := &\text{ match }a \text{ with } \\
%                   {} &\mid \Nat => \Nat^R\\
%                   {} &\mid \Bool => \Bool^R \\
%                   {} &\mid (s\PCFar t) => i_R(s) \PCFar^R i_R(t) \enspace .
%   \end{align*}
\end{example}

%% file: sts.tex
\section{Zsid\'o's Theorem Reviewed}

\label{sec:sts_ju}

We present Zsid\'o's initiality theorem \cite[\chapterautorefname~6]{ju_phd} (cf.\ \autoref{nice_thm}) for simply--typed abstract syntax.
Its formalization in the proof assistant \coq~is explained in \autoref{chap:sts_formal}.
Throughout this section the number given in the name of each definition points to the implementation
of this definition in \coq. For instance, the implementation of Simple Monad Morphisms (\autoref{def:monad_hom_simpl})
is given in \autoref{code:monad_hom_simpl}.

Our presentation follows the pattern outlined at the beginning of \autoref{sec:initial_semantics}:
in \autoref{sec:sts_arities_signatures} we present classic signatures in two different ways.
Afterwards, in \autoref{subsection:sts_rep}, we give the definition of representations of such signatures.
Finally, in \autoref{sec:sts_initiality}, we state the main theorem, proved by Zsid\'o \cite{ju_phd}.

\subsection{Signatures for Terms}\label{sec:sts_arities_signatures}

In \autoref{sec:sts_arities_syn} we give a purely syntactical definition of \emph{classic} arities.
Afterwards, in \autoref{sec:sts_arities_sem} we give a definition of arities as pairs of functors on 
suitable categories, and identify a subclass of arities which are in one--to--one correspondence 
with classic arities. We thus call arities of this subclass \emph{classic} as well.
In the following we fix a set $T$ of object types.

\subsubsection{Arities, syntactically}\label{sec:sts_arities_syn}

Syntactically, a classic arity consists of an element of $t_0\in T$ which specifies the output
type of a constructor, as well as a list of pairs $([t_{i,1},\ldots,t_{i,m_i}], t_i)$, 
where $t_{i,k}, t_i\in T$. Each such pair represents an argument of the corresponding constructor:
the element $t_i$ denotes the object type of the argument, whereas the list $[t_{i,1},\ldots,t_{i,m_i}]$
specifies the types of the variables that are bound by the constructor in this argument.

\begin{definition}[Classic $T$--Arity, $T$--Signature] \label{rem:sts_signature_list}
  A classic arity is of the form
  \[ \bigl[([t_{1,1},\ldots,t_{1,m_1}], t_1), \ldots, ([t_{n,1},\ldots,t_{n,m_n}], t_n)\bigr] \to t_0 \enspace , \]
  where $t_{i,k}$ and  $t_i$ are elements of $T$.
  We use an arrow to separate the data specifying input data and output data, respectively.
A \emph{signature} is a family of arities.
  For a formalized definition, see the \textsf{Coq} code snippets \autoref{code:sts_signature_list_notused} and
  \autoref{code:sts_signature}.
\end{definition}

\begin{example}[Signature of $\SLC$]\label{ex:tlc_sig_syntactic_sts}
The signature of the simply--typed lambda calculus (cf.\ \autoref{ex:slc_def}) is given by 
  \[\{ \abs_{s,t} :  \bigl[([s],t)\bigr] \to (s\TLCar t) \enspace , 
         \quad \app_{s,t} : \bigl[([],s\TLCar t),([],s)\bigr]\to t\}_{s,t\in\TLCTYPE} \enspace . \]
 See the code snippet \autoref{code:tlc_sig} for a \textsf{Coq} implementation of this example.
\end{example}

\subsubsection{Arities, semantically}\label{sec:sts_arities_sem}

In this section we give a definition of arities as pairs of functors between suitable categories.
The source category (cf.\ \autoref{def:sts_monads}) is a category of monads and morphisms of monads,
whereas the target category (cf.\ \autoref{def:sts_lmod}) mixes modules over different such monads.

At first, in \autoref{rem:alg_arities_sem}, we present an alternative characterization of algebraic arities.
This alternative point of view is then adapted to allow for the specification of arities for
terms.

\begin{remark}[Algebraic Arities viewed differently]\label{rem:alg_arities_sem}
An algebraic arity $j : n$ as presented in \autoref{sec:type_sigs} associates, to
any set $X$, the set $\dom(j,X) := X^n$, the \emph{domain} set. A representation $R$ of this arity $j$ 
in a set $X$ then is given by a map $j^R:X^n \to X$.
More formally, the domain set is given via a functor $\dom(j):\Set\to\Set$ which associates to any set
$X$ the set $X^n$.
Similarly, we might also speak of a \emph{codomain} functor for any arity, 
which --- for algebraic arities --- is given by the identity functor.
A representation $R$ of $j$ in a set $X$ then is given by a morphism
\[ j^R:\dom(j)(X) \to \cod(j)(X) \enspace .  \]
\end{remark}

\noindent
We take the perspective of \autoref{rem:alg_arities_sem} in order to define arities and signatures for \emph{terms}:
given a set $T$ of object types, 
an arity $\alpha$ for terms typed over $T$ 
is a pair of functors $(\dom(\alpha),\cod(\alpha))$ 
associating 
two $P$--modules $\dom(\alpha)(P)$ and $\cod(\alpha)(P)$, 
 to any suitable monad $P$. 
A suitable monad here is 
a monad $P$ on the category $\TS{T}$.
A representation $R$ of $\alpha$ in a such a monad $P$ is a module morphism 
 \[\alpha^R : \dom(\alpha)(P)\to\cod(\alpha)(P) \enspace . \]
We consider monads as in \autoref{def:monad_mu} (also: \autoref{def:endomonad}) over a category of the form $\TS{T}$ for some 
  fixed set $T$. Throughout this section, morphisms between two such monads over the same category are given by colax monad morphisms
  \emph{over the identity functor},
  i.e.\ those morphisms of \autoref{def:colax_mon_mor} (alt.\ \autoref{def:colax_mon_mor_alt})
 with $F = \Id_{\TS{T}}$. 
  For convenience, and as a reference for the implementation in \textsf{Coq}, we explicitly state the definition of these ``simple'' monad morphisms, 
  using the definition through Kleisli operation (cf.\ \autoref{def:colax_mon_mor_alt}) of monads and morphisms:

\begin{definition}[Simple Monad Morphism, \autoref{code:monad_hom_simpl}]\label{def:monad_hom_simpl}
  Let $P$ and $Q$ be two monads over a category $\C$.
 A \emph{simple morphism of monads $\tau$} %(\lstinline!Monad_Hom!) 
  from $P$ to $Q$ is given by a collection of morphisms $\tau_c\in \C(Pc,Qc)$ 
  such that the following diagrams commute for all suitable morphisms $f$:
\begin{equation*}
 \begin{xy}
  \xymatrix @=3pc{
  Pc \ar[r]^{\kl[P]{f}} \ar[d]_{\tau_c}& Pd \ar[d]^{\tau_d} & c \ar[r]^{\we^P_c} \ar[rd]_{\we^Q_c} & Pc \ar[d]^{\tau_c} \\
  Qc \ar[r]_{\kl[Q]{\comp{f}{\tau_d}}} & Qd  , & {} & Qc .
}
 \end{xy}
\end{equation*}
\end{definition}

\begin{definition} [Category $\Mon{\C}$ of Monads on $\C$] \label{def:sts_monads}
  Given a category $\C$, we define the category $\Mon{\C}$ to be the category whose objects are monads over $\C$.
  A morphism from $P$ to $Q$ in this category is a monad morphism as in \autoref{def:monad_hom_simpl}.
  We denote by $I_\C : \Mon{\C} \to \Mcol$ the inclusion functor.
\end{definition}

We define a category in which modules over different monads --- but with the same codomain category --- 
are mixed together. This category can be defined as a particular \emph{colax comma category}.
However, we also give an explicit description of the objects and morphisms 
of this category.

\begin{definition}[Colax Comma Category]
 Let $\C$ be a 2--category, and $c\in \C$ be an object of $\C$. 
 Let $\A$ be a category and let $F:\A \to \C$ be a functor.
 An object of the \emph{colax comma category} $(F\downarrow c)$ is given by a pair $(a,f:Fa \to c)$
  of an object $a\in\A$ and a morphism $f:a\to c$. A morphism to another such $(b,g:Fb\to c)$ is given by a 
  pair $(h:a\to b,\alpha)$ as in the diagram

 \[
 \begin{xy}
  \xymatrix @C=2pc @R=0.2pc{ **[l]Fa \rrtwocell<\omit>{\;\;\alpha} \ar@/^1.1pc/[rr]^{f} \ar@/_0.3pc/[rd]_{Fh} & {} & **[r]c \\
                           {} & Fb \ar@/_0.3pc/[ru]_{g} & {}
}
 \end{xy} \enspace . 
\]

\end{definition}

\noindent
While the above definition is not the most general definition possible for a colax comma category,
it is sufficient for our needs:

\begin{definition} [Large Category $\LMod{}{\C}{\D}$ of Modules]\label{def:sts_lmod}
  Given two categories $\C$ and $\D$, we define the category $\LMod{}{\C}{\D}$ to be the colax comma category $(I_\C \downarrow \Id_{\D})$.
An object of this category is a monad $P$ over $\C$ together with a $P$--module with codomain $\D$ (cf.\ \autoref{def:module}).
  A morphism $(f,h)$ to another such $(Q,N)$ is given by a morphism $f : P\to Q$ of monads over the identity functor --- i.e., a morphism
  in $\Mon{\C}$ --- and a morphism of modules
  $h : M \to f^*N = \comp{f}{N}$:

 \[
    \begin{xy}
     \xymatrix@!=4pc{ **[l]P \rtwocell<5>^M_{\comp{f}{N}}{\;\;h}
          & **[r] \Id_\D}
%          P \ar[rr]^{M} \ar[dr]_{f} & {} &(D,\Id) \\
%        {} & Q \ar[ru]_{N}& {}
    \end{xy} \enspace .
  \]

% \[
%  \begin{xy}
%   \xymatrix @C=.5pc @R=4pc{ I_\C P \ar[rr]^{I_\C h}  \ar[dr]_{M} &  & I_\C Q \ar[ld]^{N}\\
%              {}& (\D,\Id) \utwocell<\omit>{\alpha}& {}
% }
%  \end{xy}
% \]
\end{definition}

\begin{definition}[Tautological Module]\label{def:taut_large_mod}
 
%   \[
%    \begin{xy}
%  \xymatrix{    n \ar[r]^{k}     & T \ar[d]^{V} \\
%     {} & \Set
%   }
%    \end{xy}
%    \qquad\mapsto \qquad
%     \begin{xy}
%      \xymatrix{
%         n \ar[r]^{k} & T\ar[d]^{V} \ar[r]^{\id} & T\ar[d]^{RV} \\
%        {} & \Set \ar[r]_{\eta_R}& \Set,
%   }
%     \end{xy}
%   \]
To any monad $R\in\Mon{\C}$ we associate
  the tautological module $\Theta(R)$ of $R$, 
  \[\Theta(R):= (R,R) \in \LMod{}{\C}{\C} \enspace . \]
 This construction extends to a functor $\Theta : \Mon{\C} \to \LMod{}{\C}{\C}$.
\end{definition}

A half--arity associates a $P$--module towards $\Set$ to any monad $P$ over $\TS{T}$:
\begin{definition}[Half--Arity]\label{def:sts_half_arity}
 A \emph{half--arity over $T$} is a functor 
  \[ a : \Mon{\TS{T}} \to \LMod{}{\TS{T}}{\Set}\] 
from the 
 category of monads over $\TS{T}$ to the large category of modules over such monads with codomain $\Set$, such that
\begin{equation} \comp{a}{\pi_1} = \id_{\Mon{\TS{T}}}% : \Mon{\Set}\to\Mcol
 \enspace . \label{eq:sts_half_arity_left_inv}
\end{equation}
This last condition given in \autoref{eq:sts_half_arity_left_inv} ensures that each monad maps to a module \emph{over itself}. 
For a monad $R\in \Mon{\TS{T}}$, we thus sometimes omit the first component $R$ of the image $a(R)$ and consider $a(R)\in \Mod{R}{\Set}$.
\end{definition}

\begin{definition}[Arity, Signature]\label{def:sts_arity_signature}
 A \emph{$T$--arity $s$} is a pair $(\dom(s), \cod(s))$ of half--arities over $T$,
   \[\dom(s), \cod(s) : \Mon{\TS{T}} \to \LMod{}{\TS{T}}{\Set} \enspace , \] 
 written $\dom(s)\to \cod(s)$.
 A \emph{$T$--signature} is a family of $T$--arities.
\end{definition}

We give some important examples of half--arities over the set $T$. 
 Note that, by the convention of \autoref{def:sts_half_arity}, we omit the first component of objects of the large 
category of modules $\LMod{}{\TS{T}}{\Set}$.

\begin{definition}
 Let $T$ be a nonempty set, and
let $t\in T$ be an element of $T$.
  
 \begin{itemize}
  \item The map $\fibre{\Theta}{t} : \Mon{\TS{T}} \to \LMod{}{\TS{T}}{\Set}$ with object map $R\mapsto (R,\fibre{R}{t})$ is a half--arity --- the \emph{fibre with 
          respect to $t$} --- over $T$.
  \item If $M$ is a half--arity over $T$, so is $M^t : \Mon{\TS{T}} \to \LMod{}{\TS{T}}{\Set},\enspace M^T(R):= M(R)^t$ (cf.\ \autoref{def:derivation_simple}). 
        By iterating, given $t_1,\ldots,t_n\in T$, the functor 
     \[ M^{(t_1, \ldots, t_n)} : R \mapsto (\ldots (M(R)^{t_1})\ldots )^{t_n}\]
      is a half--arity.
       
  \item If $M$ and $N$ are half-arities over $T$, then so is the product $M\times N : \Mon{\TS{T}} \to \LMod{}{\TS{T}}{\Set}$:
             \[M \times N : R\mapsto M(R)\times N(R) \enspace . \]
  \item The map $R\mapsto *$, where $* : V \mapsto 1_{\Set}$ is the terminal object in $\Mod{R}{\Set}$, is a half--arity over $T$.
 \end{itemize}
\end{definition}

An \emph{arity} is a pair of half--arities. We are only interested in \emph{classic}
arities, whose domain and codomain functors are of a specific form:

\begin{definition}[Classic $T$--Arity, $T$--Signature (II)]\label{def:sts_alg_sig}
   We call \emph{classic $T$--arity} any $T$--arity $s$ of the form
   \begin{equation} %s = (t_{1,1} \ldots t_{1,m_1})t_1 , \ldots , (t_{n,1} \ldots t_{n,m_n} )t_n \to t_0 
             s = \fibre{\Theta}{t_1}^{t_{1,1}\ldots t_{1,m_1}} \times \ldots \times 
                 \fibre{\Theta}{t_n}^{t_{n,1}\ldots t_{n,m_n}} \to 
                 \fibre{\Theta}{t_0}
\label{eq:arity} \end{equation}
  
  for $t_{i,j}, t_i\in T$.
   A \emph{classic $T$--signature} is a collection of such classic arities.
\end{definition}

To an operator that binds $m_k$ variables of types 
$t_{k,1},\ldots,$ $ t_{k,m_k}$ in its $k$–th argument of type $t_k$, and which yields a term of type $t_0$,
we associate the arity given in \autoref{eq:arity}.

\begin{remark}
The classic $T$--arities and $T$--signatures of \autoref{def:sts_alg_sig} and of \autoref{rem:sts_signature_list}
 are in bijection, respectively. We can thus specify $T$--signatures by simply giving a
  term of the simple data type defined in \autoref{rem:sts_signature_list}. 
 In the \textsf{Coq} formalization, arities and signatures are defined via such data types, 
 cf.\ \autoref{code:sts_signature_list_notused} and
  \autoref{code:sts_signature}.
\end{remark}

\begin{remark}
 In \autoref{def:sts_alg_sig} we can have $n = 0$, yielding an arity for constants of, say, object type $t_0\in T$, 
 \[   s = * \to \fibre{\Theta}{t_0} \enspace . \]
 Such an arity then is given by an empty list of arguments according to \autoref{rem:sts_signature_list}.
 An example of a constant arity is given in \autoref{ex:term_sig_pcf}.
\end{remark}

As an example we discuss the classic signature of the simply typed lambda calculus:

\begin{example}[Signature of $\SLC$, \autoref{code:tlc_sig}] \label{ex:tlc_sig} 
Consider the example of the simply--typed lambda calculus (cf.\ \autorefs{ex:slc_def}, \ref{ex:tlc_syntax_monadic}).
Its signature is given syntactically in \autoref{ex:tlc_sig_syntactic_sts}.
Equivalently, it is given by the signature
\[ \Sigma_{\SLC} = \lbrace\abs_{s,t},\app_{s,t}\rbrace_{s,t \in \TLCTYPE}\] 
with
\begin{align*}
      \abs_{s,t} &:=  \fibre{\Theta}{t}^s \to \fibre{\Theta}{s\TLCar t} \quad\quad\text{and} \\
      \app_{s,t} &:=  \fibre{\Theta}{s \TLCar t} \times \fibre{\Theta}{s}\to \fibre{\Theta}{t} \enspace .
\end{align*}

\end{example}

\begin{remark}\label{rem:var_eta}
 Note that in \autoref{ex:tlc_sig} we do not need to explicitly specify 
 an arity for the \lstinline!Var! term constructor in order to obtain 
  the simply--typed lambda calculus as presented in \autoref{ex:slc_def}.
  Indeed, by building models from monads (cf.\ \autoref{def:sts_sig_rep})
  every model is by definition 
 equipped with a corresponding operation --- the unit of the underlying monad. 
\end{remark}

\subsection{Representations}\label{subsection:sts_rep}

A representation of an arity $s$ in a monad $P$ is given by a morphism of $P$--modules whose domain and
codomain are determined by $s$:

\begin{definition}[Representation of a $T$--Signature, \autoref{code:sts_arity_rep}]\label{def:sts_arity_rep}\label{def:sts_sig_rep}
A representation $R$ of a $T$--signature $\Sigma$ is given by 
\begin{itemize}
 \item a monad $P$ on the category $\TS{T}$ and
 \item for any arity $s\in \Sigma$, a morphism of modules in $\LMod{}{\TS{T}}{\Set}$,
   \[s^R : \dom(s,P) \to \cod(s,P) \enspace , \]
  such that $\pi_1(s^R) = \id_P$.
\end{itemize}
Given a representation $R$, we denote by $R$ also the underlying monad.
\end{definition}

Morphisms of representations are monad morphisms that are compatible with the representation module morphisms:

\begin{definition}[Morphism of Representations] \label{def:sts_rep_mor}
 Let $P$ and $Q$ be representations of a $T$--signature $\Sigma$. 
 A \emph{morphism of representations} $f : P \to Q$ is a morphism $f$ between the underlying monads such that the following diagram commutes for 
  any arity $s$ of $\Sigma$:
\begin{equation}\label{eq:sts_comm_diag}
\begin{xy}
\xymatrix @=4pc{
     **[l] \dom(s,P) \ar[d]_{\dom(s,f)} \ar[r]^{s^P} & **[r] \cod(s,P) \ar[d]^{\cod(s,f)}\\
     **[l] \dom(s,Q) \ar[r]_{s^Q} & **[r] \cod(s,Q) .
}
\end{xy}
\end{equation}
\end{definition}

\noindent
The preceding diagram can be seen as a diagram in two different categories, either in the category $\LMod{}{\TS{T}}{\Set}$, 
or in the category $\Mod{P}{\Set}$ of $P$--modules.

\begin{definition}[Category of Representations]\label{def:sts_cat_of_reps}
Morphisms of representations can be composed: the composition of the underlying monad morphisms again gives a morphism of representations. 
Similarly the identity morphism of monads is a morphism of representations.
Two morphisms of representations are said to be equal if their underlying morphisms of monads are equal. 
Representations and their morphisms of a signature $\Sigma$ form a category $\Rep(\Sigma)$. 
\end{definition}

\subsection{Initiality}\label{sec:sts_initiality}

The main theorem states that any $T$--signature admits an initial representation:

\begin{theorem}\label{nice_thm}
Let $\Sigma$ be a classic $T$--signature. Then the category $\Rep(\Sigma)$ of representations of $\Sigma$ has an initial object.  
\end{theorem}

\begin{remark}
 The monad underlying the initial representation associates, to any context $V\in \TS{T}$, the set of terms of the syntax of $\Sigma$
 with free variables in $V$. 
The module morphisms of the initial representation are given by the constructors of this syntax.
\end{remark}

A set--theoretic construction of the syntax as well as a proof of the theorem is given in Zsid\'o's PhD thesis \cite{ju_phd}. 
In \autoref{section:sts_initial} we explain the implementation of the main theorem in a type--theoretic setting
in the proof assistant \textsf{Coq}.

%% file: n_compilation_sig.tex
\section{Extending Zsid\'o's Theorem to Varying Types}
\label{sec:ext_zsido}

Zsid\'o's initiality result of \autoref{nice_thm} does not account for varying object sorts.
Indeed, given a signature $\Sigma$ over a set $T$ of object sorts,
any representation of $\Sigma$ ``has'' the same set of sorts $T$, i.e.\ its underlying monad
is a monad on the category $\TS{T}$.
In this section we give a new definition of signatures and their representations, and prove that the
resulting category of representations has an initial object.
The iteration operator obtained from this initiality result accounts for 
translations between languages over different sets of sorts.
We define a \emph{typed signature} to be a pair $(S,\Sigma)$ consisting of an algebraic signature
$S$ for sorts, and a signature $\Sigma$ for terms typed over the sorts specified by $S$.
A representation of such a typed signature consists of a representation of the sort signature $S$ in
some set $T$ and a representation of $\Sigma$ in a monad over the category $\TS{T}$.
Translations of sorts are given by morphisms of representations of $S$, that is, by maps of sets 
that are compatible with the representations of sorts constructors in the source and target. 
Compared to Zsid\'o, we thus restrict ourselves to sets of sorts that have \emph{inductive structure},
whereas for Zsid\'o, the set of sorts is given by an arbitrary parameter.

\subsection{Signatures for Types \& Terms} \label{sec:comp_term_sigs}

Before starting with the formal definitions, we informally consider the example of the simply--typed lambda calculus;
its signature for terms was given in the preceding section (cf.\ \autoref{ex:tlc_sig}) as:
\begin{equation}\{ \abs_{s,t} :=  \bigl[([s],t)\bigr] \to (s\TLCar t) \enspace , 
       \quad \app_{s,t} := \bigl[([],s\TLCar t),([],s)\bigr]\to t\}_{s,t\in\TLCTYPE} \enspace . 
 \label{eq:sig_tlc_simple}
\end{equation}
The parameters $s$ and $t$ range over the set $\TLCTYPE$ of types, the initial representation of the 
signature for types from \autoref{ex:type_sig_SLC}.
In particular, we have $2 \times \TLCTYPE^2$ arities in this signature.

Our goal is to consider representations of the simply--typed lambda calculus in monads over categories 
of the form $\TS{T}$ for \emph{any} set $T$ --- provided that $T$ is equipped with a representation of the signature 
$S_{\SLC}$. 
Clearly, the above signature of \autoref{eq:sig_tlc_simple}, with its strong dependence on the set $\TLCTYPE$ 
is not well--suited to express this.
Instead of the above signature, we would like to write 
\begin{equation}\{ \abs :=  \bigl[([1],2)\bigr] \to (1\TLCar 2) \enspace , 
       \quad \app := \bigl[([],1\TLCar 2),([],1)\bigr]\to 2\}\enspace . 
  \label{eq:sig_tlc_higher_order}
\end{equation}
 What is the intended meaning of such a signature?
For any representation $T$ of $S_{\SLC}$, the variables $1$ and $2$
range over elements of $T$.
In this way the number of abstractions and applications depends on the representation $T$ of $S_{\SLC}$:
intuitively, a 
representation of the above signature of \autoref{eq:sig_tlc_higher_order} 
over a representation  $T$ of $\TLCTYPE$
has 
$T^2$ 
abstractions and $T^2$ applications --- one for each pair of elements of $T$.
As an example,
for the final representation of $S_{\SLC}$ in the singleton set, one obtains only one abstraction and one application morphism.
We call arities, that contain object type variables, arities \emph{of higher degree}, 
where the degree of such
an arity denotes the number of (distinct) type variables. For instance, the arities $\abs$ and $\app$ of \autoref{eq:sig_tlc_higher_order}
are of degree $2$.

\subsubsection{Term Arities, syntactically}\label{sec:term_arities_syntactic}

In \autoref{sec:sts_ju}, arities over a fixed set of object types $T$ 
were defined purely syntactically, namely using pairs and lists, cf.\ \autoref{rem:sts_signature_list}.
We give a similar syntactic characterization of arities over a fixed algebraic signature $S$ for types
as in \autoref{def:raw_sig}.

\begin{definition}[Type of Degree $n$]
  For $n \geq 1$, we call \emph{types of $S$ of degree $n$} the elements of the set $S(n)$ 
 of types associated to the signature $S$ with free variables in the set $\{1,\ldots,n\}$.
  We set $S(0):= \hat{S}$.
  Formally, the set $S(n)$ may be obtained as the initial representation of the signature $S$ enriched by $n$ nullary arities.
\end{definition}

Types of degree $n$ are used to form classic arities of degree $n$:

\begin{definition}[Classic Arity of Degree $n$]\label{def:n_comp_classic_arity_syn}
 A classic arity for terms over the signature $S$ for types of degree $n$ is of the form
\begin{equation} 
  \bigl[([t_{1,1},\ldots,t_{1,m_1}], t_1), \ldots, ([t_{k,1},\ldots,t_{k,m_k}], t_k)\bigr] \to t_0 \enspace , 
  \label{eq:syntactic_arity_higher_degree}
\end{equation}
  where $t_{i,j}, t_i \in S(n)$.
 More formally, a classic arity of degree $n$ over $S$ is a pair 
  consisting of an element $t_0\in S(n)$ and
  a list of pairs.  where each pair itself consists of a list $[t_{i,1},\ldots,t_{i,m_i}]$ of elements of $S(n)$ and 
  an element $t_i$ of $S(n)$.
\end{definition}

  A classic arity of the form given in \autoref{eq:syntactic_arity_higher_degree} denotes a constructor --- 
    or a family of constructors, for $n \geq 1$ --- whose output type is $t_0$,
  and whose $k$ inputs are terms of type $t_i$, respectively, in each of which variables of type according to the list
  $[t_{i,1},\ldots,t_{i,m_i}]$ are bound by the constructor.

\begin{remark}
 For an arity as given in \autoref{eq:syntactic_arity_higher_degree} we also write
\begin{equation} 
      \fibre{\Theta_n^{t_{1,1},\ldots,t_{1,m_1}}}{t_1} \times \ldots\times \fibre{\Theta_n^{t_{k,1},\ldots,t_{k,m_k}}}{t_k} \to 
                   \fibre{\Theta_n^{}}{t_0} \enspace . 
  \label{eq:syntactic_arity_higher_degree22222}
\end{equation}
\end{remark}

\noindent
Examples of (classic) arities are to be found in \autoref{ex:tlc_sig_higher_order} and \autoref{ex:logic_trans}.

\begin{remark}[Implicit Degree]\label{rem:degree_implicit}
 Any arity of degree $n \in \NN$ as in \autoref{def:n_comp_classic_arity_syn} can also be considered as an arity of
degree $n+1$. We denote by $S(\omega)$ the set of types associated to the type signature $S$ with free variables in $\NN$.
 Then any arity of degree $n\in \NN$ can be considered as an arity built over $S(\omega)$.
Conversely, any arity built over $S(\omega)$ only contains a finite set of free variables in $\NN$, and can thus be considered 
to be an arity of degree $n$ for some $n\in \NN$. 
In particular, by suitable renaming of free variables, there is a \emph{minimal} degree for any arity built over $S(\omega)$.
We can thus omit the degree --- e.g., the lower inner index $n$ in \autoref{eq:syntactic_arity_higher_degree22222} ---, 
and specify any arity as an arity over $S(\omega)$, if we really want to consider this arity to be 
of minimal degree.
Otherwise we must specify the degree explicitly.
\end{remark}

\subsubsection{Term Arities, semantically}

We now attach a meaning to the purely syntactically defined arities of \autoref{sec:term_arities_syntactic}. 
More precisely, we define arities as pairs of functors over suitable categories.
Afterwards we restrict ourselves to a specific class of functors, yielding
arities which are in one--to--one correspondence to --- 
and thus can be compactly specified via --- the syntactically defined
classic arities of \autoref{sec:term_arities_syntactic}.
Accordingly, we call the restricted class of arities also \emph{classic} arities.

Throughout this section, we fix an algebraic signature $S$ for types.
An arity $\alpha$ of degree $n$ for terms over $S$ 
is a pair of functors $(\dom(\alpha),\cod(\alpha))$ 
associating 
two $P$--modules $\dom(\alpha,P)$ and $\cod(\alpha,P)$,
  each of degree $n$,
 to any suitable monad $P$. A suitable monad here is 
a monad $P$ on some category $\TS{T}$ where 
the set $T$ is equipped with a representation of $S$.  
We call such a monad an \emph{$S$--monad}.
A representation $R$ of $\alpha$ in an $S$--monad $P$ is a module morphism 
 \[\alpha^R : \dom(\alpha,P)\to\cod(\alpha,P) \enspace . \]

\noindent
 As we have seen in \autoref{ex:slc_def}, 
 constructors can in fact be \emph{families of constructors} indexed by 
 type variables. 
 For such a constructor indexed $n$ times, we consider modules \emph{of degree n} (cf.\ \autoref{rem:family_of_mods_cong_pointed_mod}).

We define a family of categories of monads which will play the role of the category defined in \autoref{def:sts_monads}:

\begin{definition}[$S$--Monad] \label{def:s-mon}
  Given an algebraic signature $S$, the \emph{2--category $\SigMon{S}$ of $S$--monads} is defined as the 2--category whose objects are pairs $(T,P)$ of
  a representation $T$ of $S$ and a monad $P : \TS{T}\to \TS{T}$.
  A morphism from $(T,P)$ to $(T', P')$ is a pair $(g, f)$ of a morphism of $S$--representations $g : T\to T'$ and a 
   monad morphism $f : P\to P'$ over the retyping functor $\retyping{g}$ (cf.\ \autoref{def:retyping_functor}). Transformations are the transformations of \Mcol.

   Given $n\in \mathbb{N}$, we write $\SigMon{S}_n$ for the 2--category whose objects are pairs $(T,P)$ of a representation $T$ of $S$ and 
  a monad $P$ over $\TS{T}_n$. A morphism from $(T,P)$ to $(T', P')$ is a pair $(g, f)$ of a morphism of 
      $S$--representations $g : T\to T'$ and a 
   monad morphism $f : P\to P'$ over the retyping functor $\retyping{g}(n)$ (cf.\ \autoref{def:retyping_functor_pointed}).
\noindent
  We call $I_{S,n} : \SigMon{S}_n \to \Mcol$ the functor which forgets the representation of $S$.
\end{definition}

We define a ``large category of modules'' in which modules over different $S$--monads are mixed together:

\begin{definition}[Large Category of Modules] %[$\LMod{n}{S}{\D}$]
  Given a natural number $n\in \mathbb{N}$, an algebraic signature $S$ and a category $\D$, 
  we call $\LMod{n}{S}{\D}$ the colax comma category $I_{S,n} \downarrow (\D, \Id)$.
  An object of this category is a pair $(P,M)$ of a monad $P \in \SigMon{S}_n$ and a $P$--module with codomain $\D$.
  A morphism to another such $(Q,N)$ is a pair $(f, h)$ of an $S$--monad morphism $f : P \to Q$ in $\SigMon{S}_n$ and a transformation $h : M \to f^*N$:
  \[
    \begin{xy}
     \xymatrix@!=4pc{ 
      **[l]P \rtwocell<5>^M_{\comp{f}{N}}{h}
          & **[r]\Id_{\D}
}
%          P \ar[rr]^{M} \ar[dr]_{f} & {} &(D,\Id) \\
%        {} & Q \ar[ru]_{N}& {}
    \end{xy} \enspace .
  \]
%
% \[
%  \begin{xy}
%   \xymatrix @C=1pc @R=4pc{ P \ar[rr]^{f}  \ar[dr]_{M} &  & Q \ar[ld]^{N}\\
%              {}& (\D,\Id) \utwocell<\omit>{h}& {}
% }
%  \end{xy}
% \]

\end{definition}

\noindent
A half--arity over $S$ of degree $n$ is given by a functor from the category of monads to the large category of modules:

\begin{definition}[Half--Arity over $S$ (of degree $n$)]
 Given an algebraic signature $S$ and $n\in \mathbb{N}$, we call \emph{half--arity over $S$ of degree $n$} a functor
  \[ \alpha : \SigMon{S} \to \LMod{n}{S}{\Set} \enspace  \]
  which is pre--inverse to the forgetful functor.
\end{definition}

Taking into account \autoref{rem:family_of_mods_cong_pointed_mod}, this means that a half--arity of degree $n$ associates to any $S$--monad $R$ 
 --- with representation of $S$ in a set $T$ ---
 a family of $R$--modules indexed $n$ times by $T$.

\begin{remark}[Module of Higher Degree corresponds to a Family of Modules]
  \label{rem:family_of_mods_cong_pointed_mod}
 Let $\C$ be a category, let $T$ be a set and $R$ be a monad on $\family{\C}{T}$. 
Suppose $n \in \NN$, and let $\D$ be a category. Then
  modules over $R_n$ with codomain $\D$ correspond precisely to families of $R$--modules indexed by $T^n$ 
  with codomain $\D$ by (un)currying.
 More precisely, let $M$ be an $R_n$--module. Given $\mathbf{t}\in T^n$, we define an $R$--module $M_{\mathbf{t}}$ by 
  \[ M_{\mathbf{t}} (c) := M(c,\mathbf{t}) \enspace . \] 
 Module substitution for $M_{\mathbf{t}}$ is given, for $f \in \family{\C}{T}(c,Rd)$, by
 \[ \mkl[M_{\mathbf{t}}]{f} := \mkl[M]{f} \]
  where we use that we also have $f \in \family{\C}{T}_n ((c,\mathbf{t}), (Rd, \mathbf{t}))$ according to \autoref{def:cat_set_pointed}.
 Going the other way round, given a family $(M_{\mathbf{t}})_{\mathbf{t}\in T^n}$, we define the $R_n$--module $M$ by
  \[ M(c,\mathbf{t}) := M_{\mathbf{t}}(c) \enspace . \]
 Given a morphism $f \in \family{\C}{T}_n ((c,\mathbf{t}), (Rd, \mathbf{t}))$ --- 
 recall that morphisms in $\family{\C}{T}_n$ are only between families with \emph{the same marker} $\mathbf{t}$ ---, 
we also have $ f \in \family{\C}{T}(c,Rd)$  and define 
 \[ \mkl[M]{f} := \mkl[M_{\mathbf{t}}]{f} \enspace .\]
The remark extends to morphisms of modules; indeed, a morphism of modules $\alpha:M\to N$ on categories with pointed index sets
 corresponds to a family of morphisms $(\alpha_{\vectorletter{t}}:M_{\vectorletter{t}}\to N_{\vectorletter{t}})_{\vectorletter{t}\in T^n}$ 
 between the associated families of modules.
\end{remark}

As in \autoref{sec:sts_ju}, we restrict our attention to half--arities which correspond, in a sense made precise
below, to the syntactically defined arities of \autoref{def:n_comp_classic_arity_syn}.
The basic brick is the \emph{tautological module of degree $n$}:

\begin{definition}
  Given a category $\C$ and $n\in \mathbb{N}$, any monad $R$ on the category $\family{\C}{T}$ induces a monad $R_n$ on $\family{\C}{T}_n$ 
with object map $(V, t_1,\ldots, t_n) \mapsto (RV, t_1,\ldots,t_n)$, 
as is already indicated for functors in \autoref{def:cat_indexed_pointed}.
\end{definition}

\begin{definition}[Tautological Module of Degree $n$]\label{def:taut_mod_pointed}
 
%   \[
%    \begin{xy}
%  \xymatrix{    n \ar[r]^{k}     & T \ar[d]^{V} \\
%     {} & \Set
%   }
%    \end{xy}
%    \qquad\mapsto \qquad
%     \begin{xy}
%      \xymatrix{
%         n \ar[r]^{k} & T\ar[d]^{V} \ar[r]^{\id} & T\ar[d]^{RV} \\
%        {} & \Set \ar[r]_{\eta_R}& \Set,
%   }
%     \end{xy}
%   \]
Let $n\in \NN$ be a natural number. 
To any $S$--monad $R$ we associate
  the tautological module of $R_n$, 
  \[\Theta_n(R):= (R_n,R_n) \in \LMod{n}{S}{\family{\Set}{T}_n} \enspace . \]
%  This construction extends to a functor.
 This construction extends to a functor $\Theta_n : \SigMon{S} \to \LMod{n}{S}{\TS{T}_n}$.
\end{definition}

Let us consider the signature $S_{\SLC}$ of types of $\SLC$.
In the syntactically defined arities (cf.\ \autoref{eq:sig_tlc_higher_order}) we write terms like $1\TLCar 2$.
We now give meaning to such a term: 
let $T$ be any representation of $S_{\TLC}$, that is, a set $T$ together with a base type $*\in T$ and a 
binary operation $(\TLCar) : T\times T \to T$.
Intuitively, the term $1\TLCar 2$ should associate, to an object $(T,V,t_1,t_2)$ with  a $T$--indexed family $V$
of sets and $t_1,t_2\in T$, the element $t_1\TLCar t_2 \in T$.
More formally, such a term is interpreted by a natural transformation (cf.\ \autoref{def:canonical_nat_trans}) over a specific 
category, whose objects are triples of a representation $T$ of $S_{\SLC}$, 
a family of sets indexed by (the set) $T$ and ``markers'' $(t_1,t_2) \in T^2$.

We go back to considering an arbitrary signature $S$ for types. 
The following are the corresponding basic categories of interest:

\begin{definition}[$S\C_n$]
  Given a category $\C$ --- think of it as the category $\Set$ of sets --- 
  we define the category $S\C_n$ to be the category an object of which 
  is a triple $(T,V,\vectorletter{t})$ where $T$ is a representation of $S$,
  the object $V \in \family{\C}{T}$ is a $T$--indexed family of objects of $\C$ and $\vectorletter{t}$ is 
  a vector of elements of $T$ of length $n$.
We denote by $SU_n:S\C_n \to \Set$ the functor mapping an object $(T,V,\vectorletter{t})$ 
to the underlying set $T$.

  We have a forgetful functor $S\C_n \to \T\C_n$ which forgets the representation
 structure.
  On the other hand, any representation $T$ of $S$ in a set $T$ gives rise to a functor 
  $\family{\C}{T}_n \to S\C_n$, which ``attaches'' the representation structure.
\end{definition}

The meaning of a term $s\in S(n)$ as a natural transformation \[s: 1 \Rightarrow SU_n : S\C_n \to \Set\]
is now given by recursion on the structure of $s$:

\begin{definition}[Canonical Natural Transformation]\label{def:nat_trans_type_indicator}\label{def:canonical_nat_trans}
  Let $s\in S(n)$ be a type of degree $n$. 
  Then $s$ denotes a natural transformation 
       \[s:1\Rightarrow SU_n : S\C_n \to \Set \enspace \]
  defined recursively on the structure of $s$ as follows: for $s = \alpha (a_1,\ldots,a_k)$
  the image of a constructor $\alpha \in S$ we set
  \[s(T,V,\vectorletter{t}) = \alpha (a_1(T,V,\vectorletter{t}),\ldots,a_k(T,V,\vectorletter{t})) \]
  and for $s = m$ with $1\leq m\leq n$  we define
  \[s(T,V,\vectorletter{t}) = \vectorletter{t}(m) \enspace . \]
  We call a natural transformation of the form $s\in S(n)$ \emph{canonical}.
\end{definition}

Canonical natural transformations are used to build \emph{classic} half--arities; 
they indicate context extension (derivation) and
selection of specific object types (fibre):

\begin{definition}[Classic Half--Arity over $S$]\label{def:alg_half_ar}
  The following clauses define an inductive set of 
  \emph{classic} half--arities, to which we restrict our attention:
  \begin{itemize}
   \item The constant functor $* : R \mapsto 1$ is a classic half--arity.
   \item Given any canonical natural transformation $\tau : 1 \to S U_n$ (cf.\ \autoref{def:canonical_nat_trans}), 
         the point-wise fibre module with respect to $\tau$ (cf.\ \autoref{def:fibre_mod_II}) of the tautological module 
         $\Theta_n : R\mapsto (R_n, R_n)$ (cf.\ \autoref{def:taut_mod_pointed}) is a classic half--arity of degree $n$, 
        \[ \fibre{\Theta_n}{\tau} : \SigMon{S} \to \LMod{n}{S}{\Set} \enspace , \quad R\mapsto \left(R,\fibre{R_n}{\tau}\right) \enspace . \]
   \item Given any (classic) half--arity $M = (M_1,M_2) : \SigMon{S} \to \LMod{n}{S}{\Set}$ 
         of degree $n$ and a canonical natural transformation $\tau : 1 \to S U_n$, 
         the point-wise derivation of $M$ with respect to $\tau$ (cf.\ \autoref{def:derived_mod_II}) is a (classic) half--arity of degree $n$, 
      \[ M^{\tau} : \SigMon{S} \to \LMod{n}{S}{\Set} \enspace , \quad R\mapsto \bigl(M(R)\bigr)^{\tau} := \left(M_1(R), M_2(R)^{\tau}\right) \enspace . \]
       Here $\bigl(M(R)\bigr)^{\tau}$ really means derivation of the module, i.e.\ derivation in the second component
         of $M(R)$.
   \item %
     Given two (classic) half--arities $M = (M_1,M_2)$ and $N = (N_1,N_2)$ of degree $n$, which coincide pointwise on the first 
          component, i.e.\ such that $M_1 = N_1$.
       Then their product $M\times N$ is again a (classic) half--arity of degree $n$. Here the product is 
        really the pointwise product in the second component, i.e.\ 
            \[ M\times N : R \mapsto \bigl(M_1(R), M_2(R)\times N_2(R)\bigr) \enspace .  \]
  \end{itemize}
\end{definition}

\begin{remark}[Classic Half--Arity, Syntactically]
 We can represent a classic half--arity of degree $n\in \NN$ over a signature $S$ for types in a purely syntactic manner:
 such a half--arity is determined by a list of the form
   \[  [(\vectorletter{t_1},s_1), \ldots, (\vectorletter{t_k}, s_k)] \enspace , \]
 where $\vectorletter{t_i}$ are vectors of finite length of elements of $S(n)$ and $s_i \in S(n)$.
 Such a list corresponds precisely to the classic half--arity
 \[ R\mapsto \fibre{R_n}{s_1}^{\vectorletter{t_1}} \times \ldots \times \fibre{R_n}{s_k}^{\vectorletter{t_k}} \enspace . \]
\end{remark}

\noindent
We use \emph{weighted sets} as indexing sets for families of arities.
The weight denotes the degree of the corresponding arity.

\begin{definition}[Weighted Set]\label{def:weighted_set}
 A weighted set is a set $J$ together with a map $d:J\to\mathbb{N}$.
\end{definition}

An arity of degree $n\in \mathbb{N}$ for terms over an algebraic signature $S$ is a pair of functors
from $S$--monads to modules in $\LMod{n}{S}{\Set}$.
The degree $n$ corresponds to the number of indices of its associated constructor.
As an example, the arities of $\Abs$ and $\App$ of \autoref{ex:slc_def} are of degree $2$, cf.\ \autoref{ex:tlc_sig_higher_order}.

\begin{definition}[Term--Arity, Signature over $S$]
  A \emph{classic arity $\alpha$ over $S$ of degree $n$} is a pair 
 \[ s = \bigl(\dom(\alpha), \cod(\alpha)\bigr) \] 
  of half--arities over $S$ of degree $n$ such that 
  \begin{packitem}
   \item $\dom(\alpha)$ is classic and
   \item $\cod(\alpha)$ is of the form $[\Theta_n]_{\tau}$ for some natural 
        transformation $\tau$ as in Def.\ \ref{def:alg_half_ar}.
  \end{packitem}
We write $\dom(\alpha) \to \cod(\alpha)$ for the arity $\alpha$, and 
 \[\dom(\alpha, R) := \dom(\alpha)(R) \] 
(and similarly for the codomain functor $\cod$).
Any classic arity is thus of the form given in \autoref{eq:syntactic_arity_higher_degree}.
Given a weighted set $(J,d)$, 
a term--signature $\Sigma$ over $S$ indexed by $(J,d)$ is a $J$-family 
 $\Sigma$ of algebraic arities over $S$, the arity $\Sigma(j)$ being of degree $d(j)$ for any $j\in J$.
\end{definition}

Finally, a \emph{typed signature} is a pair of a signature for types and a signature for terms over those types:

\begin{definition}[Typed Signature] \label{def:typed_signature}
  A \emph{typed signature} is a pair $(S,\Sigma)$ consisting of an algebraic signature $S$ and 
  a term--signature $\Sigma$ (indexed by some weighted set) over $S$.
 \end{definition}

\begin{example}[$\SLC$, \autoref{ex:slc_def} continued]\label{ex:tlc_sig_higher_order}
The terms of the simply typed lambda calculus over the type signature of \autoref{ex:type_sig_SLC} are given by the arities
 \begin{align*}     
   \abs &: \fibre{\Theta}{2}^1 \to \fibre{\Theta}{1\TLCar 2} \enspace , \\
   \app &: \fibre{\Theta}{1\TLCar 2} \times \fibre{\Theta}{1} \to \fibre{\Theta}{2} \quad ,
 \end{align*}
  both of which are of degree $2$ --- we use the convention of \autoref{rem:degree_implicit}. 
  The outer lower index and the exponent are to be interpreted as de Bruijn variables, ranging over types.
  They indicate the fibre (cf.\ \autoref{def:fibre_mod_II}) and derivation (cf.\ \autoref{def:derived_mod_II}),
  respectively, in the special case where the corresponding natural transformation is given by a natural number
  as in \autoref{def:nat_trans_type_indicator}.
  In particular, contrast that to the signature for the simply--typed lambda calculus we 
  gave in \autoref{sec:sts_ju}, \autoref{ex:tlc_sig}. 
  The difference is that now ``similar'' arities which differ only in an object type parameter, are grouped together,
  whereas this is not the case in \autoref{ex:tlc_sig}. 

  Those two arities can in fact be considered over any algebraic signature $S$ with an arrow constructor, 
  in particular over the signature $S_{\PCF}$ (cf.\ \autoref{ex:term_sig_pcf}).

\end{example}

\begin{example}[\autoref{ex:pcf_type_initial} continued] \label{ex:term_sig_pcf}
  We continue considering \PCF. The signature $S_{\PCF}$ for its types 
    is given in \autoref{ex:type_PCF}.
  The term--signature of \PCF~is given in \autoref{fig:pcf_sig_ho}: it consists of an arity for abstraction and an arity for application,
each of degree 2, an arity (of degree 1) for the fixed point operator, and 
one arity of degree 0 for each logic and arithmetic constant --- some of which we omit:
  
\begin{figure}[bt]
\centering
\fbox{%
 \begin{minipage}{7cm}

\vspace{-1ex}

 \begin{align*}     
       \abs &: \fibre{\Theta}{2}^1 \to \fibre{\Theta}{1\PCFar 2} \enspace , \\
        \app &: \fibre{\Theta}{1\PCFar 2} \times \fibre{\Theta}{1} \to \fibre{\Theta}{2} \enspace ,\\
      \PCFFix &: \fibre{\Theta}{1\PCFar 1} \to \fibre{\Theta}{1} \enspace , \\
      \PCFn{n} &: * \to \fibre{\Theta}{\Nat }  \qquad \text{for $n\in \NN$}\\
      \PCFSucc &: * \to \fibre{\Theta}{\Nat \PCFar \Nat} \\
      \PCFPred &: * \to \fibre{\Theta}{\Nat \PCFar \Nat} \\
      \PCFZerotest &: * \to \fibre{\Theta}{\Nat \PCFar \Bool } \\
       \PCFcond{\Nat} &: * \to \fibre{\Theta}{\Bool \PCFar \Nat \PCFar \Nat \PCFar \Nat} \\
       \PCFTrue, \PCFFalse &: * \to \fibre{\Theta}{\Bool}\\
	{}&\vdots
  \end{align*}

\vspace{0ex}

\end{minipage}
}
% \vspace{-2em}
 \caption{Term Signature of $\PCF$}\label{fig:pcf_sig_ho}
% \end{minipage}
% }
\end{figure}

Our presentation of $\PCF$ is inspired by Hyland and Ong's \cite{Hyland00onfull}, who --- similarly to 
Plotkin \cite{Plotkin1977223} --- consider, e.g., the successor as a constant of arrow type.
As an alternative, one might consider the successor as a constructor 
expecting a term of type $\Nat$ as argument, yielding a term of type $\Nat$.
For our purpose, those two points of view are equivalent.
\end{example}

\subsection{Representations of Typed Signatures}

A representation of a typed signature $(S,\Sigma)$ is given by a representation of $S$ (in a set)
and a representation of $\Sigma$ in a suitable monad:

\begin{definition}%[Representation of an Arity over $S$, of a Signature over $S$]\label{def:rep_term_ar}
            [Representation of a Signature over $S$]\label{def:rep_term_ar}
   Let $(S,\Sigma)$ be a typed signature. A \emph{representation $R$ of $(S,\Sigma)$} is given by
 \begin{itemize}
  \item an $S$--monad $P$ and
  \item for each arity $\alpha$ of $\Sigma$, a morphism (in the large category of modules)
          \[ \alpha^R: \dom(\alpha, P) \to \cod(\alpha, P) \enspace , \]
    such that $\pi_1(\alpha^R) = \id_P$.
 \end{itemize}
 In the following we also write $R$ for the $S$--monad underlying the representation $R$.
Note that the representation of $S$ is ``hidden'' in the $S$--monad $P$.
\end{definition}

A \emph{morphism of representations} accordingly consists of a morphism of representations of $S$ together
with a morphism of representations of $\Sigma$, that is, a monad morphism that is compatible with the term representations:

\begin{definition}[Morphism of Representations]\label{def:mor_of_reps}
  Given representations $P$ and $R$ of a typed signature $(S,\Sigma)$, a morphism of representations 
  $f : P\to R$ is given by a morphism of $S$--monads $f : P \to R$, such that, for any arity $\alpha$ of $\Sigma$,
  the following diagram of module morphisms commutes:
  \[
  \begin{xy}
   \xymatrix{
        **[l]\dom(\alpha,P) \ar[d]_{\dom(\alpha, f)} \ar[r]^{\alpha^P} & **[r]\cod(\alpha,P) \ar[d]^{\cod(\alpha,f)} \\
        **[l]\dom(\alpha,R) \ar[r]_{\alpha^R} & **[r]\cod(\alpha,R) .
    }
  \end{xy}
 \]
Again the morphism of representations of $S$ is ``hidden'' in the morphism of $S$--monads.
\end{definition}

\begin{remark}
 Taking a 2-categoric perspective, the above diagram can be read as an equality of 2-cells
\[
 \begin{xy}
   \xymatrix @!=6pc{
         **[l]P \ruppertwocell<12>^{\dom(\alpha,P)}{{\;\;\;\;\alpha^P}} \rlowertwocell_{f^*\cod(\alpha,R)}<-12>{\;\;cf}%{\cod(\alpha,f)}
                                                                    \ar[r]|{\cod(\alpha,P)}& **[r]\Id_{\Set}
    } 
  \end{xy} \quad = \quad
  \begin{xy} 
  \xymatrix @!=6pc{
         **[l]P \ruppertwocell<12>^{\dom(\alpha,P)}{\;\;df} \rlowertwocell_{f^*\cod(\alpha,R)}<-12>{\;\;\;\;\;\;\;f^*\alpha^R}%{\cod(\alpha,f)}
                                                                    \ar[r]|{f^*\dom(\alpha,R)}& **[r]\Id_{\Set}
    } 
 \end{xy} \enspace ,
\]
where we write $df$ and $cf$ instead of $\dom(\alpha, f)$ and $\cod(\alpha,f)$, respectively.

The diagram of \autoref{def:mor_of_reps} lives in the category $\LMod{n}{S}{\Set}$ --- where $n$ is the degree of $\alpha$ ---
 where objects are pairs $(P,M)$ of a $S$--monad $P$ of $\SigMon{S}_n$ and a module $M$ over $P$.
The above 2--cells are morphisms in the category $\Mod{P_n}{\Set}$, obtained by taking the second projection of 
 the diagram of \autoref{def:mor_of_reps}.
Note that for easier reading, we leave out the projection function and thus  
write $\dom(\alpha,R)$ for the $R_n$--module of $\dom(\alpha,R)$, i.e.\ for its second component, and 
similar elsewhere.
\end{remark}

 Representations of $(S,\Sigma)$ and their morphisms form a category.

\begin{remark}\label{rem:about_equiv_signatures}
 We obtain Zsid\'o's category of representations \cite[Chap.\ 6]{ju_phd} by restricting ourselves to representations of $(S,\Sigma)$ whose type representation
is the initial one.
More, precisely, 
a signature $(S,\Sigma)$ maps to a signature, say, $Z(S,\Sigma)$ over the initial set of sorts $\hat{S}$ 
 in the sense of Zsid\'o (cf.\ \autoref{sec:sts_ju} and \cite[Chap.~6]{ju_phd}), obtained by
unbundling each arity of higher degree into a family of arities of degree $0$. For instance, the 
signature of \autoref{ex:tlc_sig_higher_order} maps to the signature given in \autoref{ex:tlc_sig}.
% \[  \bigl( \App_{s,t} : [()s\SLCar t , ()s] \longrightarrow t \enspace , \enspace \Abs_{s,t}:[(s)t] \longrightarrow s\SLCar t \bigr)_{s,t \in T_{\SLC}} \enspace . \]
Representations of this latter signature in the sense of \autoref{sec:sts_ju} then are in one--to--one correspondence to 
representations in the sense of this section of the signature of \autoref{ex:tlc_sig_higher_order} \emph{over the initial representation $\hat{S}$ of sorts},
via the equivalence explained in \autoref{rem:family_of_mods_cong_pointed_mod}.
\end{remark}

\subsection{Initiality}

We have all the ingredients to state and prove an initiality theorem for typed signatures:

\begin{theorem}\label{thm:compilation}
  For any typed signature $(S,\Sigma)$, the category of representations of $(S,\Sigma)$ has an initial object.
\end{theorem}

\begin{proof}
  The proof consists of the following steps:
\begin{packenum}
 \item find the initial representation $\hat{S}$ of the type signature $S$;
 \item define the monad $\STS$ of terms specified by $\Sigma$ on the category $\TS{\hat{S}}$;
 \item equip the $S$--monad $\STS$ with a representation structure of $\Sigma$, yielding a representation
       $\hat{\Sigma}$ of $(S,\Sigma)$;
 \item for any representation $R$ of $(S,\Sigma)$, give a morphism of representations $i_R:\hat{\Sigma}\to R$;
 \item prove uniqueness of $i_R$.
\end{packenum}

We go through these points:

\begin{enumerate}
 \item 
  We have already established (cf.\ \autoref{lem:initial_sort}) that there is an initial representation of sorts, which we call $\hat S$.
   Its underlying set is called $\hat S$ as well.

 \item 
  The term monad we associate to $(S,\Sigma)$ is the same as 
  Zsid\'o's \cite[Chap.~6]{ju_phd} in the sense of \autoref{rem:about_equiv_signatures}, i.e.\ 
  it is the term monad associated to $Z(S,\Sigma)$.
  The construction of this monad in a set--theoretic setting is described in Zsid\'o's thesis.
  We will give its definition in a type--theoretic setting.

  In the following the natural transformations $\tau_i$ are in fact vectors of multiple transformations like those in \autoref{rem:nat_trans_picking_sort}
   (see also \autoref{def:derived_mod_II}), 
  iterated by successive composition. Furthermore we make use of the simplified notation as introduced in \autoref{not:tau_simpl_notation}.

  We construct the monad which underlies the initial representation of $(S,\Sigma)$,
  \[ \STS:\TS{\hat S} \to \TS{\hat S} \enspace . \]
  It associates to any set family of variables $V \in \TS{\hat S}$ an inductive set of terms with the following constructors:
  \begin{itemize}
   \item for every classic arity (of degree $n$) 
   \begin{equation} 
       \alpha = \fibre{\Theta_n}{\sigma_1}^{{{\tau_1}}} \times \ldots\times \fibre{\Theta_n}{\sigma_m}^{{\tau_m}} \to \fibre{\Theta_n}{\sigma} \label{eq:classic_arity_typed}
   \end{equation}
%
%       with vectors $\widetilde{\tau_i}$ of natural transformations as in \autoref{rem:nat_trans_picking_sort}   
      we have a family of constructors indexed $n$ times by $\mathbf{t} = (t_1,\ldots,t_n)$ as well as by the context $V \in \TS{\hat S}$:
         \[ \alpha_{{\mathbf t}}(V) : \STS^{{\tau_1}(V,\mathbf{t})}(V)_{\sigma_1(V,\mathbf{t})} \times\ldots\times 
                      \STS^{{\tau_m}(V,\mathbf{t})}(V)_{\sigma_m(V,\mathbf{t})} \to \STS(V)_{\sigma(V,\mathbf{t})} \]
   \item a family of constructors \[\Var(V)_t : V_t \to \STS(V)_t \] indexed by contexts and the set $\hat S$ of sorts.
  \end{itemize}
   The monadic structure is, accordingly, defined in the same way as in \cite{ju_phd}, by variables--as--terms --- using
   the constructor $\Var$ --- and flattening.
 \item

  The representation structure on the monad $\STS$ is defined by currying, and corresponds to Zsid\'o's: 
  given an arity $\alpha$ of degree $n$ in $\Sigma$, we must specify a module morphism 
   \[\alpha^{\hat{\Sigma}} : \dom(\alpha, \STS) \to \cod(\alpha, \STS) \enspace , \]
   where $\dom(\alpha, \STS)$ and $\dom(\alpha, \STS)$ are modules in $\Mod{\STS_n}{\Set}$.
  We define 
     \[\alpha^{\hat{\Sigma}}(V,\mathbf{t}) (a):= \alpha_{\mathbf{t}}(V)(a) \enspace , \]
  that is, the image under the constructor $\alpha$ from the definition of the monad $\STS$.
  This yields a morphism of modules $\alpha$ of degree $n$; note that according to 
  \autoref{rem:family_of_mods_cong_pointed_mod} it would be equivalent to specify a family
    $\alpha^{\hat{\Sigma}}_{\mathbf{t}}$ of module morphisms of suitable type, indexed by $\mathbf{t}$, which 
   is actually done by Zsid\'o.
 
\item
  Given any other representation $R$ over a set of sorts $T$, initiality of $\hat S$ gives
  a ``translation of sorts'' $g : \hat S \to T$.

   The morphism $i : \STS\to R$ on terms is defined by structural recursion. Unfolding the definition of 
  colax monad morphism, we need to define, for any context $V\in \family{\Set}{\hat{S}}$, a map of type
    \[i_V: \forall~t' \in T,~ \retyping{g}(\STS(V))_{t'} \to R (\retyping{g}V)_{t'} \enspace . \]
   Via the adjunction of \autoref{rem:retyping_adjunction_kan} 
%     The domain $\retyping{g}(\STS(V))_{g(t)}$ is a coproduct (cf.\ Def.\ \ref{def:retyping_functor}), 
  we equivalently define a map $i$ as a family
\[i_V:\forall~t \in \hat S,~ \STS(V)_{t} \to R (\retyping{g}V)_{g(t)} \enspace . \]
%   This equivalence corresponds to the possibility of doing pattern matching in \textsf{Coq}, cf.\ \autoref{code:pattern_match_retyping}.
%
   Let $a\in \STS(V)_t$ be a term. In case $a = \Var(V)_t(v)$ is the image of a variable $v\in V_t$, 
   we map it to 
      \[ i_V(\Var(V)_t(v)) := \eta^R(\retyping{g}V)(g(t))(\text{\lstinline!ctype!}(v)) \enspace .\]
   Otherwise the term $a = \alpha_{\vectorletter{t}}(V)(a_1,\ldots,a_k) \in \STS(V)_{\sigma(V,\mathbf{t})}$ is mapped to
   \begin{equation}i_V \bigl(\alpha_{\vectorletter{t}}(V)(a_1,\ldots,a_k)\bigr) := 
                 \alpha^R\left(\retyping{g}(n)(V,\vectorletter{t})\right) \bigl(i(a_1),\ldots,i(a_k)\bigr) \enspace .
       \label{eq:def_initial_term} \end{equation}
     This map is well--typed: note that $\retyping{g}(n)(V,\vectorletter{t}) = \left(\retyping{g}V, g_*(\vectorletter{t})\right)$ 
     by definition (\autoref{def:retyping_functor_pointed})
      and $\retyping{g}(n) ((V, \mathbf{t})^\tau) = \left(\retyping{g}V, g_*(\vectorletter{t})\right)^\tau $, 
  i.e.\ context extension and retyping permute.

   The axioms of monad morphisms, i.e.\ compatibility of this map with respect to variables--as--terms and flattening are easily
   checked: the former is a direct consequence of the definition of $i$ on variables, and the latter is proved by structural 
   induction.
   This definition yields a morphism \emph{of representations}; consider the arity $\alpha$ of $\Sigma$. 
   For this arity, the commutative diagram of  \autoref{def:mor_of_reps} informally
   reads as follows: one starts in the upper--left corner with a tuple of terms, say, $(a_1, \ldots,a_k)$ of $\STS$.
   Taking the upper--right path corresponds to the translation of the image of this tuple under the map $\alpha^{\hat{\Sigma}}$, 
   i.e.\ under the constructor $\alpha$ of $\STS$.
   The lower--left path corresponds to the image under the module morphism $\alpha^R$ of the translated tuple $(i(a_1), \ldots, i(a_k))$.
   The diagram thus precisely states the equality of \autoref{eq:def_initial_term}.
   We thus establish that $i$ is (the carrier of) a morphism of representations $(g,i):(\hat{S},\hat{\Sigma}) \to R$.
\item 
   Uniqueness of the morphism $i:(\hat{S},\hat{\Sigma}) \to R$ is proved making use of the commutative diagram of \autoref{def:mor_of_reps}.
   Suppose that $(g',i'):(\hat{S},\hat{\Sigma}) \to R$ is a morphism of representations.
   We already know that $g = g'$ by initiality of $\hat{S}$.

By structural induction on the terms of $\STS$ we prove that $i = i'$:
  using the same notation as above, for $a = \alpha_{\mathbf{t}}(V)(a_1,\ldots,a_k)$ we have
     \[ i'(a) = 
                  %i'(\alpha_{\mathbf{t}}(V)(a_1, \ldots ,a_k)) = 
          \alpha^R \left(i'(a_1), \ldots,i'(a_k)\right) \stackrel{i(a_i) = i'(a_i)}{=} \alpha^R\left(i(a_1),\ldots,i(a_k)\right) = i(a) \enspace . \]
  In case $ a = \Var(v)$ is a variable, considered as a term, 
 the fact that both $i$ and $i'$ are monad morphisms ensures that $i(\Var(v)) = i'(\Var(v)) = \eta^R_{\retyping{g}V}(\text{\lstinline!ctype!}(v))$.
  Thus we have proved $i = i'$.
 \end{enumerate}
\end{proof}

\noindent
The proof shows that the initial morphism to a representation $R$ 
depends on the representation structure on $R$ and not just on the monad $R$ itself.
We illustrate this on the example of the typed signature of \PCF:
\begin{example}\label{ex:init_pcf_translation_nosem}
 Representing the signature of \PCF~in the untyped lambda calculus leaves one with several choices to take, e.g., as to 
how to translate the fixed point operator $\mathbf{Fix}$. To represent $\mathbf{Fix}$ in $\ULC$, one must
give a unary operation on $\ULC$.
Reasonable from the semantic viewpoint are, e.g., the representations
\begin{equation} x \mapsto \App (\mathbf{Y},x) \qquad \text{ or } \qquad x \mapsto \App ({\Theta},x)  \enspace , \label{eq:2_reps_for_fix}
\end{equation} 
using, e.g., one of the fixedpoint combinators
\begin{align*}  \mathbf{Y}&:= \lambda f.(\lambda x.f (x x))(\lambda x.f (x x)) \quad \text{(Curry)} \\
                \Theta &:= (\lambda x. \lambda y.(y (x x y)))(\lambda x.\lambda y.(y(x x y))) \quad \text{(Turing).}
\end{align*}
 By initiality, those two representations yield two different compilations of \PCF~to $\LC$,
 mapping a \PCF~term of the form $\mathbf{Fix} (f)$ to 
 $\mathbf{Y} (f) = \App (\mathbf{Y},f)$ and $\Theta(f) = \App(\Theta,f)$, respectively.
 The representation module morphisms thus constitute the ``extra structure'' $\phi$, $\psi$ and $\psi'$ mentioned in \autoref{sec:trans_pcf_ulc}.
 A complete translation is given in \autoref{chap:comp_sem_formal}.
\end{example}

\section{Logics and Logic Translations}

% \begin{example}[Logic translation, \cite{TVD88}]
\label{ex:logic_trans}
  In the style of the Curry--Howard isomorphism, 
we consider propositions as types and proofs of a proposition as terms of that type.
  In this example we present the typed signatures of two different logics,
    \begin{packitem}
     \item Classical propositional logic, called \CPC, and
     \item Intuitionistic propositional logic, called \IPC.
    \end{packitem}
 According to our main theorem each of those signatures gives rise to an initial representation, a logical
   type system. We then use the \emph{iteration principle} on \CPC~in order to specify
   a translation \emph{of propositions and their proofs} from \CPC~to \IPC. 
   The translation we specify is actually the propositional fragment of the \emph{G\"odel--Gentzen negative 
  translation} \cite[Def.\ 3.4]{TVD88}.

\subsection{Signatures of Classical and Intuitionistic Logic}

We present typed signatures for classical and intuitionistic propositional logic. 
Their respective signatures for types --- \emph{propositions} --- are the same:
  let $P$ denote a set of atomic formulas.
  The \emph{types} --- propositions --- of classical (\CPC) and intuitionistic (\IPC) propositional logic are given by the following algebraic signature:
  \[ \mathcal{P} := \{p : 0, \quad \top : 0, \quad 
               \wedge : 2, \quad \bot : 0, \quad \vee : 2, \quad \impl : 2  \} \enspace . \]
where for any atomic formula $p\in P$ we have an arity $p:0$.
We call $\init{\mathcal{P}}$ the initial representation as well as its underlying set, i.e.\ the propositions of \CPC~and \IPC. 
For the set $\init{\mathcal{P}}$ we use infixed binary constructors.
Note that negation is defined as $\neg A \enspace \equiv \enspace A \impl \bot$.
 
\subsubsection{Signature of \CPC}

For the \emph{terms} of \CPC, each inference rule is given by an arity. In \autoref{tab:logic_inf_arity} (p.\ \pageref{tab:logic_inf_arity}),
 the inference rules and their corresponding arities are presented.
\def\fCenter{\mbox{$\vdash$}}
\begin{figure}[hbt]
 \centering
{ \renewcommand{\arraystretch}{2.3}
  \renewcommand{\tabcolsep}{3ex}
 \begin{tabular}{c | c}
   Inference Rule & Arity \\ \hline
   \AxiomC{}\RightLabel{$\top_\mathrm{I}$}\UnaryInfC{$\Gamma\vdash \top$} \DisplayProof & 
                  $\top_\mathrm{I} : * \to [\Theta]_{\top}$  \\

   \Axiom$\Gamma\ \fCenter\ \bot$ \RightLabel{$\bot_\mathrm{I}$}\UnaryInf$\Gamma\ \fCenter\ A$ \DisplayProof & 
                  $\bot_\mathrm{I} : [\Theta]_{\bot} \to [\Theta]_1$ \\

  \AxiomC{$\Gamma\vdash A $}\AxiomC{$\Gamma\vdash B$}\RightLabel{$\wedge_\mathrm{I}$}\BinaryInfC{$\Gamma\vdash A \wedge B$} \DisplayProof & 
                $\wedge_\mathrm{I} : [\Theta]_{1} \times[\Theta]_{2} \to [\Theta]_{1\wedge 2}$   \\
 
\Axiom$\Gamma\ \fCenter\ A \wedge B$ \RightLabel{$\wedge_\mathrm{E1}$} \UnaryInf$\Gamma\ \fCenter\ A$ \DisplayProof &
            $\wedge_\mathrm{E1} : [\Theta]_{1\wedge2} \to [\Theta]_1 $\\

\Axiom$\Gamma\ \fCenter\ A \wedge B$ \RightLabel{$\wedge_\mathrm{E2}$} \UnaryInf$\Gamma\ \fCenter\ B$ \DisplayProof &
            $\wedge_\mathrm{E1} : [\Theta]_{1\wedge2} \to [\Theta]_2 $\\

\Axiom$\Gamma , A\ \fCenter\ B$ \RightLabel{$\impl_{\mathrm{I}}$} \UnaryInf$\Gamma\ \fCenter\ A \impl B$ \DisplayProof &
            $\impl_{\mathrm{I}} :[\Theta]_2^1 \to [\Theta]_{1\impl 2}$\\

\AxiomC{$\Gamma\vdash A \impl B$}\AxiomC{$\Gamma\vdash A$}\RightLabel{$\impl_\mathrm{E}$}\BinaryInfC{$\Gamma\vdash B$} \DisplayProof & 
                  $\impl_\mathrm{E} :[\Theta]_{1\impl 2} \times [\Theta]_1 \to [\Theta]_2$ \\

\Axiom$\Gamma\ \fCenter\ A $ \RightLabel{$\vee_\mathrm{I1}$} \UnaryInf$\Gamma\ \fCenter\ A \vee B$ \DisplayProof &
            $ \vee_\mathrm{I1} :[\Theta]_1 \to [\Theta]_{1\vee 2}$\\

\Axiom$\Gamma\ \fCenter\ B $ \RightLabel{$\vee_\mathrm{I2}$} \UnaryInf$\Gamma\ \fCenter\ A \vee B$ \DisplayProof &
            $ \vee_\mathrm{I2} :[\Theta]_2 \to [\Theta]_{1\vee 2}$\\

\AxiomC{$\Gamma\vdash A \vee B$}\AxiomC{$\Gamma,A\vdash C$}\AxiomC{$\Gamma,B\vdash C$}\RightLabel{$\vee_\mathrm{E}$}\TrinaryInfC{$\Gamma\vdash C$} \DisplayProof &
                  $\vee_\mathrm{E} : [\Theta]_{1\vee 2} \times [\Theta]_3^1\times[\Theta]_3^2 \to [\Theta]_3$ \\

\AxiomC{}\RightLabel{$\mathrm{EM}$}\UnaryInfC{$\Gamma\vdash \lnot A \vee A$} \DisplayProof               &
            $\mathrm{EM} : *\to [\Theta]_{\lnot 1 \vee 1 }$
 
 \end{tabular}
}
\caption{Inference Rules of \CPC~and their Arities} \label{tab:logic_inf_arity}
\end{figure}
Each inference rule corresponds to a (family of) term --- \emph{proof} --- constructor(s),
 where inference rules without hypotheses are constants.
Note that the initial representation automatically comes with an additional inference rule
\begin{prooftree}
     \AxiomC{}
     \RightLabel{var}
     \UnaryInfC{$\Gamma,A \vdash A$}
\end{prooftree}
corresponding to
the monadic operation $\eta$, i.e.\ to the variables--as--terms constructor.
Analogously to \autoref{rem:var_eta}, it is not necessary, using our approach, to specify
this inference rule explicitly by an arity in the term signature of the logic under consideration;
any logic we specify via a typed signature automatically comes with this rule.

\subsubsection{Signature of \IPC}
The type signature and thus the formulas of intuitionistic propositional logic \IPC~are the same as for \CPC. 
However, the \emph{term} signature is missing the arity EM for excluded middle.

\subsection{Translation via Initiality}
The translation of propositions $(\_)^g : \init{\mathcal{P}} \to \init{\mathcal{P}}$, i.e.\ on the \emph{type} level, is specified by 
a representation $g$ of the algebraic signature $\mathcal{P}$ in the set $\init{\mathcal{P}}$.
According to \autoref{def:rep_alg_sig} we must specify, for any arity $s:n\in \N$ of $\mathcal{P}$, a map towards $\init{\mathcal{P}}$ 
taking a suitable number of arguments in $\init{\mathcal{P}}$,
\[  s^g : \init{\mathcal{P}}^n \to \init{\mathcal{P}} \enspace . \]
There is, of course, a canonical such map for each arity --- but this would only give us the identity morphism on $\init{\mathcal{P}}$.
We represent $\mathcal{P}$ in $\init{\mathcal{P}}$ not by this identity representation, but in such a way that 
we obtain the G\"odel--Gentzen negative translation:
\begin{align*} &p^g := \lnot\lnot p, \quad\top^g := \lnot\lnot\top, \quad \wedge^g:= \wedge, \quad \vee^g := (A,B)\mapsto \lnot (\lnot A \wedge \lnot B), \\ 
           &\impl^g := (\impl), \quad \bot^g := \lnot\lnot\bot \enspace .
\end{align*}
The proofs of \IPC~are given by the signature of \CPC~without the classical axiom EM.
We represent EM in \IPC~by giving, for any proposition $A$, a term of type $\lnot (\lnot \lnot A\wedge \lnot A)$, e.g.,
% \begin{prooftree}
%   \AxiomC{}
%      \RightLabel{var}
% \UnaryInfC{$\lnot\lnot A \wedge \lnot A\vdash\lnot\lnot A \wedge \lnot A$}
%       \RightLabel{$\wedge_\mathrm{E1}$}
% \UnaryInfC{$\lnot\lnot A \wedge \lnot A\vdash\lnot\lnot A$} 
%                                                     \AxiomC{}
%                                                    \RightLabel{var}                              
%                                                    \UnaryInfC{$\lnot\lnot A \wedge \lnot A\vdash\lnot\lnot A \wedge \lnot A$}
%                                                                      \RightLabel{$\wedge_\mathrm{E2}$}
%                                                 \UnaryInfC{$\lnot\lnot A \wedge \lnot A\vdash\lnot A$}\RightLabel{$\impl_\mathrm{E}$}
%           \BinaryInfC{$\lnot\lnot A \wedge \lnot A\vdash\bot$} \RightLabel{$\impl_\mathrm{I}$}
%           \UnaryInfC{$\vdash\lnot\lnot A \wedge \lnot A \impl \bot$}
% \end{prooftree}

\begin{prooftree} \def\fCenter{\mbox{$\vdash$}} \def\extraVskip{3pt}
  \AxiomC{}
     \RightLabel{var}
\UnaryInf$\lnot\lnot A \wedge \lnot A\ \fCenter\ \lnot\lnot A \wedge \lnot A$
      \RightLabel{$\wedge_\mathrm{E1}$}
\UnaryInf$\lnot\lnot A \wedge \lnot A\ \fCenter\ \lnot\lnot A$ 
                                                    \AxiomC{}
                                                   \RightLabel{var}                              
                                                   \UnaryInf$\lnot\lnot A \wedge \lnot A\ \fCenter\ \lnot\lnot A \wedge \lnot A$
                                                                     \RightLabel{$\wedge_\mathrm{E2}$}
                                                \UnaryInf$\lnot\lnot A \wedge \lnot A\ \fCenter\ \lnot A$
                                                         \RightLabel{$\impl_\mathrm{E}$}
          \BinaryInf$\lnot\lnot A \wedge \lnot A\ \fCenter\ \bot$ 
                      \RightLabel{$\impl_\mathrm{I}$}
          \UnaryInf$\fCenter\ \lnot\lnot A \wedge \lnot A \impl \bot$
\end{prooftree}
As another example, we give a representation of $\vee_\mathrm{I1}$, that is, for any proposition $A$ and $B$, 
we give a term of type $A^g  \to \neg(\neg A^g \wedge \neg B^g)$:
\begin{prooftree} \def\extraVskip{3pt} %\def\fCenter{\mbox{$\vdash$}} \def\extraVskip{3pt}
           \AxiomC{$A^g$}
%            \RightLabel{$(P \to \neg \neg P)$}
             \UnaryInfC{$\neg \neg A^g$}
              \RightLabel{$\vee_\mathrm{I1}$}
          \UnaryInfC{$\neg \neg A^g \vee \neg \neg B^g$} 
                      \RightLabel{De Morgan}
          \UnaryInfC{$ \neg(\neg A^g \wedge \neg B^g) $}
\end{prooftree}
Here the proof of $A^g \to \neg\neg A^g$ and of the used De Morgan law are abbreviations for 
longer proofs in \IPC.
We leave it up to the reader to find representations in \IPC~for the other arities.
 
\subsection{Remarks}
 This representation of the signature of \CPC~in \IPC~yields the (propositional fragment of the) 
  G\"odel--Gentzen translation of propositions specified in 
 Troelstra and van Dalen's book \cite[Def.\ 3.4]{TVD88}, denoted on propositions with the same name as its 
 specifying representation,
 \[ (\_)^g : \init{\mathcal{P}}\to\init{\mathcal{P}} \enspace . \]
Our translation of terms shows that any provable proposition in \CPC~translates to a provable proposition in \IPC,
since we provide the corresponding proof term via our translation:
  \begin{equation*} \Gamma \vdash_{\mathbf{C}} A \enspace \text{ implies } \enspace \Gamma^g \vdash_{\mathbf{I}} A^g \enspace .
%   \label{eq:faithful_logic} 
\end{equation*}
  
\noindent
However, a logic translation $t$ from a logic $\mathbf{L}$ to another logic $\mathbf{L'}$ should certainly satisfy an \emph{equivalence} of the form 
  \begin{equation*} \Gamma \vdash_{\mathbf{L}} A \enspace \text{ if and only if } \enspace \Gamma^t \vdash_{\mathbf{L'}} A^t \enspace .
%   \label{eq:faithful_logic} 
\end{equation*}
Our framework does \emph{not} ensure the implication from right to left, and is thus deficient from the point of view of logic translations.

Another important property of logics is normalization through cut elimination.
This aspect can be treated using the techniques presented in \autoref{chap:comp_types_sem}, where we integrate \emph{reduction rules}
into the notion of signature and their representations as presented in this chapter.

%% file: prop_arities.tex
We now would like to consider not just the \emph{terms} (and types) of a language, but also \emph{reductions}
on the terms. 
As an example, suppose we would like to equip the untyped lambda calculus with the reduction relation
generated by the beta rule given in \autoref{eq:beta}.
We could produce the syntax associated to the signature 
via the universal property explained in the preceding section
--- possibly in a computer implementation thereof ---
and define a suitable relation on the terms of the language \emph{a posteriori}.

However, in this way we would not have any guarantee concerning
compatibility of substitution with respect to this reduction relation.
Furthermore, how could we ensure any compatibility of a translation from the 
initial representation to another term language, equipped with some reduction rules, specified via the
iteration principle? There would not be any systematic way of doing so, we would need 
to check manually for each translation we consider.

The solution to this problem is to integrate reduction rules into signatures and the models
of those signatures.
Indeed, instead of considering reduction rules for just the initial representation of a signature, say, $\Sigma$,
we define \emph{inequations over $\Sigma$}, which specify rules for \emph{each} representation of $\Sigma$.
However, not all of the representations of $\Sigma$ satisfy those rules; 
we define a ``satisfaction'' predicate on the representations of $\Sigma$, 
to pick out the representations that satisfy those rules.

In order to define the satisfaction predicate, we need to consider representations 
whose codomain (read: the codomain of the underlying monad) is not the category
of plain sets, but of sets with a structure suitable to express relations between its elements.
The following monadic models come to mind:
\begin{itemize}
 \item[\color{red}{X}] $M:\Set\to\Set$ --- Terms modulo relations by quotienting\\
       We reject the idea of quotienting by the congruence relation generated by a set of inequations 
      on the grounds that we want to avoid adding a \emph{symmetry} rule and thus loose 
       the information of \emph{direction} of a reduction
 \item[\color{red}{X}] $M:\PO\to\PO$ --- Monads on preordered sets \\
       While the use of monads on preordered sets allows to retain directions of reductions, it would 
       necessitate to consider \emph{preordered} contexts. However, contexts usually 
       are given by \emph{unstructured} sets of variables.
 \item[\color{blue}{\checkmark}] $M:\Set\to\PO$ --- \emph{Relative Monads} from sets to preordered sets\\
       Relative monads from sets to preorders avoid the problems one encounters with the
       aforementioned approaches. As shown in \autoref{sec:rel_mon_defs}, the mediating functor to use
       is the functor $\Delta:\Set\to\PO$.
\end{itemize}
Before going into more detail concerning the models of signatures with inequations, we have a 
closer look at those signatures themselves.
Signatures should carry information about
\begin{packdesc}
 \item [Syntax] the terms, optionally typed over a set of sorts, and
 \item [Semantics] \emph{reductions} on the terms.
\end{packdesc}
Accordingly, we introduce a notion of \emph{$2$--signature}. A 2--signature $(\Sigma,A)$ consists
of a (higher--order) signature $\Sigma$ --- which we also call 1--signature from now on, to emphasize the existence of 
a second level, the \emph{semantic} level --- which specifies the terms of a language, as well as a set $A$ of 
\emph{inequations} over $\Sigma$. Each inequation of $A$ specifies a reduction rule.

We borrow the terms ``1--signature'' and ``2--signature'' from T.\ Hirschowitz \cite{HIRSCHOWITZ:2010:HAL-00540205:2}:
they are motivated by the point of view of Categorical Semantics. There, types and terms of a language are modelled
as the objects and morphisms of a category. 
Furthermore, reductions between terms may be modelled through 2--cells. 
In this way, a 1--signature specifies a 1--category,
whereas a 2--signature specifies a 2--category.

As the 1--signature which underlies a 2--signature, we may choose any of the notions of signature defined in the
preceding chapters (cf.\ \autorefs{def:sts_arity_signature}, \ref{def:typed_signature}).
For this chapter, however, we restrict ourselves to \emph{untyped} syntax with reductions,
allowing us to employ a simple notion of 1--signature.
The next chapter 
integrates reductions and types.

While we present 1--signatures from two perspectives, a syntactic one and a semantic one, we only present
inequations semantically. We refer to \autoref{sec:further_work} for thoughts about the syntactic aspect.

\section{1--Signatures}

We start out by defining \emph{1--signatures} in two different ways, once syntactically, and once in terms of pairs of
functors between suitable categories.

The syntactic description of arities is actually the same as in \autoref{sec:sts_ju}, even simpler:
since we only consider \emph{untyped} syntax, we just need to specify the number of arguments of a constructor, and,
for each argument, the number of variables bound in it:

\begin{definition}[Classic Arity, Signature]\label{def:classic_half_arity_prop_untyped_syntactic}
 A \emph{classic arity} is given by a list of natural numbers. 
 The length of the list indicates the number of arguments of its associated constructor, 
 whereas the $i$--th component of the list specifies the number of variables bound in the $i$--th argument.
 A \emph{classic signature} is given by a family of arities.
\end{definition}

\begin{example}[Untyped Lambda Calculus]\label{ex:sig_ulc_syn}
  The signature of the untyped lambda calculus is given by
 \[ \Sigma_{\ULC} := \{ \app : [0,0]\enspace , \quad \abs : [1] \} \enspace . \]
\end{example}
 
\noindent
For the semantic definition of arities, we define a suitable category of monads and a large category of modules.
As discussed at the beginning of the chapter, we use \emph{relative} monads and modules over relative monads.

We start by giving a simplified version of the definition of morphism of relative monads, to which we
restrict ourselves throughout this chapter. It is obtained from \autoref{def:colax_rel_mon_mor} by restricting the vertical 
functors $G$ and $G'$ to the identity functor.
Furthermore we will have $F = F'$, and the natural transformation $N$ is the identity transformation.
Given two relative monads $P$ and $Q$ on $F:\C \to \D$, a \emph{(simple) morphism of relative monads} is a
family of morphisms $\tau_c\in\D(Pc,Qc)$
that is compatible with the monadic structure:

\begin{definition}[Morphism of Relative Monads]\label{def:simp_rel_mon_mor}
Given two relative monads $P$ and $Q$ from $\C$ to $\D$ on the functor $F\colon\C\to\D$, 
 a \emph{morphism of monads} from $P$ to $Q$ is given by a collection of morphisms $\tau_c\in \D(Pc,Qc)$ such that 
the following diagrams commute for all suitable morphisms $f$:
\begin{equation*}
 \begin{xy}
  \xymatrix @=3pc{
  Pc \ar[r]^{\kl[P]{f}} \ar[d]_{\tau_c}& Pd \ar[d]^{\tau_d} & Fc \ar[r]^{\we^P_c} \ar[rd]_{\we^Q_c} & Pc \ar[d]^{\tau_c} \\
  Qc \ar[r]_{\kl[Q]{\comp{f}{\tau_d}}} & Qd & {} & Qc.\\
}
 \end{xy}
\end{equation*}
As a consequence from these commutativity properties the family $\tau$ is a natural 
transformation between the functors induced by the monads $P$ and $Q$ (cf.\ \autoref{rem:rel_monad_functorial}).
\end{definition}

\begin{definition}[Category of Relative Monads on $F$]
    Given a functor $F : \C \to \D$, we define the category $\RMon{F}$ to be the category whose objects are relative monads on $F$.
   A morphism from $P$ to $Q$ in $\RMon{F}$ is a morphism as in \autoref{def:simp_rel_mon_mor}.
\end{definition}

There is an adjunction between relative monads on $\Delta$ and monads on sets:

\begin{lemma}[Adjunction between $\Mon{\Set}$ and $\RMon{\Delta}$] \label{lem:adj_mon_rmon}
 The functors (with object functions) defined in \autoref{lem:rmon_delta_endomon} give rise to an adjunction
%   \[ \Delta \dashv U , \quad \bigl(\RMon{\Delta}\bigr)(\Delta P, Q) \cong \bigl(\Mon{\Set}\bigr)(P, UQ) \enspace . \]
\begin{equation*}   
       \begin{xy}
        \xymatrix@C=4pc{
                  **[l]\Mon{\Set} \rtwocell<5>_{U_*}^{\Delta_*}{'\bot} & **[r]\RMon{\Delta}
}
       \end{xy} \enspace .
%\Rep^{\Delta}(S) : \Delta \dashv U : \Rep(S) 
%  \label{eq:adjunction}
\end{equation*}
 
\end{lemma}
\begin{proof}
  The isomorphism $\varphi_{P,Q}:\RMon{\Delta}(\Delta_*P,Q) \cong \Mon{\Set}(P,U_*Q)$
  is defined by applying the adjunction of \autoref{lem:adj_set_po} in each morphism of the 
  family underlying a morphism of (relative) monads.
  Commuting diagrams are not modified by applying this adjunction.
  Naturality of $\varphi$ is trivial.
\end{proof}

\begin{definition}[Large Category of Modules]\label{def:lrmod_small}
  Given a functor $F: \C \to \D$ and a category $\E$, we define the category $\LRMod{}{F}{\E}$ to be the category 
  whose objects are pairs $(P,M)$ of a relative monad $P\in \RMon{F}$ and a relative $P$--module $M$ with codomain $\E$.
  A morphism to another such $(Q,N)$ is a pair $(h, f)$ of a morphism $h : P \to Q$ in $\RMon{F}$ and a morphism of $P$--modules
  $f : P \to h^* Q$ to the pullback of $Q$ along $h$ (cf.\ \autoref{subsection:rmod_examples}).
  
  For any monad $P$ on $F$ there is the injection functor
  \begin{equation*}
     I_P : \RMod{P}{\E} \to \LRMod{}{F}{\E}, \quad f \mapsto (\id, f) \enspace .     %  \label{eq:injection_rmod_lrmod}
  \end{equation*}
\end{definition}

A \emph{half--arity} associates a $P$--module towards the category $\PO$ of preorders to any relative monad $P$ on $\Delta$:
\begin{definition}[Half--Arity]
 A \emph{half--arity} $a$ is a functor 
   \[ a : \RMon{\Delta} \to \LRMod{}{\Delta}{\PO} \enspace   \]
 that is pre--inverse to the forgetful functor.
\end{definition}

Similarly to the preceding sections we restrict our attention to \emph{classic} half--arities:
\begin{definition}[Classic Half--Arity] \label{def:classic_half_arity_prop_untyped}
 The following clauses define the inductive set of \emph{classic} half--arities:
  \begin{itemize}
    \item $\Theta : P \mapsto (P,P)$, the tautological module, is classic;
    \item if $M$ is classic, so is its derivation $M' : P \mapsto (P, M(P)')$;
    \item if $M$ and $N$ are classic, so is their product $M \times N:P \mapsto (P, M(P) \times N(P))$;
    \item the constant half--arity $* : P \mapsto 1$ is classic.
  \end{itemize}

\end{definition}

\noindent
 Classic half--arities as defined in \autoref{def:classic_half_arity_prop_untyped}
  are in one--to--one correspondence to classic arities as defined in \autoref{def:classic_half_arity_prop_untyped_syntactic}:

\begin{remark}\label{rem:syn_sem_untyped_two_half_arities}
  We use the notation defined in \autoref{not:deriv_rmod_untyped}.  More generally, given a
  list of natural numbers $s = [n_1,\ldots,n_m]$, we write
   $M^s := M^{n_1} \times M^{n_2} \times\ldots\times M^{n_m}$.
  
  The same notation is used for morphisms, i.e.\ given a morphism of $R$--modules $f : M \to N$, we write 
\[  f^s := f^{n_1} \times\ldots\times f^{n_m} : M^s \to N^s \enspace . \]
  Thus any list of natural numbers specifies uniquely a classic half--arity, the empty list 
   denoting the terminal module $* : R\mapsto 1$.
\end{remark}

\begin{definition}[Arity]
 An \emph{arity} $s$ is a pair 
  $ s = (\dom(s), \cod(s))$
of half--arities
 \[ \dom(s), \cod(s) : \RMon{\Delta} \to \LRMod{}{\Delta}{\Set} \enspace . \]
We write $s = \dom(s)\to\cod(s)$, and $\dom(s,P) := \dom(s)(P)$ (and similarly for $\cod$).

%  A signature is a family (indexed by some set) of arities.
\end{definition}

\begin{definition}[Classic Arity, 1--Signature]\label{def:1--signature_untyped}
 A \emph{classic arity} is an arity of the form
   \[ \dom(s) \to \Theta \]
 such that $\dom(s)$ is a classic half--arity.
 %Using the notational convention of Def.\ \ref{def:classic_half_arity_prop_untyped}, any 
  Any classic arity as in \autoref{def:classic_half_arity_prop_untyped_syntactic} 
    uniquely specifies a classic arity by specifying its domain according to \autoref{rem:syn_sem_untyped_two_half_arities}.
  A \emph{1--signature} is a family of classic arities, or, equivalently according to \autoref{rem:syn_sem_untyped_two_half_arities}, 
a family of lists of natural numbers.
\end{definition}

\begin{example}[Untyped Lambda Calculus] \label{ex:sig_ulc}
 The 1--signature $\Sigma_{\ULC}$ of the untyped lambda calculus, already given syntactically in \autoref{ex:sig_ulc_syn}, is given by the two arities
\[ \app := \Theta \times \Theta \to \Theta \enspace , \quad \abs := \Theta' \to \Theta \enspace . \]
\end{example}

\section{Representations of 1--Signatures}

A representation of a classic arity $s$ in a monad $P$ is a module morphism $\dom(s,P) \to P$.
More generally:

\begin{definition}[Representation of an Arity]\label{def:rep_1--arity_untyped}
 A representation of an arity $s = \dom(s) \to \cod(s)$ in a monad $P$ on $\Delta$ is
 a morphism $M$ of $P$--modules
  \[ M : \dom(s,P) \to \cod(s,P) \]
 in the category $\LRMod{}{\Delta}{\Set}$, such that $\pi_1(M) = \id$.
 By abuse of notation, we also denote by $M$ the second projection of $M$, i.e.\
 we consider $M\in\RMod{P}{\Set}$.
\end{definition}

A representation of a signature is given by a relative monad on $\Delta$ 
and a representation of each arity in this monad:

\begin{definition}[Representation of a 1--Signature]\label{def:rep_in_rmonad}
A \emph{representation $R$ of a signature} $\Sigma$ is given by 

\begin{packitem}
 \item a monad $P$ on $\Delta$ and
 \item a representation $s^R:\dom(s,P) \to \cod(s,P)$ of each arity $s\in \Sigma$ in $P$ as in \autoref{def:rep_1--arity_untyped}.
\end{packitem}
 Given a representation $R$, we denote its underlying monad by $R$ as well. %By abuse of notation, we just write $R$ for the representation $(R,r)$.
\end{definition}

For any signature $\Sigma$ as in \autoref{def:1--signature_untyped}, we have representations of $\Sigma$ in monads on $\Set$ (cf.\ \autoref{def:sts_sig_rep})
and in relative monads on $\Delta$ (cf.\ \autoref{def:rep_in_rmonad}). 
The following definition links those representations:

\begin{definition}[Reps.\ in Relative Monads and Monads] \label{lem:rep_endo_rep_rel}
 To any representation of a classic signature $\Sigma$ in a relative monad $R$ as defined in \autoref{def:rep_in_rmonad} 
we associate a representation of $\Sigma$ in the monad $U_*R$ (cf.\ \autoref{lem:adj_mon_rmon}) according to the definition of representation of \autoref{def:sts_sig_rep}, 
by postcomposing with the forgetful functor from preorders to sets.

 Conversely, to any representation of $\Sigma$ in a monad $Q$ over sets we associate a representation 
of $\Sigma$ in the relative monad $\Delta_* Q$ over $\Delta$, by postcomposing with $\Delta$.
  More precisely, an arity $s = [s_1,\ldots,s_n]\in \Sigma$ and a representation of $s$ in $Q$, say,
 \[ s^Q : Q^s  \to Q \enspace , \]
  with $Q^s := Q^{s_1} \times \ldots \times Q^{s_n}$, we have to give a morphism of modules 
 \[  \Delta_* Q^{s_1} \times \ldots \times \Delta_* Q^{s_n} \to \Delta Q_* \enspace , \]
  that is, a family of monotone morphisms in the category $\PO$.
  However, the domain module is isomorphic to $\Delta_* Q^s$, hence postcomposing the map $s^Q$ with $\Delta$ 
  does the job,
  \[ \Delta_* s^Q : \Delta_* Q^s \to \Delta_* Q \enspace , \]
  and $\Delta_* s^Q$ obviously has the necessary commutation property with respect to substitution.
\end{definition}

\begin{example}[\autoref{ex:sig_ulc} continued]\label{ex:ulc_rep}

A representation $P$ of $\Sigma_{\ULC}$ is given by
\begin{packitem}
 \item a monad $P : \SET\stackrel{\Delta}{\to}\PO$ and
 \item two morphisms of $P$--modules in $\RMod{P}{\PO}$,
       \[ \app : P \times P \to P \quad\text{and}\quad \abs : P' \to P \enspace .\] 
\end{packitem}
\end{example}

\noindent
Morphisms of representations are monad morphisms which commute with the representation morphisms of modules:
\begin{definition}[Morphism of Representations]\label{def:prop_mor_of_reps}
Let $P$ and $Q$ be representations of a classic signature $\Sigma$. A \emph{morphism of representations} $f:P \to Q$ is a morphism of monads 
    $f \colon P\to Q$ such that the following diagram commutes for any arity $s \in \Sigma$: %(\vec{s}_{1})t_1,\ldots,(\vec{s}_{n})t_n \to t_0$ of $S$
\[
 \begin{xy}
  \xymatrix{ **[l] \dom(s,P) \ar[r]^{s^P} \ar[d]_{\dom(s,f)} & P \ar[d]^{f} \\
             **[l] \dom(s,Q) \ar[r]_{s^Q} & Q .
  }
 \end{xy}
\]
%
% \[
% \begin{xy}
% \xymatrix @=4pc{
% **[l] \prod\limits_{i=1}^m \partial^{n_i} P\ar[r]^{p_{\alpha}} \ar[d]_{\prod\limits_i \partial^{n_i}f} & P \ar[d]^{f}  \\
% **[l] f^* \prod\limits_{i=1}^m \partial^{n_i} Q \ar[r]_{f^* q_{\alpha}} & f^*Q \\
% }
% \end{xy}
% \]
\end{definition}

\noindent
The meaning of those diagrams might become clearer when we consider the example of the untyped lambda calculus.
In line with the abuse of notation mentioned in \autoref{def:rep_1--arity_untyped}, 
we omit the first component of objects and morphisms in $\LRMod{}{\Delta}{\Set}$:

\begin{example}[\autoref{ex:ulc_rep} continued]
  \label{ex:ulc_rep_mor}
  Let $P$ and $R$ be two representations of $\Sigma_{\ULC}$. 
A \emph{morphism} from $P$ to $R$ is given by a morphism of monads
 $f : P \to R$ such that the following diagrams of $P$--module morphisms commute:
\begin{equation*}
\begin{xy}
 \xymatrix @C=4pc  {
   **[l]P\times P \ar[r]^{\app^P} \ar[d]_{f \times f}& P \ar[d]^{f} & & **[l] P' \ar[r]^{\abs^P} \ar[d]_{f'}& P\ar[d]^{f} \\
   **[l]f^*(R\times R) \ar[r]_{f^*(\app^R)} & f^*R & & **[l]f^*R' \ar[r]_{f^*(\abs^R)}& f^*R .
}
\end{xy}
\end{equation*}
\end{example}

\noindent
To make sense of these diagram it is necessary to recall the constructions on modules of \autoref{subsection:rmod_examples}. 
The diagrams live in the category $\RMod{P}{\PO}$. 
The vertices are obtained from the tautological modules $P$ resp.\ the $Q$ over the monads $P$ resp.\ $Q$ by applying 
the pullback (for $Q$) and derivation functors as well as by the use of the product in the category of $P$--modules into $\PO$.
The vertical morphisms are module morphisms induced by $f$, to which --- on the left--hand side --- 
functoriality of derivation and products are applied. 
Furthermore instances of \autorefs{lem:rel_pb_prod} and \ref{lem:rel_pb_comm} are hidden in the lower left corner. 
The lower horizontal morphism makes use of the functoriality of the pullback operation.

\begin{definition}[Category of Representations] \label{def:cat_of_reps_relmons}
 Representations of $\Sigma$ and their morphisms form a category $\Rep^{\Delta}(\Sigma)$.
 %This category has a terminal object, whose underlying monad is the terminal monad over $\Delta$ (cf.\ lemma \ref{rem:cat_final}).
 \end{definition}

\begin{lemma}[Adj.\ between Reps.\ in Rel.\ Monads and Reps.\ in Monads] \label{lem:adjunction_reps}
 The assignment of \autoref{lem:rep_endo_rep_rel} extends to an adjunction between the 
  category of representations in relative monads on $\Delta$ 
 and the category of representations in monads on sets (cf.\ \autoref{def:sts_cat_of_reps}):
\begin{equation*}   
       \begin{xy}
        \xymatrix@C=4pc{
                  **[l]\Rep(\Sigma) \rtwocell<5>_{U_*}^{\Delta_*}{'\bot} & **[r]\Rep^{\Delta}(\Sigma)
}
       \end{xy} \enspace . 
\end{equation*}
  
\end{lemma}

\begin{lemma}[Initiality for 1--Signatures]\label{lem:init_no_eqs_untyped}
The category of 
  representations of a signature $\Sigma$ in relative monads as defined in \autoref{def:cat_of_reps_relmons}
  has an initial object. Its underlying monad associates, to any set of variables, the set of terms of $\Sigma$,
  equipped with the equality preorder.
\end{lemma}
\begin{proof}
 This is a direct consequence of \autoref{lem:left_adj_cocont} which says that left adjoints preserve colimits --- thus, in particular, initial objects ---,
 applied to the adjunction of \autoref{lem:adjunction_reps}.
\end{proof}

\section{Inequations}

Consider the beta rule of lambda calculus,
\[ \lambda M(N) \leadsto M[*:=N] \enspace . \]
In our formalism, abstraction and application are considered as morphisms of modules 
(cf.\ \autoref{ex:ulcb_constructor_mod_mor}), 
and so is substitution (cf.\ \autoref{def:hat_P_subst}). 
This suggests to define (in)equations over a 1--signature $\Sigma$ 
as \emph{parallel pairs} of module morphisms, indexed by representations of $\Sigma$.
Put differently, an (in)equation associates a parallel pair of module morphisms to any representation of $\Sigma$.
Hirschowitz and Maggesi \cite{journals/corr/abs-0704-2900} specify equations through such pairs of (indexed) module morphisms over (plain) monads.
We adapt their definition to our use of \emph{relative} monads and modules over such monads.
Afterwards we simply interpret a pair of half--equations as \emph{inequation} rather than equation.

\begin{definition}[Category of Half--Equations, \cite{journals/corr/abs-0704-2900}] \label{def:half_eq_untyped}
Let $\Sigma$ be a signature. A \emph{$\Sigma$--module} $U$
is a functor from the category of representations of $\Sigma$ to the category $\LRMod{}{\Delta}{\PS}$
commuting with the forgetful functors %--- denoted by $\pi_1$, respectively --- 
to the category of relative monads over $\Delta$:
\[
 \begin{xy}
  \xymatrix{
    **[l]  \Rep^{\Delta}(\Sigma) \ar[rd]%_{\pi_1} 
                    \ar[rr]^{U} 
                                       &               & **[r]\LRMod{}{\Delta}{\PS} \ar[ld]%^{\pi_1}
     \\
                              & \RMon{\Delta} . &
}
 \end{xy}
\]

\noindent
Such a $\Sigma$--module $U$ associates, to any representation of $\Sigma$ with underlying monad $P$, a module over $P$.

We define a morphism of $\Sigma$--modules to be a natural transformation which
becomes the identity when composed with the forgetful functor. 
We call these morphisms \emph{half--equations}.
These definitions yield a category
which we call the \emph{category of $\Sigma$--modules
(or the category of half--equations)}.
We sometimes write \[U^R_X := U(R)(X)\] for the value of a $\Sigma$--module at the representation $R$ and the set $X$.
Similarly, for a half--equation $\alpha : U \to V$ we write \[\alpha^R_X := \alpha(R)(X) : U^R_X \to V^R_X \enspace.\]
\end{definition}

\begin{remark}\label{rem:wOrd_instead_Ord}
 We define $\Sigma$--modules over the signature $\Sigma$ as functors into the category $\LRMod{}{\Delta}{\PS}$, 
 whose objects are modules \emph{with codomain category $\PS$ instead of $\PO$} to accommodate an important example:
  recall that substitution of \emph{one} variable (cf.\ \autoref{def:hat_P_subst}) 
  is not necessarily monotone in the second argument.
  Thus, in order to build a half--equation from this substitution (cf.\ \autoref{def:subst_half_eq}),
   we need to use the category $\PS$ as codomain category.
\end{remark}

\begin{remark}\label{rem:half_equation_comm}
 A half--equation $\alpha$ from $\Sigma$--module $U$ to $V$ associates, to 
 any representation $R$, a morphism of $R$--modules $\alpha^R : U(R) \to V(R)$ in $\RMod{R}{\PS}$ such that
for any morphism $f : P \to R$ of representations of $\Sigma$ the following diagram commutes:
 \begin{equation*}
\begin{xy}
 \xymatrix @=4pc  {
   **[l](P,U(P)) \ar[r]^{\alpha^P} \ar[d]_{(f,U(f) )}& **[r](P,V(P)) \ar[d]^{(f, V(f))} \\
   **[l](R, f^*(U(R))) \ar[r]_{\alpha^R} & **[r](R,f^* (V(R))) \enspace .
}
\end{xy}
\end{equation*}
\end{remark}

\begin{remark}\label{rem:half--equation_variant}
 Pierre--Louis Curien suggested the following alternative definition of a half--equation,
  where its domain and codomain only depend on \emph{the monad underlying each representation}: 
  domain and codomain are specified by functors $U$ and $V$ on the category $\RMon{\Delta}$, and 
  a half--equation $\alpha$ from  $U$ to $V$ is given by a natural transformation 
  \[ \alpha : \comp{\pi_1}{U} \to \comp{\pi_1}{V} \enspace , \]
 where $\pi_1 : \Rep^{\Delta}(\Sigma) \to \RMon{\Delta}$ is the forgetful functor.
 Indeed, in all the examples of half--equations we consider, the domain and codomain $\Sigma$--modules
 only depend on the monads underlying a representation, not the representation structure itself.
 Both variants, the one presented here in detail as well as the one suggested by Curien,
 are implemented in our \textsf{Coq} library.
\end{remark}

% \noindent
Given a 1--signature $\Sigma$, we restrict ourselves to \emph{classic} inequations:
these are inequations whose codomain $\Sigma$--module
is of a specific form. The restriction to these inequations allows us to ensure a technical condition
which we prove, for classic inequations, in \autoref{lem:useful_lemma}. 
Analogously to the preceding chapters, we only write the second component of objects in the 
large category $\LRMod{}{\Delta}{\PS}$ of modules.

\begin{definition}[Classic $\Sigma$--Module]\label{def:alg_s_mod}
 We call \emph{classic} any $\Sigma$--module satisfying the following inductive predicate. %contained in the following inductive set:

 \begin{packitem}
  \item The map $\hat{\Theta} : R \mapsto \widehat{\pi_1R}$ (cf.\ \autoref{def:forget_hat_module} and \autoref{rem:half--equation_variant}) 
      is a classic $\Sigma$--module.
  \item If the $\Sigma$--module $M : R\mapsto M(R)$ is classic, so is 
              \[M' : R\mapsto M(R)' \enspace .\]
  \item If $M$ and $N$ are classic, so is \[M\times N : R\mapsto M(R)\times N(R) \enspace . \]
  \item The terminal module $* : R \mapsto 1$ is classic.
 \end{packitem}
Using the same notation as in \autoref{rem:syn_sem_untyped_two_half_arities}, any list of natural numbers specifies uniquely a 
 classic $\Sigma$--module.
\end{definition}

We now present some particular classic half--equations:

\begin{definition}\label{def:subst_half_eq}
 The substitution operation of \autoref{def:hat_P_subst},
 \[\subst  : R\mapsto \subst^R : \hat{R}' \times \hat{R}\to \hat{R} \]
 is a half--equation over any 1--signature $\Sigma$.
Its domain and codomain are classic.

\end{definition} 
\begin{example}[\autoref{ex:sig_ulc} continued]\label{ex:app_circ_half}
  The map
 \[ \comp{(\abs \times\id)}{\app}:R \mapsto \comp{(\abs^R \times\id^R)}{\app^R}  : \hat{R}' \times \hat{R} \to \hat{R}\]
 is a half--equation over the signature $\Sigma_{\ULC}$.
\end{example}

\begin{definition}
 \label{def:arity_classic_module_untyped}
 Any arity $s = [n_1, \ldots, n_m] \in \Sigma$ defines a classic $\Sigma$--module 
    \[\dom(s) : R\mapsto R^{n_1} \times \ldots \times R^{n_m} \enspace . \] 
\end{definition}

\noindent
An \emph{inequation} is given by a pair of parallel half--equations:

\begin{definition}[Inequations, 2--Signature] \label{def:ineq_untyped}
 Given a 1--signature $\Sigma$, a \emph{$\Sigma$--inequation} is a pair of parallel half--equations between $\Sigma$--modules.
We write 
     \[\alpha \leq \gamma : U\to V\] 
for the inequation $(\alpha, \gamma)$ with domain $U$ and codomain $V$.
A \emph{2--signature} is a pair $(\Sigma,A)$ of a 1--signature $\Sigma$ and a set $A$ of 
 $\Sigma$--inequations.
\end{definition}

Given a 2--signature $(\Sigma,A)$, we can test whether a given representation $R$ of $\Sigma$ 
satisfies the inequations of $A$. Those representations satisfying any inequation of $A$
form the category of representations of $(\Sigma,A)$:

\begin{definition}[Representation of Inequations]
 \label{def:rep_ineq_untyped}
 A \emph{representation of a $\Sigma$--inequation $\alpha\leq \gamma : U\to V$} is any representation $R$ of $\Sigma$ such that 
  $\alpha^R \leq \gamma^R$ pointwise, i.e.\ such that for any set $X$ and any $y\in U(R)(X)$, 
    \[\alpha^R_X(y) \leq \gamma^R_X(y) \enspace .\]
We say that such a representation $R$ \emph{satisfies} the inequation $\alpha\leq \gamma$.

For a set $A$ of $\Sigma$--inequations, we call \emph{representation of $(\Sigma,A)$} any representation of $\Sigma$ that
satisfies each inequation of $A$.
We define the
category of representations of the 2--signature $(\Sigma, A)$ to be the full subcategory $\Rep^{\Delta}(\Sigma,A)$ of the category of
representations of $\Sigma$ whose objects are representations of $(\Sigma, A)$.

\end{definition}

\begin{example}[\autoref{ex:app_circ_half} continued] \label{ex:sig_ulc_prop}
We denote by $\beta$ the $\Sigma_{\ULC}$--inequation 
\[  \quad \app \circ (\abs \times \id) \leq \subst \enspace .  \tag{$(\beta)$} \label{eq:beta_ineq}\]
We write $(\Sigma_{\ULC},\beta) := (\Sigma_{\ULC},\{\beta\})$.
 A representation $P$ of $(\Sigma_{\ULC},\beta)$ is given by
\begin{packitem}
 \item a monad $P : \SET\stackrel{\Delta}{\to}\PO$ and
 \item two morphisms of $P$--modules 
       \[ \app : P \times P \to P \quad\text{and}\quad \abs :  P' \to P\] 
\end{packitem}
 such that for any set $X$ and any $y\in P(X')$ and $z\in P X$
\begin{equation*} \label{eq:prop_arity}
\app_X (\abs_X (y), z) \enspace \leq \enspace y [*:=z] \enspace .  
\end{equation*} 
\end{example}

\section{Initiality for 2--Signatures}

Given a 2--signature $(\Sigma,A)$, we would like to exhibit an initial object
in its associated category of representations of $(\Sigma,A)$. 
However, we have to rule out inequations which are never satisfied, 
since an empty category obviously does not not have an initial object.
%We hence restrict ourselves to inequations which are satisfied by at least the terminal representation of $S$.  
We restrict ourselves to inequations with a classic codomain:
\begin{definition}[Classic Inequation]
 \label{def:classic_ineq_untyped}
 A $\Sigma$--inequation is \emph{classic} if its codomain is classic.
\end{definition}

\begin{theorem}\label{thm:init_w_ineq_untyped}
 For any set of classic $\Sigma$--inequations $A$, the category of representations of $(\Sigma,A)$ has an initial object.
\end{theorem}

\begin{proof}

The basic ingredients for building the initial representation 
  are given by the initial representation $\Delta\init{\Sigma}$ in the category $\Rep^{\Delta}(\Sigma)$ (cf.\ \autoref{lem:init_no_eqs_untyped}) 
  or, equivalently, by the initial representation $\init{\Sigma}$ in $\Rep(\Sigma)$.
  We call $\init{\Sigma}$ the monad underlying the representation $\init{\Sigma}$.
  
  The proof consists of three steps: at first, we define a preorder $\leq_A$ on the terms of $\init{\Sigma}$, induced by the set $A$ of inequations.
   Afterwards we show that the data of the representation $\init{\Sigma}$ --- substitution, representation morphisms etc. --- 
  is compatible with the preorder $\leq_A$ in a suitable sense. This will yield a representation $\init{\Sigma}_A$ of $(\Sigma,A)$.
  Finally we show that $\init{\Sigma}_A$ is the initial such representation.

\noindent
\emph{--- The monad underlying the initial representation:}

\noindent
    For any set $X$, we equip $\init{\Sigma} X$ with a preorder $A$ by setting, for $x,y\in \init{\Sigma} X$,

    \begin{equation} x \leq_A y \quad :\Leftrightarrow \quad \forall R : \Rep^{\Delta}(\Sigma,A), \quad%\text{ R satisfies } A \Rightarrow
                      i_R (x) \leq_R i_R (y) \enspace ,
      \label{eq:order_untyped}
    \end{equation}
where $i_R : \Delta\init{\Sigma} \to R$ is the initial morphism of representations of $\Sigma$, cf.\ \autoref{lem:init_no_eqs_untyped}. 
We have to show that the map \[ X\mapsto \init{\Sigma}_A X := (\init{\Sigma} X, \leq_A) \] yields a relative monad on $\Delta$. 
The missing fact to prove is that the substitution with a morphism 
\[ f\in\PO(\Delta X, \init{\Sigma}_A Y) \cong \SET(X,\init{\Sigma} Y)\] 
is compatible with the order $\leq_A$:
given any $f \in \PO(\Delta X, \init{\Sigma}_A Y)$ we show that $\sigma^{\init{\Sigma}}(f): \Set(\init{\Sigma} X,\init{\Sigma} Y)$ is monotone with respect
to $\leq_A$ and hence (the carrier of) a morphism $\sigma(f) : \PO(\init{\Sigma}_A X, \init{\Sigma}_A Y)$.
We overload the infix symbol $\bind{}{}$ to denote monadic substitution.
Suppose $x\leq_A y$, we show 
        \[\bind{x}{f} \enspace \leq_A \enspace \bind{y}{f} \enspace .\] 
Using the definition of $\leq_A$, 
we must show, for any representation $R$ of $(\Sigma,A)$, 
\[  i_R(\bind{x}{f}) \enspace \leq_R \enspace i_R(\bind{y}{f}) \enspace .\]
Since $i_R$ is a morphism of representations, it is compatible with the substitution of $\init{\Sigma}$ and $U_*R$; we have
 \[ i_R(\bind{x}{f}) \enspace = \enspace \bind{i_R(x)} {\comp{f}{i_R}} \enspace . \]
Rewriting this equality and its equivalent for $y$ in the current goal yields the goal 
\[ \bind{i_R(x)} {\comp{f}{i_R}} \quad \leq_A \quad \bind{i_R(y)} {\comp{f}{i_R}} \enspace ,\]
which is true since the substitution of $R$ (whose underlying map is that of $U_*R$) is monotone 
in the first argument (cf.\ \autoref{rem:about_substitution})
 and $i_R (x) \leq_R i_R(y)$ by assumption.
We hence have defined a monad $\init{\Sigma}_A$ over $\Delta$.
We interrupt the proof for an important lemma:
\begin{lemma}\label{lem:useful_lemma}
  Given a classic $\Sigma$--module $V : \Rep^{\Delta}(\Sigma) \to \LMod{}{\Delta}{\PS}$ from the category of representations of $\Sigma$ in monads on $\Delta$
  to the large category of modules over such monads, we have
   \[ x \leq_A y \in V(\init{\Sigma})(X) \quad \Leftrightarrow \quad \forall R : \Rep^{\Delta}(\Sigma,A), \quad V(i_R)(x)\leq_{V^R_X} V(i_R)(y) \enspace ,\]
 where now and later we omit the argument $X$, e.g., in $V(i_R)(X)(x)$.
\end{lemma}
\begin{proof}[Proof of \autoref{lem:useful_lemma}.]
   The proof is done by induction on the derivation of ``$V$ classic''. The only interesting case is where $V = M\times N$ is a product:
    \begin{align*} 
	  (x_1, y_1) \leq (x_2,y_2) &\Leftrightarrow x_1 \leq x_2 \wedge y_1 \leq y_2 \\
                         {}         &\Leftrightarrow  \forall R, M(i_R) (x_1) \leq M(i_R) (x_2) \wedge \forall R, N (i_R) (y_1) \leq N(i_R) (y_2) \\
                         {}         &\Leftrightarrow  \forall R, M(i_R) (x_1) \leq M(i_R) (x_2) \wedge  N (i_R) (y_1) \leq N(i_R) (y_2) \\
                         {}         &\Leftrightarrow  \forall R, V(i_R) (x_1,y_1) \leq V(i_R) (x_2,y_2) \enspace .% \wedge \forall R, N (i_R) y_1 \leq N(i_R) y_2 \\
    \end{align*}

\end{proof}

\noindent
\emph{--- Representing $\Sigma$ in $\init{\Sigma}_A$:}

\noindent
Any arity $s \in \Sigma$ should be represented by the module morphism $s^{\init{\Sigma}}$, i.e.\ by the representation of $s$ in $\init{\Sigma}$. 
We have to show that those representations are compatible with the preorder $\leq_A$.
Given $x\leq_A y$ in $\dom(s,\init{\Sigma})(X)$, we show (omitting the argument $X$ in $s^{\init{\Sigma}}(X)(x)$)
 \[ s^{\init{\Sigma}} (x) \quad \leq_A \quad s^{\init{\Sigma}}(y) \enspace. \]
By definition, we have to show that, for any representation $R$ as before,
\[ i_R (s^{\init{\Sigma}} (x)) \quad \leq_R \quad i_R (s^{\init{\Sigma}}(y)) \enspace. \]
Since $i_R$ is a morphism of representations, it commutes with the representational module morphisms ---
the corresponding diagram is similar to the diagram of \autoref{def:prop_mor_of_reps}.
By rewriting with this equality we obtain the goal
\[ s^R \Bigl(\bigl(\dom(s)(i_R)\bigr)(x)\Bigr) \quad \leq_R \quad s^R\Bigl(\bigl(\dom(s) (i_R)\bigr)(y)\Bigr) \enspace. \]
This goal is proved by instantiating \autoref{lem:useful_lemma} with the classic $\Sigma$--module $\dom(s)$ 
(cf.\ \autoref{def:arity_classic_module_untyped}) and the fact that $s^R$ is monotone.
We hence have established a representation --- which we call $\init{\Sigma}_A$ --- of $\Sigma$ in the monad $\init{\Sigma}_A$.

\noindent
\emph{--- $\init{\Sigma}_A$ satisfies $A$:}

\noindent
The next step is to show that the representation $\init{\Sigma}_A$ satisfies $A$. 
Given an inequation 
     \[\alpha \leq \gamma : U \to V\] 
of $A$ with a classic $\Sigma$--module $V$, 
   we must show that for any set $X$ and any $x\in U(\init{\Sigma}_A)(X)$ in the domain of $\alpha$ we have 
       \begin{equation} \alpha^{\init{\Sigma}_A}_X(x) \quad \leq_A \quad \gamma^{\init{\Sigma}_A}_X(x) \enspace . \label{eq:sigma_a} \end{equation}
In the following we omit the subscript $X$. By \autoref{lem:useful_lemma} the goal is equivalent to

\begin{equation}
 \forall R : \Rep^{\Delta}(\Sigma,A), \quad V(i_R) (\alpha^{\init{\Sigma}_A}(x)) \quad \leq_{V^R_X} \quad V(i_R) (\gamma^{\init{\Sigma}_A} (x)) \enspace . \label{eq:sigma_a_alt}
\end{equation}
Let $R$ be a representation of $(\Sigma,A)$. We continue by proving \autoref{eq:sigma_a_alt} for $R$.
   By \autoref{rem:half_equation_comm} and the fact that $i_R$ is also the carrier of a 
   morphism of representations of $\Sigma$ from $\Delta\init{\Sigma}$ to $R$ (cf.\ \autoref{lem:adjunction_reps}) we can rewrite the goal as
       \[ \alpha^R\bigl(U(i_R)(x)\bigr) \quad \leq_{V^R_X} \quad \gamma^R \bigr(U(i_R)(x)\bigr) \enspace , \]
  which is true since $R$ satisfies $A$.

\noindent
\emph{--- Initiality of $\init{\Sigma}_A$:}

\noindent
Given any representation $R$ of $(\Sigma,A)$, the morphism $i_R$ is monotone with respect to the preorders on $\init{\Sigma}_A$ and $R$ by construction of $\leq_A$.
It is hence a morphism of representations from $\init{\Sigma}_A$ to $R$.
Uniqueness of the morphisms $i_R$ follows from its uniqueness in the category of representations of $\Sigma$, i.e.\ without inequations.
Hence $\init{\Sigma}_A$ is the initial object in the category of representations of $(\Sigma,A)$.
\end{proof}

\begin{remark}
 Note that the proof of the main theorem uses the equivalence proved in \autoref{lem:useful_lemma} 
 in \emph{both} directions.
 The implication from left to right %Note that the property proved in Lem.\ \ref{lem:useful_lemma} 
  would be ensured automatically if we had
 defined $\Sigma$--modules to be functors into the category $\LRMod{}{\Delta}{\PO}$ instead of $\LRMod{}{\Delta}{\PS}$.
 See \autoref{rem:wOrd_instead_Ord} for an explanation why we still choose the latter category as codomain category.
\end{remark}

\begin{remark}
 Note that for a classic $\Sigma$--module $V$ we can actually prove 
  the implication from left to right of \autoref{lem:useful_lemma}  
  more generally: for \emph{any} morphism of representations $f : P \to R$ (not just an initial one as in \autoref{lem:useful_lemma})
 the module morphism $V(f) : V(P) \to V(R)$ is monotone. Again the only interesting case is where $V = V_1 \times V_2$ is a product.
  Let $X$ be a set and $x = (x_1, x_2)$ and $y = (y_1, y_2)$ in $V(P)(X)$:
 \begin{align*}
      (x_1, x_2) \leq_{V^P} (y_1,y_2)     &\Leftrightarrow x_1 \leq_{V_1^P} y_1 \wedge x_2 \leq_{V_2^P} y_2 \\
                         {}         &\Rightarrow  V_1(f)(x_1) \leq_{V^R_1} V_1(f)(y_1) \wedge V_2(f)(x_2) \leq_{V^R_2} V_2(f)(y_2) \\
                         {}         &\Leftrightarrow  \left(V_1(f)(x_1),V_2(f)(x_2)\right) \leq_{V^R} \left( V_1(f)(y_1), V_2(f)(y_2)\right) \\
		    {}         &\Leftrightarrow V(f)(x_1, x_2) \leq_{V^R} V(f)(y_1,y_2) \enspace .
 \end{align*}

\end{remark}

\begin{example}[\autoref{ex:sig_ulc_prop} continued]
  \label{ex:2--sig_ulcbeta}
 The only inequation \autoref{eq:beta_ineq} of the signature $(\Sigma_{\ULC},\beta)$ is classic. The initial representation of $(\Sigma_{\ULC},\beta)$
 is given by the monad $\ULCB$ together with the $\ULCB$--module morphisms $\Abs$ and $\App$ (cf.\ \autoref{ex:ulcb_constructor_mod_mor})
  as representation structure. 
\end{example}

We conclude this section with some remarks about ``generating inequalities'', (regular) monads and \emph{fully} faithful morphisms:

\begin{remark}[about ``Generating'' Inequations]\label{rem:gen_ineqs}
 Given a 2--signature $(\Sigma,A)$ and a representation $R$ of $\Sigma$,
  the representation morphism of modules $s^R$ of any $s \in \Sigma$ of $R$
     is monotone. % (cf.\ Remark \ref{rem:monotone}).
For the initial representation of $(\Sigma,A)$ this means that any relation between terms of $\Sigma$ which comes from $A$
 is automatically propagated into subterms.
Similarly, the relation on those terms is by construction reflexive and transitive, since we consider representations 
in monads with codomain $\PO$.

For the example of $\ULCB$ this means that in order to obtain a complete reduction relation, 
it is sufficient to enforce only one rule by an inequation, which is
   \[ (\lambda M) N\leq M[*:=N] \enspace . \]
\end{remark}

\begin{remark}[about Finite Contexts]\label{rem:finite_contexts}
Altenkirch et al.\ \cite{DBLP:conf/fossacs/AltenkirchCU10} characterize the untyped lambda calculus as
a relative monad on the inclusion functor $i:\Fin\to\Set$ from finite sets to sets.
 An anonymous referee suggested combining our viewpoint --- syntax as monad over $\Delta:\Set\to\PO$ ---  
with Altenkirch et al.'s %\citeauthor{DBLP:conf/fossacs/AltenkirchCU10}'s: 
 one might consider the 
 lambda calculus as a relative monad on the composition $\comp{i}{\Delta}:\Fin\to\PO$, and, more generally,
 one might consider representations of a signature $(\Sigma,A)$ over monads on $\comp{i}{\Delta}:\Fin\to\PO$.
 The above theorem remains true when replacing monads on $\Delta$ by monads on $\comp{i}{\Delta}$
 everywhere.
 An equivalence between the theorem thus obtained and our \autoref{thm:init_w_ineq_untyped} might be established in a
  way similar to what Zsid\'o \cite{ju_phd} does in her PhD thesis:
 she shows, by means of adjunctions between the respective categories of models, the equivalence between the approach
 of Fiore et al.\ \cite{fpt} --- based on monoids over finite contexts --- and the approach of Hirschowitz and Maggesi 
 \cite{DBLP:conf/wollic/HirschowitzM07},
  where models are built from monads on the category $\Set$, i.e.\ over arbitrary contexts.
 
\end{remark}

\begin{remark}[about Monads on $\PO$]\label{rem:endo_ord}
 As mentioned in \autoref{sec:rel_work}, Ghani and L\"uth \cite{DBLP:journals/njc/GhaniL03} and 
 Hirschowitz and Maggesi \cite{DBLP:journals/iandc/HirschowitzM10}
 suggest the use of monads over the category $\PO$ of preordered sets for modelling syntax with a rewriting relation.
 Indeed, representations of a signature $(\Sigma,A)$ could be analogously defined for such monads.
The above construction of the initial representation of $(\Sigma,A)$ 
  carries over to representations in such monads,
 thus yielding an initiality result in which syntax is modelled as monad on $\PO$.
 It might be interesting to establish a precise connection --- e.g., in form of adjunctions ---
  between the resulting categories of representations in monads on $\PO$ and representations
  in relative monads on $\Delta$.
\end{remark}

\begin{remark}[about \emph{Fully} Faithful Translations]
  By construction any morphism $f : P \to Q$ of representations of a 2--signature $(\Sigma,A)$ is \emph{faithful},
  i.e.\ it sends related terms $x \leadsto y$ in $P(X)$ to related terms $f_X(x) \leadsto f_X(y)$ in $Q(X)$.
  It is natural to ask whether $f$ is also \emph{full}, that is, whether each $f_X : P(X) \to Q(X)$ is
  a full functor between the preorders $P(X)$ and $Q(X)$, considered as functors.
  Explicitly, this means to ask whether for any $x, y\in P(X)$ such that $f_X(x) \leadsto f_X(y)$ in $Q(X)$
  we have $x \leadsto y$.
\end{remark}

%% file: comp_sem_monads.tex
This chapter aims to combine the contents of \autorefs{sec:compilation} and \ref{sec:prop_arities} in order
to obtain an initiality result for simple type systems with reductions on the term level.
This result thus accounts for our example from \autoref{sec:trans_pcf_ulc}: the translation 
from $\PCF$ with its usual reduction relation to the untyped lambda calculus with beta reduction.
The goal thus is to define a notion of 
\emph{signature} and suitable
 \emph{representations} for such signatures, such that the types and terms generated by the signature, 
              equipped with reductions according to the inequations specified by the signature, 
              form the
 \emph{initial representation}.
Analogously to the previous chapter, we define a notion of \emph{2--signature} with two levels:
a \emph{syntactic} level specifying types and terms of a language, and, on top of that, a \emph{semantic} level
specifying reduction rules on the terms.

\section{1--Signatures}

From the \emph{syntactic} point of view presented in \autoref{sec:term_arities_syntactic}, 
1--signatures for types and terms are the same as in \autoref{sec:compilation}, \autoref{def:typed_signature}.
We have to adapt the \emph{semantic} definition of signatures for terms, however, since we now work 
with relative monads on $\TDelta{T}$ for some set $T$ instead of monads over families of sets.
The following definition is the analogue of \autoref{def:s-mon}, adapted to the use of relative monads:
\begin{definition}[Relative $S$--Monad] \label{def:s-rmon}
  Given an algebraic signature $S$, the \emph{category $\SigRMon{S}$ of relative $S$--monads} is defined as the category whose objects are pairs $(T,P)$ of
  a representation $T$ of $S$ and a relative monad 
\[P : \TS{T}\stackrel{\TDelta{T}}{\longrightarrow} \TP{T} \enspace .\]
  A morphism from $(T,P)$ to $(T', P')$ is a pair $(g, f)$ of a morphism of $S$--representations $g : T\to T'$ and a 
    morphism of relative monads $f : P\to P'$ over the retyping functor $\retyping{g}$ as in \autoref{rem:rel_mon_mor_case}.

   Given $n\in \mathbb{N}$, we write $\SigRMon{S}_n$ for the category whose objects are pairs $(T,P)$ of a representation $T$ of $S$ and 
  a relative monad $P$ over $\TDelta{T}_n$. A morphism from $(T,P)$ to $(T', P')$ is a pair $(g, f)$ of a morphism of 
      $S$--representations $g : T\to T'$ and a 
   monad morphism $f : P\to P'$ over the retyping functor $\retyping{g}(n)$ defined in \autoref{def:retyping_functor_pointed}.

\end{definition}

Similarly, we have a large category of modules over relative monads:

\begin{definition}[Large Category $\LRMod{n}{S}{\D}$ of Modules]
  \label{def:lrmod_typed}
  Given a natural number $n\in \mathbb{N}$, an algebraic signature $S$ and a category $\D$, 
  we call $\LRMod{n}{S}{\D}$ %the colax comma category $I_{S,n} \downarrow (\D, \Id)$.
  the category an object of which is a pair $(P,M)$ of a relative $S$--monad $P \in \SigRMon{S}_n$ and a $P$--module with codomain $\D$.
  A morphism to another such $(Q,N)$ is a pair $(f, h)$ of a morphism of relative $S$--monads 
$f : P \to Q$ in $\SigRMon{S}_n$ and a morphism of relative modules $h : M \to f^*N$.
\end{definition}

As before, we sometimes just write the module --- i.e.\ the second --- component of an object or morphism
 of the large category of modules.
Given $M\in \LRMod{n}{S}{\D}$, we thus write $M(V)$ or $M_V$ for the value of the module on the object $V$.

A \emph{half--arity over $S$ of degree $n$} is a functor from relative $S$--monads to the category of large modules of degree $n$:

\begin{definition}[Half--Arity over $S$ (of degree $n$)]
  \label{def:half_arity_degree_semantic_typed}
 Given an algebraic signature $S$ and $n\in \mathbb{N}$, we call \emph{half--arity over $S$ of degree $n$} a functor
  \[ \alpha : \SigRMon{S} \to \LRMod{n}{S}{\PO} \enspace . \]
 which is pre--inverse to the forgetful functor.
\end{definition}

As before we restrict ourselves to a class of such functors. Again, we start with the \emph{tautological} module:

\begin{definition}[Tautological Module of Degree $n$]
  Given $n\in \mathbb{N}$, any relative monad $R$ over $\TDelta{T}$ induces a monad $R_n$ over $\TDelta{T}_n$ 
%   \[
%    \begin{xy}
%  \xymatrix{    n \ar[r]^{k}     & T \ar[d]^{V} \\
%     {} & \Set
%   }
%    \end{xy}
%    \qquad\mapsto \qquad
%     \begin{xy}
%      \xymatrix{
%         n \ar[r]^{k} & T\ar[d]^{V} \ar[r]^{\id} & T\ar[d]^{RV} \\
%        {} & \Set \ar[r]_{\eta_R}& \Set,
%   }
%     \end{xy}
%   \]
 with object map $(V, t_1,\ldots, t_n) \mapsto (RV, t_1,\ldots,t_n)$.
To any relative $S$--monad $R$ we associate
  the tautological module of $R_n$, 
  \[\Theta_n(R):= (R_n,R_n) \in \LRMod{n}{S}{\TP{T}_n} \enspace . \]
\end{definition}

\noindent
Furthermore, we again use \emph{canonical natural transformations} (cf.\ \autoref{def:canonical_nat_trans}) to build \emph{classic} half--arities; 
these transformations specify context extension (derivation) and
selection of specific object types (fibre):

\begin{definition}[Classic Half--Arity]
  As with monads (cf.\ \autoref{sec:ext_zsido}), we restrict our attention to \emph{classic} half--arities, which we define 
  analogously to \autoref{def:alg_half_ar} as constructed using derivations and products,
  starting from the fibres of the tautological module and the constant singleton module.
  We omit the precise statement of this definition.
\end{definition}

A half--arity of degree $n$ thus associates, to any relative $S$--monad $P$ over a set of types $T$, 
a \emph{family of $P$--modules} indexed by $T^n$:

\begin{remark}[Module of Higher Degree corresponds to a Family of Modules (II)]
  \label{rem:family_of_mods_cong_pointed_mod_relative}
   \Autoref{rem:family_of_mods_cong_pointed_mod} applies analogously to modules over relative modules.
    More precisely, let $T$ be a set and let $R$ be a monad on the functor $\family{\Delta}{T}$.
    Then a module $M$ over the monad $R_n$ corresponds precisely to a family of $R$--modules 
    $(M_{\vectorletter{t}})_{\vectorletter{t}\in T^n}$ by (un)currying.
  Similarly, a morphism $\alpha:M\to N$ of modules of degree $n$ is equivalent to a family 
   $(\alpha_{\vectorletter{t}})_{\vectorletter{t}\in T^n}$ of morphisms of modules of degree zero with
     $\alpha_{\vectorletter{t}}:M_{\vectorletter{t}}\to N_{\vectorletter{t}}$.
     
\end{remark}

An arity of degree $n\in \mathbb{N}$ for terms over an algebraic signature $S$ is defined to be a pair of functors
from relative $S$--monads to modules in $\LRMod{n}{S}{\PO}$.
The degree $n$ corresponds to the number of object type indices of its associated constructor.
As an example, the arities of $\Abs$ and $\App$ of \autoref{ex:slc_def} are of degree $2$.

\begin{definition}[Term--Arity, Signature over $S$]
  A \emph{classic arity $\alpha$ over $S$ of degree $n$} is a pair 
 \[ s = \bigl(\dom(\alpha), \cod(\alpha)\bigr) \] 
  of half--arities over $S$ of degree $n$ such that 
  \begin{packitem}
   \item $\dom(\alpha)$ is classic and
   \item $\cod(\alpha)$ is of the form $\fibre{\Theta_n}{\tau}$ for some canonical natural 
        transformation $\tau$ as in \autoref{def:canonical_nat_trans}.
  \end{packitem}
Any classic arity is thus \emph{syntactically} of the form given in \autoref{eq:classic_arity_typed}.
 Note, however, that the definition of $\Theta$ in \autoref{sec:ext_zsido} differs from the one
used in the present chapter.
We write $\dom(\alpha) \to \cod(\alpha)$ for the arity $\alpha$, and $\dom(\alpha, R) := \dom(\alpha)(R)$
and similar for the codomain and morphisms of relative $S$--monads. 
Given a weighted set $(J,d)$ as in \autoref{def:weighted_set},
a term--signature $\Sigma$ over $S$ indexed by $(J,d)$ is a $J$-family 
 $\Sigma$ of classic arities over $S$, the arity $\Sigma(j)$ being of degree $d(j)$ for any $j\in J$.
\end{definition}

\begin{definition}[Typed Signature]
  A \emph{typed signature} is a pair $(S,\Sigma)$ consisting of an algebraic signature $S$ for sorts and 
  a term--signature $\Sigma$ (indexed by some weighted set) over $S$.
 \end{definition}

\begin{example}
 \Autorefs{ex:tlc_sig_higher_order} and \ref{ex:term_sig_pcf} still apply. Note, however, that
 the underlying definition of $\Theta$ differs from that of Sec.\ \ref{sec:compilation}, and that
 fibre and derivation are adapted accordingly.
\end{example}

\section{Representations of 1--Signatures}\label{sec:rep_of_1--sigs_typed}

\begin{definition}[Representation of an Arity, a Signature over $S$]
   \label{def:1--rep_typed}
   A representation of an arity $\alpha$ over $S$ in an $S$--monad $R$ is a morphism of relative modules 
    \[ \dom(\alpha,R) \to \cod(\alpha, R) \enspace . \]
  A representation $R$ of a signature over $S$ is a given by a relative $S$--monad --- called $R$ as well ---  
  and a representation $\alpha^R$ of each arity $\alpha$ of $S$ in $R$.
\end{definition}

Representations of $(S,\Sigma)$ are the objects of a category $\Rep^\Delta(S,\Sigma)$, whose morphisms are defined as follows:
\begin{definition}[Morphism of Representations]
  \label{def:rel_mor_of_reps_typed}
  Given representations $P$ and $R$ of a typed signature $(S,\Sigma)$, a morphism of representations 
  $f : P\to R$ is given by a morphism of relative $S$--monads $f : P \to R$, such that for any arity $\alpha$ of $\Sigma$
  the following diagram of module morphisms commutes:
  \[
  \begin{xy}
   \xymatrix{
        **[l]\dom(\alpha,P) \ar[d]_{\dom(\alpha, f)} \ar[r]^{\alpha^P} & **[r]\cod(\alpha,P) \ar[d]^{\cod(\alpha,f)} \\
        **[l]\dom(\alpha,R) \ar[r]_{\alpha^R} & **[r]\cod(\alpha,R) .
    }
  \end{xy}
 \]
\end{definition}

\begin{lemma}\label{lem:init_no_eqs_typed}
  For any typed signature $(S,\Sigma)$, the category of representations of $(S,\Sigma)$ has an initial object.
\end{lemma}
\begin{proof}
  The initial object is obtained, analogously to the untyped case (cf.\ \autorefs{lem:adj_mon_rmon}, \ref{lem:adjunction_reps}, \ref{lem:init_no_eqs_untyped}), 
 via an adjunction $\Delta_* \dashv U_*$ between the categories of representations of $(S,\Sigma)$
 in relative monads and those in monads as in \autoref{sec:compilation}.
  
  In more detail, to any relative $S$--monad $(T,P) \in \SigRMon{S}$ we associate the $S$--monad
  $U (T,P) := (T,UP)$ where $U_{*}P$ is the monad obtained by postcomposing with the forgetful functor $\family{U}{T} : \TP{T} \to \TS{T}$.
  Substitution for $U_{*}P$ is defined, in each fibre, as in \autoref{lem:rmon_delta_endomon}.
  For any arity $s\in \Sigma$ we have that 
       \[U_{*} \dom(s,P) \cong \dom(s,U_{*}P) \enspace , \] 
  and similar for the codomain.
  The postcomposed representation morphism $U_*s(P)$ hence represents $s$ in $U_{*}P$ in the sense of \autoref{sec:compilation}.
  This defines the functor $U_* : \Rep^\Delta(S,\Sigma) \to \Rep(S,\Sigma)$.
  Conversely, to any $S$--monad we can associate a relative $S$--monad by postcomposing with $\family{\Delta}{T} : \TS{T} \to \TP{T}$,
  analogous to the untyped case in \autoref{lem:rep_endo_rep_rel}, yielding $\Delta_* : \Rep(S,\Sigma) \to \Rep^\Delta(S,\Sigma)$.
  In summary, the natural isomorphism
    \[ \varphi_{R,P}:\bigl(\Rep^\Delta(S,\Sigma)\bigr)(\Delta_* R, P) \cong \bigl(\Rep(S,\Sigma)\bigr)(R, U_*P) \]
  is given by postcomposition with the forgetful functor (from left to right) resp.\ the functor $\Delta$ (from right to left).

\end{proof}

\section{Inequations}

Analogously to the untyped case (cf.\ \autorefs{def:half_eq_untyped}, \ref{def:ineq_untyped}), an inequation associates, 
to any representation of $(S,\Sigma)$ in a relative monad $P$, two parallel morphisms of $P$--modules.
However, similarly to arities, an inequation may now be, more precisely, a \emph{family of inequations}, indexed by object types.
Consider the simply--typed lambda calculus, which was defined with \emph{typed} abstraction and application.
Similarly, we have a \emph{typed substitution} operation for $\TLC$, which substitutes a term of type $s\in \TLCTYPE$ for a free variable of type $s$
in a term of type $t\in \TLCTYPE$, yielding again a term of type $t$.
For $s,t\in \TLCTYPE$ and $M\in\SLC(V^{*s})_t$ and $N \in \SLC(V)_s$, beta reduction is specified by

  \[ \lambda_{s,t} M(N) \leadsto M [* := N] \enspace , \]
where our notation hides the fact that not only abstraction, but also application and substitution are typed operations.
More formally, such a reduction rule might read as a family of inequations between morphisms of modules
\[  \comp{(\abs_{s,t} \times\id)}{\app_{s,t}} \enspace \leq \enspace \_ [*^s :=_t \_ ] \enspace , \] 
where $s,t\in \TLCTYPE$ range over types of the simply--typed lambda calculus.
Analogously to \autoref{sec:ext_zsido}, we want to specify the beta rule without referring to the set $\TLCTYPE$,
but instead express it for an arbitrary representation $R$ of the typed signature $(S_{\TLC},\Sigma_{\TLC})$ 
(cf.\ \autorefs{ex:type_sig_SLC}, \ref{ex:tlc_sig_higher_order}),
as in

\[  \comp{(\abs^R \times \id)}{\app^R} \enspace \leq \enspace \_ [* := \_ ] \enspace , \] 
where both the left and the right side of the inequation are given by suitable $R$--module morphisms of degree 2.
Source and target of a \emph{half--equation} accordingly are given by functors from representations of a typed signature $(S,\Sigma)$
to a suitable category of modules.
A half--equation then is a natural transformation between its source and target functor:

\begin{definition}[Category of Half--Equations] \label{def:half_eq_typed}
Let $(S,\Sigma)$ be a signature. An \emph{$(S,\Sigma)$--module} $U$ of degree $n\in \mathbb{N}$
is a functor from the category of representations of $(S,\Sigma)$ as defined in \autoref{sec:rep_of_1--sigs_typed} 
to the category $\LRMod{n}{S}{\PS}$ (cf.\ \autoref{def:lrmod_typed})
commuting with the forgetful functor to the category of relative monads.
We define a morphism of $(S,\Sigma)$--modules to be a natural transformation which
becomes the identity when composed with the forgetful functor. 
We call these morphisms \emph{half--equations} (of degree $n$).
We write $U^R := U(R)$ for the image of the representation $R$ under the $S$--module $U$, and similar for 
morphisms.
\end{definition}

\begin{definition}[Substitution as Half--Equation]\label{def:subst_half_eq_typed}
  Given a relative monad on $\family{\Delta}{T}$, its associated substitution--of--one--variable operation (cf.\ \autoref{def:hat_P_subst_typed}) 
  yields a family of module morphisms, indexed by pairs $(s,t)\in T$.
  By \autoref{rem:family_of_mods_cong_pointed_mod_relative} this family is equivalent to a module morphism of degree 2.
 The assignment
 \[\subst  : R\mapsto \subst^R :  \fibre{\hat{R}_2}{2}^1 \times \fibre{\hat{R}_2}{1} \to \fibre{\hat{R}_2}{2} \]
 thus yields a half--equation of degree $2$ over any signature $S$.
 Its domain and codomain are classic.
\end{definition} 

\begin{example}[\autoref{ex:tlc_sig_higher_order} continued]\label{ex:app_circ_half_typed}
  The map
 \[ \comp{(\abs\times\id)}{\app} : R \mapsto \comp{(\abs^R \times \id^R)}{\app^R}  : \fibre{\hat{R}_2}{2}^1 \times \fibre{\hat{R}_2}{1} \to \fibre{\hat{R}_2}{2} \]
 is a half--equation over the signature $\SLC$, as well as over the signature of \PCF.
\end{example}

\begin{definition}
 \label{def:arity_classic_module_typed}
 Any classic arity of degree $n$,
  \[s = \fibre{\Theta_n}{\sigma_1}^{{{\tau_1}}} \times \ldots\times \fibre{\Theta_n}{\sigma_m}^{{\tau_m}} \to \fibre{\Theta_n}{\sigma} \enspace , \] 
      defines a classic $S$--module 
    \[\dom(s) : R\mapsto \fibre{R_n}{\sigma_1}^{\tau_1} \times \ldots \times \fibre{R_n}{\sigma_m}^{\tau_m}  \enspace . \] 
\end{definition}

\begin{definition}[Inequation] \label{def:ineq_typed}
 Given a signature $(S,\Sigma)$, an \emph{inequation over $(S,\Sigma)$}, or \emph{$(S,\Sigma)$--inequation}, of degree $n\in \NN$ is a pair of 
parallel half--equations between $(S,\Sigma)$--modules of degree $n$.
We write $\alpha \leq \gamma$ for the inequation $(\alpha, \gamma)$.
We leave the degree implicit whenever possible, analogously to \autoref{rem:degree_implicit}.
\end{definition}

\begin{example}[Beta Reduction]
  For any suitable 1--signature --- i.e.\ for any 1--signature that has an arity for abstraction and an arity for application ---
  we specify beta reduction through an inequation of degree $2$ using the parallel half--equations of \autoref{def:subst_half_eq_typed} and \autoref{ex:app_circ_half_typed}:
  \[   \comp{(\abs \times \id)}{\app} \leq \subst : \fibre{\Theta}{2}^1 \times \fibre{\Theta}{1} \to \fibre{\Theta}{2} \enspace . \]
\end{example}

\begin{example}[Fixpoints and Arithmetics of \PCF]\label{ex:pcf_ineqs}
 The reduction rules for $\PCF$ are informally given in \autoref{eq:pcf_reductions}.
 We specify these reduction rules as inequations over the 1--signature of $\PCF$ (cf.\ \autoref{ex:term_sig_pcf})
 as follows: 
   \begin{align*}
        \comp{(\abs \times \id)}{\app} &\leq \subst : \fibre{\Theta}{2}^1 \times \fibre{\Theta}{1} \to \fibre{\Theta}{2}\\
      \PCFFix &\leq \comp{(\id,\PCFFix)}{\app} : \fibre{\Theta}{1\PCFar 1} \to \fibre{\Theta}{1} \\
      \comp{(\PCFSucc,\PCFn{n})}{\app} &\leq \PCFn{n+1} : * \to \fibre{\Theta}{\Nat}\\
      \comp{(\PCFPred,\PCFn{0})}{\app} &\leq \PCFn{0} : * \to \fibre{\Theta}{\Nat}\\
      \comp{\left(\PCFPred,\comp{({\PCFSucc},{\PCFn{n}})}{\app}\right)}{\app} &\leq \PCFn{n} : * \to \fibre{\Theta}{\Nat}\\
      \comp{(\PCFZerotest,\PCFn{0})}{\app} &\leq \PCFTrue  : * \to \fibre{\Theta}{\Bool}\\
      \comp{\left(\PCFZerotest,\comp{(\PCFSucc,\PCFn{n})}{\app}\right)}{\app} &\leq \PCFFalse  : * \to \fibre{\Theta}{\Bool}\\
                             &\vdots
   \end{align*}

\end{example}

\begin{definition}[Representation of Inequations]\label{def:rep_ineq_typed}
  \label{def:2--rep_typed}
 A \emph{representation of an $(S,\Sigma)$--inequation $\alpha\leq \gamma : U \to V$} (of degree $n$) is any representation 
  $R$ over a set of types $T$ of $(S,\Sigma)$ such that 
  $\alpha^R \leq \gamma^R$ pointwise, i.e.\ if for any pointed context $(X,\vectorletter{t}) \in \TS{T}\times T^n$, 
      any $t\in T$ and any $y\in U^R_{(X, \vectorletter{t})}(t)$, 
    \begin{equation}\alpha^R(y) \enspace \leq \enspace \gamma^R(y) \enspace , \label{eq:comp_sem_rep_ineq}\end{equation}
where we omit the sort argument $t$ as well as the context $(X,\vectorletter{t})$ from $\alpha$ and $\gamma$.
We say that such a representation $R$ \emph{satisfies} the inequation $\alpha \leq \gamma$.

For a set $A$ of $(S,\Sigma)$--inequations, we call \emph{representation of $((S,\Sigma),A)$} any representation of $(S,\Sigma)$ that
satisfies each inequation of $A$.
We define the category of representations of the 2--signature $((S,\Sigma), A)$ to be the full subcategory of the category of
representations of $S$ whose objects are representations of $((S,\Sigma), A)$.
We also write $(\Sigma,A)$ for $((S,\Sigma), A)$.

According to \autoref{rem:family_of_mods_cong_pointed_mod_relative}, 
 the inequation of \autoref{eq:comp_sem_rep_ineq} is equivalent to ask whether, for any $\vectorletter{t} \in T^n$, 
    any $t\in T$ and any $y\in U_{\vectorletter{t}}^R(X)(t)$, 
\begin{equation*}\alpha_{\vectorletter{t}}^R(y) \enspace \leq \enspace \gamma_{\vectorletter{t}}^R(y) \enspace .
%      \label{eq:comp_sem_rep_ineq_indexed}
\end{equation*}

\end{definition}

\section{Initiality for 2--Signatures}

We are ready to state and prove an initiality result for typed signatures with inequations:

\begin{theorem}\label{thm:init_w_ineq_typed}
 For any set of classic $(S,\Sigma)$--inequations $A$, the category of representations of $((S,\Sigma),A)$ has an initial object.
\end{theorem}

\begin{proof}
 
 The proof is analogous to that of the untyped case (c.f.\ \autoref{thm:init_w_ineq_untyped}).
 The fact that we now consider \emph{typed} syntax introduces a minor complication, on the presentation of which 
  we put the emphasis during the proof.
The basic ingredients for building the initial representation 
  are given by the initial representation 
$(\init{S},\init{\Sigma})$
 --- or just $\init{\Sigma}$ for short --- in the category $\Rep(S,\Sigma)$ of representations in monads on set families (cf.\ \autoref{thm:compilation}).
 Equivalently, the ingredients come from the initial object 
$(\init{S},\Delta_*\init{\Sigma})$  --- or just $\Delta_*\init{\Sigma}$ for short --- of representations
 \emph{without inequations} in the category $\Rep^{\Delta}(S,\Sigma)$ (cf.\ \autoref{lem:init_no_eqs_typed}).
    We call $\init{\Sigma}$ resp.\ $\Delta_*\init{\Sigma}$  the monad resp.\ relative monad underlying the initial representation

  The proof consists of 3 steps: at first, we define a preorder $\leq_A$ on the terms of $\init{\Sigma}$, induced by the set $A$ of inequations.
   Afterwards we show that the data of the representation $\init{\Sigma}$ --- substitution, representation morphisms etc. --- 
  is compatible with the preorder $\leq_A$ in a suitable sense. This will yield a representation $\init{\Sigma}_A$ of $(\Sigma,A)$.
  Finally we show that $\init{\Sigma}_A$ is the initial such representation.

\noindent
\emph{--- The monad underlying the initial representation:}

\noindent
    For any context $X\in \TS{\init{S}}$ and $t\in \init{S}$, we equip $\init{\Sigma}X(t)$ with a preorder $A$ by setting --- \emph{morally}, cf.\ below ---, 
      for $x,y\in \init{\Sigma}X(t)$,

    \begin{equation} x \leq_A y \quad :\Leftrightarrow \quad \forall R : \Rep(\Sigma,A), \quad%\text{ R satisfies } A \Rightarrow
                      i_R (x) \leq_R i_R (y) \enspace ,
      \label{eq:order_typed}
    \end{equation}
where $i_R : \Delta_*\init{\Sigma} \to R$ is the initial morphism of representations of $(S,\Sigma)$, 
 cf.\ \autoref{lem:init_no_eqs_typed}. 
Note that the above definition in \autoref{eq:order_typed} is ill--typed:
 we have $x\in \init{\Sigma}X(t)$, which cannot be applied to (a fibre of) $i_R(X) : \retyping{g}(\init{\Sigma}X) \to R(\retyping{g}X)$.
We denote by $\varphi = \varphi_R$ the natural isomorphism induced by the adjunction of \autoref{def:retyping_functor}
and \autoref{rem:retyping_adjunction_kan} obtained by retyping --- along the initial morphism of types $g:\init{S}\to T = T_R$ --- 
towards the set $T$ of ``types'' of $R$,
\[ \varphi_{X,Y} : \family{\PO}{T}\left(\retyping{g}(\init{\Sigma}X), R(\retyping{g}X)\right) \cong 
                     \family{\PO}{\init{S}}\left(\init{\Sigma}X, \comp{g}{R(\retyping{g}X)}\right) \enspace . \]
Instead of the above definition in \autoref{eq:order_typed}, we should really write 
    \begin{equation} x \leq_A y \quad :\Leftrightarrow \quad \forall R : \Rep(\Sigma,A), \quad%\text{ R satisfies } A \Rightarrow
                      \left(\varphi(i_{R,X})\right) (x) \leq_R \left(\varphi (i_{R,X})\right) (y) \enspace ,
      \label{eq:order_typed_corrected}
    \end{equation}
where we omit the subscript ``$R$'' from $\varphi$.
We have to show that the map \[ X\mapsto \init{\Sigma}_A X := (\init{\Sigma} X, \leq_A) \] yields a relative monad on $\TDelta{\init{S}}$. 
The missing fact to prove is that the substitution with a morphism 
\[ f\in\TP{\init{S}}(\Delta X, \init{\Sigma}_A Y) \cong \TS{\init{S}}(X,\init{\Sigma} Y)\] 
is compatible with the order $\leq_A$:
given any $f \in \TP{\init{S}}(\Delta X, \init{\Sigma}_A Y)$ we show that 
        \[\sigma^{\init{\Sigma}}(f) \in \TS{\init{S}}(\init{\Sigma} X,\init{\Sigma} Y)\] 
is monotone with respect
to $\leq_A$ and hence (the carrier of) a morphism 
\[\sigma^{\init{\Sigma}_A}(f) \in \TP{\init{S}}(\init{\Sigma}_A X, \init{\Sigma}_A Y) \enspace . \]
We overload the infix symbol $\bind{}{}$ to denote monadic substitution.
Note that this notation now hides an implicit argument giving the sort of the term in which we substitute.
Suppose $x, y \in \init{\Sigma}X(t)$ with $x\leq_A y$, we show 
        \[\bind{x}{f} \enspace \leq_A \enspace \bind{y}{f} \enspace .\] 
Using the definition of $\leq_A$, 
we must show, for a given representation $R$ of $(\Sigma,A)$, 
\begin{equation}  \left(\varphi(i_R)\right)(\bind{x}{f}) \enspace \leq_R \enspace \left(\varphi(i_R)\right)(\bind{y}{f}) \enspace .
   \label{eq:proof_of_chap5_1}
 \end{equation}
Let $g$ be the initial morphism of types towards the types of $R$.
Since $i:= i_R$ is a morphism of representations --- and thus in particular a \emph{monad} morphism, 
it is compatible with the substitution of $\init{\Sigma}$ and $R$; we have

\begin{equation}
 \begin{xy}
  \xymatrix @R=4pc @C=5pc{
    \retyping{g}(\init{\Sigma}X) \ar[r]^{\retyping{g}(\sigma(f))} \ar[d]_{i_X} & \retyping{g}(\init{\Sigma}Y) \ar[d]^{i_Y} \\
    R(\retyping{g}X) \ar[r]_{\sigma^R(\comp{\retyping{g}f}{i_Y})}& R (\retyping{g} Y).
}
 \end{xy}
\label{eq:comp_sem_monad_mon_mor_diag}
\end{equation}
By applying the isomorphism $\varphi$ on the diagram of \autoref{eq:comp_sem_monad_mon_mor_diag}, we obtain
\begin{align} \comp{\sigma(f)}{\varphi(i_Y)} &= \varphi\left(\comp{\retyping{g}(\sigma(f))}{i_Y}\right) \notag \\
                                              &= \varphi\left(\comp{i_X}{\sigma(\comp{\retyping{g}f}{i_Y})}\right) \notag \\
                                              &= \comp{\varphi(i_X)}{g^*\left(\sigma^R(\comp{\retyping{g}f}{i_Y}) \right)} \enspace . \label{eq:p_chap5_2}
\end{align}
Rewriting the equality of \autoref{eq:p_chap5_2} twice in the goal \autoref{eq:proof_of_chap5_1} yields the goal
\[ g^*\left(\sigma^R(\comp{\retyping{g}f}{i_Y}) \right)\left((\varphi(i_X))(x) \right) = 
           g^*\left(\sigma^R(\comp{\retyping{g}f}{i_Y}) \right)\left((\varphi(i_X))(y) \right) \enspace ,  \]
which is true since $g^*\left(\sigma^R(\comp{\retyping{g}f}{i_Y}) \right)$ is monotone and
  $(\varphi(i_X))(x) \leq_R (\varphi(i_X))(y)$ by hypothesis.
We hence have defined a monad $\init{\Sigma}_A$ over $\TDelta{\init{S}}$.

\begin{lemma}\label{lem:useful_lemma_typed}
 \Autoref{lem:useful_lemma} generalizes to the typed setting of this chapter. 
\end{lemma}

\begin{proof}[Proof of \autoref{lem:useful_lemma_typed}]
 The proof is analogous to the proof of \autoref{lem:useful_lemma}: we apply 
  the same reasoning in the corresponding fibre.
 
\end{proof}

\noindent
\emph{--- Representing $\Sigma$ in $\init{\Sigma}_A$:}

\noindent
Any arity $s \in \Sigma$ should be represented by the module morphism $s^{\init{\Sigma}}$, i.e.\ by the representation of $s$ in $\init{\Sigma}$. 
We have to show that those representations are compatible with the preorder $A$.
Given $x\leq_A y$ in $\dom(s,\init{\Sigma})(X)$, we show (omitting the argument $X$ in $s^{\init{\Sigma}}(X)(x)$)
 \[ s^{\init{\Sigma}} (x) \quad \leq_A \quad s^{\init{\Sigma}}(y) \enspace. \]
By definition, we have to show that, for any representation $R$ with initial morphism $i = i_R : \init{\Sigma} \to R$ as before,
\[ \varphi(i_X) (s^{\init{\Sigma}} (x)) \quad \leq_R \quad \varphi(i_X) (s^{\init{\Sigma}}(y)) \enspace. \]
But these two sides are precisely the images of $x$ and $y$ under the upper--right composition of the 
diagram of \autoref{def:rel_mor_of_reps_typed} for the morphism of representations $i_R$.
By rewriting with this diagram we obtain the goal
\[ s^R \Bigl(\bigl(\dom(s)(i_R)\bigr)(x)\Bigr) \quad \leq_R \quad s^R\Bigl(\bigl(\dom(s) (i_R)\bigr)(y)\Bigr) \enspace. \]
We know that $s^R$ is monotone, thus it is sufficient to show 
\[  \bigl(\dom(s)(i_R)\bigr)(x) \quad \leq_R \quad \bigl(\dom(s) (i_R)\bigr)(y) \enspace. \]
This goal follows from  \autoref{lem:useful_lemma_typed} 
(instantiated for the classic $S$--module $\dom(s)$, cf.\ \autoref{def:arity_classic_module_typed}) 
and the hypothesis $x \leq_A y$.
We hence have established a representation --- which we call $\init{\Sigma}_A$ --- of $S$ in $\init{\Sigma}_A$.

\noindent
\emph{--- $\init{\Sigma}_A$ satisfies $A$:}

\noindent
The next step is to show that the representation $\init{\Sigma}_A$ satisfies $A$. 
Given an inequation 
     \[\alpha \leq \gamma : U \to V\] 
of $A$ with a classic $S$--module $V$, 
   we must show that for any context $X \in \TS{\init{S}}$, any $t\in \init{S}$ and any $x\in U(\init{\Sigma}_A)(X)_t$ in the domain of $\alpha$ we have 
\begin{equation*} \alpha^{\init{\Sigma}_A}(x) \quad \leq_A \quad \gamma^{\init{\Sigma}_A}(x) \enspace , 
%   \label{eq:sigma_a_typed} 
\end{equation*}
where here and later we omit the context argument $X$ and the sort argument $t$.
By \autoref{lem:useful_lemma_typed} the goal is equivalent to

\begin{equation}
 \forall R : \Rep(\Sigma,A), \quad V(i_R) (\alpha^{\init{\Sigma}_A}(x)) \quad \leq_{V^R_X} \quad V(i_R) (\gamma^{\init{\Sigma}_A} (x)) \enspace . 
 \label{eq:sigma_a_alt_typed}
\end{equation}
Let $R$ be a representation of $(\Sigma,A)$. We continue by proving \autoref{eq:sigma_a_alt_typed} for $R$.
    \Autoref{rem:half_equation_comm} holds analogously in the typed setting of this chapter. 
     The fact that $i_R$ is the carrier of a 
   morphism of $(S,\Sigma)$--representations from $\Delta\init{\Sigma}$ to $R$  allows to rewrite the goal as
       \[ \alpha^R\bigl(U(i_R)(x)\bigr) \quad \leq_{V^R_X} \quad \gamma^R \bigr(U(i_R)(x)\bigr) \enspace , \]
  which is true since $R$ satisfies $A$.

\noindent
\emph{--- Initiality of $\init{\Sigma}_A$:}

\noindent
Given any representation $R$ of $(\Sigma,A)$, the morphism $i_R$ is monotone with respect to the orders on $\init{\Sigma}_A$ and $R$ by construction of $\leq_A$.
It is hence a morphism of representations from $\init{\Sigma}_A$ to $R$.
Uniqueness of the morphisms $i_R$ follows from its uniqueness in the category of representations of $(S,\Sigma)$, i.e.\ without inequations.
Hence $(\init{S},\init{\Sigma}_A)$ is the initial object in the category of representations of $((S,\Sigma),A)$.

\end{proof}

\begin{remark}[Iteration Principle by Initiality]\label{rem:comp_sem_iteration}
 The universal property of the language generated by a 2--signature yields an \emph{iteration principle}
  to define maps --- translations --- on this language, which are certified to be compatible with substitution and 
  reduction in the source and target languages. How does this iteration principle work? 
More precisely, what data (and proof) needs to be specified in order to define such a translation via initiality from 
a language, say, $(\init{S},\init{\Sigma}_A)$ to another language $(\init{S}',\init{\Sigma}'_{A'})$, generated by signatures $(S,\Sigma,A)$ 
and $(S',\Sigma',A')$, respectively?
The translation is a morphism --- an initial one --- in the category of representations of the signature $(S,\Sigma,A)$ 
of the source language.
It is obtained by equipping the relative monad $\init{\Sigma}'_{A'}$ underlying the target language with a representation
of the signature $(S,\Sigma,A)$. In more detail:
\begin{enumerate}
 \item we give a representation of the type signature $S$ in the set $\init{S}'$. By initiality of $\init{S}$, this yields a 
       translation $\init{S} \to \init{S}'$ of sorts.
 \item Afterwards, we specify a representation of the term signature $\Sigma$ in the monad $\init{\Sigma}'_{A'}$ by 
        defining suitable (families) of morphisms of $\init{\Sigma}'_{A'}$--modules. This yields a representation $R$
           of $(S, \Sigma)$ in the monad $\init{\Sigma}'_{A'}$.
       \newcounter{tempcounter} \setcounter{tempcounter}{\value{enumi}}
\end{enumerate}
By initiality, we obtain a morphism $f : (\init{S},\init{\Sigma}) \to R$ of representations of $(S,\Sigma)$, that is,
we obtain a translation from $(\init{S},\init{\Sigma})$ to $(\init{S}',\init{\Sigma}')$ as the colax monad morphism underlying 
the morphism $f$.
However, we have not yet ensured that the translation $f$ is compatible with the respective reduction preorders
in the source and target languages. 
\begin{enumerate}\setcounter{enumi}{\value{tempcounter}}
 \item Finally, we verify that the representation $R$ of $(S,\Sigma)$ satisfies the inequations of $A$, that is, we check
       whether, for each $\alpha \leq \gamma : U \to V \in A$, and for each context $V$, each $t\in \init{S}$ and $x \in U^R_V(t)$,
        \[    \alpha^R (x) \enspace \leq \enspace \gamma^R (x) \enspace . \]
\end{enumerate}
After verifying that $R$ satisfies the inequations of $A$, the representation $R$ is in fact a representation of $(S,\Sigma,A)$.
The initial morphism $f$ thus yields a faithful translation from $(\init{S},\init{\Sigma}_A)$ to $(\init{S}',\init{\Sigma}'_{A'})$.
\end{remark}

\begin{example}[Translation from $\PCF$ to $\LC$, \autorefs{ex:init_pcf_translation_nosem} and \ref{ex:pcf_ineqs} cont.]
Recall the translations from  $\PCF$  to the
untyped lambda calculus of \autoref{ex:init_pcf_translation_nosem}.
We might attempt to specify the same translations using the iteration operator obtained by 
\autoref{thm:init_w_ineq_typed}, where $\PCF$ is equipped with the reduction relation generated 
by the inequations of \autoref{ex:pcf_ineqs} and $\LC$ is equipped with beta reduction as in \autoref{ex:sig_ulc_prop}.
However, representing the fixedpoint operator of $\PCF$ by the lambda term $\Theta$ fails, for reasons 
explained at the end of \autoref{chap:comp_sem_formal}.

For the translation of $\PCF$ to the lambda calculus mapping the fixedpoint operator of $\PCF$ to the Turing 
fixedpoint combinator, we have formalized its specification via initiality in the proof assistant Coq \cite{coq}.
After constructing the category of representations of \PCF, we equip the untyped lambda calculus 
with a representations of \PCF, representing the arity $\mathbf{Fix}$ 
by 
the Turing operator ${\Theta}$.
The formalization is explained in \autoref{chap:comp_sem_formal}.
Note that the translation is given by a \textsf{Coq} function and hence executable.
\end{example}

%% file: about_coq.tex
\section{About the Proof Assistant \texorpdfstring{$\mathsf{Coq}$}{Coq}}

The proof assistant \textsf{Coq} \cite{coq} is an implementation of the 
\emph{Calculus of Inductive Constructions (CIC)} which itself is a constructive \emph{type theory}. 
Its objects are \emph{terms} built according to a grammar 
(see the \textsf{Coq} manual \cite{CoqManualV83} for the term forming rules). 
Each valid term has its associated \emph{type} which is itself a term and which is automatically computed by \textsf{Coq}.
In \textsf{Coq} a typing judgment is written \lstinline!t : T!, meaning that $t$ is a term of type $T$. 
Typing judgments are for example \lstinline!1 : Nat! and \lstinline!plus : Nat -> Nat -> Nat!. 
Function application is simply denoted by a blank, i.e.\ we write \lstinline!f x! for $f(x)$.

The CIC also treats propositions as types via the \emph{Curry--Howard isomorphism}, hence a proof of a proposition $P$ is in fact a term of type $P$.  
Accordingly, a proof of a proposition $A \Rightarrow B$ is a function $A \to B$, i.e.\ a term which associates a proof of $B$ to any proof of $A$. 
As an example, the function $\id:P\to P$ is a proof of the tautology $P \Rightarrow P$. 
In the proof assistant \textsf{Coq} a user hence proves a proposition \lstinline!P! by providing a term \lstinline!p! of 
type \lstinline!P!. \textsf{Coq} checks the validity of the proof \lstinline!p! by checking whether \lstinline!p : P!.

\textsf{Coq} comes with extensive support to \emph{interactively build} the proof terms of a given proposition. 
In \emph{proof mode} so-called \emph{tactics} 
help the users to reduce the proposition they want to prove --- the \emph{goal} --- into one or more simpler subgoals, 
until reaching trivial subgoals which can be solved directly.

Particular concepts of \textsf{Coq} such as records and type classes, setoids, implicit arguments and coercions are explained 
in a call--by--need fashion in the course of the thesis.
One important feature is the \lstinline!Section! mechanism (cf.\ also the \textsf{Coq} manual \cite{CoqManualV83}). 
Parameters and hypotheses declared in a section automatically get 
discharged when closing the section. Constants of the section then become functions,
depending on an argument of the type of the parameter they mentioned.
We illustrate this concept by means of a small example; consider the following \textsf{Coq} declarations:
\begin{lstlisting}
Section def_double.
Variable n : nat.
Definition double : nat := 2 * n.
Check double.
 double
      : nat
Print Assumptions double.
 Section Variables:
 n : nat
\end{lstlisting}
Inside the section \lstinline!def_double!, the constant \lstinline!double! is of type \lstinline!nat!, 
 as we verify using the \lstinline!Check! command.
 Furthermore, it depends on the section variable \lstinline!n : nat!
declared using the \lstinline!Variable! vernacular command. 
After closing the section, the constant \lstinline!double! is a closed term of function type:
\begin{lstlisting}
End def_double.
Check double.
 double
      : nat -> nat
Print Assumptions double.
 Closed under the global context
Eval compute in double 4.
      = 8 : nat
\end{lstlisting}
In our formalization, we use the \lstinline!Section! mechanism extensively.
When presenting a definition depending on section variables, 
we either give a slightly modified, fully discharged version of the statement --- compared to the actual \textsf{Coq} code ---, 
or mention the section variables informally in the text.

%% file: formalizing_alg.tex
\section{Formalizing Algebraic Structures}
\label{sec:alg_structure}

An algebraic structure typically is given by some data --- i.e. sets and operations on them ---
that satisfies given properties.
For instance, a group is given by a set, together with a binary associative 
multiplication and a unit element, such that any 
element of the set has a multiplicative inverse.
Such algebraic structures are defined in a \emph{hierarchic} way:
for instance, any  \emph{group} is a particular \emph{monoid} that admits inverses. Thus any group is a monoid. 
The other way round, given a group, if multiplication is commutative, then this group is actually abelian,
and the group is an element of the class of abelian groups.

This hierarchic structure poses a major problem in the formalization of classic mathematics, and
the question of how to formalize algebraic structures is a subject of active research. 
Put simply, the main question is how tightly one should pack together the data and properties 
of an algebraic structure.
If data and properties are packed together tightly, then operations and properties can easily be associated 
to their respective underlying sets, and this allows for overloading notation and coercions.
On the other hand, this tight packing makes it difficult to ``add'' data and properties to an instance of 
an algebraic structure, e.g., to consider a group, for which one has proved commutativity of multiplication, as an 
abelian group.
We do not attempt to propose a solution to the challenge of how to formalize algebraic structures. 
However, we need to choose from the existing solutions.
In \textsf{Coq} there are basically two possible answers: 
\emph{records}, employed e.g., by Garillot et al.~\cite{packing}, correspond to a tight packing of algebraic structure, 
whereas \emph{type classes} \cite{sozeau.Coq/classes/fctc}, as used 
by Spitters and v.~d.~Weegen \cite{DBLP:journals/mscs/SpittersW11}, correspond to a rather loose packaging.

\textsf{Coq} records are implemented as an inductive data type with one constructor,
However, use of the vernacular command \lstinline!Record! (instead of plain \lstinline!Inductive!)
allows the optional automatic definition of the projection functions to the 
constructor arguments -- the ``fields'' of the record.
Additionally, one can declare those projections as \emph{coercions}, i.e.\ they can be inserted 
automatically by \textsf{Coq}, and left out in printing.
As an example for a coercion, it allows us to write \lstinline!c : C! for an object \lstinline!c!
of a category \lstinline!C!. Here the projection from the category type to the type of objects of 
a category is declared as a coercion (cf.\ \autoref{lst:Cat}). 
This is the formal counterpart to the convention introduced in the informal definition of categories 
in \autoref{def:category}.
Another example of coercion 
is given in the definition of monad (cf.\ \autoref{def:monad_mu}), where it corresponds precisely to 
the there--mentioned \emph{abuse of notation}.
Finally, an example of coercion that is \emph{not} given by a projection is given by the 
tautological module, i.e.\ the map that associates to any monad $P$ the tautological $P$--module (cf.\ \autoref{def:taut_mod}).

Type classes are implemented as records. Similarly to the difference between records and inductive types,
type classes are distinguished from records only in that some meta--theoretic features are automatically enabled when 
declaring an algebraic structure as a class rather than a record.
For details we refer to Sozeau's article about the implementation of type classes \cite{sozeau.Coq/classes/fctc}
and Spitters and v.\ d.\ Weegen's work \cite{DBLP:journals/mscs/SpittersW11}.
Type classes differ from records in their usage, more specifically, in which data one declares as a \emph{parameter} of the
structure and which one declares as a \emph{field}.
The following example, borrowed from \cite{DBLP:journals/mscs/SpittersW11}, illustrates the different uses; we give 
two definitions of the algebraic structure of reflexive relation, one in terms of classes and one in terms of records:
\begin{lstlisting}
Class Reflexive {A : Type}{R : relation A} := 
   reflexive : forall a, R a a.

Record Reflexive := {
  carrier : Type ;
  car_rel : relation carrier ;
  rel_refl : forall a, car_rel a a }.
\end{lstlisting}

\noindent
Our main interest in classes comes from the fact that by using classes many of the arguments of projections 
are automatically declared as \emph{implicit arguments}. This leads to more readable code since arguments that can be 
deduced by \textsf{Coq} do not have to be written down. Thus it corresponds precisely to 
the mathematical practice of not mentioning arguments (e.g.\ indices) which ``are clear from the context''. An instance of 
this behaviour can be seen in the definition of category in \autoref{def:category}, where we omit 
the 3 ``object'' arguments --- written as an index --- of the dependent composition of morphisms.
In particular, the structure argument of the projection, that is, the argument specifying the instance whose field 
we want to access, is implicit and deduced automatically by \textsf{Coq}. This mechanism allows for \emph{overloading},
a prime example being the implementation of setoids (cf.\ \autoref{sec:setoids}) as a type class; 
in a term ``\lstinline!a == b!'' denoting setoidal equality, \textsf{Coq} automatically finds 
the correct setoid instance from the type of \lstinline!a! and \lstinline!b!% 
\footnote{Beware! In case several instances of setoid have been declared on one and the same \textsf{Coq} type,
   the instance chosen by \textsf{Coq} might not be the one intended by the user. 
This is the main reason for Spitters and v.\ d.\ Weegen to restrict the 
     fields of type classes to \emph{propositions}.}.

We decide to define our algebraic structures in terms of type classes first, 
and bundle the class together with some of the class parameters in a record afterwards,
as is shown in the following example for the type class \lstinline!Cat_struct! (cf.\ \autoref{lst:cat}) 
and the bundling record \lstinline!Cat!.  
\begin{form}[Bundling a type class into a record]\label{lst:Cat}
\begin{lstlisting}
Record Cat := {
     obj:> Type ;
     mor: obj -> obj -> Type ;
     cat_struct:> Cat_struct mor }.
\end{lstlisting}
\end{form}

\noindent
This duplication of \textsf{Coq} definitions is a burden rather than a feature. We still proceed like this for the following reasons:
in our case the use of records is unavoidable since we want to have a \textsf{Coq} \emph{type} of categories, 
of functors between two given categories, etc. 
This is necessary when those objects --- functors, for instance --- shall themselves be the objects or morphisms of some category,
as is clear from \autoref{lst:Cat}.
However, we profit from aforementioned features of type classes, notably automatic declaration of some arguments as \emph{implicit} and
the resulting overloading.

Apart from that, we do not employ any feature that makes the use of type classes comfortable --- 
such as maximally inserted arguments, operational classes, etc. --- 
since we usually work with the bundled versions.
Readers who are interested in how to use type classes in \textsf{Coq} properly, 
are advised to take a look at Spitters and v.\ d.\ Weegen's paper \cite{DBLP:journals/mscs/SpittersW11}.
There, the authors employ the mentioned bundling of type classes in records only when necessary, 
e.g., when the considered structures are to be the objects or morphisms of some category.

%% file: formalizing_cats.tex
\section{Formalizing Categories}
\label{sec:about_cat_formal}

As seen in \autoref{sec:categories_functors}, there are two definitions of category (\autoref{def:category}, \autoref{rem:def_cat_alt}), 
which are equivalent from 
the point of view of a mathematician.
When implementing category theory in dependent type theory, however, one needs to 
choose the one or the other definition.
This section explains how we implement categories in \textsf{Coq} and some consequences of 
our design choice.

\subsection{Which Definition to Formalize --- Dependent Hom--Sets?} \label{subsec:dep_hom_types}

The main difference concerning formalization between these two definitions is that
of \emph{composability of morphisms}.
The first definition can be implemented directly only in type theories featuring \emph{dependent types},
such as the Calculus of Inductive Constructions (CIC). The ambient type system, i.e.\ the prover,
then takes care of composability -- terms with compositions of non--composable morphisms are 
rejected as ill--typed terms.

The second definition can be implemented also in provers with a simpler type system such as the 
family of HOL theorem provers. However, since those (as well as the CIC) are theories where functions are total,
one is left with the question of how to implement composition.
Composition might then be implemented either as a functional relation or as a total function
about which nothing is known (deducible) on non--composable morphisms.
The second possibility is implemented in O'Keefe's library \cite{OKeefe}.
There the author also gives an overview of available formalizations
in different theorem provers with particular attention to the choice of the 
definition of category.

In our favourite prover \textsf{Coq}, both definitions have been employed in significant
developments: the second definition is used in Simpson's construction of the Gabriel--Zisman
localization \cite{Simpson_localization}, whereas Huet and Sa\"ibi's \textsf{ConCaT} \cite{concat}
uses type families of morphisms as in \autoref{def:category}.
To our knowledge there is no library in a prover with dependent types such as \textsf{Coq} or \textsc{NuPrl} \cite{nuprl}
 which develops and compares
both definitions with respect to provability, readability, and other criteria.

We decided to construct our library using type families of morphisms. 
In this way the proof of composability of two morphisms is done by \textsf{Coq}
type computation automatically.
As a consequence, we are sometimes obliged to insert trivial isomorphisms between equal --- but not convertible ---
objects of some category, in order to make compositions typecheck.
For an example see \autoref{subsec:sts_formal_reps}.

\textsf{Coq}'s \emph{implicit argument} mechanism allows us 
to omit the deducible arguments, as we do in \autoref{def:category}
for the ``object arguments'' $c,d$ and $e$ of the composition.
Together with the possibility to define infix notations,
this brings our formal syntax close to informal mathematical syntax.

\subsection{Setoidal Equality on Morphisms}\label{sec:setoid_eq_on_mor}

All the properties of a category $\C$ concern equality of two parallel morphisms, i.e. morphisms with same source and target.
In \textsf{Coq} there is a polymorphic equality, called \emph{Leibniz equality}, readily available for any type.
However, this equality actually denotes \emph{syntactic equality}, which already in the case of maps 
does not coincide with the ``mathematical'' equality on maps -- given by pointwise equality -- that we would rather consider.
With the use of axioms --- for the mentioned example of maps the axiom \lstinline!functional_extensionality! from 
the \textsf{Coq} standard library --- 
one can often deduce Leibniz equality from the ``mathematical equality'' in question.
But this easily gets cumbersome, in particular when the morphisms --- as will be in our case --- 
are sophisticated algebraic structures composed of a lot of data and properties.
Instead, we require any collection of morphisms $\C(c,d)$ for objects $c$ and $d$ of $\C$
to be equipped with an equivalence relation, which plays the r\^ole of 
equality on this collection. In the \textsf{Coq} standard library 
equivalence relations are implemented as a type class with the 
underlying type as a parameter \lstinline!A!, 
and the relation as well as a proof of it being an equivalence as fields:
\begin{form}[Setoid Type Class]\label{lst:setoid}
\begin{lstlisting}
Class Setoid A := {
  equiv : relation A ;
  setoid_equiv :> Equivalence equiv }.
\end{lstlisting}
\end{form}

\noindent
Setoids as morphisms of a category have been used by Aczel \cite{aczel_galois} in LEGO (there a setoid is simply called ``set'')
and Huet and Sa\"ibi (HS) \cite{concat} in \textsf{Coq}.
HS's setoids are implemented as records of which the underlying type is a component
instead of a parameter. This choice makes it necessary to duplicate 
the definitions of setoids and categories in order to make them available with a ``higher'' type
\footnote{In HS's \textsc{ConCaT}, a type \lstinline!T! that is defined after the type of setoids cannot be the carrier of a setoid
 itself. As a remedy, HS define a type \lstinline!Setoid'! isomorphic to \lstinline!Setoid!
  \emph{after} the definition of \lstinline!T!. The type of \lstinline!Setoid'! now being higher 
  than that of \lstinline!T!, one can define a term of type \lstinline!Setoid'! whose carrier is \lstinline!T!.}.

\subsection{\textsf{Coq} Setoids and Setoid Morphisms}\label{sec:setoids}

Setoids in \textsf{Coq} are implemented as a type class (cf.\ \autoref{lst:setoid})
with a type parameter \lstinline!A!
and a relation on \lstinline!A! as well as a proof of this relation being an equivalence
as fields. For the term \lstinline!equiv a b! the infix notation ``\lstinline!a == b!''
is introduced. The instance argument of \lstinline!equiv! is implicit (cf.\ \autoref{sec:alg_structure}).

A \emph{morphism of setoids} between setoids \lstinline!A! and 
\lstinline!B! is a \textsf{Coq} function on the underlying types
which is compatible with the setoid relations on the source and target. That is, 
it maps equivalent terms of \lstinline!A! to equivalent terms of \lstinline!B!, or, in mathematical notation,
\begin{equation} 
        a \equiv_A a' \quad \text{implies} \quad f(a) \equiv_B f(a') \enspace . \label{eq:proper} 
\end{equation}
In the \textsf{Coq} standard library such morphisms are implemented as a type class
\begin{lstlisting}
Class Proper {A} (R : relation A) (m : A) : Prop :=
  proper_prf : R m m.
\end{lstlisting}
where the type \lstinline!A! is instantiated with a function type \lstinline!A -> B! and the 
relation \lstinline!R! on \lstinline!A -> B! is instantiated with pointwise 
compatibility\footnote{
In the \textsf{Coq} standard library the definition of \lstinline!respectful! is actually a special case of 
a more general definition of a heterogeneous relation \lstinline!respectful_hetero!.
}:
\begin{lstlisting}
Definition respectful (A B : Type) (R : relation A) (R' : relation B) :=
  fun f g => forall x y, R x y -> R' (f x) (g y).
Notation " R ==> R' " := (@respectful _ _ (R%signature) (R'%signature))
    (right associativity, at level 55) : signature_scope.
\end{lstlisting}

\noindent
Given \textsf{Coq} types \lstinline!A! and \lstinline!B! equipped with relations 
\lstinline!R : relation A! and \lstinline!R' : relation B!, respectively, and a map
\lstinline!f : A -> B!, 
the statement \lstinline!Proper (R ==> R') f! --- replacing aforementioned notation --- really means
\begin{lstlisting}
Proper (respectful R R') f , 
\end{lstlisting}
which is the same as \lstinline!respectful R R' f f!,
which itself just means
\begin{lstlisting}
forall x y, R x y -> R' (f x) (f y) .
\end{lstlisting}
This is indeed the statement of \autoref{eq:proper} in the special case that \lstinline!R!
and \lstinline!R'! are \emph{equivalence} relations.

For any component of an algebraic structure that is a map defined on setoids,
we add a condition of the form \lstinline!Proper...! in the formalization.
Examples are the categorical composition (\autoref{lst:cat}) and the 
monadic substitution map (\autoref{code:endomonad}).
Rewriting related terms under those equivalence relations is tightly integrated in the \lstinline!rewrite!
tactic of \textsf{Coq}.

\subsection{\textsf{Coq} Implementation of Categories}\label{sec:coq_impl_cat}
As a result of the aforementioned considerations, we adopt Sozeau's definition of category \cite{sozeau.Coq/classes/fctc}, 
which itself is a variant of the definition given by Huet and Sa\"ibi \cite{concat}.
Unlike Huet and Sa\"ibi's contribution \textsf{ConCaT}, Sozeau's approach uses \emph{type classes} for algebraic structures and thus 
avoids the universe inconsistencies that have to be circumvented by duplicating definitions in \textsf{ConCaT} (cf.\ \autoref{sec:setoid_eq_on_mor}).
More precisely, in Sozeau's imple\-men\-ta\-tion of setoids (cf.\ \autoref{lst:setoid}), the carrier type is a \emph{parameter}
instead of a \emph{field} as in Huet and Sa\"ibi's.
Our type class of categories is parametrized by a type of objects and a type family of morphisms, whose parameters are the source and target objects.
\begin{form}[Type Class of Categories]\label{lst:cat}
%  [label={lst:cat}, caption={Type class of categories}]
\begin{lstlisting}
Class Cat_struct (obj : Type)(mor : obj -> obj -> Type) := {
  mor_oid :> forall a b, Setoid (mor a b) ;
  id : forall a, mor a a ;
  comp : forall {a b c}, mor a b -> mor b c -> mor a c ;
  comp_oid :> forall a b c, Proper (equiv ==> equiv ==> equiv) (@comp a b c) ;
  id_r : forall a b (f: mor a b), comp f (id b) == f ;
  id_l : forall a b (f: mor a b), comp (id a) f == f ;
  assoc : forall a b c d (f: mor a b) (g:mor b c) (h: mor c d),
      comp (comp f g) h == comp f (comp g h) }.
\end{lstlisting}
\end{form}

\noindent
Compared to \autoref{def:category} there are two additional fields:
the field 
\begin{lstlisting}
mor_oid :> forall a b, Setoid (mor a b)
\end{lstlisting}
equips each collection of morphisms \lstinline!mor a b! with
a custom equivalence relation.
The field \lstinline!comp_oid! states that the composition \lstinline!comp! of the category is compatible with the setoidal structure on the morphisms
given by the field \lstinline!mor_oid! as explained in \autoref{sec:setoids}.
We recall that setoidal equality is overloaded and denoted by the infix symbol `\lstinline!==!'. 
In the following we write `\lstinline!a ---> b!' for \lstinline!mor a b! and \lstinline!f;;g! 
for the composition of morphisms \lstinline!f : a ---> b! and \lstinline!g : b ---> c!
\footnote{\textsf{Coq} deduces and inserts the missing ``object'' arguments \lstinline!a!, \lstinline!b! and \lstinline!c! of the 
composition automatically from the type of the morphisms. 
For this reason those arguments are called \emph{implicit} (cf.\ \autoref{sec:alg_structure}).}.

\subsection{The Categories of Interest}

The category $\SET$ is formalized in \coq~as the category of \coq~types.
By using \textsf{Coq} types and functions as objects and morphisms of 
this category, we obtain executable \textsf{Coq} substitution and translation maps, 
cf.\ \autoref{code:translation_pcf_ulc_example}.
\begin{form}[$\SET$, \autoref{def:SET}]\label{code:SET}
\begin{lstlisting}
Program Instance TYPE_struct : Cat_struct (fun a b => a -> b) := {
   mor_oid a b := TYPE_hom_oid a b ;
   id a := fun x : a => x ;
   comp a b c := fun (f : a -> b) (g : b -> c) => fun x => g (f x) }.
\end{lstlisting}
\end{form}

\noindent
In this instance declaration, the fields \lstinline!id_r!, \lstinline!id_l! and \lstinline!assoc!
are filled automatically by the \lstinline!Program! framework, cf.\ \autoref{subsec:program}.
For a set $T$, the category $\TS{T}$ of \autoref{def:TST} has, as objects, 
\textsf{Coq} type families indexed by $T$.
Morphisms between two such objects are suitable families of \textsf{Coq} functions :
\begin{form}[Category of Type Families] \label{code:cat_type_families}
\begin{lstlisting}
Program Instance ITYPE_struct : Cat_struct (obj := T -> Type)
          (fun A B => forall t, A t -> B t) := {
  mor_oid := INDEXED_TYPE_oid ;  (* pointwise equal. in each component *)
  comp A B C f g := fun t => fun x => g t (f t x) ;
  id A := fun t x => x }.
\end{lstlisting}
\end{form}

\noindent
We also employ categories whose objects are families of \emph{preordered} sets (i.e.\ \textsf{Coq} types), 
and morphisms are monotone \textsf{Coq} functions. We omit their definition.

\subsection{Initial Objects}

Initial objects have been defined in \autoref{def:init_object}.
Formally, we implement the initiality structure as a type class, parametrized by categories.
Its fields are given by an object \lstinline!Init! of the category, a map \lstinline!InitMor! mapping
each object \lstinline!a! of the category to a morphism from \lstinline!Init! to \lstinline!a! 
and a proposition stating that \lstinline!InitMor a! is unique for any object \lstinline!a!.
\begin{lstlisting}
Class Initial (C : Cat) := {
  Init : C ;
  InitMor: forall a : C, Init ---> a ;
  InitMorUnique: forall a (f : Init ---> a), f == InitMor a }.
\end{lstlisting}
Note that the initial morphism is \emph{not} given by an existential statement of the form $\forall a, \exists f : \ldots$,
or, in \textsf{Coq} terms, using an \lstinline!exists! statement.
This is because the \textsf{Coq} existential lies in \lstinline!Prop! and hence does not allow
for elimination --- witness extraction --- when building anything but proofs. 

\subsection{\texorpdfstring{Interlude on the \lstinline!Program! feature}{Interlude on the Program feature}}\label{subsec:program}

The \lstinline!Program Instance! vernacular allows to fill in fields of an instance of a type class by means of tactics.
Indeed, when omitting a field in an instance declaration --- such as the proofs of associativity \lstinline!assoc! 
and left and right identity \lstinline!id_l! and \lstinline!id_r! in \autoref{code:SET}.
--- the \lstinline!Program! framework creates an \emph{obligation}
for each missing field, making use of the information that the user provided for the other fields.
As an example, the obligation created for the field \lstinline!assoc! of the previous example is to prove associativity for the composition defined by
\begin{lstlisting}
comp f g := fun x => g (f x) .
\end{lstlisting}
It then tries to solve the resulting obligations using the tactic that the user has specified via the \lstinline!Obligation Tactic!
command. In case the automatic resolution of the obligation fails, the user can enter the interactive proof mode finish the proof manually.

It is technically possible to fill in both \emph{data} and \emph{proof} fields automatically via the \lstinline!Program! framework. 
However, in order to avoid the automatic inference of data which we cannot control,
we always specify \emph{data} directly as is done in \autoref{code:SET}, and rely on automation via \lstinline!Program! only for \emph{proofs}.

%% file: basic_formal.tex
\subsection{Retyping and Option}

We present the formalization of some commonly used definitions.
The reader might want to skip this section and come back to it when
being pointed back here.

We define retyping (cf.\ \autoref{def:retyping_functor}) 
 for families of sets and preordered sets through an inductive type:
\begin{form}[Retyping Functor, \autoref{def:retyping_functor}]\label{code:retype}
 \begin{lstlisting}
Variables (T T' : Type) (g : T -> T').
Inductive retype (V : ITYPE T) : ITYPE T' :=
  | ctype : forall t, V t -> retype V (g t).
\end{lstlisting}
The constructor \lstinline!ctype : forall V t, V t -> retype V (g t)! is the carrier of the natural 
transformation of the same name of \autoref{def:retyping_functor}.
Given a family $V$ of \emph{preordered} sets, the preorder on $\retyping{g}V$ is induced by the preorder
on $V$:
%\begin{form}[retyping on families of preorders] \label{code:retype_po}
\begin{lstlisting}
Inductive retype_ord (V : IPO T) : forall u, relation (retype g V (u)) :=
  | ctype_ord : forall t (x y : V t), x <<< y
            -> retype_ord (ctype g x) (ctype g y).
\end{lstlisting}
\end{form}

\noindent 
The option data type is implemented in the module \lstinline!Coq.Init.Datatypes! of the \coq~standard library.
\begin{form}[Option, \autoref{sec:endo_deriv}]\label{code:option}
 \begin{lstlisting}
Inductive option (A:Type) : Type :=
  | Some : A -> option A
  | None : option A.  
 \end{lstlisting}
\end{form}

\noindent
We can turn the map $A \mapsto A':= A+\{*\}$ into a monad as follows:
\begin{form}[Option Monad]\label{code:option_monad}
 \begin{lstlisting}
Program Instance option_monad_s :
  Monad_struct (C:=TYPE) (option) := {
  weta := @Some ;
  kleisli a b f := fun t => match t with
                   | Some y => f y
                   | None => None
                   end }.
 \end{lstlisting}
\end{form}

\noindent
There is also a typed variant of the option data type:
\begin{form}[Typed Option, \autoref{sec:endo_deriv}]\label{code:typed_option}
\begin{lstlisting}
Inductive opt (u : T) (V : ITYPE T) : ITYPE T :=
  | some : forall t : T, V t -> opt u V t
  | none : opt u V u.
\end{lstlisting}
\end{form}

\noindent
Given a list \lstinline!l! over \lstinline!T!, 
the multiple addition of variables with (object language) types according to \lstinline!l! to a set of variables \lstinline!V! 
is defined by recursion over \lstinline!l!. 
For this enriched set of variables we introduce the notation \lstinline!V ** l!.

\begin{lstlisting}
Fixpoint pow (l : [T]) (V : ITYPE T) : ITYPE T :=
  match l with
  | nil => V
  | b::bs => pow bs (opt b V)
  end.
\end{lstlisting}

\noindent
The map \lstinline!opt! is functorial, as is the multiple addition of variables \lstinline!pow!. 
On morphisms the \lstinline!pow! operation is defined by recursively applying the functoriality of \lstinline!opt!, 
where for the latter we use a special notation with a prefixed hat.
\begin{lstlisting}
Fixpoint pow_map (l : [T]) V W (f : V ---> W) : 
         V ** l ---> W ** l :=
  match l return V ** l ---> W ** l with
  | nil => f
  | b::bs => pow_map (^f)
  end.
\end{lstlisting}

%% file: endomonads_formal.tex
\section{Monads, Modules and their Morphisms}

Implementing monads leaves one with the choice between the definitions given in 
\autoref{def:monad_mu} and \autoref{def:endomonad}.
The first definition, while preferred by category theorists, has the inconvenience that defining instances 
of monads such as monadic syntax would require proving commutativity of the square, 
e.g., using multiple induction for monadic syntax.
Furthermore the second definition is well--known in the programming community for its use in \textsc{Haskell}.
We thus decide to implement the definition of \autoref{def:endomonad}.
Since we are mainly interested in its instances over the category of (families of) sets, we can define convenient 
infix notation for its substitution.

Formally, a monad (cf.\ \autoref{def:endomonad}) is a type class parametrized by a category \lstinline!C! 
and a function \lstinline!F : C -> C! on the objects of \lstinline!C!:

\begin{form}[Monad, \autoref{def:endomonad}] \label{code:endomonad}
 \begin{lstlisting}
Class Monad_struct (C : Cat) (F : C -> C) := {
  weta : forall c, c ---> (F c);
  kleisli : forall a b, (a ---> F b) -> (F a ---> F b);
  kleisli_oid :> forall a b, Proper (equiv ==> equiv) (kleisli (a:=a) (b:=b));
  eta_kl : forall a b (f : a ---> F b), weta a ;; kleisli f == f;
  kl_eta : forall a, kleisli (weta a) == id _;
  dist : forall a b c (f : a ---> F b) (g : b ---> F c),
      kleisli f ;; kleisli g == kleisli (f ;; kleisli g) }.
\end{lstlisting}
\end{form}

\noindent
Monads admit a functorial structure:
\begin{form}[Functoriality for Monads, \autoref{rem:monad_kleisli_funct}]\label{code:endomonad_functor}
 \begin{lstlisting}
Variable T : Monad C.
Definition lift : forall a b (f: a ---> b), T a ---> T b :=
       fun a b f => kleisli (f ;; weta b).
 \end{lstlisting}
\end{form}

\noindent
We present two different implementations of monad morphisms. 
The more general definition implements colax monad morphisms as defined in \autoref{def:colax_mon_mor_alt}:

\begin{form}[Colax Monad Morphism, \autoref{def:colax_mon_mor_alt}]
  \label{code:colax_mon_mor_alt}
 \begin{lstlisting}
  Class colax_Monad_Hom_struct (Tau : forall c, F (P c) ---> Q (F c)) := {
  gen_monad_hom_kl : forall c d (f : c ---> P d),
       #F (kleisli f) ;; Tau _ ==
          Tau _ ;; (kleisli (#F f ;; Tau _ )) ;
  gen_monad_hom_weta : forall c : C,
       #F (weta c) ;; Tau _ == weta _ }.
 \end{lstlisting}
\end{form}

\noindent
When working exclusively with a special case of a more general definition, it is 
more convenient to implement this special case as a separate definition:
for two monads $P$ and $Q$ over the same category $C$, a \emph{simple morphism of monads} 
 --- as used in \autoref{sec:sts_ju} --- 
is given by a 
family of morphisms $\tau_c\in\C(Pc,Qc)$ 
that is compatible with the monadic structure:

\begin{form}[Simple Monad Morphism, \autoref{def:monad_hom_simpl}] \label{code:monad_hom_simpl}
\begin{lstlisting}
Class Monad_Hom_struct (Tau: forall c, P c ---> Q c) := {
  monad_hom_kl: forall c d (f: c ---> P d),
      kleisli f ;; Tau d == Tau c ;; kleisli (f ;; Tau d) ;
  monad_hom_weta: forall c: C, weta c ;; Tau c == weta c }.
\end{lstlisting}
\end{form}

\noindent
It follows from these commutativity properties that the family $\tau$ is a natural transformation between the functors induced by the monads $P$ and $Q$.
Given a monad $P$ over $\C$, a $P$--module with codomain $D$ is formalized as follows:

\begin{form}[Module, \autoref{def:endo_module}] \label{code:endo_module}
\begin{lstlisting}
Variable P : Monad C.
Class Module_struct (M : C -> D) := {
  mkleisli: forall c d, (c ---> P d) -> (M c ---> M d);
  mkleisli_oid :> forall c d, Proper (equiv ==> equiv) (mkleisli (c:=c)(d:=d));
  mkl_weta: forall c, mkleisli (weta c) == id _ ;
  mkl_mkl: forall c d e (f : c ---> P d) (g : d ---> P e),
      mkleisli f ;; mkleisli g == mkleisli (f ;; kleisli g) }.
\end{lstlisting}
\end{form}

\noindent
For two modules $S$ and $T$ with codomain $\D$ over a monad $P$ as above, 
a module morphism from $S$ to $T$ is given by a family of maps, indexed by the objects of $\C$,
commuting with module substitution:

\begin{form}[Module Morphism, \autoref{def:endo_mod_hom}] \label{code:endo_mod_hom}
\begin{lstlisting} 
Class Module_Hom_struct (N: forall x, S x ---> T x) := {
  mod_hom_mkl: forall c d (f: c ---> P d),
        mkleisli f ;; N _ == N _ ;; mkleisli f }.
\end{lstlisting}
\end{form}

%% file: relative_monads_formal.tex
\section{Relative Monads, Formalized}

As opposed to (plain) monads, we have only one definition of relative monads
available.
The implementation of this definition in \textsf{Coq} is similar to that of monads (cf.\ \autoref{code:endomonad}).
Given a functor $F : \C \to \D$, a relative monad is given by a map $T : \C\to\D$ on 
the objects of the categories involved, and data analogous to that of a monad:
\begin{form}[Relative Monad, \autoref{def:relative_monad}] 
   \label{code:relative_monad}
 \begin{lstlisting}
Variables C D : Cat.
Variable F : Functor C D.
Class RMonad_struct (T : C -> D) := {
  rweta: forall c : C,  F c ---> T c ;
  rkleisli: forall a b : C, F a ---> T b -> T a ---> T b ;
  rkleisli_oid:> forall a b, Proper (equiv ==> equiv) (rkleisli (a:=a) (b:=b)) ;
  reta_kl : forall a b: obC, forall f : F a ---> T b, rweta a ;; rkleisli f == f ;
  rkl_eta : forall a, rkleisli (rweta a) == id _ ;
  rdist: forall a b c (f : F a ---> T b) (g : F b ---> T c),
           rkleisli f ;; rkleisli g == rkleisli (f ;; rkleisli g) }.
 \end{lstlisting}
\end{form}

\noindent
Analogously to monads we define functoriality for a given relative monad \lstinline!P!:
\begin{form}[Functoriality for Relative Monads, \autoref{rem:rel_monad_functorial}]
 \label{code:rel_monad_functorial}
\begin{lstlisting}
Variable P : RMonad C.
Definition rlift : forall a b (f : a ---> b), P a ---> P b :=
       fun a b f => rkleisli (#F f ;; rweta b). 
\end{lstlisting}
\end{form}

\noindent
In the following we consider morphisms of relative monads in varying generality: 
one definition (\autoref{def:simp_rel_mon_mor}) is analogous to the simple morphisms of monads 
(cf.\ \autoref{code:monad_hom_simpl}),
another implements the colax version of \autoref{def:colax_rel_mon_mor}.
For the statement of the second, general, definition, we place ourselves in the 
 environment given in \autoref{def:colax_rel_mon_mor}. In short, we have 
   a natural transformation $N : F'G \Rightarrow G'F : \C \to \D'$.
\begin{form}[Colax Morphism of Relative Monads, \autoref{def:colax_rel_mon_mor}]
  \label{code:colax_rel_mon_mor}
\begin{lstlisting}
Variable N : NT (CompF G F') (CompF F G').
Class colax_RMonad_Hom_struct (tau: forall c : C,  G' (P c) ---> Q (G c)):={
  gen_rmonad_hom_rweta : forall c : C,
    N _ ;; #G' (rweta c) ;; tau c == rweta (G c) ;
  gen_rmonad_hom_rkl : forall (c d : C) (f : F c ---> P d),
    #G' (rkleisli f) ;; tau d == tau c ;; rkleisli (a:=G c) (N c ;; #G' f ;; tau _ ) }.
 \end{lstlisting}

\end{form}

% Variable F : Functor C D.
% Variable F' : Functor C' D'.
% Variable P : RMonad F.
% Variable Q : RMonad F'.
% Variable G : Functor C C'.
% Variable G' : Functor D D'.
% 
% Variable N : NT (CompF G F') (CompF F G').

\noindent
A \emph{module} $M$ over a relative monad $P$ (on a functor $F$) is given by data similar to that of a module over a monad,
except for the insertion of applications of $F$ where necessary.

\begin{form}[Module over a Relative Monad, \autoref{def:rmodule}]
 \label{code:rmodule}
 \begin{lstlisting}
Class RModule_struct (M : C -> E) := {
  rmkleisli: forall c d (f : F c ---> P d), M c ---> M d ;
  rmkleisli_oid :> forall c d, Proper (equiv ==> equiv) (rmkleisli (c:=c)(d:=d)) ;
  rmkl_rweta: forall c : C, rmkleisli (rweta c) == id (M c) ;
  rmkl_rmkl: forall c d e (f : F c ---> P d) (g : F d ---> P e),
         rmkleisli f ;; rmkleisli g == rmkleisli (f ;; rkleisli g) }.
 \end{lstlisting}
\end{form}

\noindent
Given two modules $M$ and $N$ with codomain $\D$ over a relative monad $P$, a module morphism from $M$ to $N$
is given by a collection of maps $(S_c: Mc \to Nc)_{c\in \C}$ commuting with module substitution:
\begin{form}[Morphism of Relative Modules, \autoref{def:rel_mod_mor}]\label{code:rel_mod_mor}
 \begin{lstlisting}
Variables M N : RModule P D.
Class RModule_Hom_struct (S : forall c : C, M c ---> N c) := {
  rmod_hom_rmkl: forall c d (f : F c ---> P d),
        rmkleisli f ;; S d == S c ;; rmkleisli f }.
 \end{lstlisting}
\end{form}

%% file: sts_formal.tex
\chapter{Formalization of Zsid\'o's theorem}\label{chap:sts_formal}

In this chapter we describe the formalization in the proof assistant \textsf{Coq} \cite{coq} of 
Zsid\'o's initiality theorem presented in \autoref{sec:sts_ju}.
In particular, we explain what we omitted in the informal presentation --- the construction of 
the initial representation of a given simply--typed signature.

\section{Signatures \& Representations}\label{section:sig_formal}

An \emph{arity} determines the type and binding behaviour of a \emph{constructor}, and a \emph{signature} is a family of arities.
A \emph{representation} of a signature $S$ is given by a monad $P$ (over a suitable category) 
and a morphism of $P$--modules for each arity $\alpha$ of $S$, where the source and target module of this morphism 
are determined by $\alpha$.
Among those representations the object of interest is the \emph{initial} one, 
i.e.\ the representation from which there exists exactly one \emph{morphism of representations} to any other representation. 
The initial representation is called \emph{syntax generated by $S$}.
% whereas the other representations (with their respective morphisms from the syntax) are the \emph{semantics} of $S$.

\subsection{Using Lists for Algebraic Arities \& Signatures}
For the formal definitions let us fix a set $T$ of object language types.
As explained in \autoref{def:sts_alg_sig}, an algebraic arity over $T$ is determined by a pair of a list of data
and an element $t_0 \in T$, yielding an efficient and concise way to specify algebraic arities.
An algebraic signature could thus  be implemented --- as in \autoref{rem:sts_signature_list} --- as a pair consisting of a 
type \lstinline!sig_index! --- which is used for indexing the arities --- and a map from the indexing type to 
the actual arity type, which is simply built using lists --- for which we employ a Haskell--like notation --- and products:
\begin{form}[Signature, \autoref{rem:sts_signature_list}] \label{code:sts_signature_list_notused}
\begin{lstlisting}
Notation "[ T ]" := (list T) (at level 5).
Record Signature : Type := {
  sig_index : Type;
  sig : sig_index -> [[T] * T] * T }.
\end{lstlisting}
\end{form}

\noindent
However, a slight modification turns out to be useful. 
During the construction of the initial representation, a universal quantification over arities of a signature $S$ 
 with a given target type $t\in T$ is needed. 
Using the above hypothetical implementation, this quantification could be achieved by using a sigma type:
\begin{lstlisting}
Definition Signature_t (t : T) : Type := {s : sig_index S | snd (sig s) = t}.
\end{lstlisting}
This definition would be awkward to use since we would be obliged to handle equality proofs when
talking about indices, i.e.\ terms of \lstinline!sig_index S!, with a specific output type.
 We can
in fact do better: while the \emph{propositional} equality as used above
would need our intervention, \emph{definitional} equality --- \emph{conversion} --- 
is handled by \textsf{Coq}.
Hence we decide to implement a signature over a set of types \lstinline!T! 
as a function that maps each \lstinline!t : T! to the collection of arities 
whose output type is the given \lstinline!t!. 
In other words, the parameter \lstinline!t! of \lstinline!Signature_t! in the definition of 
signature replaces the second component of the arities:

\begin{form}[Signature, \autoref{rem:sts_signature_list}] \label{code:sts_signature}
\begin{lstlisting}
Record Signature_t (t : T) : Type := {
  sig_index : Type ;
  sig : sig_index -> [[T] * T] }.
Definition Signature := forall t, Signature_t t.
\end{lstlisting}
\end{form}

\noindent
We discuss the formalization of the example signature of the simply--typed lambda calculus (cf.\ \autoref{ex:tlc_sig}).
At first we define an indexing type \lstinline!TLC_index_t! for each \lstinline!t : T!. 
After that, we build an indexed signature \lstinline!TLC_sig! mapping each index to its arity:
\begin{form}[Signature of $\TLC$, \autoref{ex:tlc_sig}] \label{code:tlc_sig}
\begin{lstlisting}
Inductive TLC_index : T -> Type := 
   | TLC_abs : forall s t : T, TLC_index (s ~> t)
   | TLC_app : forall s t : T, TLC_index t.

Definition TLC_arguments : forall t, TLC_index t -> [[T] * T] :=
    fun t' r => match r with 
      | TLC_abs s t => (s::nil,t)::nil
      | TLC_app s t => (nil,s ~> t)::(nil,s)::nil
      end.
Definition TLC_sig t := Build_Signature_t t (@TLC_arguments t).
\end{lstlisting}
\end{form}

\subsection{Modules and Morphisms for Arities}\label{subsec:sts_formal_reps}

To any signature given as a dependent function of type \lstinline!Signature! as in \autoref{code:sts_signature} 
 we associate the \emph{actual} signature in the sense of \autoref{def:sts_arity_signature}.
More precisely, for an arity $s = \bigl( [(\vectorletter{s_1},t_1), \ldots, (\vectorletter{s_n},t_n)], t_0\bigr)$ 
given by lists we define the functors
 $\dom(s)$ and $\cod(s)$, each of which, given a monad $P \in \Mon{\TS{T}}$ (cf.\ \autoref{def:sts_monads}),
 yield a $P$--module with codomain $\Set$. Note that the bold face letters $\vectorletter{s_i}$ denote lists of sorts.
 
It would in principle be possible to build the module $\dom(s,P)$ associated to a monad $P$ using the
category--theoretic machinery defined in \autorefs{subsection:mod_examples} and \ref{sec:monads_on_set_families},
i.e.\ by applying iteratively the derivation functor to the tautological module $P$ 
as often as indicated by the arity $s$ and finally the suitable fibre functor,
glueing everything together via the product on module categories.
% e.g.\ the module morphism induced by a monad morphism and the fibre and derivation functors.
However, we choose not to, for reasons we explain now.
Consider again the diagram of \autoref{eq:sts_comm_diag}, instantiated for the classic arity $s$:

\begin{equation} \label{eq:sts_diag}
\begin{xy}
\xymatrix @=4pc{
**[l] \prod\limits_{i=1}^n \fibre{P^{\mathbf{s_i}}}{t_i}\ar[r]^{{\alpha}^P} \ar[d]_{\prod\limits_i \fibre{f^{\mathbf{s_i}}}{t_i}} & \fibre{P}{t_0} \ar[d]^{f_{t_0}}  \\
**[l] f^* \prod\limits_{i=1}^n \fibre{Q^{\mathbf{s_i}}}{t_i} \ar[r]_{f^*({\alpha}^Q)} & f^*\fibre{Q}{t_0} \\
}
\end{xy}
\end{equation}

\noindent
This diagram actually makes use of many instances of the equalities mentioned in 
\autoref{rem:pullback_deriv_fibre_endo},
in order to justify composability of module morphisms.
For instance, in the lower right corner, the fact that pullback and fibre may be permuted, is used.
In \textsf{Coq} the aforementioned equalities of modules hold propositionally (if one uses appropriate axioms, such as proof irrelevance),
but not definitionally, i.e.\ the modules are not convertible (see also \autoref{rem:about_module_eq}). 
In order to be able to compose a module morphism with target $\rho^*\fibre{M}{u}$, for instance,
with a module morphism with source module $\fibre{\rho^*M}{u}$, one needs to insert a suitable isomorphism of modules 
 $\rho^*\fibre{M}{u} \cong \fibre{\rho^*M}{u}$.
The carriers of these isomorphisms are families of identity functions, respectively, 
since the carriers of the source and target modules are convertible. %As modules, however, source and target are not convertible in \textsf{Coq}.
In our formalization we would have to insert these isomorphisms 
(called \lstinline!PROD_PB!, \lstinline!ITDER_PB! and \lstinline!ITFIB_PB! in our \textsf{Coq} library) 
in order to make some compositions typecheck
--- as illustrated by the diagram in \autoref{eq:diag} --- 
which in turn would result in quite a cumbersome formalization with decreased readability. 
Instead we decide to implement the left vertical morphism from scratch. 
\begin{table}[hbt]
\begin{equation}\label{eq:diag}
 \begin{xy}
  \xymatrix @=3.5pc{
  \prod\limits_{i=1}^n \fibre{P^{\mathbf{s_i}}}{t_i}\ar[rr]^{{\alpha}^P} \ar[d]_{\prod_i \fibre{f^{\mathbf{s_i}}}{t_i}}& {} & P_{t_0} \ar[dddd]^{f_{t_0}}  \\
  \prod_{i=1}^n \fibre{(f^*Q)^{\mathbf{s_i}}}{t_i} \ar[d]_{\prod_i \fibre{\cong}{t_i}} & {} &{} \\
  %\prod_{i=1}^n (f^* D_i Q)_{t_i} & (f^*Q)_{t_0} \\
  \prod_{i=1}^n \fibre{f^*(Q^{\mathbf{s_i}})}{t_i} \ar[d]_{\prod_i \cong}& {} &{} \\
  \prod_{i=1}^n f^*\fibre{Q^{\mathbf{s_i}}}{t_i} \ar[d]_{\cong}& {} &{} \\
  f^* \prod_{i=1}^n \fibre{Q^{\mathbf{s_i}}}{t_i} \ar[r]_-{f^* ({\alpha}^Q)} & f^* \fibre{Q}{t_0} \ar[r]_{\cong}& \fibre{f^*Q}{t_0}
}
 \end{xy}
\end{equation}

\end{table}
For this to work it is most convenient to define the carrier of the product modules as an inductive type, 
instead of applying the product in the module category recursively. 
Hence also the product modules are built manually rather than using the categorical devices of derivation, fibre and product.

\subsubsection{Domain, Codomain, Representations}

Given an arity $s = (\vectorletter{s_1},t_1), \ldots ,(\vectorletter{s_n},t_n)\to t_0 $ (or shorter $\ell \to t_0$) and a monad $P$, 
we have to construct the module 
 \[\dom(s,P) =  \prod_{i=1}^n \fibre{P^{\mathbf{s_i}}}{t_i} = \prod_{\ell} P \enspace . \]
Its carrier, being a kind of heterogeneous list, is given as an inductive 
type parame\-tri\-zed by a set family of variables and a list such as the list $\ell$ indicating the domain of an arity. 
In fact, for defining the carrier, only an object map $M$ of the type indicated below is necessary:
%\begin{form}
\begin{lstlisting}
Variable M : (ITYPE T) -> (ITYPE T).
Inductive prod_mod_c (V : ITYPE T) : [[T] * T] -> Type :=
  | TTT :  prod_mod_c V nil 
  | CONSTR : forall b bs, 
     M (V ** (fst b)) (snd b) -> prod_mod_c V bs -> prod_mod_c V (b::bs).
\end{lstlisting}
%\end{form}

\noindent
Now, for a list \lstinline!l : [[T]*T]!,  if $M$ is equipped with a module structure over a monad $P$, 
we equip the map \lstinline!fun V => prod_mod_c V l!
with a module structure.
Its substitution is given by a function \lstinline!pm_mkl!, which is defined by recursion 
on the argument of type \lstinline!prod_mod_c ... !, applying 
the module substitution of $M$ in each component:
\begin{lstlisting}
Fixpoint pm_mkl l V W (f : V ---> P W) (X : prod_mod_c M V l) : 
                                               prod_mod_c M W l :=
     match X in prod_mod_c _ _ l return prod_mod_c M W l with
     | TTT => TTT M W
     | CONSTR b bs elem elems => CONSTR (M:=M) (V:=W) 
                 (mkleisli (Module_struct := M) (lshift f) (snd b) elem)
                 (pm_mkl f elems)
     end.
\end{lstlisting}

\noindent
Proving its module property ---  by induction on the argument \lstinline!X! --- yields a
module \lstinline!prod_mod l! for each list \lstinline!l : [[T] * T]!. 
For $s = \ell \to t_0$, this defines the object function of the functor $\dom(s)$.
The object function of $\cod(s)$ is easy to define, since it simply associates, to any monad $P$, 
the fibre module with respect to $t_0$ of the tautological module $P$.
Again, this is defined more generally for any $P$--module $M$ with codomain category $\TS{T}$.
Putting both domain and codomain together, we associate, 
to any algebraic arity $s$ and any $P$--module $M$, a type of module
morphisms
\[ \dom(s,M) \to \cod(s,M) \]
as in \autoref{code:sts_arity_rep} below. 
Note that $M$ is later instantiated by the tautological $P$--module $P$.

\begin{form}[Representation of an Arity, \autoref{def:sts_arity_rep}]\label{code:sts_arity_rep}
\begin{lstlisting}
Variable M : Module P (ITYPE T).
Definition modhom_from_arity (ar : [[T] * T] * T) : Type :=
  Module_Hom (prod_mod M (fst ar)) (M [(snd ar)]).
\end{lstlisting} 
\end{form}

\noindent
where \lstinline!M[(s)]! denotes the fibre of the module \lstinline!M! over \lstinline!s!.
Finally a representation of a signature \lstinline!S! in a monad \lstinline!P! is given by a module morphism for each arity \lstinline!i!, 
i.e.\ by specifying a function of type 
\[\forall s\in S, \dom(s,P) \to \cod(s,P) \enspace , \]
where $P$ denotes the tautological $P$--module.
Since the set of arities is indexed by the target type of the arities, the representation structure is indexed as well:
\begin{form}[Representation of a Signature, \autoref{def:sts_arity_rep}]\label{code:sts_sig_rep}
\begin{lstlisting}
Variable P : Monad (ITYPE T).
Definition Repr_t (t : T) := 
   forall i : sig_index (S t), modhom_from_arity P ((sig i), t).
Definition Repr := forall t, Repr_t t.
\end{lstlisting}
\end{form}

\noindent
We bundle the data and define a representation as a monad together with a representation structure over this monad\footnote{
Here an example of \emph{coercion} occurs. 
The special notation \lstinline!:>! allows us to omit the projection \lstinline!rep_monad! when accessing the monad 
which underlies a given representation \lstinline!R!. We can hence also write \lstinline!R x! for the value of the 
monad of \lstinline!R! on an object \lstinline!x! of the underlying category. 
%This is precisely the kind of \emph{abus de notation} which is used in mathematics, and such coercions are used very often in our theory files. Other examples can be found e.\ g.\ in \lstinline!Cat!, \lstinline!Functor! and \lstinline!Monad!.
}:
\begin{lstlisting}
Record Representation := {
  rep_monad :> Monad (ITYPE T);
  repr : Repr rep_monad }.
\end{lstlisting}

\subsubsection{Morphisms of Representations}

The carrier of the domain module $\dom(s,P) = \prod_l P$ of a representation (cf.\ \autoref{eq:sts_diag}) is defined as an inductive type. 
This suggests the use of structural recursion for defining the left vertical morphism of the commutative diagram of \autoref{eq:sts_diag}. 
Given a monad morphism $f : P\to Q$, we apply $f$ to every component of 
%$\prod_{i=1}^n (\partial_{\vec{s}_i} P)_{t_i}$
$\prod_{\ell} P$
:

\begin{lstlisting}
Fixpoint Prod_mor_c (l : [[T] * T]) (V : ITYPE T) (X : prod_mod P l V) : 
                  f* (prod_mod Q l) V :=
  match X in prod_mod_c _ _ l return f* (prod_mod Q l) V with
  | TTT => TTT _ _ 
  | CONSTR b bs elem elems => 
    CONSTR (f _ _ elem) (Prod_mor_c elems)
  end.
\end{lstlisting}
This function is easily proved to be a morphism of $P$--modules 
%\[ \text{\tt Prod\sub mor}\colon\prod_{i=1}^n (\partial_{\vec{s}_i} P)_{t_i} \to f^* \prod_{i=1}^n (\partial_{\vec{s}_i} Q)_{t_i}.\]
\[ \dom(s,f) := \text{\lstinline!Prod_mor!}\colon\prod_{\ell} P \to f^* \prod_{\ell} Q \enspace.\]
We thus are able to avoid mentioning all those trivial isomorphisms in the definition of the arrow map of the functor $\dom(s)$
that are present in the diagram of \autoref{eq:diag}.

The codomain arrow $\cod(s,f) = f_{t_0}$ is obtained by taking the fibre module of the 
module morphism induced by $f$, cf.\ \autoref{subsection:mod_examples}.
The \textsf{Coq} function \lstinline!PbMod_ind_Hom!, which associates to any monad morphism the induced module morphism,
can even be declared as a coercion
\begin{lstlisting}
Coercion PbMod_ind_Hom : Monad_Hom >-> mor.
\end{lstlisting}
such that the abuse of notation introduced in the informal \autoref{def:induced_module_mor} has a 
counterpart in the formal development.

The isomorphism in the lower right corner however remains in the formalization, 
appearing as \lstinline!ITPB_FIB!. Its underlying family of morphisms, however, is simply a family of identity functions. 
For an arity \lstinline!a! and module morphisms \lstinline!RepP! and \lstinline!RepQ! representing 
this arity in monads \lstinline!P! and \lstinline!Q! respectively, the definition of the commutative diagram 
reads as follows:
\begin{form}[Commutative Diagram for Representation Morphism, \autoref{def:sts_rep_mor}]
\begin{lstlisting}
Definition commute f RepP RepQ : Prop :=
   RepP ;; f [(snd a)] == Prod_mor (fst a) ;; f* RepQ ;; ITPB_FIB f _ _ 
\end{lstlisting}
\end{form}

\noindent
A morphism of representations from \lstinline!P! to \lstinline!Q! of the signature \lstinline!S! is 
just a monad morphism from \lstinline!P! to \lstinline!Q! together with the commutativity property 
for each arity. More precisely, since arities are indexed by their target type, we have a commutative diagram
for any object type \lstinline!t : T! and each arity (index) \lstinline!i! in the indexing set of \lstinline!S t!:
\begin{form}[Morphism of representations, \autoref{def:sts_rep_mor}]\label{code:sts_rep_mor}
\begin{lstlisting}
Variables P Q : Representation S.
Class Representation_Hom_struct (f : Monad_Hom P Q) :=
  repr_hom_s : forall t (i : sig_index (S t)), commute f (repr P i) (repr Q i).
Record Representation_Hom : Type := {
  repr_hom_c :> Monad_Hom P Q;
  repr_hom :> Representation_Hom_struct repr_hom_c }.
\end{lstlisting}
\end{form}

\noindent
As mentioned in \autoref{subsection:sts_rep}, representations of \lstinline!S! and their 
morphisms form a category \lstinline!REPRESENTATION S!. 
Composition of representations is defined by composing the underlying monad morphisms:
\begin{lstlisting}
Program Instance Rep_comp_struct :
   Representation_Hom_struct (Monad_Hom_comp f g).
\end{lstlisting}
where the commutation property is proved by some tactic defined beforehand.
Accordingly, the identity morphism of representations is built upon the identity
monad morphism:
\begin{lstlisting}
Program Instance Rep_Id_struct : 
   Representation_Hom_struct (Monad_Hom_id P).
\end{lstlisting}
Since equality on morphisms of representations is defined as equality of the underlying monad morphisms, 
the properties of composition necessary for representations to form a category are a consequence of those 
for the category \lstinline!MONAD (ITYPE T)!.
The construction of the initial representation (and hence the proof of \autoref{nice_thm}) is explained in the next section.

\section{Construction of the Initial Object}\label{section:sts_initial}

The initial object of the category of representations of the signature $S$ is constructed in several steps:
\begin{packenum}
 \item the syntax associated to $S$ as an inductive data type \lstinline!STS!,
  \item definition of a monad structure \lstinline!STS_Monad! on said data type,
  \item construction of the representation structure \lstinline!STSRepr! on \lstinline!STS_Monad!,
 \item for any representation \lstinline!R!, construction of morphism \lstinline! init R! from \lstinline!STSRepr! to \lstinline!R!,
  \item uniqueness of \lstinline!init R! for any representation \lstinline!R!.
\end{packenum}

\subsection{The Terms Generated by a Signature} %\texorpdfstring{The Syntax associated to $S$}{The Syntax associated to S}}
The first step is to define a map \lstinline!STS : ITYPE T ---> ITYPE T! --- the monad carrier --- mapping each type family $V$ of variables to the type family of 
terms with free variables in $V$. Since objects of \lstinline!ITYPE T! really are dependent \textsf{Coq} types (cf.\ \autoref{code:cat_type_families}),
this map is implemented as a \textsf{Coq} inductive family of types, parametrized by a context and dependent on object types.
Apart from the use of dependent types, the ``data'' parts of this section could indeed be done in any programming language featuring inductive types.

Mutual induction is used, defining at the same time a type \lstinline!STS_list! of heterogeneous lists of terms, 
yielding the arguments to the constructors of $S$. This list type is indexed by arities, such that the constructors can be fed with precisely the 
right kind of arguments.

\begin{form}[Terms of the Initial Representation] \label{code:sts_initial_term}
\begin{lstlisting}
Inductive STS (V : ITYPE T) : ITYPE T :=
  | Var : forall t, V t -> STS V t
  | Build : forall t (i : sig_index (S t)), STS_list V (sig i) -> STS V t
with 
 STS_list (V : ITYPE T) : [[T] * T] -> Type :=
  | TT : STS_list V nil
  | constr : forall b bs, 
      STS (V ** (fst b)) (snd b) -> STS_list V bs -> STS_list V (b::bs).
\end{lstlisting}
\end{form}

\noindent
The constructor \lstinline!Build! takes 3 arguments: 
\begin{packitem}
 \item an object type \lstinline!t! indicating its output type,
 \item an arity \lstinline!i! (resp.\ its index) from the set of indices with output type \lstinline!t! and
 \item a term of type \lstinline!STS_list V (sig i)! carrying the subterms of the term to construct.
\end{packitem}
Note that \textsf{Coq} typing ensures the correct typing of all constructible terms of \lstinline!STS!, 
a techique called \emph{intrinsic typing}.
The \lstinline!Scheme! command generates a mutual induction scheme for the defined pair of types.
The latter type is actually isomorphic to the type \lstinline!prod_mod_c STS!. This duplication of data could hence have been avoided by 
defining a nested inductive type as follows, instead of using mutual induction.
\begin{lstlisting}
 Inductive STS (V : ITYPE T) : ITYPE T :=
  | Var : forall t, V t -> STS V t
  | Build : forall t (i : sig_index (S t)), prod_mod_c STS V (sig i) -> STS V t.
\end{lstlisting}

\noindent
However, we use the mutual inductive version because it allows us to define functions on those types by mutual recursion
rather than by nested recursion; the latter are significantly more difficult to reason about.

\subsection{Monad Structure on the Set of Terms}%{\texorpdfstring{Monad Structure on \lstinline!STS!}{Monad Structure on STS}}

We continue by defining a monad structure on the map \lstinline!STS!. Again, due to our choice of implementing sets as \textsf{Coq} types 
(cf.\ \autoref{code:cat_type_families}), the maps we need really are \textsf{Coq} functions.
As in the special case of $\LC$ (cf.\ \autoref{ex:ulc_monad}) and $\SLC$ (cf.\ \autoref{ex:tlc_syntax_monadic}),
the monadic map $\eta$ is given by the variable--as--term constructor \lstinline!Var!.
The substitution map \lstinline!subst! is defined using two helper functions \lstinline!rename! (providing functoriality, cf.\ \autoref{rem:monad_kleisli_funct}) 
and \lstinline!_shift!
(used when substituting under binders, cf.\ \autoref{ex:ulc_mod_mor_kl}). Renaming and substitution 
are implemented using mutual recursion on the mutually inductive data types \lstinline!STS! and \lstinline!STS_list!:
\begin{lstlisting}
 Fixpoint rename V W (f : V ---> W) t (v : STS V t):=
    match v in STS _ t return STS W t with
    | Var t v => Var (f t v)
    | Build t i l => Build (i:=i) (list_rename l f)
    end
with
  list_rename V t (l : STS_list V t) W (f : V ---> W) : STS_list W t :=
     match l in STS_list _ t return STS_list W t with
     | TT => TT W
     | constr b bs elem elems =>
             constr (elem //- ( f ^^ (fst b)))
                               (elems //-- f)
     end
where "x //- f" := (rename f x)
and "x //-- f" := (list_rename x f).
...
(* a lot more code *)
...
Fixpoint subst (V W : ITYPE T) (f : V ---> STS W) t (v : STS V t) :
  STS W t := match v in STS _ t return STS W t with
    | Var t v => f t v
    | Build t i l => Build (l >>== f)
    end
with
  list_subst V W t (l : STS_list V t) (f : V ---> STS W) : STS_list W t :=
     match l in STS_list _ t return STS_list W t with
     | TT => TT W
     | constr b bs elem elems =>
       constr (elem >== (_lshift f)) (elems >>== f)
     end
where "x >== f" := (subst f x)
and "x >>== f" := (list_subst x f).
\end{lstlisting}
The monadic properties that the substitution should satisfy, are similar to the lemmas one would prove in order to establish ``programm correctness''.
As an example, the third monad law reads as
\begin{lstlisting}
 Lemma subst_subst V t (v : STS V t) W X (f : V ---> STS W)
             (g : W ---> STS X) :
     v >== f >== g = v >== f;; subst g.
Proof.
  apply (@STSind
    (fun (V : T -> Type) (t : T) (v : STS V t) => forall (W X : T -> Type)
          (f : V ---> STS W) (g : W ---> STS X),
        v >== f >== g = v >== (f;; subst g))
   (fun (V : T -> Type) l (v : STS_list V l) =>
       forall (W X : T -> Type)
          (f : V ---> STS W) (g : W ---> STS X),
        v >>== f >>== g = v >>== (f;; subst g) ));
  t5.
Qed.
\end{lstlisting}
Its proof script is a typical example; most of those lemmas are proved using the induction scheme \lstinline!STSind! 
--- instantiated with suitable properties --- followed by a single custom tactic which finishes off the resulting subgoals, 
mainly by rewriting with equalities proved beforehand.
After a quite lengthy series of lemmas we obtain that the function \lstinline!subst! and the variable--as--term constructor \lstinline!Var! turn \lstinline!STS! into a monad:
\begin{lstlisting}
Program Instance STS_monad : Monad_struct STS := {
  weta := Var ;
  kleisli := subst }.
\end{lstlisting}

\subsection{A Representation in the Monad of Terms} %{\texorpdfstring{A representation in \lstinline!STS!}{A representation in STS}}

The representational structure on \lstinline!STS! is defined using the \lstinline!Build! constructor. 
For each arity \lstinline!i! in the index set \lstinline!sig_index (S t)!, we must give a morphism of modules from \lstinline!prod_mod STS (sig i)! to 
\lstinline!STS [(t)]!. 
Since the constructor \lstinline!Build! takes its argument from \lstinline!STS_list! and not from the isomorphic \lstinline!prod_mod STS!, 
we precompose with one of the isomorphisms between those two types:
\begin{lstlisting}
Program Instance STS_arity_rep (t : T) (i : sig_index (S t)) : Module_Hom_struct 
       (S := prod_mod STS (sig i)) (T := STS [(t)]) 
   (fun V X => Build (STSl_f_pm X)).
\end{lstlisting}
The only property to verify is the compatibility of this map with the module substitution, which we happily leave to \textsf{Coq}.
We obtain a representation of \lstinline!S!:
\begin{lstlisting}
Record STSRepr : REPRESENTATION S := Build_Representation (@STSrepr).
\end{lstlisting}

\subsection{Weak Initiality for the Representation in the Term Monad}

In the introduction, we gave the equations that a morphism of representations of the natural numbers should satisfy.
Reading those equations as a rewrite system from left to right yields a way to define iterative functions on the
natural numbers.
This idea is also used in order to define a morphism from \lstinline!STSRepr! to any representation \lstinline!R! of the signature \lstinline!S!:
a term of \lstinline!STS!, whose root is a constructor \lstinline!Build t i! for some object type \lstinline!t! and an arity \lstinline!i!,
is mapped recursively to the image --- of the recursively computed argument --- 
under the corresponding representation \lstinline!repr R i! of \lstinline!R!. 
This definition for a morphism of representations will turn out to be the only one possible, leading to uniqueness.
Formally, the carrier \lstinline!init! of what will be the initial morphism  from \lstinline!STSRepr! to \lstinline!R! 
is defined as a mutually recursive \textsf{Coq} function:
\begin{lstlisting}
Fixpoint init V t (v : STS V t) : R V t :=
   match v in STS _ t return R V t with
   | Var t v => weta (Monad_struct := R) V t v
   | Build t i X => repr R i V (init_list X)
   end
with
 init_list l (V : ITYPE T) (s : STS_list V l) : prod_mod R l V :=
   match s in STS_list _ l return prod_mod R l V with
   | TT => TTT _ _
   | constr b bs elem elems =>
        CONSTR (init elem) (init_list elems)
   end.
\end{lstlisting}
where the function \lstinline!init_list! applies \lstinline!init! to (heterogeneous) lists of arguments.
We have to show that this function is a morphism of monads and a morphism of representations.
A series of lemmas show that \lstinline!init! commutes with renaming resp.\ lifting (\lstinline!init_lift!), 
shifting (\lstinline!init_shift!) and substitution (\lstinline!init_kleisli!):
\begin{lstlisting}
Lemma init_lift V t x W (f : V ---> W) : init (x //- f) = lift f t (init x).
Lemma init_shift a V W (f : V ---> STS W) : forall (t : T) (x : opt a V t),
    init (x >>- f) = x >>- (f ;; @init _).
Lemma init_kleisli V t (v : STS V t) W (f : V ---> STS W) :
  init (v >== f) = kleisli (f ;; @init _ ) t (init v).
\end{lstlisting}
 
\noindent
The latter property is precisely one of the axioms of morphisms of monads (cf.\ \autoref{def:monad_hom_simpl}, rectangular diagram). 
The second monad morphism axiom which states compatibility with the $\eta$s of the monads involved is 
fulfilled by definition of \lstinline!init! --- it is exactly the first branch of the pattern matching by which the function \lstinline!init!
is defined. 
We hence have established that \lstinline!init! is (the carrier of) a morphism of monads:
\begin{lstlisting}
Program Instance init_monadic : Monad_Hom_struct (P:=STSM) init.
Record init_mon := Build_Monad_Hom init_monadic.
\end{lstlisting}

\noindent
Very much less work is then needed to show that \lstinline!init! also is a morphism of representations:
\begin{lstlisting}
Program Instance init_representic : Representation_Hom_struct init_mon.
\end{lstlisting}

\subsection{Uniqueness and Initiality}

Uniqueness of the morphism of representations \lstinline!init_rep! (obtained from packaging \lstinline!init_representic! into 
a record instance) is expressed by the following lemma:
\begin{lstlisting}
Lemma init_unique : forall f : STSRepr ---> R , f == init_rep.
\end{lstlisting}
Instead of directly proving the lemma, we prove at first an unfolded version which allows to directly apply 
the mutual induction scheme \lstinline!STSind!:
\begin{lstlisting}
Variable f : Representation_Hom STSRepr R.
Hint Rewrite one_way : fin.
Ltac ttt := tt; 
 (try match goal with [t:T, s : STS_list _ _ |-_] => rewrite <- (one_way s);
             let H:=fresh in assert (H:=repr_hom f (t:=t));
             unfold commute in H; simpl in H end);
             repeat (app (mh_weta f) || tinv || tt).

Lemma init_unique_prepa V t (v : STS V t) : f V t v = init v.
Proof.
  apply (@STSind
     (fun V t v => f V t v = init v)
     (fun V l v => Prod_mor f l V (pm_f_STSl v) = init_list v));
  ttt.
Qed.
\end{lstlisting}

\noindent
Finally we declare an instance of the \lstinline!Initial! type class for the category of representations \lstinline!REPRESENTATION S! with  
\lstinline!STSRepr! as initial object and \lstinline!init_rep R! as the initial morphism towards any other representation \lstinline!R!.
\begin{form}[Instance of Initial for Category of Representations]\label{lst:initial_rep}
\begin{lstlisting}
Program Instance STS_initial : Initial (REPRESENTATION S) := {
  Init := STSRepr ;
  InitMor R := init_rep R }.
\end{lstlisting}
\end{form}

\noindent
In this instance declaration, the proof field \lstinline!InitMorUnique! is filled 
automatically by the \lstinline!Program! feature, using the preceding lemma \lstinline!init_unique!.

\section{Remarks}

The nature of the theorem made it convenient for computer theorem proving: the proofs are %
straightforward, carrying no %
surprises. Moreover, they are highly technical using (mutual) induction, something 
\textsf{Coq} offers good support for.

Some aspects remain unsatisfactory: using type classes and records simultaneously is at least confusing for the reader, 
even if there are reasons from the implementor's point of view to do so. 
Also, the weak support for nested induction in \textsf{Coq} obliged us to use mutual induction instead, 
leading to some duplication of data and hence another unnecessary source of confusion.
Other aspects, such as the implementation of syntax in an efficient way, 
i.e.\ without any extrinsic typing device, could be done due to \textsf{Coq}'s good support for dependent types.

According to \lstinline!coqwc!\footnote{%
The tool \lstinline!coqwc!, part of the standard \textsf{Coq} tools, counts the number of lines in a \textsf{Coq} source file, 
classified into the 3 categories \emph{specification}, \emph{proof} and \emph{comment}.
} 
the \textsf{Coq} files that are specific to the proved theorem consist of approximately 400 lines of specification and 600 lines of proof. 
The proofs are done in a semi--automated way, employing a proof style promoted by Chlipala in his online book \cite{cpdt}, 
as well as in a published user tutorial \cite{jfr:cpdt}.
An earlier version using a more standard proof style included about 900 lines of proof. 
This reduction is mainly due to the fact that proof automation also stimulates reuse of code -- here reuse of proof code -- 
similarly to how polymorphism does for data structures and functions.
However, we do not claim to be experts in proof automation, nor do we have ``one tactic to rule them all''.

%% file: prop_arities_formal.tex
In this chapter we present the formalization in the proof assistant \textsf{Coq}
of \autoref{thm:init_w_ineq_untyped} of \autoref{sec:prop_arities}.
We first define arities and 1--signatures in terms of lists.
Afterwards we define representations for 1--arities and construct the initial such representation.
We then formalize inequations over 1--signatures and construct, for any suitable 2--signature,
the initial representation.
Finally we show how to specify the untyped lambda calculus with beta reduction via a 2--signature.

\section{Arities by Lists}

According to \autoref{def:classic_half_arity_prop_untyped_syntactic},
a 1--signature consists of an indexing type and, for each index, a list of natural numbers,
indicating the number of arguments of a constructor, as well as the number of variables 
bound in each argument.
Formally, 1--signatures are an untyped version of \autoref{code:sts_signature_list_notused}.
In the formalization they are simply called ``signatures'':

\begin{form}[1--Signature, \autoref{def:classic_half_arity_prop_untyped_syntactic}]
 \label{code:1--signature_untyped}
 \begin{lstlisting}
Notation "[ T ]" := (list T) (at level 5).
Record Signature : Type := {
  sig_index : Type ;
  sig : sig_index -> [nat] }.
 \end{lstlisting}
\end{form}

\noindent
Next we formalize context extension according to a natural number, cf.\ \autoref{sec:deriv_and_fibre}.
These definitions are important for the definition of the module morphisms we 
associate to an arity, cf.\ below.
Context extension is actually functorial.
Given a natural number \lstinline!n! %$n$ 
and a set of variables \lstinline!V!, %$V$
 we recursively
define the set \lstinline!V ** n! to be the set \lstinline!V! enriched with \lstinline!n!
additional variables. 

\begin{form}[Adding fresh variables]
 \begin{lstlisting}
Fixpoint pow (n : nat) (V : TYPE) : TYPE :=
  match n with
  | 0 => V
  | S n' => pow n' (option V)
  end.
Notation "V ** n" := (pow n V) (at level 10).
Fixpoint pow_map (l : nat) V W (f : V ---> W) :
         V ** l ---> W ** l :=
  match l return V ** l ---> W ** l with
  | 0 => f
  | S n' => pow_map (^ f)
  end.
Notation "f ^^ l" := (pow_map (l:=l) f) (at level 10).
 \end{lstlisting}
\end{form}

\section{Representations of a 1--Signature}

Given a classic arity $s$, i.e.\ a list of natural numbers $s$ (cf.\ \autoref{code:1--signature_untyped}), 
and a relative monad $P$ on the functor $\Delta$,
we define the product module $P^{s}$ as in \autoref{rem:syn_sem_untyped_two_half_arities}. 
More generally, we define $M^s$ for any $P$--module $M$ with codomain $\PO$.
Analogously to the implementation of \autoref{chap:sts_formal}, we build this module from scratch instead of 
relying on the category--theoretic constructions such as product and derivation functor for the module categories,
allowing us to omit the insertion of isomorphisms in the style of \autorefs{lem:rel_pb_prod} and \ref{lem:rel_pb_comm}.
Given any module \lstinline!M! over a monad \lstinline!P! 
from sets to preordered sets, we define the 
product type \lstinline!prod_mod_c! as a dependent type parametrized by a 
set of variables and dependent on a list of naturals. 
Actually we define at first the carrier depending not on a module, but just on a carrier function \lstinline!M!.
The relation on the product is induced by that on \lstinline!M!.
% The bound variable \lstinline!a! below is of type \lstinline![nat]!.
\begin{form}[Product Module, Carrier map] \label{code:prop_prod_mod}
\begin{lstlisting}
Variable M : TYPE -> Ord.
Inductive prod_mod_c (V : TYPE) : [nat] -> Type :=
  | TTT : prod_mod_c V nil
  | CONSTR : forall b bs,
         M (V ** b)-> prod_mod_c V bs -> prod_mod_c V (b::bs) .
Notation "a -:- b" := (CONSTR a b) (at level 60).
Inductive prod_mod_c_rel (V: TYPE) : forall n, relation (prod_mod_c M V n):=
  | TTT_rel : forall x y : prod_mod_c M V nil, prod_mod_c_rel x y
  | CONSTR_rel : forall n l, forall x y : M (V ** n),
          forall a b : prod_mod_c M V l, x << y -> 
     prod_mod_c_rel a b -> prod_mod_c_rel (x -:- a) (y -:- b).
\end{lstlisting}
\end{form}

\noindent
Note that the infixed ``\lstinline!<<!'' is overloaded notation and denotes the relation of any preordered set.
For any given list \lstinline!a! of naturals and any set \lstinline!V! 
of variables, the set \lstinline!prod_mod_c V a! equipped
with the relation \lstinline!prod_mod_c_rel V a! is in fact a preordered set. For the proof of transitivity we rely
on the \textsf{Coq} tactic \lstinline!dependent induction!, thus on the axioms
\begin{lstlisting}
JMeq.JMeq_eq : forall (A : Type) (x y : A), x ~= y -> x = y 
Eqdep.Eq_rect_eq.eq_rect_eq : forall (U : Type) (p : U) 
                                (Q : U -> Type) (x : Q p) 
                                (h : p = p), x = eq_rect p Q x p h
\end{lstlisting}
from the \textsf{Coq} standard library.

Now, if \lstinline!M! is not just a map of type \lstinline!TYPE -> Ord!, 
but a module over some relative monad \lstinline!P! over \lstinline!Delta!,
we equip the product map with a modulic substitution in form of a recursive function:

\begin{form}[Product module, substitution]\label{code:prop_prod_subst}
 \begin{lstlisting}
Variable M : RMOD P Delta.
Fixpoint pm_mkl l V W (f : Delta V ---> P W)
      (X : prod_mod_c (fun V => M V) V l) : prod_mod_c _ W l :=
     match X in prod_mod_c _ _ l return prod_mod_c (fun V => M V) W l with
     | TTT => TTT _ W
     | elem -:- elems =>
      rmkleisli (RModule_struct := M) (lshift _ f) elem -:- pm_mkl f elems
     end.
(* ... *)
Definition prod_mod (a : [nat]) := Build_RModule (prod_mod_struct a).
 \end{lstlisting}
\end{form}

\noindent
Afterwards we prove by induction that this map is indeed monotone with respect to the preorder defined in 
\autoref{code:prop_prod_mod}. Altogether, \autoref{code:prop_prod_mod} and
\ref{code:prop_prod_subst} define a module \lstinline!prod_mod M l! for any module \lstinline!M : RMOD P Ord! and
any list of naturals \lstinline!l!.

To any arity \lstinline!ar : [nat]! and a module \lstinline!M! over a monad \lstinline!P!
we associate a type of module morphisms \lstinline!modhom_from_arity ar M!.
Representing \lstinline!ar! in \lstinline!M! then means giving a term of type \lstinline!modhom_from_arity ar M!.
Note that in the corresponding \autoref{def:rep_1--arity_untyped}
we have defined representations in \emph{monads} only. Indeed we instantiate \lstinline!M! with 
the tautological module later.

\begin{form}[Type of Representations of an Arity, \autoref{def:rep_1--arity_untyped}]
 \begin{lstlisting}
Variable P : RMonad Delta.
Definition modhom_from_arity (M : RModule P Ord) (ar : [nat]) : Type := RModule_Hom (prod_mod M ar) M.
\end{lstlisting}
\end{form}

\noindent
For the rest of the section, we suppose a signature \lstinline!S! to be given via a \textsf{Coq} section variable, 
\lstinline!Variable S : Signature.!
As just mentioned, representing the signature \lstinline!S! in a monad \lstinline!P! (cf.\ \autoref{def:rep_in_rmonad}) means
providing a suitable module morphism for any arity of \lstinline!S!, i.e.\ providing, for any 
element of the indexing set \lstinline!sig_index S!, a term of type \lstinline!modhom_from_arity P (sig i)!:

\begin{form}[Representation of 1--Signature, \autoref{def:rep_in_rmonad}]\label{code:rep_in_rmonad}
 \begin{lstlisting}
Definition Repr (P : RMonad Delta) := 
  forall i : sig_index S, modhom_from_arity P (sig i).
Record Representation := {
  rep_monad :> RMonad Delta ;
  repr : Repr rep_monad }.
 \end{lstlisting}
\end{form}

\noindent
The projecton \lstinline!rep_monad! is declared as a \emph{coercion} by using the special syntax \lstinline!:>!.
This coercion allows for abuse of notation in \textsf{Coq} as we do informally according to \autoref{def:rep_in_rmonad}.
See the first paragraph of \autoref{subsec:ineq_formal} for a use of this abuse.

\section{Morphisms of Representations}

A morphism of representations from $P$ to $Q$ ist given by a monad morphism $f : P \to Q$
between the underlying monads such that a diagram commutes for any arity, cf.\ \autoref{def:prop_mor_of_reps}.
The main task in the implementation is to define this diagram for a given arity $\ell$, and, more
specifically, the left vertical morphism 
\[\dom(\ell,f) = f ^ {\ell} : P^{\ell} \to f^* Q^{\ell} \enspace . \] 
using the notation of \autoref{rem:syn_sem_untyped_two_half_arities}.
Since $P^{\ell}$ is defined as an inductive type, it makes sense to define $f^{\ell}$ by recursion on the inductive type underlying $P^{\ell}$,
named \lstinline!prod_mod_c P V l! (cf.\ \autoref{code:prop_prod_mod}):
\begin{form}[Carrier of Domain Module Morphism of \autoref{def:prop_mor_of_reps}]
 \begin{lstlisting}
Variables P Q : RMonad Delta.
Variable f : RMonad_Hom P Q.
Fixpoint Prod_mor_c (l : [nat]) (V : TYPE) (X : prod_mod_c (fun V => P V) V l) :
                   (prod_mod_c _ V l) :=
  match X in prod_mod_c _ _ l
  return f* (prod_mod Q l) V with
  | TTT => TTT _ _
  | elem -:- elems => f _ elem -:- Prod_mor_c elems
  end.
 \end{lstlisting}
\end{form}

\noindent
Proving this map monotone is a simple exercise, as well as its commutation property with substitution,
yielding the aforementioned module morphism.
Now we have all the ingredients we need in order to define the diagram of 
\autoref{def:prop_mor_of_reps}. For an arity $a$ the diagram reads as follows:

\begin{form}[Commutative Diagram of \autoref{def:prop_mor_of_reps}]
 \begin{lstlisting}
Variable a : [nat].
Variable RepP : modhom_from_arity P a.
Variable RepQ : modhom_from_arity Q a.
Notation "f * M" := (# (PbRMOD f _ ) M).
Definition commute := Prod_mor a ;; f * RepQ == RepP ;; f^.
 \end{lstlisting}
\end{form}

\noindent
Here \lstinline!f^! denotes the module morphism induced by a monad morphism, cf.\ \autoref{def:ind_rmod_mor}.
Using the preceding definition, we define morphisms of representations of \lstinline!S!:

\begin{form}[Morphism of Representations, \autoref{def:prop_mor_of_reps}]
 \begin{lstlisting}
Variables P Q : Representation.
Class Representation_Hom_struct (f : RMonad_Hom P Q) :=
   repr_hom_s : forall i : sig_index S,
            commute f (repr P i) (repr Q i).
Record Representation_Hom : Type := {
  repr_hom_c :> RMonad_Hom P Q;
  repr_hom :> Representation_Hom_struct repr_hom_c }.
 \end{lstlisting}
\end{form}

\section{Category of Representations}

In this section we describe in more detail the category of representations of a 1--signature, cf.\ \autoref{def:cat_of_reps_relmons}.
The composition of morpisms of representations $f : P \to Q$ and $g : Q \to R$ is essentially done by composing the underlying 
monad morphisms. One has to show that this morphism does indeed commute with the representation morphisms of $P$ and $R$.
Similarly, the identity monad morphism of (the monad underlying) a representation $P$ yields a morphism of representations.
Fed with some suitable lemma, the \lstinline!Program! framework does the job for us:

\begin{form}[Composition and Identity of Representations]
 \begin{lstlisting}
Variables P Q R : Representation S.
Variable f : Representation_Hom P Q.
Variable g : Representation_Hom Q R.
Program Instance Rep_comp_struct : 
      Representation_Hom_struct (RMonad_comp f g).
Program Instance Rep_Id_struct : Representation_Hom_struct (RMonad_id P).
 \end{lstlisting}
\end{form}

\noindent
Since equality of morphisms of representations is defined as equality of the underlying monad 
morphisms, the categorical properties of compositition are established already as part of the 
definition of the category \lstinline!RMONAD F! for any functor \lstinline!F!.
\begin{form}[Category of Representations, \autoref{def:cat_of_reps_relmons}]\label{code:prop_cat_of_reps}
 \begin{lstlisting}
Program Instance REP_struct : Cat_struct (@Representation_Hom S) := {
  mor_oid a c := eq_Rep_oid a c;
  id a := Rep_Id a;
  comp P Q R f g := Rep_Comp f g }.
Definition REP := Build_Cat REP_struct.
 \end{lstlisting}
\end{form}

\section{Initiality without Inequations}

We construct the initial object of the category \lstinline!REP! (cf.\ \autoref{code:prop_cat_of_reps}).
In the informal proof of \autoref{lem:init_no_eqs_untyped} this initial object is the image under a left adjoint of the initial object 
in a category of representations as defined in \autoref{sec:sts_ju} with the set of object sorts $T = \{*\}$.
For the formal proof we decide to implement the initial object of \lstinline!REP! directly, 
in order to obtain a compact formalization.
However, the initial object %
is constructed in a way similar to
that of \autoref{chap:sts_formal}. 
The carrier of the initial representation is just a simplified --- because untyped --- version of \autoref{code:sts_initial_term}.
The only significant difference to \autoref{chap:sts_formal} is that we equip the set of terms with the 
trivial diagonal preorder by applying the functor $\Delta$, in \textsf{Coq} called \lstinline!Delta!:
\begin{form}
 \begin{lstlisting}
Inductive UTS (V : TYPE) : TYPE :=
  | Var : V -> UTS V
  | Build : forall (i : sig_index S), UTS_list V (sig i) -> UTS V
with
UTS_list (V : TYPE) : [nat] -> Type :=
  | TT : UTS_list V nil
  | constr : forall b bs,
      UTS (V ** b) -> UTS_list V bs -> UTS_list V (b::bs).
Notation "a -::- b" := (constr a b).
Definition UTS_sm V := Delta (UTS V).
 \end{lstlisting}
\end{form}

\noindent
We define renaming and, built on top of renaming, substitution:

\begin{lstlisting}
Fixpoint rename (V W: TYPE ) (f : V ---> W) (v : UTS V):=
    match v in UTS _ return UTS W with
    | Var v => Var (f v)
    | Build i l => Build (l //-- f)
    end
with
  list_rename V t (l : UTS_list V t) W (f : V ---> W) : UTS_list W t :=
     match l in UTS_list _ t return UTS_list W t with
     | TT => TT W
     | constr b bs elem elems => elem //- f ^^ b -::- elems //-- f
     end
where "x //- f" := (rename f x)
and "x //-- f" := (list_rename x f).
Fixpoint subst (V W : TYPE) (f : V ---> UTS W) (v : UTS V) :
  UTS W := match v in UTS _ return UTS _ with
    | Var v => f v
    | Build i l => Build (l >>== f)
    end
with
  list_subst V W t (l : UTS_list V t) (f : V ---> UTS W) : UTS_list W t :=
     match l in UTS_list _ t return UTS_list W t with
     | TT => TT W
     | elem -::- elems =>
       elem >== _lshift f -::- elems >>== f
     end
where "x >== f" := (subst f x)
and "x >>== f" := (list_subst x f).
\end{lstlisting}

\noindent
Accordingly, the definition of a monadic structure on $V \mapsto \Delta\UTS(V)$ differs from the monad
\lstinline!STS_monad! of \autoref{section:sts_initial} only in the occasional use of the functor 
$\Delta$ (\lstinline!Delta!) on the morphisms --- 
corresponding to the definition of the left adjoint for \autoref{lem:adj_mon_rmon}: %, e.g.\ the variable--as--term constructor:
\begin{form}[Relative Monad Freely Generated by 1--Signature]
 \begin{lstlisting}
Program Instance UTS_sm_rmonad : RMonad_struct Delta UTS_sm := {
  rweta c := #Delta (@Var c);
  rkleisli a b f := #Delta (subst f) }.
Canonical Structure UTSM := Build_RMonad UTS_sm_rmonad.
 \end{lstlisting}
\end{form}

\noindent
The monad \lstinline!UTSM! is easily equipped with a representation of the signature \lstinline!S!;
the carrier of the representation of \lstinline!i : sig_index S! is given by 
the function 
\begin{lstlisting}
fun (X : prod_mod_c _ V (sig i)) => Build (i:=i) (UTSl_f_pm (V:=V) X)
\end{lstlisting}
that is, by the constructor \lstinline!Build i! of the type \lstinline!UTS!, precomposed
with an isomorphism \lstinline!UTSl_f_pm! from \lstinline!prod_mod_c UTS! to \lstinline!UTS_list!.
We thus obtain a representation \lstinline!UTSRepr! of the signature \lstinline!S!.

Given another representation, say, \lstinline!R!, of \lstinline!S!, the morphism \lstinline!init!
from \lstinline!UTSRepr! to \lstinline!R!
is defined by recursion:
\begin{lstlisting}
Fixpoint init V (v : UTS V) : R V :=
        match v in UTS _ return R V with
        | Var v => rweta (RMonad_struct := R) V v
        | Build i X => repr R i V (init_list X)
        end
with
 init_list l (V : TYPE) (s : UTS_list V l) : prod_mod R l V :=
    match s in UTS_list _ l return prod_mod R l V with
    | TT => TTT _ _
    | elem -::- elems => init elem -:- init_list elems
    end.
\end{lstlisting}
This map \lstinline!init! is compatible with lifting and substitution in \lstinline!UTSM! and \lstinline!R!, respectively:
\begin{lstlisting}
Lemma init_lift V x W (f : V ---> W) :
   init (x //- f) = rlift R f (init x).
Lemma init_kleisli V (v : UTS V) W (f : Delta V ---> UTS_sm W) :
  init (v >== f) = rkleisli (f ;; @init_sm W) (init v).
\end{lstlisting}
where \lstinline!init_sm W ! is the (trivially) monotone version of \lstinline!init W! --- 
 the adjunct of \lstinline!init W ! under the adjunction of \autoref{lem:adj_set_po}.
The latter of those lemmas constitutes an important part of the proof that \lstinline!init! is the carrier of a module 
morphism from \lstinline!UTSM! to \lstinline!R!.
It is trivial to prove that \lstinline!init! is also compatible with the representation structure of \lstinline!UTSRepr! and \lstinline!R!,
thus the carrier of a morphism of representations called \lstinline!init_rep : UTSRepr ---> R!.
Afterwards uniqueness of \lstinline!init_rep! is proved:

\begin{lstlisting}
Lemma init_unique :forall f : UTSRepr ---> R , f == init_rep.
\end{lstlisting}
\noindent
Finally we establish initiality by an instance declaration of the corresponding class:
\begin{lstlisting}
Program Instance UTS_initial : Initial (REP S) := {
  Init := UTSRepr ;
  InitMor R := init_rep R }.
\end{lstlisting}

\section{Inequations and Initial Representation of a 2--Signature}\label{sec:ineq_formal}\label{subsec:ineq_formal}

For a 1--signature $S$, an $S$-module is defined to be a functor from representations of $S$ to the category 
whose objects are pairs of a monad $P$ and a module $M$ over $P$, cf.\ \autoref{def:half_eq_untyped}. 
We do not need the functor properties, and use dependent types instead of the cumbersome category of pairs,
in order to ensure that a representation in a monad $P$ is mapped to a $P$--module.

The below definition makes use of two \emph{coercions}. Firstly, we may write $a:\C$ because the 
``object'' projection of the category record (cf.\ \autoref{lst:cat}) is declared as a coercion.
Secondly, the monad underlying any representation can be accessed without explicit projection
using the coercion in \autoref{code:rep_in_rmonad} we mentioned above.

\begin{lstlisting}
Record S_Module := {
  s_mod :> forall R : REP S, RMOD R wOrd ;
  s_mod_hom :> forall (R T : REP S)(f : R ---> T),
         s_mod R ---> PbRMod f (s_mod T)  }.
Notation "U @ f" := (s_mod_hom U f)(at level 4). 
\end{lstlisting}
Note that we write \lstinline!U@f! for the image of the morphism of representations \lstinline!f! under
the $S$--module \lstinline!U!. Source and target module of \lstinline!f! are implicit arguments in this application.

A half-equation is a natural transformation between $S$-modules. We need the naturality condition in the following.
Since we have not formalized $S$-modules as functors, we have to state naturality explicitly:

\begin{form}[Half--Equation, \autoref{def:half_eq_untyped}]
\label{code:half_equation_untyped}
\begin{lstlisting}
Class half_equation_struct (U V : S_Module) 
    (half_eq : forall R : REP S, U R ---> V R) := {
  comm_eq_s : forall (R T : REP S)  (f : R ---> T), 
      U @ f ;; PbRMod_Hom _ (half_eq T) ==  half_eq R ;; V @ f }.
Record half_equation (U V : S_Module) := {
  half_eq :> forall R : REP S, U R ---> V R ;
  half_eq_s :> half_equation_struct half_eq }.
\end{lstlisting}
\end{form}

\noindent
We now formalize \emph{classic} $S$--modules.
Any list of natural numbers uniquely specifies a classic $S$--module, cf.\ \autoref{def:alg_s_mod}.
Given a list of naturals \lstinline!codl!, we call this $S$--module \lstinline!S_Mod_classic codl!.
A \emph{classic half--equation} is any half--equation with a classic co\-do\-main, 
and a classic inequation is a pair of parallel classic half--equations
(cf.\ \autoref{def:classic_ineq_untyped}):

\begin{lstlisting}
Definition half_eq_classic (U : S_Module)(codl : [nat]) := 
                               half_equation U (S_Mod_classic codl).
Record ineq_classic := {
  Dom : S_Module ;
  Cod : [nat] ;
  eq1 : half_eq_classic Dom Cod ;
  eq2 : half_eq_classic Dom Cod }.
\end{lstlisting}

\noindent
Give a representation \lstinline!P! and a (classic) inequation \lstinline!e!, 
we check whether \lstinline!P! satisfies  \lstinline!e! by pointwise comparison (cf.\ \autoref{def:rep_ineq_untyped}):

\begin{lstlisting}
Definition satisfies_ineq (e : ineq_classic) (P : REP S) :=
  forall c (x : Dom e P c),
        eq1 _ _ _ x << eq2 _ _ _ x.
(* for a family of inequations indexed by a set A *)
Definition Inequations (A : Type) := A -> ineq_classic.
Definition satisfies_ineqs A (T : Inequations A) (R : REP S) :=
      forall a, satisfies_ineq (T a) R.
\end{lstlisting}

\noindent
We formalize sets of classic inequations as pairs of an indexing type \lstinline!A! together with a term of type \lstinline!Inequations A!,
that is, a map from \lstinline!A! 
to the type of classic inequations \lstinline!ineq_classic!.
The category of representations of $(S,A)$ is obtained as a full subcategory of the category of representations of $S$.
The following declaration produces a subcategory from predicates on the type of representations and 
on the (dependent) type of morphisms of representations, yielding the category \lstinline!PROP_REP! of representations of
$(S,A)$:

\begin{lstlisting}
Variable A : Type.
Variable T : Inequations A.
Program Instance Ineq_Rep : SubCat_compat (REP S)
     (fun P => satisfies_ineqs T P) (fun a b f => True).
Definition INEQ_REP : Cat := SubCat Ineq_Rep.
\end{lstlisting}

\noindent
We now construct the initial object of \lstinline!INEQ_REP!.
The relation on the initial object is defined precisely as in the paper proof, cf.\ \autoref{eq:order_untyped}:
\begin{lstlisting}
Definition prop_rel_c X (x y : UTS S X) : Prop :=
      forall R : PROP_REP, init (FINJ _ R) x << init (FINJ _ R) y.
\end{lstlisting}
Here, \lstinline!FINJ _ R! denotes the representation \lstinline!R! as a representation of \lstinline!S!,
i.e.\ the injection of \lstinline!R! in the category \lstinline!REP S! of representations of \lstinline!S!.
The relation defined above is indeed a preorder, and 
we define the monad \lstinline!UTSP! to be the monad whose underlying sets are identical 
to \lstinline!UTSM!, namely the sets defined by \lstinline!UTS!, but equipped with this new preorder.
This monad \lstinline!UTSP! is denoted by $\Sigma_A$ in the paper proof.

The representation module morphisms of the initial representation \lstinline!UTSRepr! can be ``reused''
after having proved their compatibility with the new order, yielding a representation \lstinline!UTSProp!.
An important lemma states that this representation satisfies the inequations of \lstinline!T!:
\begin{lstlisting}
Lemma UTSPRepr_sig_prop : satisfies_ineqs T UTSProp.
\end{lstlisting}
We have to explicitly inject the representation into the category of representations of $(S,A)$:
\begin{lstlisting}
Definition UTSPROP : INEQ_REP :=
 exist (fun R : Representation S => satisfies_ineqs T R) UTSProp
  UTSPRepr_sig_prop.
\end{lstlisting}

\noindent
For building the initial morphism towards any representation \lstinline!R : INEQ_REP!, 
we first build the corresponding morphism in the category of representations of $S$:
\begin{lstlisting}
Definition init_prop_re : UTSPropr ---> (FINJ _ R) := ...
\end{lstlisting}
which we then inject, analogously to the initial representation, into the subcategory of 
representations of $(S,A)$:
\begin{lstlisting}
Definition init_prop : UTSPROP ---> R := exist _ (init_prop_re R) I.
\end{lstlisting}

\noindent
Finally we prove \autoref{thm:init_w_ineq_untyped}: 
An initial object of a category is given by
an object \lstinline!Init! of this category, a map associating go any object \lstinline!R!
a morphism \lstinline!InitMor R : Init ---> R!, and a proof of uniqueness of any such morphism.
We instanciate the type class \lstinline!Initial! for the category \lstinline!INEQ_REP!
of representations of $(S,A)$:

\begin{lstlisting}
Program Instance INITIAL_INEQ_REP : Initial INEQ_REP := {
  Init := UTSPROP ;
  InitMor := init_prop ;
  InitMorUnique := init_prop_unique }.
\end{lstlisting}

\noindent
We check its type after closing all the sections --- and thus abstracting from the section variables:
\begin{lstlisting}
Check INITIAL_INEQ_REP.
INITIAL_INEQ_REP
     : forall (S : Signature) (A : Type) (T : Inequations S A),
       Initial (INEQ_REP (S:=S) (A:=A) T)
\end{lstlisting}

\section{\texorpdfstring{$\Lambda\beta$}{Lambdabeta}: Lambda Calculus with \texorpdfstring{beta}{Beta} reduction}

We implement the example 2--signature $\Lambda\beta$, cf.\ \autoref{ex:2--sig_ulcbeta}.
Throughout this section, we use use a custom notation in \textsf{Coq} for the datatype of lists:
\begin{lstlisting}
Notation "[[ x ; .. ; y ]]" := (cons x .. (cons y nil) ..). 
\end{lstlisting}

% \subsection{The 1--Signature \texorpdfstring{$\Lambda$}{Lambda}}

In order to specify the 1--signature $\Lambda$ (cf.\ \autoref{def:1--signature_untyped}, \autoref{ex:sig_ulc_syn}),
we first 
define an indexing set \lstinline!Lambda_index! consisting of two elements, \lstinline!ABS! and \lstinline!APP!.
This indexing set reflects the fact that the signature $\Lambda$ consists of two arities.
The record instance \lstinline!Lambda! is a term of type \lstinline!Signature! (cf.\ \autoref{code:1--signature_untyped}).
The map \lstinline!sig Lambda! then associates the corresponding lists of naturals to each of these elements, according to 
 \autoref{ex:sig_ulc_syn}:
\begin{lstlisting}
Inductive Lambda_index := ABS | APP.
Definition Lambda : Signature := {|
  sig_index := Lambda_index ;
  sig := fun x => match x with 
                  | ABS => [[ 1 ]] 
                  | APP => [[ 0 ; 0]]
                  end |}.
\end{lstlisting}

% \subsection{The \texorpdfstring{$\Lambda$}{Lambda}--Inequation \texorpdfstring{$\beta$}{beta}}

The definition of the inequation $\beta$ (cf.\ \autoref{ex:sig_ulc_prop}) is a more challenging task, since a half--equation is not 
just an element of a simple datatype like a 1--arity, but given by suitable module morphisms.

At first, we define the substitution of \emph{one} variable (cf.\ \autoref{def:subst_half_eq}) as a half--equation.
The carrier \lstinline!subst_carrier! of the substitution is defined as in \autoref{def:hat_P_subst}.
Afterwards we prove that this carrier satisfies the properties of a module morphism, that is, is compatible with 
substitution in the source and target modules. After abstracting from the section variable \lstinline!R!,  
we obtain a function \lstinline!subst_module_mor!
which, given any representation \lstinline!R! of \lstinline!S!, yields the
substitution module morphism associated to (the monad underlying) \lstinline!R!.

\begin{lstlisting}
Variable S : Signature.
Variable R : REP S.
Definition subst_carrier :
   (forall c : TYPE, (S_Mod_classic_ob [[1; 0]] R) c ---> 
         (S_Mod_classic_ob [[0]] R) c) := ...
Program Instance sub_struct : RModule_Hom_struct
      (M:=S_Mod_classic_ob [[1 ; 0]] R) 
      (N:=S_Mod_classic_ob [[0]] R) 
   subst_carrier.
Definition subst_module_mor := Build_RModule_Hom (sub_struct R).
\end{lstlisting}

\noindent
The last step is to prove ``naturality'', that is, the commutativity of the family of diagrams of \autoref{code:half_equation_untyped}.
We recall that we do not implement $S$--modules as functors, but just as the data part of functors. 
This is why we put the word \emph{naturality} in quotes. After the proof we define our first half--equation, \lstinline!subst_half_eq!.

\begin{lstlisting}
Program Instance subst_half_s : half_equation_struct 
      (U:= S_Mod_classic [[1 ; 0]]) 
      (V:= S_Mod_classic [[0]]) 
   subst_module_mor.
Definition subst_half_eq := Build_half_equation subst_half_s.
\end{lstlisting}

The definition of the second half--equation of \autoref{ex:app_circ_half} is possible for any 1--signature with abstraction and application, such as 
the 1--signature $\Lambda$. To keep the example simple, we only define the half--equation for $\Lambda$.
The needed steps are precisely the same as for the substitution half--equation, so we just give the statements.
\begin{lstlisting}
Definition beta_carrier :
   (forall c : TYPE, (S_Mod_classic_ob [[1; 0]] R) c ---> 
         (S_Mod_classic_ob [[0]] R) c) := ...
Program Instance beta_struct : RModule_Hom_struct
      (M:=S_Mod_classic_ob [[1 ; 0]] R) 
      (N:=S_Mod_classic_ob [[0]] R) 
   beta_carrier.
Definition beta_module_mor := Build_RModule_Hom beta_struct.
Program Instance beta_half_s : half_equation_struct
      (U:=S_Mod_classic Lambda [[1 ; 0]])
      (V:=S_Mod_classic Lambda [[0]])
   beta_module_mor.
Definition beta_half_eq := Build_half_equation beta_half_s.
\end{lstlisting}

\noindent
In the end we package both half--equations into one inequation specifying the beta rule of \autoref{ex:sig_ulc_prop}.

\begin{lstlisting}
Definition beta_rule : ineq_classic Lambda := {|
   eq1 := beta_half_eq ;
   eq2 := subst_half_eq Lambda |}.
\end{lstlisting}

\noindent
We can now associate a short name to the category of representations of $\Lambda\beta$, where, for increased clarity,
we specify the implicit arguments:
\begin{lstlisting}
Definition Lambda_beta_Cat := INEQ_REP
    (S:=Lambda)(A:=unit)(fun x : unit => beta_rule).
\end{lstlisting}
Note that our formal definition allows that an inequation appears multiple times in a 2--signature, whereas in the informal definition we
have \emph{sets} of inequations.
Unlike for arities, having several copies of the same inequation 
does not change the resulting category neither the initial object, of course.
The initial representation is obtained via the specification
\begin{lstlisting}
Definition Lambda_beta := @Init _ _ _ 
        (INITIAL_INEQ_REP (fun x : unit => beta_rule)).
\end{lstlisting}

%% file: comp_sem_monads_formal.tex
In this chapter we describe the implementation of the category of representations of \PCF, equipped with 
reduction rules --- we refer to it as \emph{semantic \PCF} from now on --- 
as described informally in \autoref{subsec:semantic_ulc_pcf}. 
We state the reduction rules more precisely later.
This theorem is an instance of \autoref{thm:init_w_ineq_typed} proved in \autoref{chap:comp_types_sem}. 
However, for the implementation in \textsf{Coq} of this  instance we make several simplifications compared to the general theorem:

\begin{itemize}
 \item we do not define a notion of 2--signature, but specify directly a \textsf{Coq} type of
       representations of semantic \PCF ;
 \item we use dependent \textsf{Coq} types to formalize arities of higher degree (cf.\ \autoref{def:half_arity_degree_semantic_typed}),
        instead of relying on modules on categories with pointed index sets. A representation of an arity of degree $n$
        is thus given by a family of module morphisms (of degree zero), indexed $n$ times over the respective object type as 
        described in \autoref{rem:family_of_mods_cong_pointed_mod_relative};
 \item the relation on the initial object is not defined via the formula of \autoref{eq:order_typed}, but directly 
        through an inductive type, cf.\ \autoref{code:pcf_eval}, and various closures, cf.\ \autoref{code:pcf_propag}.
\end{itemize}

\section{Representations of \texorpdfstring{\PCF}{PCF}}

In this section we explain the formalization of representations of semantic \PCF.
According to \autoref{def:1--rep_typed} and \autoref{def:2--rep_typed}, such a representation
consists of 
\begin{packenum}
   \item a representation of the types of \PCF~(in a \textsf{Coq} type \lstinline!U!), cf.\ \autoref{ex:type_PCF},
   \item a relative monad \lstinline!P! over the functor $\family{\Delta}{U}$ (in the formalization: \lstinline!IDelta U!) and
   \item representations of the arities of \PCF~(cf.\ \autoref{ex:term_sig_pcf}), i.e. morphisms of $P$--modules with suitable source and target modules such that
   \item the inequations defining the reduction rules of \PCF~are satisfied. \label{list:last_item}
\end{packenum}

\noindent
A representation of \PCF~should be a ``bundle'', i.e.\ a record type, whose components --- or ``fields'' --- 
are these \ref{list:last_item} items. % should be bundled into a record type
In order to ease the definitions, we first define what a representation of the term signature of \PCF~in a monad $P$ is, 
 in the presence of an $S_\PCF$--monad (cf.\ \autoref{def:s-rmon}).
Unfolding the definitions, we suppose given a type \lstinline!Sorts!, a relative monad \lstinline!P! over \lstinline!IDelta Sorts!
and three operations on \lstinline!Sorts!: a binary function \lstinline!Arrow! --- denoted by an infixed ``\lstinline!~~>!''
--- and two constants \lstinline!Bool! and \lstinline!Nat!.
\begin{lstlisting}
Variable Sorts : Type.
Variable P : RMonad (IDelta Sorts).
Variable Arrow : Sorts -> Sorts -> Sorts.
Variable Bool : Sorts.
Variable Nat : Sorts.
Notation "a ~~> b" := (Arrow a b) (at level 60, right associativity).
\end{lstlisting}
In this context, a representation of \PCF~is given by a bunch of module morphisms satisfying some conditions.
We split the definition into smaller pieces.
Note that \lstinline!M[t]! denotes the fibre module of module \lstinline!M! with respect to \lstinline!t!, 
and \lstinline!d M // u! denotes derivation of module \lstinline!M! with respect to \lstinline!u!.
The module denoted by a star \lstinline!*! is the terminal module, which is the constant singleton module.
\begin{form}[1--Signature of $\PCF$]\label{code:rpcf_1-sig}
\begin{lstlisting}
Class PCFPO_rep_struct := {
  app : forall u v, (P[u ~~> v]) x (P[u]) ---> P[v];
  abs : forall u v, (d P // u)[v] ---> P[u ~~> v];
  rec : forall t, P[t ~~> t] ---> P[t];
  tttt : * ---> P[Bool];
  ffff : * ---> P[Bool];
  nats : forall m:nat, * ---> P[Nat];
  Succ : * ---> P[Nat ~~> Nat];
  Pred : * ---> P[Nat ~~> Nat];
  Zero : * ---> P[Nat ~~> Bool];
  CondN: * ---> P[Bool ~~> Nat ~~> Nat ~~> Nat];
  CondB: * ---> P[Bool ~~> Bool ~~> Bool ~~> Bool];
  bottom: forall t, * ---> P[t];
  ...
\end{lstlisting}
\end{form}

\noindent
These module morphisms are subject to some inequations specifying the reduction rules of \autoref{subsec:semantic_ulc_pcf},
 or, equivalently, \autoref{ex:pcf_ineqs}.
 The beta rule reads as
\begin{form}[Beta Rule for Representations of $\PCF$] \label{code:rpcf_beta}
\begin{lstlisting}  
  beta_red : forall r s V y z, app r s V (abs r s V y, z) << y[*:= z] ;
  ...
\end{lstlisting}
\end{form}

\noindent
where \lstinline!y[*:= z]! is the substitution of the freshest variable (cf.\ \autoref{def:hat_P_subst_typed}) 
as a special case of simultaneous monadic substitution. 
The rule for the fixed point operator says that $\mathbf{Y}(f) \leadsto f\left(\mathbf{Y}(f)\right)$:
\begin{form}[Inequation for Fixedpoint Operator]\label{code:rpcf_rec}
\begin{lstlisting}
  Rec_A: forall V t g, rec _ _ g << app t t V (g, rec _ _ g) ;
  ...
\end{lstlisting}
\end{form}

\noindent
The other inequations concern the arithmetic and logical constants of \PCF.
Firstly, we have that the conditionals reduce according to the truth value
they are applied to:
\begin{form}[Logic Inequations of $\PCF$ Representations]\label{code:rpcf_cond}
\begin{lstlisting}
  CondN_t: forall V n m,
     app _ _ _ (app _ _ _
          (app _ _ _ (CondN V tt, tttt _ tt), n), m) << n ;
  CondN_f: forall V n m,
     app _ _ _ (app _ _ _
          (app _ _ _ (CondN V tt, ffff _ tt), n), m) << m ;
  CondB_t: forall V u v,
     app _ _ _ (app _ _ _
          (app _ _ _ (CondB V tt, tttt _ tt), u), v) << u ;
  CondB_f: forall V u v,
     app _ _ _ (app _ _ _
          (app _ _ _ (CondB V tt, ffff _ tt), u), v) << v ;
   ...
\end{lstlisting}
\end{form}

\noindent
Furthermore, we have that $\Succ(n)$ reduces to $n+1$ (which in \textsf{Coq} is written \lstinline!S n!), 
reduction of the $\zeroqu$ predicate
 according to whether its argument is zero or not, and that the predecessor  is 
 post--inverse to the successor function:
\begin{form}[Arithmetic Inequations of $\PCF$ Representations] \label{code:rpcf_arith}
\begin{lstlisting}
  Succ_red: forall V n,
     app _ _ _(Succ V tt, nats n _ tt) << nats (S n) _ tt ;
  Zero_t: forall V,
     app _ _ _(Zero V tt, nats 0 _ tt) << tttt _ tt ;
  Zero_f: forall V n,
     app _ _ _(Zero V tt, nats (S n) _ tt) << ffff _ tt ;
  Pred_Succ: forall V n,
     app _ _ _(Pred V tt, app _ _ _ (Succ V tt, nats n _ tt)) << nats n _ tt;
  Pred_Z: forall V,
     app _ _ _(Pred V tt, nats 0 _ tt) << nats 0 _ tt  }.
\end{lstlisting}
\end{form}

\noindent
Unfortunately, at this stage of the definition, we were not able to introduce a more convenient notation for application,
neither to omit the arguments denoted by an underscore as instances of implicit arguments.
After abstracting over the section variables we package all of this into a record type:

\begin{lstlisting}
Record PCFPO_rep := {
  Sorts : Type;
  Arrow : Sorts -> Sorts -> Sorts;
  Bool : Sorts ;
  Nat : Sorts ;
  pcf_rep_monad :> RMonad (IDelta Sorts);
  pcf_rep_struct :> PCFPO_rep_struct pcf_rep_monad  Arrow Bool Nat }.
Notation "a ~~> b" := (Arrow a b) (at level 60, right associativity).
\end{lstlisting}
The type \lstinline!PCFPO_rep! later constitutes the type of objects of the category of representations of semantic \PCF.

\section{Morphisms of Representations}

A morphism of representations (cf.\ \autoref{def:rel_mor_of_reps_typed}) is built from a morphism $g$ of \emph{type} representations
and a colax monad morphism over the retyping functor associated to the map $g$. 
The implementation of retyping is explained in \autoref{code:retype}.
In the particular case of $\PCF$, a morphism of representations from $P$ to $R$ consists of a 
morphism of representations of the types of \PCF~--- with underlying map \lstinline!Sorts_map! --- and
a colax morphism of relative monads which makes commute the diagrams of the form given in \autoref{def:rel_mor_of_reps_typed}.
We first define the diagrams we expect to commute, before packaging everything into a record type of morphisms.
The context is given by the following declarations: 
\begin{lstlisting}
Variables P R : PCFPO_rep.
Variable Sorts_map : Sorts P -> Sorts R.
Hypothesis HArrow : forall u v, Sorts_map (u ~~> v) = Sorts_map u ~~> Sorts_map v.
Hypothesis HBool : Sorts_map (Bool _ ) = Bool _ .
Hypothesis HNat : Sorts_map (Nat _ ) = Nat _ .

Variable f : colax_RMonad_Hom P R
    (G1:=RETYPE (fun t => Sorts_map t))
    (G2:=RETYPE_PO (fun t => Sorts_map t))
  (RT_NT (fun t => Sorts_map t)).
\end{lstlisting}

\noindent
We explain the commutative diagrams of \autoref{def:rel_mor_of_reps_typed} for some of the arities. 
For the successor arity we ask the following diagram to commute:
\begin{form}[Commutative Diagram for Successor Arity]
\label{code:pcf_succ_diag}
\begin{lstlisting}
Program Definition Succ_hom' :=
  Succ ;; f [(Nat ~~> Nat)] ;; Fib_eq_RMod _ _ ;; IsoPF
           == 
  *--->* ;; f ** Succ.
\end{lstlisting}
\end{form}

\noindent
Here the morphism \lstinline!Succ! refers to the representation of the successor arity either of \lstinline!P! (the
first appearance) or \lstinline!R! (the second appearance) --- \textsf{Coq} is able to figure this out itself.
The domain of the successor is given by the terminal module $*$. Accordingly, we have that $\dom(\SUCC, f)$ is 
the trivial module morphism with domain and codomain given by the terminal module. We denote this module morphism by \lstinline!*--->*!.
The codomain is given as the fibre of $f$ of type $\Nat \PCFar \Nat$.
The two remaining module morphisms are isomorphisms which do not appear in the informal description.
The isomorphism \lstinline!IsoPF! is needed to permute fibre with pullback (cf.\ \autoref{lem:rel_pb_fibre}).
The morphism \lstinline!Fib_eq_RMod M H! takes a module \lstinline!M! and a proof \lstinline!H! of equality of two object types as 
arguments, say, \lstinline!H : u = v!. Its output is an isomorphism \lstinline! M[u] ---> M[v]!. Here the proof is of type
 \begin{lstlisting}
Sorts_map (Nat ~~> Nat) = Sorts_map Nat ~~> Sorts_map Nat 
 \end{lstlisting}
and \textsf{Coq} is able to figure out the proof itself. % by specialising the hypothesis \lstinline!HArrow!. 
We expand on this kind of modules in \autoref{sec:digression_coq_equality}
The diagram for application uses the product of module morphisms, denoted by an infixed \lstinline!X!:
\begin{form}[Commutative Diagram for Application Arity]
\label{code:pcf_app_diag}
\begin{lstlisting}
Program Definition app_hom' := forall u v,
    app u v;; f [( _ )] ;; IsoPF 
          ==
    (f [(u ~~> v)] ;; Fib_eq_RMod _ (HArrow _ _);; IsoPF ) 
         X
    (f [(u)] ;; IsoPF ) ;; 
       IsoXP ;; f ** (app _ _ ).
\end{lstlisting}
\end{form}

\noindent
In addition to the already encountered isomorphism \lstinline!IsoPF! we have to insert an isomorphism \lstinline!IsoXP! 
which permutes pullback and product (cf.\ \autoref{lem:rel_pb_prod}).
As a last example, we present the property for the abstraction:
\begin{form}[Commutative Diagram for Abstraction Arity]
\label{code:pcf_abs_diag}
\begin{lstlisting}
Program Definition abs_hom' := forall u v,
    abs u v ;; f [( _ )]
          ==
    DerFib_RMod_Hom _ _ _ ;;  IsoPF ;;
    f ** (abs (_ u) (_ v)) ;; IsoFP ;;
    Fib_eq_RMod _ (eq_sym (HArrow _ _ )) .
\end{lstlisting}
\end{form}

\noindent
Here the module morphism \lstinline!DerFib_RMod_Hom f u v! corresponds to the morphism $\dom(\Abs(u,v),f) = \fibre{f^u}{v}$, and 
\lstinline!IsoFP! permutes fibre with pullback, just like its sibling \lstinline!IsoPF!, but the other way round.

We bundle all those properties into a type class:
\begin{lstlisting}
Class PCFPO_rep_Hom_struct := {
  CondB_hom : CondB_hom' ;
  CondN_hom : CondN_hom' ;
  Pred_hom : Pred_hom' ;
  Zero_hom : Zero_hom' ;
  Succ_hom : Succ_hom' ;
  fff_hom : fff_hom' ;
  ttt_hom : ttt_hom' ;
  bottom_hom : bottom_hom' ;
  nats_hom : nats_hom' ;
  app_hom : app_hom' ;
  rec_hom : rec_hom' ;
  abs_hom : abs_hom' }.
\end{lstlisting}

\noindent
Similarly to what we did for representations, we abstract over the section variables and define a record type of 
morphisms of representations from \lstinline!P! to \lstinline!R! :
\begin{lstlisting}
Record PCFPO_rep_Hom := {
  Sorts_map : Sorts P -> Sorts R ;
  HArrow : forall u v, Sorts_map (u ~~> v) = Sorts_map u ~~> Sorts_map v;
  HNat : Sorts_map (Nat _ ) = Nat R ;
  HBool : Sorts_map (Bool _ ) = Bool R ;
  rep_Hom_monad :> colax_RMonad_Hom P R (RT_NT Sorts_map);
  rep_colax_Hom_monad_struct :> PCFPO_rep_Hom_struct
                 HArrow HBool HNat rep_Hom_monad }.
\end{lstlisting}

\section{Digression on Equal Fibre Modules  in \textsf{Coq}} \label{sec:digression_coq_equality}

  Suppose $Q$ is a relative monad on some functor $F : \C\to \D$ and $M$ is a $Q$--module with codomain $\TP{T}$.
 Let $u,t \in T$ and suppose given a proof $H$ of the proposition $u = t$. We can now prove $\fibre{M}{u} = \fibre{M}{t}$,
but unfortunately this is not sufficient for composing a morphism with 
 codomain $\fibre{M}{u}$ with one whose domain is $\fibre{M}{t}$ in \textsf{Coq} (cf.\ \autoref{subsec:dep_hom_types}).
 Indeed, the problem we encounter here is even worse than that of permutation of pullback with fibre, derivation and products 
 (see e.g.\ \autoref{subsec:sts_formal_reps}), since not even
 the carriers of $\fibre{M}{u}$ and $\fibre{M}{t}$ are convertible. This means that the isomorphism we have to insert
  does not even allow for an underlying family of \emph{identity} maps as carriers, but instead is a transport of the 
form \lstinline!eq_rect!.

  In more detail, the carrier of $M$ is a map from the objects of $\C$ to $\TP{T}$, that is, for each $c\in \C$, its
  image $Mc \in \TP{T}$ is basically a dependent type (with some structure). 
  The fibre is then simply computed by application.
  The carrier of a module morphism $\rho:\fibre{M}{u} \to \fibre{M}{t}$ thus consists of a family of maps of sets
   indexed by objects $c\in \C$,
   \[ \rho_c : M(c)(u) \to M(c)(t) \enspace . \]
 In intensional type theory, we have an \emph{explicit} cast operator \lstinline!eq_rect! which allows the definition of precisely 
 such a map:
\begin{lstlisting}
Check eq_rect.
   eq_rect
     : forall (A : Type) (x : A) (P : A -> Type),
       P x -> forall y : A, x = y -> P y
\end{lstlisting}
 Note that this operator is equivalent to the operator $\mathrm{J}$ in Hofmann's PhD thesis \cite{hofmann_phd}, 
   whose typing rule is called \textsc{Id-Elim-J}.

Here we instantiate \lstinline!A! by set of object types $T$ and the dependent type \lstinline!P! by $M(c)$,
 allowing us to define a map \lstinline!transport! from $M(c)(u)$ to $M(c)(t)$:
\begin{lstlisting}
Variable T : Type.
Variables u t : T.
Variable M : RMOD Q (IPO T).
Hypothesis H : u = t.
Definition transport (c : C) : M c u -> M c t :=
   fun (s : M c u) => eq_rect u (fun t : T => (M c) t) s t H.
\end{lstlisting}

\noindent
Fortunately it is possible to get rid of the transport via a computation rule equivalent to a rule named \textsc{Id-Comp} in Hofmann's thesis. 
In \textsf{Coq} this rule says that the term
\begin{lstlisting}
eq_rect u P a a eq_refl 
\end{lstlisting}
reduces to --- and thus in particular is provably equal to --- the term \lstinline!a! itself.
Thus a considerable part of proof code in the following is about elimination of explicit casts. Indeed, the scheme is as follows:
we start with a goal

\begin{lstlisting}
 ______________________________________(1/1)
G
\end{lstlisting}
such that \lstinline!G! contains a subterm \lstinline!eq_rect u P a b H!, i.e.\ with \lstinline!H : a = b!.
We then generalize \lstinline!H!, yielding the goal
\begin{lstlisting}
 ______________________________________(1/1)
forall H : a = b, G
\end{lstlisting}
Now rewriting with a proof of \lstinline!a = b! (using a copy of \lstinline!H!) turns the goal into
\begin{lstlisting}
 ______________________________________(1/1)
forall H : a = a, G
\end{lstlisting}
After introducing \lstinline!H!, we can rewrite \lstinline!H! in the goal into \lstinline!eq_refl! using the 
 axiom \lstinline!UIP_refl! which says that any proof of \lstinline!a = a! is equal to \lstinline!eq_refl!.
 Thus the goal \lstinline!G! contains the subterm \lstinline!eq_rect u P a a eq_refl!,
  which simplifies to \lstinline!a! --- the transport has disappeared.
Note that for the rewrite of \lstinline!forall H : a = b! into \lstinline!forall H : a = a! in the goal,
 many other terms from the context have to be generalized, as well as structures broken into their constituent pieces,
 in order to obtain sufficient flexibility in the goal for the rewrite to result in a well--typed term.

\section{Equality of Morphisms, Category of Representations}

We have already seen how some definitions that are trivial in informal mathematics, turn into something awful 
in intensional type theory. Equality of morphisms of representations is another such definition.
Informally, two such morphisms $a, c : P \to R$ of representations are equal if 
\begin{enumerate}
 \item their map of object types $f_a$ and $f_c$ (\lstinline!Sorts_map!) are equal and
 \item their underlying colax morphism of monads --- also called $a$ and $c$ --- are equal.
\end{enumerate}
In our formalization, the second condition is not even directly expressable, since these monad morphisms 
  do not have the same type: 
 we have, for a context $V \in \TS{P}$,
\[ a_V : \retyping{f_a}(PV) \to R(\retyping{f_a}V) \]
and 
\[ c_V : \retyping{f_c}(PV) \to R(\retyping{f_c}V)  \enspace . \]
where $\TS{P}$ is a notation for contexts typed over the set of object types the representation $P$ comes with, 
formally the type \lstinline!Sorts P!.
We can only compare $a_V$ to $c_V$ by composing each of them with a suitable transport \lstinline!transp! again, yielding morphisms
\[ \comp{a_V}{R(\text{\lstinline!transp!})} : \retyping{f_a}(PV) \to R(\retyping{f_a}V) \to R(\retyping{f_c}V)\]
and
\[ \comp{\text{\lstinline!transp'!}}{c_V}   : \retyping{f_a}(PV) \to \retyping{f_c}(PV) \to R(\retyping{f_c}V)  \enspace . \]
As before, for equal fibres $\fibre{M}{u}$ and $\fibre{M}{t}$ with $u = t$, the carriers of those transports 
\lstinline!transp! and \lstinline!transp'!
are terms of the form \lstinline!eq_rect _ _ _ H!, where \lstinline!H! is a proof term which depends on 
the proof of 
\begin{lstlisting}
forall x : Sorts P, Sorts_map c x = Sorts_map a x
\end{lstlisting}
of the first condition.
Altogether, the definition of equality of morphisms of representations is given by the following inductive proposition:
\begin{lstlisting}
Inductive eq_Rep (P R : PCFPO_rep) : relation (PCFPO_rep_Hom P R) :=
 | eq_rep : forall (a c : PCFPO_rep_Hom P R), 
            forall H : (forall t, Sorts_map c t = Sorts_map a t),
            (forall V, a V ;; rlift R (Transp H V)
                                    == 
                      Transp_ord H (P V) ;; c V ) -> eq_Rep a c.
\end{lstlisting}
The formal proof that the relation thus defined is an equivalence is inadequately long when compared to 
its mathematical complexity, due to the transport elimination.

Composition of representations is done by composing the underlying maps of sorts, as well as composing
the underlying monad morphisms pointwise. Again, this operation, which is trivial from a mathematical point
of view, yields a difficulty in the formalization, due to the fact that in the formalization
\[\retyping{g}(\retyping{f}V) \not\equiv \retyping{(\comp{f}{g})}V \enspace . \]
More precisely, suppose given two morphisms of representations $a : P \to Q$ and $b : Q \to R$,
given by families of morphisms indexed by $V$ resp.\ $W$,
\begin{align*} a_V &: \widetilde{PV}^a \to Q(\widetilde{V}^a)  \quad \text{and} \\
        b_W &: \widetilde{QW}^b \to R(\widetilde{W}^b) \enspace ,
\end{align*}
where we write $\widetilde{V}^a$ for $\retyping{f_a}V$.
The monad morphism underlying the composite morphism of representations is given by the following definition:

\[
 \begin{xy}
    \xymatrix{     \widetilde{PV}^{\comp{a}{b}} \ar[rr]^{\comp{a}{b}_V} \ar[d]_{\text{\lstinline!match!}} & &  R(\widetilde{V}^{\comp{a}{b}}) \\
                        PV  \ar[d]_{\text{\lstinline!ctype!}}                     &  &      R\left(\widetilde{\widetilde{V}^a}^b\right)   \ar[u]_{R(\cong)}  \\
                        \widetilde{PV}^a  \ar[r]_{a_V}        & Q(\widetilde{V}^a) \ar[r]_{\text{\lstinline!ctype!}}& \widetilde{Q(\widetilde{V}^a)}^b \ar[u]_{b_{\widetilde{V}^a}}
       }
 \end{xy}
\]
or, in \textsf{Coq} code,
\begin{lstlisting}
Definition comp_rep_car : (forall c : ITYPE U,
        RETYPE (fun t => f' (f t)) (P c) --->
     R ((RETYPE (fun t => f' (f t))) c)) :=
  fun (V : ITYPE U) t (y : retype (fun t => f' (f t)) (P V) t) =>
    match y with ctype _ z =>
      lift (M:=R) (double_retype_1 (f:=f) (f':=f') (V:=V)) _
          (b _ _ (ctype (fun t => f' t)
               (a _ _ (ctype (fun t => f t) z ))))
     end.
\end{lstlisting}
where \lstinline!double_retype_1! denotes the isomorphism in the upper right corner.
The proof of the commutative diagrams for the composite monad morphism is lengthy due to the number of arities of the signature of \PCF.
Definition of the identity morphisms is routine, and in the end we define the category of representations of semantic \PCF:
\begin{lstlisting}
Program Instance REP_s :
     Cat_struct (obj := PCFPO_rep) (PCFPO_rep_Hom) := {
  mor_oid P R := eq_Rep_oid P R ;
  id R := Rep_id R ;
  comp a b c f g := Rep_comp f g }.
\end{lstlisting}

\section{One Particular Representation}

We define a particular representation, which we later prove to be initial.
First of all, the set of object types of \PCF~is given as follows:
\begin{lstlisting}
Module PCF.
Inductive Sorts := 
  | Nat : Sorts
  | Bool : Sorts
  | Arrow : Sorts -> Sorts -> Sorts.
End PCF.
\end{lstlisting}

\noindent
For this section we introduce some notations:
\begin{lstlisting}
Notation "'TY'" := PCF.Sorts.
Notation "'Bool'" := PCF.Bool.
Notation "'Nat'" := PCF.Nat.
Notation "'IT'" := (ITYPE TY).
Notation "a '~>' b" := (PCF.Arrow a b) (at level 69, right associativity).
\end{lstlisting}

\noindent
We specify the set of \PCF~constants through the following inductive type, indexed by
the sorts of \PCF:
\begin{lstlisting}
Inductive Consts : TY -> Type :=
 | Nats : nat -> Consts Nat
 | ttt : Consts Bool
 | fff : Consts Bool
 | succ : Consts (Nat ~> Nat)
 | preds : Consts (Nat ~> Nat)
 | zero : Consts (Nat ~> Bool)
 | condN: Consts (Bool ~> Nat ~> Nat ~> Nat)
 | condB: Consts (Bool ~> Bool ~> Bool ~> Bool).
\end{lstlisting}
The set family of terms of \PCF~is given by an inductive family, parametrized by a context
\lstinline!V! and indexed by object types: 
\begin{lstlisting}
Inductive PCF (V: TY -> Type) : TY -> Type:=
 | Bottom: forall t, PCF V t
 | Const : forall t, Consts t -> PCF V t
 | Var : forall t, V t -> PCF V t
 | App : forall t s, PCF V (s ~> t) -> PCF V s -> PCF V t
 | Lam : forall t s, PCF (opt t V) s -> PCF V (t ~> s)
 | Rec : forall t, PCF V (t ~> t) -> PCF V t.
Notation "a @ b" := (App a b)(at level 43, left associativity).
Notation "M '" := (Const _ M) (at level 15).
\end{lstlisting}
Monadic substitution is defined recursively on terms:
\lstset{mathescape=false}
%Reserved Notation "y >>= f" (at level 42, left associativity).
\begin{lstlisting}
Fixpoint subst (V W: TY -> Type)(f: forall t, V t -> PCF W t)
           (t : TY)(v : PCF V t) : PCF W t :=
    match v with
    | Bottom t => Bottom W t
    | c ' => c '
    | Var t v => f t v
    | u @ v => u >>= f @ v >>= f
    | Lam t s u => Lam (u >>= shift f)
    | Rec t u => Rec (u >>= f)
    end
where "y >>= f" := (@subst _ _ f _ y).
\end{lstlisting}

\noindent
Here \lstinline!shift f! is the substitution map \lstinline!f! extended to account for 
an extended context under the binder \lstinline!Lam!.
It is equal to the shifted map of \autoref{def:rel_module_deriv}.

\lstset{mathescape=true}
Finally, we define a relation on the terms of type \lstinline!PCF! via the inductive definition

\begin{form}[Reduction Rules for $\PCF$] \label{code:pcf_eval}
\begin{lstlisting}
Inductive eval (V : IT): forall t, relation (PCF V t) :=
 | app_abs : forall (s t:TY) (M: PCF (opt s V) t) N, 
               eval (Lam M @ N) (M [*:= N])
 | condN_t: forall n m, eval (condN ' @ ttt ' @ n @ m) n 
 | condN_f: forall n m, eval (condN ' @ fff ' @ n @ m) m 
 | condB_t: forall u v, eval (condB ' @ ttt ' @ u @ v) u 
 | condB_f: forall u v, eval (condB ' @ fff ' @ u @ v) v
 | succ_red: forall n, eval (succ ' @ Nats n ') (Nats (S n) ')
 | zero_t: eval ( zero ' @ Nats 0 ') (ttt ')
 | zero_f: forall n, eval (zero ' @ Nats (S n)') (fff ')
 | pred_Succ: forall n, eval (preds ' @ (succ ' @ Nats n ')) (Nats n ')
 | pred_z: eval (preds ' @ Nats 0 ') (Nats 0 ')
 | rec_a : forall t g, eval (Rec g) (g @ (Rec (t:=t) g)).
\end{lstlisting}
\end{form}
\noindent
which we then propagate into subterms (cf.\ \autoref{code:pcf_propag}) and close with respect to transitivity and reflexivity:
\begin{form}[Propagation of Reductions into Subterms]\label{code:pcf_propag}
 \begin{lstlisting}
Reserved Notation "x :> y" (at level 70).
Variable rel : forall (V:IT) t, relation (PCF V t).
Inductive propag (V: IT) : forall t, relation (PCF V t) :=
| relorig : forall t (v v': PCF V t), rel v v' -> v :> v'
| relApp1: forall s t (M M' : PCF V (s ~> t)) N, M :> M' -> M @ N :> M' @ N
| relApp2: forall s t (M : PCF V (s ~> t)) N N', N :> N' -> M @ N :> M @ N'
| relLam: forall s t (M M':PCF (opt s V) t), M :> M' -> Lam M :> Lam M'
| relRec: forall t (M M' : PCF V (t ~> t)), M :> M' -> Rec M :> Rec M'
where "x :> y" := (@propag _ _ x y).
 \end{lstlisting}
\end{form}

\noindent
The data thus defined constitutes a relative monad \lstinline!PCFEM! on the functor $\TDelta{T_\PCF}$ (\lstinline!IDelta TY!). 
We omit the details.

Now we need to define a suitable morphism (resp.\ family of morphisms) of \lstinline!PCFEM!--modules for any arity
(of higher degree).
Let $\alpha$ be any such arity, for instance the arity $\App$. We need to verify two things:
\begin{enumerate}
 \item we show that the constructor of \lstinline!PCF! which corresponds to $\alpha$ is monotone with respect to 
        the order on \lstinline!PCFEM!. For instance, we show that for any two terms \lstinline!r s:TY! and any 
          \lstinline!V : IDelta TY!, 
         the function 
         \begin{lstlisting}
fun y => App (fst y) (snd y): PCFEM V (r~>s) x PCFEM V r -> PCFEM V s 
         \end{lstlisting}
is monotone.
 \item We show that the monadic substitution defined above distributes over the constructor in the sense of \autoref{ex:ulc_mod_mor_kl},
    i.e.\ we prove that the constructor is the carrier of a \emph{module} morphism.
\end{enumerate}

\noindent
All of these are very straightforward proofs, resulting in a representation \lstinline!PCFE_rep! of semantic \PCF:
\begin{lstlisting}
Program Instance PCFE_rep_struct :
       PCFPO_rep_struct PCFEM PCF.arrow PCF.Bool PCF.Nat := {
  app r s := PCFApp r s;
  abs r s := PCFAbs r s;
  rec t := PCFRec t ;
  tttt := PCFconsts ttt ;
  ffff := PCFconsts fff;
  Succ := PCFconsts succ;
  Pred := PCFconsts preds;
  CondN := PCFconsts condN;
  CondB := PCFconsts condB;
  Zero := PCFconsts zero ;
  nats m := PCFconsts (Nats m);
  bottom t := PCFbottom t }.
Definition PCFE_rep : PCFPO_rep := Build_PCFPO_rep  PCFE_rep_struct.
\end{lstlisting}
Note that in the instance declaration \lstinline!PCFE_rep_struct!, the \lstinline!Program! framework proves 
automatically the properties of \autoref{code:rpcf_beta}, \ref{code:rpcf_rec}, \ref{code:rpcf_cond} and
 \ref{code:rpcf_arith}.

\section{Initiality}\label{sec:comp_sem_formal_initiality}

In this section we define a morphism of representations from \lstinline!PCFE_rep! to any representation \lstinline!R : PCFPO_rep!.
At first we need to define a map between the underlying sorts, that is, a map \lstinline!Sorts PCFE_rep -> Sorts R!. 
In short, each \PCF~type goes to its representation in \lstinline!R!:
\begin{lstlisting}
Fixpoint Init_Sorts_map (t : Sorts PCFE_rep) : Sorts R := 
    match t with
    | PCF.Nat => Nat R 
    | PCF.Bool => Bool R
    | u ~> v => (Init_Sorts_map u) ~~> (Init_Sorts_map v)
    end.
\end{lstlisting}

\noindent
The function \lstinline!init! is the carrier of what will later be proved to be the initial morphism to 
the representation \lstinline!R!. It maps each constructor of \PCF~recursively to its counterpart in the
representation \lstinline!R!:

\begin{lstlisting}
Fixpoint init V t (v : PCF V t) :
    R (retype (fun t0 => Init_Sorts_map t0) V) (Init_Sorts_map t) :=
  match v with
  | Var t v => rweta R _ _ (ctype _ v)
  | u @ v => app _ _ _ (init u, init v)
  | Lam _ _ v => abs _ _ _ (rlift R
         (@der_comm TY (Sorts R) (fun t => Init_Sorts_map t) _ V ) _ (init v))
  | Rec _ v => rec _ _ (init v)
  | Bottom _ => bottom _ _ tt
  | y ' => match y in Consts t1 return
              R (retype (fun t2 => Init_Sorts_map t2) V) (Init_Sorts_map t1) with
                 | Nats m => nats m _ tt
                 | succ => Succ _ tt
                 | condN => CondN _ tt
                 | condB => CondB _ tt
                 | zero => Zero _ tt
                 | ttt => tttt _ tt
                 | fff => ffff _ tt
                 | preds => Pred _ tt
                 end
  end.
\end{lstlisting}

\noindent
We write $i_V$ for \lstinline!init V! and $g$ for \lstinline!Init_Sorts_map!. 
Note that $i_V : \PCF(V)\to g^*\left(R(\retyping{g}V)\right)$ really is \emph{the image of the initial morphism under
the adjunction $\varphi$} of \autoref{rem:retyping_adjunction_kan}.
Intuitively, passing from \lstinline!init V!$= i_V$ to its adjunct $\varphi^{-1}(i_V)$ is done by 
precomposing with pattern matching on the constructor \lstinline!ctype! (cf.\ \autoref{rem:about_coproduct_matching}).
We informally denote $\varphi^{-1}(i_V)$ by $\comp{\text{\lstinline!match!}}{\text{\lstinline!init V!}}$.

The map \lstinline!init! is compatible with renaming and substitution in \lstinline!PCF! and
\lstinline!R!, respectively, in a sense made precise by the following two lemmas.
The first lemma states that, for any morphism $f : V \to W$ in $\TS{T_\PCF}$, the  following
diagram commutes:
\[
 \begin{xy}
  \xymatrix @C=3.5pc{
        **[l]\PCF(V) \ar[d]_{\text{\lstinline!init V!}} \ar[r]^{\PCF(f)}    &  **[r]\PCF(W) \ar[d]^{\text{\lstinline!init W!}} \\
        **[l]g^*R(\retyping{g}V) \ar[r]_{g^*R(g^*f)}& **[r] g^*R(\retyping{g}W).
}
 \end{xy}
\]

\begin{lstlisting}
Lemma init_lift (V : IT) t (y : PCF V t) W (f : V ---> W) : 
   init (y //- f) = rlift R (retype_map f) _ (init y).
\end{lstlisting}

\noindent
The next commutative diagram concerns substitution; for any $f : V \to \PCF(W)$, 
the diagram obtained by applying $\varphi$ to the diagram given in \autoref{eq:comp_sem_monad_mon_mor_diag}
 --- i.e.\ the diagram corresponding to \autoref{eq:p_chap5_2} ---, commutes:

\[
 \begin{xy}
  \xymatrix @C=8pc {
        **[l]\PCF(V) \ar[d]_{\text{\lstinline!init V!}} \ar[r]^{\kl[\PCF]{f}}    &  **[r]\PCF(W) \ar[d]^{\text{\lstinline!init W!}} \\
        **[l] g^*R(\tilde{V}) \ar[r]_{g^*\kl[R]{\comp{(g^*f)}{\varphi^{-1}(\text{\lstinline!init W!})}}}& **[r] g^*R(\tilde{W}).
}
 \end{xy}
\]
In \textsf{Coq} the lemma \lstinline!init_subst! proves commutativity of this latter diagram:
% \begin{form}%\label{code:colax_mon_mor_subst_variant}
\begin{lstlisting}
Lemma init_subst V t (y : PCF V t) W (f : IDelta _ V ---> PCFE W):
  init (y >>= f) =
    rkleisli (RMonad_struct := R)
        (SM_ind (V:= retype (fun t => _ t) V)
                (W:= R (retype (fun t => _ t) W))
                (fun t v => match v with ctype t p => init (f t p) end)) 
         _ (init y).
\end{lstlisting}
% \end{form}

\noindent
This latter lemma establishes almost the commutative diagram for the family $\varphi^{-1}(i_V)$  to constitute a 
(colax) \emph{monad} morphism, which reads as follows:

\begin{equation}\label{eq:init_subst_monadic}
 \begin{xy}
  \xymatrix @C=5.5pc{
       **[l] \retyping{g}\left(\PCF(V)\right) \ar[d]_{\comp{\text{ \lstinline!match!}}{\text{\lstinline!init V! }}} \ar[r]^{\retyping{g}\left(\kl[\PCF]{f}\right)}    
                           & **[r]\retyping{g}\left({\PCF(W)}\right) \ar[d]^{\comp{\text{ \lstinline!match!}}{\text{\lstinline!init W! }}} \\
        **[l]R(\retyping{g}V) \ar[r]_{\kl[R]{\comp{\comp{(\retyping{g}f)}{\text{ \lstinline!match! }}}{\text{\lstinline!init! }}}}& **[r]R(\retyping{g}{W}) .
}
 \end{xy}
\end{equation}

%  draw this diagram, which reads 
%  \[ \overline{\bind{}{f}} ;; \comp{match}{i} == \comp{match}{i} ;; \bind{}{\comp{\overline{f}}{\comp{match}{i}}}\]

\noindent
Before we can actually build a monad morphism with carrier map 
  $\comp{\text{\lstinline!match!}}{\text{\lstinline!init V!}}$, we need to verify that \lstinline!init! --- and thus its adjunct --- is monotone.
We do this in 3 steps, corresponding to the 3 steps in which we built up the preorder on the terms
of \PCF:
\begin{packenum}
 \item \lstinline!init! monotone with respect to the relation \lstinline!eval! (cf.\ \autoref{code:pcf_eval}):
\begin{lstlisting}
Lemma init_eval V t (v v' : PCF V t) : eval v v' -> init v <<< init v'.
\end{lstlisting}
 \item \lstinline!init! monotone with respect to the propagation into subterms of \lstinline!eval!;
\begin{lstlisting}
Lemma init_eval_star V t (y z : PCF V t) : eval_star y z -> init y <<< init z.
\end{lstlisting}
 \item \lstinline!init! monotone with respect to reflexive and transitive closure of above relation.
\begin{lstlisting}
Lemma init_mono c t (y z : PCFE c t) : y <<< z -> init y <<< init z.
\end{lstlisting}
\end{packenum}

\noindent
We now have all the ingredients to define the initial morphism from \PCF~to \lstinline!R!.
As already indicated by the diagram \autoref{eq:init_subst_monadic}, its carrier 
is not given by just the map \lstinline!init!, since this map does not have the right type:
its domain is given, for any context $V\in \TS{T_\PCF}$, by $\PCF(V)$ and not, as needed, by $\retyping{g}\left(\PCF(V)\right)$.
We thus precompose with pattern matching in order to pass to its adjunct: 
for any context $V$, the carrier of the initial morphism is given by

% \begin{form}%\label{code:pattern_match_retyping}
\begin{lstlisting}
fun t y => match y with
     | ctype _ p => init p
     end
 : retype _ (PCF V) ---> R (retype _ W)
\end{lstlisting}
% \end{form}

\noindent
We recall that the constructor \lstinline!ctype! is the carrier of the natural transformation
of the same name of \autoref{def:retyping_functor}, and that precomposing with pattern matching 
corresponds to specifying maps on a coproduct via its universal property.

\noindent
Putting the pieces together, we obtain a morphism of representions of semantic \PCF:
\begin{lstlisting}
Definition initR : PCFPO_rep_Hom PCFE_rep R :=
        Build_PCFPO_rep_Hom initR_s.
\end{lstlisting}

\noindent
Uniqueness is proved in the following lemma:
\begin{lstlisting}
Lemma initR_unique : forall g : PCFE_rep ---> R, g == initR.
\end{lstlisting}
The proof consists of two steps: first, one has to show that the translation of \emph{sorts}
coincide. Since the source of this translation is an inductive type --- the initial representation of
 the signature of \autoref{ex:type_PCF} --- this proof is done by induction.
Afterwards the translations of terms are proved to be equal. The proof is done by induction on terms of \PCF. 
It makes essentially use of the commutative diagrams (cf.\ \autoref{def:rel_mor_of_reps_typed}) which we exemplarily presented for 
the arities of successor 
(\autoref{code:pcf_succ_diag}), application (\autoref{code:pcf_app_diag}) and abstraction (\autoref{code:pcf_abs_diag}).
Finally we can declare an instance of \lstinline!Initial! for the category \lstinline!REP! of representations:
\begin{lstlisting}
Instance PCF_initial : Initial REP := {
  Init := PCFE_rep ;
  InitMor R := initR R ;
  InitMorUnique R := @initR_unique R }.
\end{lstlisting}
Checking the axioms used for the proof of initiality (and its dependencies) yields the use of 
non--dependent functional extensionality (applied to the translations of sorts) and 
uniqueness of identity proofs, which in the \textsf{Coq} standard library is implemented as a consequence of 
another --- logically equivalent --- axiom \lstinline!eq_rect_eq!:
\begin{lstlisting}
Print Assumptions PCF_initial.
Axioms:
CatSem.AXIOMS.functional_extensionality.functional_extensionality : 
    forall (A B : Type) (f g : A -> B),
                            (forall x : A, f x = g x) -> f = g
Eq_rect_eq.eq_rect_eq : forall (U : Type) (p : U) (Q : U -> Type) 
                          (x : Q p) (h : p = p), x = eq_rect p Q x p h
\end{lstlisting}

\section{A Representation of \texorpdfstring{\PCF}{PCF} in the Untyped Lambda Calculus}

We use the iteration principle explained in \autoref{rem:comp_sem_iteration} in order to 
specify a translation from $\PCF$ to the untyped lambda calculus which is compatible with
reduction in the source and target.
According to the principle, it
is sufficient to define a representation of $\PCF$ in the relative 
monad of the lambda calculus 
(cf.\ \autorefs{ex:ulc_def} and \ref{ex:ulcbeta})
and to verify that this representation satisfies the inequations of \autoref{eq:pcf_reductions},
formalized in the \textsf{Coq} code snippets \ref{code:rpcf_beta}, \ref{code:rpcf_rec}, \ref{code:rpcf_cond} 
and \ref{code:rpcf_arith}.
The first task, specifying a representation of the types of $\PCF$, in the singleton set of types of $\LC$,
is trivial. We furthermore specify representations of the term arities of $\PCF$, presented in 
\autoref{code:rpcf_1-sig}, by giving an instance of the corresponding type class.

\begin{lstlisting}
Program Instance PCF_ULC_rep_s :
 PCFPO_rep_struct (Sorts:=unit) ULCBETAM (fun _ _ => tt) tt tt := {
  app r s := ulc_app r s;
  abs r s := ulc_abs r s;
  rec t := ulc_rec t ;
  tttt := ulc_ttt ;
  ffff := ulc_fff ;
  nats m := ulc_N m ;
  Succ := ulc_succ ;
  CondB := ulc_condb ;
  CondN := ulc_condn ;
  bottom t := ulc_bottom t ;
  Zero := ulc_zero ;
  Pred := ulc_pred }.
\end{lstlisting}

\noindent
Before taking a closer look at the module morphisms we specify in order to represent the arities of $\PCF$,
we note that in the above instance declaration, we have not given the proofs corresponding to code snippets
\ref{code:rpcf_beta} to \ref{code:rpcf_arith}. 
In the terms of \autoref{rem:comp_sem_iteration}, we have not completed the third task, the verification that
the given representation satisfies the inequations. 
The \lstinline!Program! feature we use during the above instance declaration is able to detect that the fields
called \lstinline!beta_red!, \lstinline!rec_A!, etc., are missing, and enters into interactive proof mode to allow us
to fill in each of the missing fields.

We now take a look at some of the lambda terms representing arities of \PCF.
The carrier of the representations \lstinline!ulc_app! is the application of lambda calculus, of course,
 and similar for \lstinline!ulc_abs!. Here the parameters \lstinline!r! and \lstinline!s! vary over
  terms of type \lstinline!unit!, the type of sorts underlying this representation.
We use an infixed application and a de Bruijn notation instead of the more abstract notation of nested data types:
\begin{lstlisting}
Notation "a @ b" := (App a b) (at level 42, left associativity).
Notation "'1'" := (Var None) (at level 33).
Notation "'2'" := (Var (Some None)) (at level 24). 
\end{lstlisting}
 
\noindent
The truth values \True~and \False~are represented by
\begin{lstlisting}
Eval compute in ULC_True. 
    = Abs (Abs 2)
Eval compute in ULC_False. 
    = Abs (Abs 1)
\end{lstlisting}

\noindent
Natural numbers are given in Church style, the successor function is given by the term $\lambda nfx. f(n~f~x)$.
The predecessor is represented by the constant 
\[\lambda nfx.n~(\lambda gh.h (g~f)) (\lambda u.x) (\lambda u.u), \]
and the test for zero is represented by $\lambda n. n (\lambda x.F) T$, where $F$ and $T$ are the lambda terms representing \False~and \True, respectively.
\begin{lstlisting}
Eval compute in ULC_Nat 0.
    = Abs (Abs 1)
Eval compute in ULC_Nat 2.
    = Abs (Abs (2 @ (Abs (Abs (2 @ (Abs (Abs 1) @ 2 @ 1))) @ 2 @ 1)))
Eval compute in ULCsucc.
    = Abs (Abs (Abs (2 @ (3 @ 2 @ 1))))
Eval compute in ULC_pred.
    = Abs (Abs (Abs (3 @ Abs (Abs (1 @ (2 @ 4))) @ Abs 2 @ Abs 1)))
Eval compute in ULC_zero.
    = Abs (1 @ Abs (Abs (Abs 1)) @ Abs (Abs 2))
\end{lstlisting}

\noindent
The conditional is represented by the lambda term $\lambda p a b. p~a~b$:
\begin{lstlisting}
Eval compute in ULC_cond.
    = Abs (Abs (Abs (3 @ 2 @ 1)))
\end{lstlisting}

\noindent
The constant arity $\bot_A$ is represented by $\Omega$:
\begin{lstlisting}
Eval compute in ULC_omega.
    = Abs (1 @ 1) @ Abs (1 @ 1)
\end{lstlisting}

\noindent
The fixed point operator \PCFFix~(\lstinline!rec!) is represented by the \emph{Turing} fixed--point combinator, that is, the lambda term
\begin{lstlisting}
Eval compute in ULC_theta.
    = Abs (Abs (1 @ (2 @ 2 @ 1))) @ Abs (Abs (1 @ (2 @ 2 @ 1)))
\end{lstlisting}

\noindent
The reason why we use the Turing operator instead of, say, the combinator $\mathbf{Y}$,
\begin{lstlisting}
Eval compute in ULC_Y.
    = Abs (Abs (2 @ (1 @ 1)) @ Abs (2 @ (1 @ 1)))
\end{lstlisting}
is that the latter does not have a property that is crucial for us:
It is 
\[  \Theta(f) \rightsquigarrow^* f\left(\Theta(f)\right) \]
but only
\[  \mathbf{Y}(f) \stackrel{*}{\leftrightsquigarrow} f\left(\mathbf{Y}(f)\right) \]
via a common reduct. 
Thus if we would attempt to represent the arity \lstinline!rec! by the fixed--point combinator $\mathbf{Y}$,
we would not be able to prove the condition expressed in \autoref{code:rpcf_rec}.
A way to allow for the use of $\mathbf{Y}$ as representation of \lstinline!rec! would by
to consider \emph{symmetric} relations on terms, e.g., relative monads into a category of setoids.

As a final remark, we emphasize that while reduction is given as a relation in our formalization, 
and as such is not computable, the obtained translation from \PCF~to the untyped lambda calculus
is executable in \textsf{Coq}.
For instance, we can translate the \PCF~term negating boolean terms as follows:
\begin{form}\label{code:translation_pcf_ulc_example}
\noindent
\begin{lstlisting}
Eval compute in 
  (PCF_ULC_c ((fun t => False)) tt (ctype _        
   (Lam (condB ' @@ x_bool @@ fff ' @@ ttt ')))).
   = Abs (Abs (Abs (Abs (3 @ 2 @ 1))) @ 1 @ Abs (Abs 1) @ Abs (Abs 2))
\end{lstlisting}
\end{form}
\noindent
Here we use infixed ``\lstinline!@@!'' to denote application of \PCF, and \lstinline!x_bool! is
simply a notation for a de Bruijn variable of type \lstinline!Bool! of the lowest level, i.e.\ a variable
that is bound by the \lstinline!Lam! binder of \PCF~in above term.

%% file: conclusion.tex
We summarize the contributions of this thesis and discuss further work.

\section{Contributions}

We have proved an initiality result for \emph{simply--typed syntax} equipped with 
\emph{reduction rules}.
The category--theoretic iteration principle obtained through the universal property of initiality
is sufficiently general to allow for the specification of translations from the term representation to
languages typed over \emph{different} sets of sorts.

We have characterized binding syntax with a reduction relation --- for instance the lambda calculus with beta reduction 
--- as a \emph{relative monad}
over the functor $\Delta$ (cf.\ \autoref{ex:ulcbeta}), encoding not only
commutativity properties of substitution, 
but also its \emph{monotonicity} in the first--order argument.
By a suitable strengthening of the definition of relative monad in a 2--categorical context, an additional
monotonicity property for the higher--order argument of substitution can be assured, cf.\ \autoref{rem:about_substitution}.
We have also carried the definition of \emph{module over a monad} and 
several constructions of modules over to 
modules over \emph{relative} monads.

We then have proved several theorems in the proof assistant \textsf{Coq}:
firstly, we implemented Zsid\'o's initiality theorem \cite[Chap.~6]{ju_phd}, 
summed up in this work as a reference in \autoref{sec:sts_ju}.
Secondly, we have proved the initiality theorem of \autoref{sec:prop_arities}, yielding a tool, which, 
when fed with a 2--signature $(S,A)$, provides the syntax associated to $S$ equipped with the reduction relation generated by 
the inequations of $A$.
Thirdly, we have proved an instance of our main theorem, \autoref{thm:init_w_ineq_typed} of \autoref{chap:comp_types_sem}, for the particular 
2--signature of the programming language \PCF~equipped with reduction rules as in \autoref{eq:pcf_reductions}.
The representation of the signature of \PCF~in the monad of the untyped lambda calculus with beta reduction results in 
an executable translation from \PCF~to $\LC$ which is certified to be compatible with substitution and reduction in the source and target languages.

\section{Further Work}\label{sec:further_work}

In the future, we hope to prove and implement initiality theorems for richer type systems. In particular, \emph{dependent} types and
\emph{polymorphism}, two important steps towards certified programs and code reusability, respectively, should be accounted for.

Furthermore, the modelling of semantics should be improved to allow reasoning about important properties such as \emph{termination}.

As mentioned before, the implementation of initiality results in a proof assistant may serve as a framework for 
research about programming languages and logics. For this reason we envisage the implementation in a proof assistant of \autoref{thm:init_w_ineq_typed}
in its full generality.

% Another very useful feature would be the automatic generation 
% in the computer implementation of a \emph{reduction function}, i.e.\ the choice of a reduction strategy and automatic validation
% against the semantic modeling (in form of preorders or similar) of the initial model.

We present these points in detail:

\begin{description}

\item [Fine--grained modelling of reduction]

For a given 2--signature (a signature together with a set of inequations), models of this 2--signature so far 
were basically functors which associate, to any set ``of variables'', a preordered set ---
 intuitively a model of ``terms'' over the set of variables\footnote{%
We ignore the typed case for the moment, which is analogous.}.
The preorder $\leq$ on such a model corresponds to the reduction relation on the term model, i.e.\ the 
``term'' $t$ reduces to $t'$ if and only if $t \leq t'$.

The modelling of reductions via \emph{preorders} may be considered too coarse in several aspects:
\begin{itemize}
 \item different reductions might lead from one term to another. However, the use of preorders to model reduction
                   does not allow to distinguish two reductions with the same source and target.
 \item The hard--coded reflexivity rule makes reasoning about \emph{normalization} --- in particular \emph{termination} --- difficult.
\end{itemize}

Instead of considering \emph{preordered} sets (indexed by sets of free variables) as models of a 2--signature, 
it would thus be interesting to consider a structure 
which allows for more fine--grained treatment of reduction, such as 
 graphs or categories.
In other words, we might build models of 2--signatures from relative monads into the category of \emph{graphs} or (small) \emph{categories}.
Using this new definition of model, one might then envisage to prove an initiality theorem analogous to the one already proven,
and to use the additional structure obtained by switching to graphs or categories to reason about 
the aforementioned properties.

\item [Inequations, Syntactically]
  Fiore and Hur \cite{DBLP:conf/csl/FioreH10} develop a syntactic theory of equations over a 
  higher--order signature, allowing for 
  proofs of soundness and completeness with respect to the models of the signature and the equations.
  Similar techniques should allow for a syntactic presentation of our \emph{in}equations.
  Apart from the obvious goal of soundness and completeness, such a syntactic presentation would also 
  facilitate the specification of reductions in the computer implementation in \textsf{Coq}:
  in particular, it would make it possible to specify reductions without any knowledge about 
  category--theoretic concepts.
  
   A minimal goal would be to have a data type --- dependent on a 1--signature --- which 
  allows to specify the usual half--equations, mainly obtained from substitution and 
  from composition of arities, e.g., $\comp{(\abs\times\id)}{\app}$.
  To a term of this data type, on could associate a family of morphisms of modules which
  constitutes the carrier of a half--equation: the algebraic properties (being a 
  morphism of modules, which corresponds to the compatibility of substitution with
  meta--substitution in \cite{DBLP:conf/csl/FioreH10}, could be proved once and
  for all by induction.

\item [More sophisticated type systems]

New programming languages tend to be equipped with more and more sophisticated type systems: 
\emph{dependent types} allow to ensure properties of function output and thus 
secure plugging together of functions. 
\emph{Polymorphism} allows for the reuse of code in various situations.
An algebraic characterization of such sophisticated type systems with variable binding 
via a universal property is still missing.
We hope to extend initiality results to encompass these type systems.

\item [A wider class of arities]

 The present initiality theorems encompass arities, i.e.\ term constructors, of quite simple
 nature: the only operations considered are product --- for constructors with \emph{multiple} arguments --- 
 and context extension, for modelling variable binding.

 It would be desirable to consider more general term formers. 
Hirschowitz and Maggesi \cite{hirschowitz_maggesi_fics2012} 
have introduced a notion of strengthened arity which allows, for instance, to treat
 a term former of explicit flattening
  $\mu : \comp{T}{T} \to T$.
 Ultimately, we hope to find a very general \emph{simple} criterion %--- similar to the \emph{strict positivity} criterion --- 
for arities and signatures for which an initial model can be provided.

\item [A certified research tool] % research for programming languages and logics]

%  Finally, 
The obtained results should --- as we have already done for untyped syntax with reductions --- 
be implemented in a theorem prover such as \textsf{Coq}. In this way,
an initiality theorem may be used as a practical tool for easily experimenting with different languages.
Changing a language would be done by simply changing its specifying signature, whereas all necessary data and 
properties such as certified substitution and iteration, but also reductions, would be provided by the system.
For this computer implementation and suitable reduction rules, it would also be desirable to obtain automatically a reduction \emph{function} $r$ in addition 
to the reduction relation. This reduction function might be validated against the relation
 in the sense that one may prove that for any term $t$, one has $t \leq r(t)$.

\end{description}

%% file: appendix.tex
\appendix

\input{syntax_pcf_ulc}

%% file: syntax_pcf_ulc.tex
\chapter{Syntax and Semantics of Lambda Calculus and \texorpdfstring{\PCF}{PCF}} \label{chap:syntax_semantics_pcf_ulc}

The following section informally introduces the syntax and semantics of $\PCF$ and $\LC$, as it might be introduced
in some computer science textbook. 
Our presentation of the lambda calculus is inspired by Barendregt and Barendsen's course \cite{barendregt_barendsen}, 
and that of $\PCF$ by Hyland and Ong's paper \cite{Hyland00onfull}.

\section{Syntax of Lambda Calculus and \texorpdfstring{\PCF}{PCF}} \label{subsec:intro_pcf_lc_syntax}

Let $V$ be a countably infinite set (of variables).
The syntax of $\LC$ is given by
\[ \Lambda \enspace ::= \quad v \mid \Lambda@\Lambda \mid \lambda v. \Lambda \enspace ,\]
where $v\in V$ varies over variables.

The programming language \PCF~is a \emph{typed} language, more precisely a \emph{simply--typed} language.
It is given by 
 \begin{packitem}
  \item a set of \emph{sorts},
  \item a set of \emph{terms} and
  \item a \emph{typing} map associating a sort to any term.
 \end{packitem}

We take the presentation of \PCF~from Hyland and Ong's paper on full abstraction \cite{Hyland00onfull}.
The sorts of \PCF~are constructed from two base sorts and a function type constructor: 
 \[ T_{\PCF} \enspace ::= \enspace  \Nat \mid \Bool \mid T_\PCF \PCFar T_\PCF \enspace . \]

The \emph{terms} of \PCF~are defined in two steps: at first, we define a set of \emph{raw} terms,
  which actually contains more elements than we want. Afterwards, we define a \emph{welltypedness}
  predicate on those raw terms. The terms of \PCF~then are the well--typed raw terms.
The raw terms of $\PCF$ are given by the grammar of Fig.\ \ref{fig:pcf_grammar}.

\begin{figure}[tbh]
\centering
\fbox{%
 \begin{minipage}{8cm}

\centering

\vspace{1ex}

\begin{tabular}{l c l l}
     $s$ & ::= & $\bot_A$ & undefined \\
                   & |  & $c_A$ & constant \\
                  &  | & $x_A$ & variable \\
                  & | &  $s@s$ & application\\
                  & | & $\lambda x:A.s$ & abstraction \\
                  & | & $\PCFFix_A (s)$ & fixed point operator
\end{tabular}

\vspace{1ex}

\end{minipage}
}
% \vspace{-2em}
  \caption{Grammar of \PCF}\label{fig:pcf_grammar}
% \end{minipage}
% }
\end{figure}

Note that we use the same infix notation $\_@\_$ for application in \PCF~and $\ULC$. We also write $f(x)$ for $f@x$ when no 
 confusion can arise.
The constants $c_A$ of sort $A$ are the basic constants from logic and arithmetic, i.e. booleans \True~and \False, natural numbers
$n$, successor and predecessor as well as test for zero, and conditionals. They are listed in \autoref{fig:pcf_constants}.
\begin{figure}[hbt]
 \centering

\fbox{%
 \begin{minipage}{11cm}

\vspace{2ex}

\centering

\begin{tabular}{r c l l}
     $n$ & : & $\Nat$ & naturals (for $n\in \NN$)\\
   $\True, \False$  & :  & $\Bool$ & boolean constants \\
   %$\mathbf{Z}$ & : &  $\Nat$ & zero\\     
    $\mathbf{S}$ &  : & $\Nat \PCFar \Nat$ & successor \\
    $\pred$ & : & $\Nat \PCFar \Nat $  & predecessor \\
    \zeroqu& : & $\Nat \PCFar \Bool$ & test on zero \\
   \condn & : & $\Bool \PCFar \Nat \PCFar \Nat \PCFar\Nat$& conditional for naturals \\
   \condb & : & $\Bool \PCFar \Bool \PCFar \Bool \PCFar\Bool$& conditional for booleans 
\end{tabular}

\vspace{2ex}

\end{minipage}
}

 \caption{Constants of \PCF} \label{fig:pcf_constants}
\end{figure}

Instead of all raw terms from the definition of Fig.\ \ref{fig:pcf_grammar} we only consider \emph{well--typed} terms, that is,
 those raw  terms that are typable according to the typing judgements of \autoref{fig:pcf_typing}.

\begin{figure}[htb]

 \centering

\fbox{%
 \begin{minipage}{13cm}

\centering

{
\renewcommand{\arraystretch}{1.8}
\begin{tabular}{c c c}
 $c_A : A$   & $\bot_A : A$ &  \\   
\AxiomC{$M:A\PCFar A$}\UnaryInfC{$\PCFFix_A(M) : A$}\DisplayProof  
                                 & \AxiomC{$M:A_2$}\UnaryInfC{$\lambda x: A_1.M : A_1 \PCFar A_2$}\DisplayProof 
                                &\AxiomC{$M:A_1\PCFar A_2$} \AxiomC{$N:A_1$}\BinaryInfC{$M@N : A_2$}\DisplayProof
\end{tabular}
}

\vspace{2ex}
\end{minipage}
}

\caption{Typing rules of \PCF}\label{fig:pcf_typing}
\end{figure}

\section{\texorpdfstring{Semantics of Lambda Calculus and \PCF}{Semantics of Lambda Calculus and PCF}}\label{subsec:semantic_ulc_pcf}

Functional programming languages such as $\PCF$ and $\LC$ allow for \emph{computation}
by \emph{reduction}, as explained in \autoref{subsec:adding_semantics}.
The prime example of \emph{reduction rule} is the \emph{beta rule} of $\LC$,
\begin{equation}     (\lambda x.M)N \enspace \rightsquigarrow_{\beta} \enspace M[x:=N] \enspace , \label{eq:beta} %\tag{$(\beta)$}
\end{equation}
where $M[x:=N]$ denotes the term $M$ where free occurrences of the variable $x$ have been replaced by $N$ in a capture--avoiding 
manner.

% The semantics of functional programming languages is given by \emph{reductions}.
% A prime example for reduction is the \emph{beta} reduction, expressed by 
% \[ (\lambda M)N \rightsquigarrow M[*:=N] \enspace . \tag{$\beta$}\]
The above rule may be considered to ``generate'' beta reduction in the sense that we also consider 
\begin{packenum}
 \item reductions in \emph{sub}terms such as in $ \lambda x.(\lambda y. M)N$ and
 \item \emph{chains of reductions}, that is, reductions consisting of multiple steps.
\end{packenum}
Thus, to be more precise, what is usually called ``beta reduction'', is in fact the closure of the 
relation specified by the rule given in \autoref{eq:beta}
under propagation into subterms as well as transitivity and reflexivity, denoted by $\twoheadrightarrow_{\beta}$ in 
Barendregt and Barendsen's course \cite{barendregt_barendsen}.
In general we associate three different relations to any set of reduction rules, see \autoref{subsec:adding_semantics}.

% We consider the semantics of a programming language to be given by a \emph{reduction relation} on the terms
% of the language. For instance, beta reduction is the closure under aforementioned operations of the relation
% generated by \autoref{eq:beta}.

Reduction in $\PCF$ is given by a beta rule similar to \autoref{eq:beta} and
several additional reduction rules concerning the fixed point operator and the 
logical and arithmetic constants. We list them using a small--step semantics as given in
\cite{Hyland00onfull} or in Pitts' lecture notes on denotational semantics \cite{pitts_dens}.
Analogously to the lambda calculus with beta reduction, we denote by ``$\twoheadrightarrow_{\PCF}$'' the reduction relation obtained as closure under 
propagation into subterms as well as reflexivity and transitivity.
\begin{figure}[hbt]

\centering

\fbox{%
 \begin{minipage}{8cm}
\centering
\begin{align}
       \lambda x : A. M (N) & \rightsquigarrow M [x:= N] \notag\\
       \PCFFix(g) &\rightsquigarrow g(\PCFFix(g)) \notag\\
       \mathbf{S} (n) &\rightsquigarrow n + 1 \notag\\
       \text{pred}(0) &\rightsquigarrow 0 \notag\\
       \text{pred}(\textbf{S}(n)) &\rightsquigarrow n \notag\\ %\label{eq:pcf_reductions}\\
       \zeroqu(0) &\rightsquigarrow \True \notag\\
       \zeroqu(\mathbf{S}(n)) &\rightsquigarrow \False \notag\\
       \PCFcond{\sigma}(\True)(M)(N) &\rightsquigarrow M \quad (\sigma\in\{\Bool,\Nat\})\notag\\
       \PCFcond{\sigma}(\False)(M)(N) &\rightsquigarrow N \quad (\sigma\in\{\Bool,\Nat\})\notag \\
               \notag
\end{align}
\end{minipage}
}

% \vspace{-2em}
\caption{Reduction rules of \PCF}\label{eq:pcf_reductions}
\end{figure}